\documentclass[aps,pra,onecolumn,superscriptaddress,nofootinbib,longbibliography]{revtex4-2}

\usepackage[a4paper,left=0.9in,right=0.9in,top=0.85in,bottom=0.9in]{geometry}
\usepackage{cmap}
\usepackage[T1]{fontenc}
\usepackage[utf8]{inputenc}
\usepackage{microtype}
\usepackage{amsmath,amssymb,amsthm,mathtools}
\usepackage{booktabs}
\usepackage{xcolor}
\usepackage{soul}
\usepackage{graphicx}
\usepackage{tikz}
\usetikzlibrary{arrows.meta,positioning,fit,calc,backgrounds,decorations.pathmorphing,shapes.geometric}
\usepackage[ruled,vlined,linesnumbered]{algorithm2e}
\DontPrintSemicolon
\SetAlgoCaptionSeparator{:}
\SetAlCapNameFnt{\small\bfseries}
\SetAlCapFnt{\small}
\SetAlgoNlRelativeSize{-2}
\SetAlFnt{\normalsize}
\SetAlgoInsideSkip{smallskip}

\usepackage[colorlinks=true,linkcolor=blue!45!black,citecolor=green!35!black,urlcolor=red!45!black]{hyperref}
\usepackage[nameinlink,noabbrev]{cleveref}

\crefname{algocf}{Algorithm}{Algorithms}
\Crefname{algocf}{Algorithm}{Algorithms}

\usepackage{mathrsfs} 
\usepackage{enumitem}

\newtheorem{theorem}{Theorem}
\newtheorem{proposition}{Proposition}
\newtheorem{lemma}{Lemma}
\newtheorem{corollary}{Corollary}

\newtheorem{problem}{Problem}

\newcommand{\ket}[1]{\lvert #1\rangle}
\newcommand{\bra}[1]{\langle #1\rvert}
\newcommand{\braket}[2]{\langle #1\mid #2\rangle}
\newcommand{\Tr}{\operatorname{Tr}}

\newcommand{\vac}{\mathrm{vac}}

\newcommand{\eps}{\varepsilon}

\newcommand{\set}[1]{\left\{#1\right\}}
\newcommand{\anticomm}[2]{\left\{#1,#2\right\}}
\newcommand{\abs}[1]{\left|#1\right|}
\newcommand{\norm}[1]{\left\lVert#1\right\rVert}

\definecolor{inputblue}{HTML}{DCEBFA}
\definecolor{blockblue}{HTML}{4C78A8}
\definecolor{blockorange}{HTML}{F28E2B}
\definecolor{blockgreen}{HTML}{59A14F}
\definecolor{softgray}{HTML}{F4F5F7}
\definecolor{linegray}{HTML}{69707A}
\definecolor{nearcol}{HTML}{EAF2FB}
\definecolor{strongcol}{HTML}{FCECDD}
\definecolor{classicalcol}{HTML}{E9F5E8}

\newcommand{\MJ}[1]{\textcolor{purple}{[MJ: #1]}}

\allowdisplaybreaks

\begin{document}

\makeatletter
\addtocontents{toc}{\protect\let\protect\savedcontentsline\protect\contentsline}
\addtocontents{toc}{\protect\let\protect\contentsline\protect\@gobblefour}
\makeatother

\title{Efficient Learning of Fermionic Magic States under Free-Fermion Evolution}

\author{Jiwon Heo}
\thanks{These authors contributed equally to this work.}
\affiliation{Graduate School of Quantum Science and Technology, Korea Advanced Institute of Science and Technology, Daejeon 34141, Korea}

\author{Myeongjin Shin}
\thanks{These authors contributed equally to this work.}
\affiliation{School of Computational Sciences, Korea Advanced Institute of Science and Technology, Daejeon 34141, Korea}

\author{Changhun Oh}
\email{changhun0218@gmail.com}
\affiliation{Graduate School of Quantum Science and Technology, Korea Advanced Institute of Science and Technology, Daejeon 34141, Korea}
\affiliation{Department of Physics, Korea Advanced Institute of Science and Technology, Daejeon 34141, Korea}

\date{\today}

\begin{abstract}
We establish efficient learning for a family of fermionic magic states under unknown number-conserving free-fermion evolution. Each input block has a definite particle number and is a superposition of Fock states, with occupied mode sets disjoint both within and across blocks. The key idea is to exploit the spectral structure of particle reduced density matrices (RDMs) to separate contributions from individual blocks from those involving several blocks, allowing us to reconstruct the hidden block structure. For a fixed upper bound on the particle number per block, our algorithm uses single-copy measurements and polynomial sample and classical computational complexity to recover a compact description of the state with prescribed fidelity and high probability, without prior knowledge of the block decomposition or the evolution. RDMs up to this upper bound suffice for reconstruction. We further show that this RDM order is necessary in general: two orthogonal states in the family can have identical RDMs at every lower order. These results show that an extensive number of non-Gaussian blocks can be compatible with efficient state learning.
\end{abstract}

\maketitle

\section{Introduction}
\label{sec:introduction}

Reliable quantum state characterization is essential for verifying state preparation, benchmarking quantum devices, and probing correlations and entanglement in quantum simulators~\cite{eisert2020quantum,lanyon2017efficient,lange2023adaptive}. Quantum state tomography provides a standard approach by reconstructing a classical description of an unknown state from measurements on repeated preparations~\cite{anshu2024survey,qin2026statistical,d2003quantum,paris2004quantum}. However, without structural assumptions, both the number of parameters needed to describe the state and the number of copies required to reconstruct it with a fixed accuracy grow exponentially with system size~\cite{haah2017sample}. This growth makes tomography increasingly demanding in both experimental and classical computational resources.

Despite this general difficulty, a suitable structure can make efficient learning possible. For example, pure stabilizer states in qubit systems and Gaussian states in bosonic and fermionic systems admit classical descriptions of polynomial size: stabilizer generators for stabilizer states, covariance matrices for fermionic Gaussian states, and first moments and covariance matrices for bosonic Gaussian states. These compact descriptions form the basis of efficient learning algorithms for stabilizer and Gaussian states~\cite{montanaro2017learning,aaronson2023fermion,bittel2025optimal,mele2025learning}. The optimal sample complexity of Gaussian-state tomography has also been established~\cite{chen2026optimal}.
Beyond these families, non-Clifford operations in qubit systems and non-Gaussian operations in bosonic and fermionic systems generally produce states that are no longer fully characterized by the compact descriptions above. Although this removes the basis for the preceding learning methods, efficient learning remains possible for certain families of states prepared using a small number of such operations~\cite{grewal2025efficient,mele2025learning,mele2025few}.

Efficient learning has also been established for other structured families beyond the basic stabilizer and Gaussian settings. In qubit systems, efficient algorithms have been developed for certain instantaneous quantum polynomial-time circuit states~\cite{arunachalam2022optimal} and for broader families of states prepared by shallow circuits~\cite{huang2024learning,landau2025learning}. In bosonic systems, efficient learning has likewise been shown for certain states obtained by applying Gaussian unitaries to Fock inputs~\cite{iosue2025higher}. Together, these results suggest that states outside the basic stabilizer and Gaussian families can remain learnable when their preparation retains suitable structure.

In fermionic computation, pure non-Gaussian states of definite parity serve as magic resources beyond Gaussian operations~\cite{hebenstreit2019all}, raising the question of which families of fermionic magic states admit efficient learning. We consider states obtained by preparing non-Gaussian blocks on mutually disjoint groups of modes and mixing them through an unknown passive Gaussian circuit. Here, free-fermion evolution refers specifically to number-conserving Gaussian evolution. A representative example is a product of the four-mode two-particle state $(\ket{0011}+\ket{1100})/\sqrt{2}$~\cite{ivanov2016computational,oszmaniec2022fermion}. Related block-product structures also arise in electronic-structure calculations and simulations of fermionic dynamics~\cite{kottmann2022optimized,huggins2022unbiasing,alam2025dynamics,oh2026classical}, and in fermionic quantum machine learning~\cite{bako2025fermionic,kerenidis2026scalable}. Although the input has a simple block structure, the unknown Gaussian evolution can mix modes across all blocks and hide their original mode decomposition.

This motivates the question of whether the full state remains efficiently learnable when the number of non-Gaussian blocks grows linearly with the system size and their mode decomposition is unknown. Previous guarantees for fermionic states prepared with at most logarithmically many non-Gaussian gates~\cite{mele2025few} do not directly cover this setting. Here, we establish efficient learning in this regime for a structured family of fermionic magic states. Each input block has a definite particle number bounded by a constant and is a coherent superposition of Fock states with mutually disjoint occupied mode sets; distinct blocks also occupy disjoint mode sets. The learner receives only copies of the output state, without knowing the input block decomposition, its coefficients, or the Gaussian unitary.

Our approach exploits the structure of particle reduced density matrices (RDMs), with the $k$-RDM describing correlations involving $k$ particles~\cite{zhao2021fermionic,gigena2021many,cianciulli2024bipartite}. Passive Gaussian evolution transforms each $k$-RDM within the $k$-particle sector~\cite{low2022classical}, preserving spectral information about the input blocks even as their modes become delocalized. For fixed $k$, these RDMs have polynomial size and can be estimated efficiently from independent copies~\cite{zhao2021fermionic,low2022classical,heyraud2025unified,koizumi2026provably}. Efficient estimation alone, however, does not solve the learning problem: the RDMs contain both coherent contributions from individual blocks and contributions involving particles from different blocks.

When every block contains $p$ particles, the eigenvalue-one space of the exact $p$-RDM is precisely the span of the block states. Since contributions involving several blocks can have eigenvalues arbitrarily close to one, selecting this space directly can be unstable under estimation errors. We use the $1$-RDM to identify high-occupation modes and construct compressions of the $p$-RDM from which the relevant block contributions can be selected stably. A classical Gram-splitting procedure then separates the individual factors by exploiting the orthogonality of their one-particle supports.

For blocks with different, unknown particle numbers, we extend this procedure by processing the RDMs in increasing order and removing product contributions formed from previously recovered factors. By controlling how RDM estimation errors propagate through these steps, we prove that the reconstruction returns a compact classical approximation to the full state with prescribed fidelity and high probability. For a fixed upper bound on the particle number per block, both the sample and classical computational costs are polynomial in the number of modes and the inverse target accuracy, using RDMs only up to this bound. We also show that this RDM order is necessary in the worst case: two orthogonal states in our family can have identical RDMs at every lower order.

Sec.~\ref{sec:problem-main-results} defines the learning problem and states our main results, and Sec.~\ref{sec:algorithm-overview} presents the reconstruction strategy. Secs.~\ref{sec:equal-reconstruction} and~\ref{sec:bounded-reconstruction} establish the learning guarantees, while Sec.~\ref{sec:rdm-order-necessity} proves the RDM-order lower bound. We conclude with a discussion of implications and open questions.

\section{Problem setup and main results}\label{sec:problem-main-results}

\subsection{Fermionic Fock space and state family}\label{Sec:State}
Consider $m$ fermionic modes with creation and annihilation operators $\hat c_j^\dagger$ and $\hat c_j$, where $j\in[m]:=\set{1,\ldots,m}$. They obey the canonical anticommutation relations,
\begin{align}
    \anticomm{\hat c_i}{\hat c_j}  = 0,
    \qquad  \{\hat c_i^\dagger,\hat c_j^\dagger\}  =0,
    \qquad  \{\hat c_i, \hat c_j^\dagger\}  =\delta_{ij}\hat{I},
    \label{eq:car}
\end{align}
where $\{A,B\}:=AB+BA$. For an ordered set $I=\set{i_1<\cdots<i_k}$, we define
$\hat c_I^\dagger:=\hat c_{i_1}^\dagger\cdots \hat c_{i_k}^\dagger, \hat c_I:=(\hat c_I^\dagger)^\dagger=\hat c_{i_k}\cdots \hat c_{i_1}$,
where $\hat{c}_{\varnothing} = \hat{c}_{\varnothing}^\dagger = \hat{I}$. Then, a normalized $k$-particle Fock state corresponding to $I$ is defined as $|I\rangle:=\hat{c}_I^\dagger|\text{vac}\rangle$.
Here, $|\text{vac}\rangle$ is the vacuum state defined as $\hat{c}_j|\text{vac}\rangle = 0$ for all $j$. Using these Fock states as a basis, we write a normalized $k$-particle pure state as
$|\psi\rangle=\sum_{|I| = k} \psi_I |I\rangle$, with $\sum_{|I|=k} |\psi_I|^2=1$,
where $\psi_I\in\mathbb{C}$.
Since this basis is indexed by the $k$-element subsets of $[m]$, the $k$-particle subspace has dimension
$D_k:=\binom{m}{k}$.
We define the creation polynomial associated with $|\psi\rangle$ by
$\hat c^\dagger[\psi]:=\sum_{\abs{I}=k} \psi_I\hat c_I^\dagger$.

We consider inputs built from non-Gaussian states on disjoint groups of modes~\cite{oszmaniec2022fermion, huggins2022unbiasing, kerenidis2026scalable, bako2025fermionic, alam2025dynamics, oh2026classical, ivanov2016computational}.
A representative example is the four-mode, two-particle state
\begin{align}
    |\psi_4\rangle &:=\frac{1}{\sqrt{2}} \left(|0011\rangle+|1100\rangle\right).
\end{align}
Product inputs of the form $|\psi_4\rangle^{\otimes n}$, followed by fermionic linear optics and occupation-number measurements, have been studied as a quantum computational advantage scheme~\cite{ivanov2016computational, oszmaniec2022fermion} and for simulating many-body systems~\cite{alam2025dynamics}.

We consider a broader family of block-product inputs, with each block formed by a superposition of Fock states with the same particle number and mutually disjoint occupied mode sets.
To formally define this family, for each block $b\in[n]$, fix integers $p_b\geq2$ and $s_b\geq2$, and choose subsets $I_{b,1},\ldots,I_{b,s_b}\subseteq[m]$ of size $p_b$. For each $b\in[n]$, let $\omega_{b,1},\ldots,\omega_{b,s_b}$ be nonzero complex numbers satisfying $\sum_{l=1}^{s_b}|\omega_{b,l}|^2=1$, and define the state of block $b$ by $\ket{\omega_b^{\mathrm{in}}}:=\sum_{l=1}^{s_b}\omega_{b,l}\ket{I_{b,l}}$. We also allow a set $I_0\subseteq[m]$ of always-occupied modes, requiring $I_0$ and all sets $I_{b,l}$ to be mutually disjoint. All remaining input modes are in the vacuum state. The input state is then
\begin{align}
    \ket{\Psi_{\mathrm{in}}} := \hat c_{I_0}^\dagger \prod_{b=1}^{n} \left( \sum_{l=1}^{s_b} \omega_{b,l} \hat c_{I_{b,l}}^\dagger \right) \ket{\vac} = \ket{\mathbf{1}}\otimes\left(\bigotimes_{b=1}^{n}\ket{\omega_b^{\mathrm{in}}}\right)\otimes\ket{\mathbf{0}},
    \label{eq:block-source}
\end{align}
where the product is taken in increasing order of $b$. The tensor-product expression uses the fermionic identification with the always-occupied modes first, followed by the blocks in increasing order of $b$ and then the vacuum modes. We refer to each basis state $|I_{b,l}\rangle$ as a \textit{branch} of block $b$.
The input $\ket{\psi_4}^{\otimes n}$ is recovered by taking $m=4n$, $I_0=\varnothing$, $p_b=s_b=2$, and $\omega_{b,1}=\omega_{b,2}=1/\sqrt{2}$ for every block.
The same definition also includes blocks with different coefficients, as in perfect-pairing states~\cite{huggins2022unbiasing}, and analogous fixed-particle-number states on larger sets of modes~\cite{kerenidis2026scalable}.

Our target state is obtained by applying an unknown passive Gaussian unitary $\hat U$ to the input state. It is specified by a single-particle unitary $U\in\operatorname{U}(m)$ satisfying
\begin{align}
    \hat U\hat c_j^\dagger\hat U^\dagger &= \sum_{i=1}^{m} U_{ij}\hat c_i^\dagger.
    \label{eq:passive-unitary}
\end{align}
Thus, $\hat U$ preserves the total particle number and therefore satisfies $\hat{U} |\vac\rangle = |\vac\rangle$ up to a global phase. Applying it to the input gives the target state
\begin{align}
    \ket{\Psi} &:= \hat U\ket{\Psi_{\mathrm{in}}}.
    \label{eq:output-state}
\end{align}
Fig.~\ref{fig:Problem}(a) illustrates the target state. For each block $b$, define the transformed block state by
$\ket{\omega_b}:=\hat U\ket{\omega_b^{\mathrm{in}}}$.
Although particle-number conservation ensures that $\ket{\omega_b}$ remains a $p_b$-particle state, its expansion in the output Fock basis is unknown.

\begin{figure}
    \centering
    \includegraphics[width=0.9\linewidth]{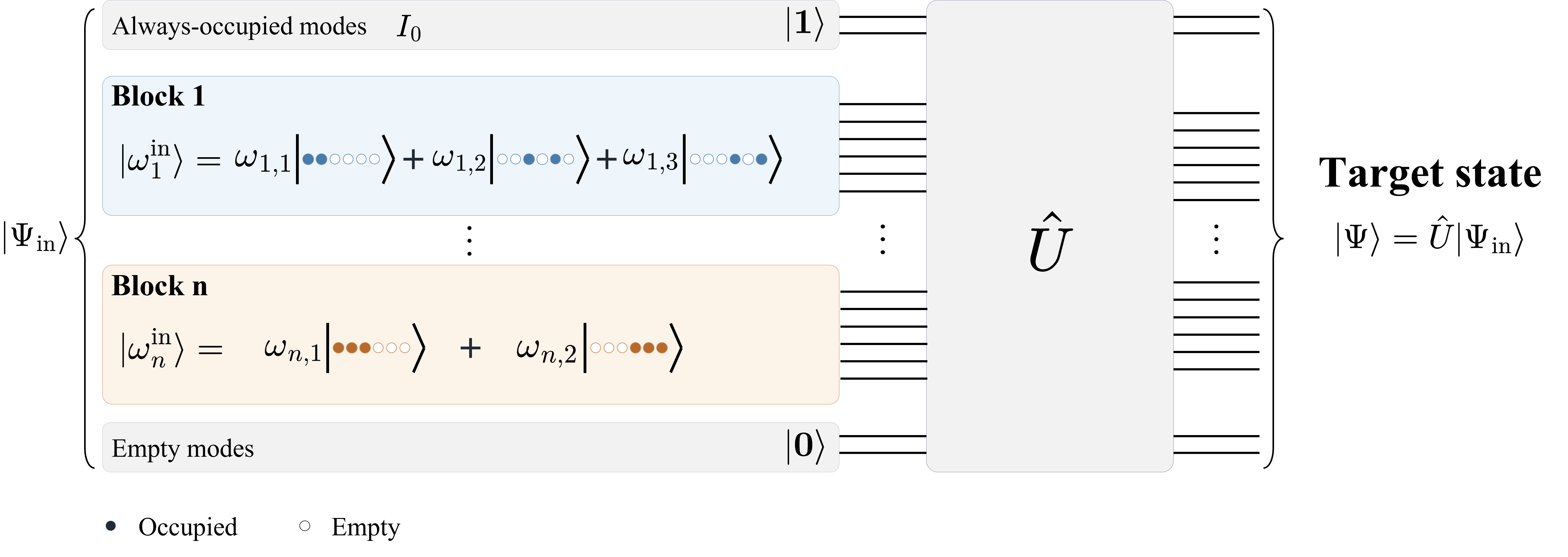}
    \caption{Example of the input state family $\ket{\Psi_{\rm in}}$ and target state $\ket{\Psi}$. Mode numbers within each illustrated block are local to that block; different blocks occupy disjoint sets of physical modes. Block $1$ has $s_1=3$ branches and $p_1=2$ particles, with occupied local mode sets $\{1,2\}$, $\{3,5\}$, and $\{4,6\}$. Block $n$ has $s_n=2$ branches and $p_n=3$ particles, with occupied local mode sets $\{1,2,3\}$ and $\{4,5,6\}$. These block states, together with always-occupied modes $\ket{1}$ and empty modes $\ket{0}$, form the input state $\ket{\Psi_{\rm in}}$ in Eq.~\eqref{eq:block-source}. Applying the unknown passive Gaussian unitary $\hat U$ gives the target state $\ket{\Psi}$ in Eq.~\eqref{eq:output-state}.}
    \label{fig:Problem}
\end{figure}

\subsection{Learning problem and main results}\label{Sec:Main results}

We now describe the learning task considered in this work. Although the learner is promised that the unknown target state $\ket{\Psi}$ has the form in Eq.~\eqref{eq:output-state}, neither the passive Gaussian unitary $\hat U$ nor the input block data are known. Using measurements on independent copies of $\ket{\Psi}$ and classical post-processing of their outcomes, the learner must construct a classical description of a pure state close to $\ket{\Psi}$ in fidelity. More precisely, we consider the following learning problem.

\begin{problem}\label{prob:learning}
Fix $\varepsilon_{\mathrm{fid}},\delta\in(0,1)$. Given independent copies of $\ket{\Psi}$, output a classical description of a normalized pure state $\ket{\widetilde\Psi}$ such that, with probability at least $1-\delta$,
\begin{align}
    \abs{ \braket{\widetilde\Psi}{\Psi} }^2&\geq 1-\varepsilon_{\mathrm{fid}}.
    \label{eq:learning-goal}
\end{align}
\end{problem}
This problem is also illustrated in Fig.~\ref{fig:Problem}. Since the learning procedure combines measurements on independent copies with classical post-processing, we measure its efficiency by both the number of copies of $\ket{\Psi}$ and the classical post-processing time. In this work, we study Problem~\ref{prob:learning} in two settings, \textit{homogeneous} and \textit{heterogeneous} cases.

We first consider the \textit{homogeneous} case, in which every input mode belongs to a block, with no always-occupied modes ($I_0=\varnothing$) or additional vacuum modes. Every block contains the same known number $p\geq2$ of particles, i.e., $p_b=p$ for all $b\in[n]$.

\begin{theorem}[Homogeneous case]
\label{thm:equal-blocks}
Fix $p\geq2$. There exists a learning algorithm for the homogeneous setting that uses only single-copy measurements and solves Problem~\ref{prob:learning}. The required sample complexity and classical time complexity are
\begin{align}
    O\left(m^{O(p)}\varepsilon_{\mathrm{fid}}^{-2} \log\left(\frac{2m}{\delta}\right)\right).
    \label{eq:equal-copy-simple}
\end{align}
\end{theorem}
For fixed $p$, both the sample and classical computational costs remain polynomial in the number of modes, even when the number of non-Gaussian blocks grows linearly with the system size. Although the unknown Gaussian evolution hides the original separation of the input blocks, the $1$- and $p$-RDMs retain enough information to reconstruct the state. Efficient learning therefore extends to this family without restricting the number of non-Gaussian blocks to be small. We describe the reconstruction algorithm in Sec.~\ref{sec:equal-reconstruction}.

We next consider the \textit{heterogeneous} case, in which the non-Gaussian blocks can have different particle numbers $p_b\geq2$, and the input may also include always-occupied modes and additional vacuum modes. The block particle numbers $p_b$ are unknown to the learner, who is given only a common upper bound $2\leq r\leq m$ such that $p_b\leq r$ for all $b\in[n]$.

\begin{theorem}[Heterogeneous case]
\label{thm:heterogeneous-blocks}
Fix $r\geq2$.  There exists a learning algorithm for the heterogeneous setting that uses only single-copy measurements and solves Problem~\ref{prob:learning}. The required sample complexity and classical time complexity are
\begin{align}
    O\left(m^{O(r)} \varepsilon_{\mathrm{fid}}^{-2} \log\left( \frac{mr}{\delta}\right)\right).
    \label{eq:mixed-copy-simple}
\end{align}
\end{theorem}

Thm.~\ref{thm:heterogeneous-blocks} extends efficient learning to inputs containing blocks with different, unknown particle numbers. Since the $k$-RDM can contain contributions from both individual $k$-particle blocks and products of smaller blocks, the reconstruction processes RDMs of orders $1$ through $r$ in increasing order, removing products of previously recovered blocks before identifying new ones. For fixed $r$, the cost of this recursive reconstruction remains polynomial in the number of modes. Thus, a common particle number is not required: the learner can reconstruct the state using only an upper bound on the particle number per block. The algorithm is described in Sec.~\ref{sec:bounded-reconstruction}.

To specify the RDM inputs used in both settings, for $k$-element subsets $I,J\subseteq[m]$ we define the $k$-RDM of a pure state $\ket{\psi}$ by
\begin{align}
    \left( \Gamma_{\psi}^{(k)} \right)_{I,J} &:= \bra{\psi} \hat c_J^\dagger \hat c_I \ket{\psi}.
    \label{eq:rdm-definition}
\end{align}
We regard $\Gamma_\psi^{(k)}$ as an operator on the $k$-particle sector, with the entries above giving its $D_k\times D_k$ matrix in the Fock basis. We use the same notation for the operator and its matrix. The homogeneous algorithm uses only $\Gamma_{\Psi}^{(1)}$ and $\Gamma_{\Psi}^{(p)}$, whereas the heterogeneous algorithm uses $\Gamma_{\Psi}^{(1)},\ldots,\Gamma_{\Psi}^{(r)}$. We use the single-copy fermionic partial-tomography protocol of Ref.~\cite{zhao2021fermionic} to estimate these RDMs from independent copies of $\ket{\Psi}$. The classical time includes both processing the measurement outcomes and reconstructing the state, and is counted in arithmetic operations. In the remainder of the paper, we explain how these RDM estimates are used to reconstruct the target state, starting with an overview in Sec.~\ref{sec:algorithm-overview}. The copy and classical time bounds are derived in App.~\ref{app:rdm-estimation} and Secs.~\ref{sec:homogeneous-time-complexity} and~\ref{sec:heterogeneous-time-complexity}.

Finally, reconstruction in the heterogeneous setting requires RDMs of order at least $r$ in the worst case: RDMs of orders below $r$ do not determine every state in the learning class, even in the single-block case.

\begin{theorem}[Necessity of the highest RDM order]
\label{thm:rdm-order-necessary}
For every $r\geq2$ and $m\geq2r$, there exist two orthogonal states $\ket{\Omega_+}$ and $\ket{\Omega_-}$ in the single-block subclass of Problem~\ref{prob:learning}, each containing $r$ particles, such that
\begin{align}
    \Gamma_{\Omega_+}^{(k)}&=\Gamma_{\Omega_-}^{(k)},
    \qquad 1\leq k<r.
    \label{eq:lower-rdms-identical}
\end{align}
Consequently, RDMs of orders below $r$ cannot determine every state in the heterogeneous family.
\end{theorem}

\section{Algorithm overview}
\label{sec:algorithm-overview}

%

In this section, we outline the reconstruction algorithm for the homogeneous setting. We illustrate the reconstruction through a two-particle example and then extend it to general $p$. To isolate the main ideas, we use exact RDMs throughout this section and defer the analysis of estimation errors to Sec.~\ref{sec:equal-reconstruction}.

\subsection{A guiding two-particle example}
\label{sec:two-particle-overview}

Consider a target state $\ket{\Psi}=\hat U\ket{\Psi_{\mathrm{in}}}$, where $\hat U$ is an unknown passive Gaussian unitary. The input consists of two two-particle blocks on eight modes, $\ket{\Psi_{\mathrm{in}}}:=\ket{\omega_1^{\mathrm{in}}}\otimes\ket{\omega_2^{\mathrm{in}}}$, with
\begin{align}
\ket{\omega_1^{\mathrm{in}}} &=\sqrt{w_1}\ket{12}+\sqrt{1-w_1}\ket{34},
\qquad \ket{\omega_2^{\mathrm{in}}} =\sqrt{w_2}\ket{56}+\sqrt{1-w_2}\ket{78}.
\label{eq:overview-two-block-input}
\end{align}
We label the branches so that $1/2\leq w_b<1$.

\medskip
\noindent\textbf{Identifying the block-state span from the $2$-RDM.}
We first show how the exact $2$-RDM reveals the span of the two block states. To do so, we compute the input $2$-RDM and then examine how it transforms under the unknown passive Gaussian unitary.

By Eq.~\eqref{eq:rdm-definition}, each $2$-RDM entry is the inner product of the unnormalized states obtained by removing the corresponding pairs of particles. Since each block in the product input $\ket{\Psi_{\mathrm{in}}}$ has a definite particle number, removing two particles from the first block, two from the second, or one from each leaves distinct particle-number configurations across the two blocks. Because residual states from different sectors are mutually orthogonal, the input $2$-RDM has the direct-sum decomposition
\begin{align}
\Gamma_{\Psi_{\mathrm{in}}}^{(2)} &=\Gamma_{\mathrm{complete,in}}^{(2)} \oplus\Gamma_{\mathrm{mix,in}}^{(2)},
\end{align}
where the \emph{complete-block} part combines the first two sectors and the \emph{mix-block} part corresponds to the third.

For the first block, we compute $\hat c_2\hat c_1\ket{\Psi_{\mathrm{in}}} =\sqrt{w_1}\ket{\omega_2^{\mathrm{in}}}$ and $\hat c_4\hat c_3\ket{\Psi_{\mathrm{in}}} =\sqrt{1-w_1}\ket{\omega_2^{\mathrm{in}}}$. On the span of $\ket{12}$ and $\ket{34}$, the $2$-RDM is therefore represented by
\begin{align}
    \begin{pmatrix}
        w_1 & \sqrt{w_1(1-w_1)}\\
        \sqrt{w_1(1-w_1)} & 1-w_1
    \end{pmatrix}
\end{align}
in this ordered basis, while all other rows and columns within the first-block sector vanish because the corresponding pairs yield zero residual states. The displayed matrix is the outer product of the coefficient vector of $\ket{\omega_1^{\mathrm{in}}}$ with itself, identifying the entire first-block contribution as $\ket{\omega_1^{\mathrm{in}}}\bra{\omega_1^{\mathrm{in}}}$. Applying the same calculation to the second block gives $\Gamma_{\mathrm{complete,in}}^{(2)}=\sum_{b=1}^2\ket{\omega_b^{\mathrm{in}}}\bra{\omega_b^{\mathrm{in}}}$. 
The complete-block part is the orthogonal projector onto the span of the two block states: its eigenvalue-one space is precisely this span, and it vanishes on the orthogonal complement within the complete-block sectors.

In the mix-block part, removing one particle from each block, as for the pair $15$, leaves one particle in each block. Since distinct mixed pairs leave orthogonal residual states, this restriction is diagonal in the input pair basis. The four groups of mixed pairs yield
\begin{align}
\Gamma_{\mathrm{mix,in}}^{(2)}
&=w_1w_2 I_4\oplus w_1(1-w_2)I_4
\oplus(1-w_1)w_2 I_4
\oplus(1-w_1)(1-w_2)I_4,
\end{align}
where $I_4$ denotes the identity on the corresponding four-dimensional pair space. For instance, the first group consists of $15,16,25,26$. 
With all mix-block eigenvalues strictly below one, the eigenvalue-one space of the full input $2$-RDM is exactly the span of the two block states.

Although this decomposition is explicit in the input mode basis, the unknown passive Gaussian evolution mixes modes across blocks and hides their separation in the output basis. 
The key to reconstruction is that passive Gaussian evolution preserves this spectral distinction: the eigenvalue-one space of the output $2$-RDM is exactly the span of the transformed block states. To establish this, apply the creation-operator transformation in Eq.~\eqref{eq:passive-unitary} to a two-particle basis state. Choosing the irrelevant global phase of $\hat U$ so that $\hat U\ket{\vac}=\ket{\vac}$, we obtain, for $i<j$,
\begin{align}
    \hat U\ket{ij}
    &=\left(\sum_a U_{ai}\hat c_a^\dagger\right)
      \left(\sum_b U_{bj}\hat c_b^\dagger\right)\ket{\vac}
     =\sum_{a<b}\bigl(U_{ai}U_{bj}-U_{aj}U_{bi}\bigr)\ket{ab}.
\end{align}
The minus sign from fermionic anticommutation expresses the antisymmetry of the two-particle state. Accordingly, the transformation above is the restriction of $U\otimes U$ to the antisymmetric two-particle space, denoted by $\wedge^2 U$, with pair-basis entries $(\wedge^2 U)_{ab,ij}=U_{ai}U_{bj}-U_{aj}U_{bi}$. Since $U\otimes U$ is unitary and preserves this space, $\wedge^2 U$ is also unitary.

The same transformation governs the $2$-RDM of the full state. Substituting the creation- and annihilation-operator transformations into Eq.~\eqref{eq:rdm-definition} yields
\begin{align}
    \Gamma_\Psi^{(2)}&=(\wedge^2 U)\Gamma_{\Psi_{\mathrm{in}}}^{(2)}(\wedge^2 U)^\dagger,
    \qquad \ket{\omega_b}=(\wedge^2 U)\ket{\omega_b^{\mathrm{in}}}.
\end{align}
Hence, the input $2$-RDM and the output $2$-RDM are unitarily equivalent: their eigenvalues are identical, and their eigenspaces are mapped by $\wedge^2 U$. Applying this transformation to the complete-block and mix-block contributions yields
\begin{align}
\Gamma_\Psi^{(2)} =\sum_{b=1}^2\ket{\omega_b}\bra{\omega_b} +\Gamma_{\mathrm{mix}}^{(2)},
\label{eq:overview-two-rdm-warmup}
\end{align}
where $\Gamma_{\mathrm{mix}}^{(2)} :=(\wedge^2 U)\Gamma_{\mathrm{mix,in}}^{(2)}(\wedge^2 U)^\dagger$. Since unitary conjugation preserves orthogonality and eigenvalues, the two contributions retain their orthogonal supports, with eigenvalue one on the block-state span and strictly smaller eigenvalues in the mix-block part. Diagonalizing the output $2$-RDM therefore identifies $T=\operatorname{span}\{\ket{\omega_1},\ket{\omega_2}\}$ without knowing $U$ or the input block supports.

\medskip
\noindent\textbf{Recovering the individual blocks by Gram splitting.}
Since the two block states share the same eigenvalue $1$, diagonalizing the $2$-RDM identifies only their span $T$. To recover the individual blocks, we use the procedure we call \emph{Gram splitting}. The key observation is that the $1$-RDM of a superposition of the blocks has no cross-block terms: for $\ket{x}=x_1\ket{\omega_1}+x_2\ket{\omega_2}$,
\begin{align}
\Gamma_x^{(1)}=\sum_{b=1}^2|x_b|^2\Gamma_{\omega_b}^{(1)}.
\end{align}
The block RDMs on the right have mutually orthogonal one-particle supports.
Since the same decomposition holds for a fixed reference vector $\ket{z}=z_1\ket{\omega_1}+z_2\ket{\omega_2}\in T$, orthogonality eliminates all contributions from different blocks in the trace overlap:
\begin{align}
\operatorname{Tr}\left[\Gamma_x^{(1)}\Gamma_z^{(1)}\right] &= \sum_{b=1}^2 |x_b|^2|z_b|^2\operatorname{Tr}\left[(\Gamma_{\omega_b}^{(1)})^2\right].
\label{eq:overview-pair-splitting}
\end{align}
Thus, the overlap is a quadratic form in $x$ that is diagonal in the unknown block basis.
Accordingly, we sample a uniformly random normalized vector $\ket{z}\in T$, making the diagonal coefficients distinct with probability one. With this degeneracy removed, diagonalizing the quadratic form in a known orthonormal basis $\ket{u_1},\ket{u_2}$ of $T$ recovers the block states up to phases and ordering.

More explicitly, write $\ket{x}=v_1\ket{u_1}+v_2\ket{u_2}$ and $v=(v_1,v_2)^T$. Define the transition $1$-RDMs by $(\Gamma_{u_\beta,u_\alpha}^{(1)})_{i,j}:=\bra{u_\alpha}\hat c_j^\dagger\hat c_i\ket{u_\beta}$, for $\alpha,\beta\in\{1,2\}$. The Hermitian matrix representing the overlap in this basis is
\begin{align}
\Tr\left[\Gamma_x^{(1)}\Gamma_z^{(1)}\right]&=v^\dagger S_zv,
\qquad
(S_z)_{\alpha,\beta}:=\Tr\left[\Gamma_{u_\beta,u_\alpha}^{(1)}\Gamma_z^{(1)}\right].
\end{align}
Diagonalizing $S_z$ gives the coefficients of the block states in the basis $\ket{u_1},\ket{u_2}$, up to phases and ordering. Using their creation operators, we reconstruct the full state as
\begin{align}
    \ket{\Psi}
    =\hat c^\dagger[\omega_1]\hat c^\dagger[\omega_2]\ket{\vac},
\end{align}
up to a global phase. Since the transition RDMs in this splitting step are calculated from known vectors, the full reconstruction requires only the exact $2$-RDM and no additional measurements.

\medskip
The same construction works for any number of blocks with mutually orthogonal one-particle supports, as formalized below.
\begin{lemma}[Exact Gram splitting]\label{lem:main-exact-gram-splitting}
Let $k\geq2$, and let $T$ be a subspace of the $k$-particle space $\wedge^k\mathbb C^m$, spanned by normalized states $\ket{\psi_1},\ldots,\ket{\psi_d}$ with mutually orthogonal one-particle supports. Given an exact orthonormal basis of $T$, Gram splitting recovers these states up to individual phases and permutation with probability one. For fixed $k$, the algorithm uses polynomially many arithmetic operations in $m$ and requires no additional copies of the target state.
\end{lemma}
The general construction and its proof are given in App.~\ref{app:exact-gram-splitting}. Lem.~\ref{lem:main-robust-gram-splitting} below provides the corresponding error guarantee when only an estimated subspace is available.

\noindent\textbf{An alternative reconstruction for robustness.}
The stability of the reconstruction above depends on the spectral gap between the complete-block and mix-block parts of the $2$-RDM. Since the block-state span has eigenvalue one and the largest mix-block eigenvalue is $w_1w_2$, the separating gap is
\begin{align}
    1-w_1w_2.
\end{align}
As both $w_b$ approach one, the gap vanishes, allowing estimation errors to mix the complete-block and mix-block subspaces. We therefore replace direct extraction of the block-state span with a reconstruction that first uses the $1$-RDM to separate the modes with large occupation. We then use different compressions of the $2$-RDM to recover blocks with and without a dominant branch, as we illustrate below using exact RDMs. Estimation errors are analyzed in Sec.~\ref{sec:equal-reconstruction}.

To see how the $1$-RDM identifies these modes, note that a one-body operator cannot connect the two branches of a block because their occupied mode pairs are disjoint. The input $1$-RDM and its transformation under passive Gaussian evolution are therefore
\begin{align}
    \Gamma_{\Psi_{\mathrm{in}}}^{(1)}&=\operatorname{diag}\bigl(w_1,w_1,1-w_1,1-w_1,w_2,w_2,1-w_2,1-w_2\bigr),\qquad \Gamma_\Psi^{(1)}=U\Gamma_{\Psi_{\mathrm{in}}}^{(1)}U^\dagger.
\end{align}
Thus, the branch weights remain eigenvalues, while the corresponding one-particle directions become $U\ket{j}$. Choose $\theta\in[2/3,3/4]$ and call a branch \emph{dominant} when its weight exceeds $\theta$. To illustrate both components of this reconstruction in a single example, suppose $w_1>\theta$ and $1/2\leq w_2<\theta$, so that only the first block has a dominant branch. Although the original spectral gap is bounded below in this case, this choice lets us demonstrate separately how the procedure handles blocks with and without a dominant branch. The spectral subspace of $\Gamma_\Psi^{(1)}$ above $\theta$ is then
\begin{align}
    S_\theta=\operatorname{span}\{U\ket{1},U\ket{2}\}.
    \label{eq:overview-high-occupation-space}
\end{align}
Diagonalizing the output $1$-RDM determines this space without knowledge of $U$. The remaining branch of the first block and both branches of the second block have all their occupied modes in $S_\theta^\perp$. Thus, the first block is split between the two mode spaces, while the second remains entirely in the latter.

To recover the first block, we extract the coherence between its two branches. Let $Q_{\mathrm{hi}}$ project onto two-particle states with both particles in $S_\theta$, and let $Q_{\mathrm{lo}}$ project onto those with both particles in $S_\theta^\perp$. Both projectors are constructed from $S_\theta$ alone. While $Q_{\mathrm{lo}}\neq I-Q_{\mathrm{hi}}$ because pairs with one particle in each space lie in neither range, no complete branch is lost: both of its modes have the same occupation and therefore lie in the same space. Write the transformed branches of the first block as $\ket{f_1}:=(\wedge^2 U)\ket{12}$ and $\ket{g_1}:=(\wedge^2 U)\ket{34}$. Since every mixed pair contains a particle from the second block, $Q_{\mathrm{hi}}$ removes all such pairs and has range spanned by $\ket{f_1}$. Applying $Q_{\mathrm{lo}}$ on the other side then isolates the coherence between the two branches:
\begin{align}
    C_2:=Q_{\mathrm{lo}}\Gamma_\Psi^{(2)}Q_{\mathrm{hi}}
    =\sqrt{w_1(1-w_1)}\ket{g_1}\bra{f_1}.
    \label{eq:overview-pair-cross}
\end{align}
While the right singular vector associated with the nonzero singular value $\sigma_1=\sqrt{w_1(1-w_1)}$ identifies $\ket{f_1}$ up to a phase, applying $C_2$ to that vector gives the remaining branch component with its relative phase. Together with the weight obtained from the $2$-RDM, this determines the full block state:
\begin{align}
    w_1&=\bra{f_1}\Gamma_\Psi^{(2)}\ket{f_1},
    \qquad
    \ket{\omega_1}=\sqrt{w_1}\ket{f_1}
    +\frac{C_2\ket{f_1}}{\sqrt{w_1}}.
\end{align}

The same cross compression removes the mix-block contribution even when both blocks have dominant branches, including the regime $w_1,w_2\to1$ that closes the original spectral gap. Indeed, the mix-block part is diagonal in the transformed pair basis and therefore has no matrix elements between the all-high and all-low sectors. In that case, $C_2$ retains the branch coherences of both blocks.

Returning to our example, both branches of the second block lie entirely in the modes selected by $S_\theta^\perp$, so we recover it by restricting the $2$-RDM to pairs in that space:
\begin{align}
    A_2:=Q_{\mathrm{lo}}\Gamma_\Psi^{(2)}Q_{\mathrm{lo}}
    =(1-w_1)\ket{g_1}\bra{g_1}
      +\ket{\omega_2}\bra{\omega_2}
      +Q_{\mathrm{lo}}\Gamma_{\mathrm{mix}}^{(2)}Q_{\mathrm{lo}}.
    \label{eq:overview-pair-low}
\end{align}
This preserves the full second block, including the coherence between its branches. Each surviving mixed pair contains one mode occupied in $\ket{g_1}$, contributing a factor $1-w_1$; its eigenvalue is therefore either $(1-w_1)w_2$ or $(1-w_1)(1-w_2)$. Thus, $\ket{\omega_2}$ spans the eigenvalue-one space of $A_2$, while every other eigenvalue is at most $1-w_1<1-\theta$. This fixed spectral gap makes recovery of the second block stable, completing the reconstruction together with the first block obtained from $C_2$.

Although the $A_2$ step has a fixed spectral gap, the singular value $\sigma_1$ of $C_2$ can be arbitrarily small. To avoid resolving singular values that estimation errors may obscure, we introduce a threshold $t>0$ and reconstruct the first block from $C_2$ only when $\sigma_1>t$. If $\sigma_1\leq t$, we approximate the first block by its dominant branch, with an error controlled by $t$. Its occupied space $S_\theta$ is still known: occupying both modes of any orthonormal basis of this space gives $\ket{f_1}$ up to a phase, without resolving the small singular value. The approximation has infidelity
\begin{align}
    1-|\langle f_1|\omega_1\rangle|^2
    =1-w_1
    =\frac{\sigma_1^2}{w_1}
    \leq\frac{t^2}{\theta}.
\end{align}
Combining this branch with the exactly recovered second block implies the same full-state infidelity. Thus, the threshold prepares the reconstruction for estimated RDMs while introducing a controlled approximation already in the exact setting.

\medskip
\noindent\textbf{Constructing the input and unitary.}
Let $\ket{\Phi}$ denote the state reconstructed from exact RDMs with the chosen truncation threshold. To express it as a block-product input followed by a passive Gaussian unitary, we first recover the branch coefficients and occupied modes within each block, then combine these modes into one unitary. For the second block, recovered in full by the $A_2$ step, we seek
\begin{align}
    \ket{\omega_2}&=\sum_{\ell=1}^2 a_\ell\hat c^\dagger[u_\ell]\hat c^\dagger[v_\ell]\ket{\vac},\qquad a_\ell>0,\quad a_1^2+a_2^2=1,
\end{align}
where $u_1,v_1,u_2,v_2$ are orthonormal modes to be determined, with coefficient phases absorbed into the modes.

The key observation is that annihilating a particle in a mode of one branch eliminates the other branches because their occupied spaces are orthogonal:
\begin{align}
    \hat c[u_\ell]\ket{\omega_2}&=a_\ell\hat c^\dagger[v_\ell]\ket{\vac},\qquad \hat c[u_\ell]:=\bigl(\hat c^\dagger[u_\ell]\bigr)^\dagger.
\end{align}
Thus, once one mode is identified, contraction recovers its partner together with the branch amplitude and phase. To find that first mode, represent $\ket{\omega_2}$ by its known antisymmetric coefficient matrix $M$, with $M_{ij}=\langle ij|\omega_2\rangle$ for $i<j$. The matrix $MM^\dagger$ identifies the squared branch amplitudes and their occupied eigenspaces. For a normalized eigenvector $u_1$ with positive eigenvalue $a_1^2$, contraction with the original $M$ determines the partner:
\begin{align}
    MM^\dagger&=\sum_{\ell=1}^2a_\ell^2(u_\ell u_\ell^\dagger+v_\ell v_\ell^\dagger),\qquad v_1:=-M\overline{u_1}/a_1.
\end{align}
Using $M$ fixes the pairing and phase information needed to reproduce the block. The resulting pair is orthonormal, and subtracting its branch leaves a residual on modes orthogonal to both. Repeating this procedure recovers every branch, for any number of two-particle branches. It also works for coincident weights, when the recovered pairs may differ from the original ones while representing the same block; Prop.~\ref{prop:iu-two-particle-decomposition} proves these claims.

If the first block was fully recovered, apply the same procedure to $\ket{\omega_1}$. If it was replaced by its dominant branch $\ket{f_1}$, any orthonormal basis of its occupied space $S_\theta$ gives two occupied-core modes representing $\ket{f_1}$ up to phase~\cite{terhal2002classical}. Each block's recovered modes stay within its one-particle support. Since these supports and the core are mutually orthogonal, all recovered modes form one orthonormal set. Assign them to columns of $U_{\mathrm{out}}$ and complete them to an orthonormal basis. On the corresponding input modes, use the recovered coefficients to form each block's branch superposition, occupy the core modes, and leave the additional modes empty. This gives a block-product input $\ket{\Xi_{\mathrm{in}}}$ whose output satisfies, up to a global phase,
\begin{align}
    \ket{\widetilde\Psi}&:=\hat U_{\mathrm{out}}\ket{\Xi_{\mathrm{in}}}=\ket{\Phi}.
\end{align}
Preparation therefore preserves the reconstruction error: the output is exact when both blocks were recovered and has infidelity at most $t^2/\theta$ when the first block was truncated. For general particle numbers, finding the branch modes requires an additional step, while assembling them into a single unitary uses the same orthogonality principle.

\subsection{Reconstruction in the homogeneous setting}
\label{sec:homogeneous-overview}

We extend the example to the general homogeneous setting, where every block contains the same number $p\geq2$ of particles, $I_0=\varnothing$, and no additional vacuum modes are present. We again use exact $1$- and $p$-RDMs to describe the reconstruction. The implementation with estimation errors is given in Sec.~\ref{sec:equal-reconstruction}.

\begin{enumerate}
    \setcounter{enumi}{-1}
    \item \textbf{Block structure of the $p$-RDM.}
As in the two-particle example, selecting all $p$ particles from one block gives its rank-one projector, while selecting particles from several blocks gives the mix-block contribution. Thus,
\begin{align}
    \Gamma_\Psi^{(p)}
    =\sum_{b=1}^{n}\ket{\omega_b}\bra{\omega_b}
      +\Gamma_{\mathrm{mix}}^{(p)},
    \label{Eq:Decomposition of p RDM}
\end{align}
where the two contributions have orthogonal supports. The nonzero mixed eigenvalues are products of branch weights from at least two blocks and can approach one, so the eigenvalue-one space again need not have a uniform spectral gap.

    \item \textbf{Separate the modes using the $1$-RDM.}
    \label{step:overview-high}
    As in the example, the input $1$-RDM is diagonal: each mode in $I_{b,\ell}$ has occupation $|\omega_{b,\ell}|^2$. Each branch weight therefore appears on $p$ input modes. The passive Gaussian unitary $\hat U$ maps these modes to the occupied space $F_{b,\ell}:=\operatorname{span}\{U\ket{j}:j\in I_{b,\ell}\}$. Writing $P_F$ for the orthogonal projector onto $F$, the output $1$-RDM is
    \begin{align}
        \Gamma_\Psi^{(1)}
        =U\Gamma_{\Psi_{\mathrm{in}}}^{(1)}U^\dagger
        =\sum_{b,\ell}|\omega_{b,\ell}|^2P_{F_{b,\ell}}.
        \label{eq:overview-one-rdm-form}
    \end{align}
    The occupied spaces are mutually orthogonal, so the branch weights remain eigenvalues of the output $1$-RDM.

    Choose $\theta\in[2/3,3/4]$ and let $S_\theta$ be the spectral subspace of $\Gamma_\Psi^{(1)}$ above $\theta$. As in the example, we call a branch \emph{dominant} when its weight exceeds $\theta$. Thus,
    \begin{align}
        S_\theta=\bigoplus_{b,\ell:\,|\omega_{b,\ell}|^2>\theta}F_{b,\ell}.
    \end{align}
    Since $\theta>1/2$, each block has at most one dominant branch. In the example, $S_\theta$ contained only the occupied modes of the first block's dominant branch; here it collects the occupied spaces of all dominant branches. Diagonalizing the output $1$-RDM determines this combined space without knowing $U$ or resolving the individual branches.

    To form the corresponding compressed RDMs, let $Q_{\mathrm{hi}}$ project onto $p$-particle states with all particles in $S_\theta$, and let $Q_{\mathrm{lo}}$ project onto those with all particles in $S_\theta^\perp$. As in the example, states with particles in both mode spaces are omitted, but every complete branch lies in one of the two selected sectors. We form the analogues of $C_2$ and $A_2$:
    \begin{align}
        C_p:=Q_{\mathrm{lo}}\Gamma_\Psi^{(p)}Q_{\mathrm{hi}},
        \qquad
        A_p:=Q_{\mathrm{lo}}\Gamma_\Psi^{(p)}Q_{\mathrm{lo}}.
        \label{eq:overview-two-compressions}
    \end{align}

    \item \textbf{Recover the retained dominant blocks using $C_p$.}
    \label{step:overview-spaces}
    \label{step:overview-split}
    The cross compression again retains coherence between a dominant branch and the remaining branches of the same block. For the analysis, set $w_b:=\max_\ell|\omega_{b,\ell}|^2$ and, for $w_b>\theta$, write $\ket{\omega_b}=\sqrt{w_b}\ket{f_b}+\ket{r_b}$. Here, $\ket{f_b}$ is the normalized dominant branch, with its coefficient's phase absorbed into it, and $\ket{r_b}$ is the sum of the remaining branches, with $\|r_b\|^2=1-w_b$. In the two-branch example, $\ket{r_1}=\sqrt{1-w_1}\ket{g_1}$.

    Each proper lower-order RDM contributing to the mix-block part has no coherence between distinct branches. Since each branch lies entirely in either $S_\theta$ or $S_\theta^\perp$, these contributions cannot connect the two projected sectors. Therefore, $Q_{\mathrm{lo}}\Gamma_{\mathrm{mix}}^{(p)}Q_{\mathrm{hi}}=0$, and
    \begin{align}
        C_p
        =\sum_{b:w_b>\theta}\sqrt{w_b}\ket{r_b}\bra{f_b}.
        \label{eq:overview-cross-operator-form}
    \end{align}
    Since different blocks have mutually orthogonal one-particle supports, the vectors $\ket{f_b}$ are orthonormal and the residual vectors $\ket{r_b}$ are mutually orthogonal. Thus, $\ket{f_b}$ is a right singular vector with singular value $\sigma_b=\sqrt{w_b(1-w_b)}$.

    As in the example, small singular values can be obscured by RDM-estimation errors. Therefore, we retain the right singular space $ R_p$ of $ C_p$ above a threshold $t>0$,
    \begin{align}
        R_p :=\operatorname{span} \{\ket{f_b}:w_b>\theta,\ \sigma_b>t\}.
    \end{align}
    The selected space had dimension at most one in the example; here it can contain several dominant branches. Since different branches can have the same singular value, the singular-value decomposition (SVD) alone need not identify the individual branches. Thus, we apply the Gram-splitting procedure, guaranteed by Lem.~\ref{lem:main-exact-gram-splitting}, to $ R_p$ to separate them.
Once an individual branch is recovered, its weight and full block state are obtained from the same identities as in the example:
    \begin{align}
        w_b&=\bra{f_b}\Gamma_\Psi^{(p)}\ket{f_b},
        \qquad
        \ket{\omega_b}
        =\sqrt{w_b}\ket{f_b}
        +\frac{C_p\ket{f_b}}{\sqrt{w_b}}.
        \label{eq:overview-full-block-reconstruction}
    \end{align}
Equivalently, the complete-block restriction yields $\Gamma_\Psi^{(p)}\ket{f_b}=\sqrt{w_b}\ket{\omega_b}$. Thus, applying the full $p$-RDM to the recovered branch and normalizing its image also recovers the block. We use this form with estimated RDMs in Sec.~\ref{sec:equal-reconstruction}.

    \item \textbf{Recover the blocks without a dominant branch using $A_p$.}
    \label{step:overview-low}
    As $A_2$ preserved the second block in the example, $A_p$ preserves every block whose branches all lie outside $S_\theta$. If $w_b\leq\theta$, its complete-block contribution remains $\ket{\omega_b}\bra{\omega_b}$. If $w_b>\theta$, only $\ket{r_b}\bra{r_b}$ remains, with eigenvalue $1-w_b<1-\theta$. Since the surviving mix-block eigenvalues are products of branch weights at most $\theta$, they are also at most $\theta$. Consequently,
    \begin{align}
        W_p^{\mathrm H} :=\ker(A_p-I) =\operatorname{span}\{\ket{\omega_b}:w_b\leq\theta\}.
        \label{eq:overview-low-weight-block-space}
    \end{align}
    This eigenvalue-one space is separated from the remaining spectrum by a gap of at least $1-\theta\geq1/4$. Thus, we apply Gram splitting to $W_p^{\mathrm H}$ to recover the individual block states.

    \item \textbf{Approximate the excluded dominant blocks.}
    \label{step:overview-core}
    For an excluded dominant block, $\sigma_b\leq t$ implies $1-w_b=\sigma_b^2/w_b\leq t^2/\theta$. Thus, the block can be approximated by its dominant branch, with infidelity $1-w_b$.
The occupied directions of the excluded branches are still included in $S_\theta$. In the example, $S_\theta$ was precisely the occupied space of the excluded dominant branch. In general, we must remove the occupied spaces of the dominant branches already recovered in Step~\ref{step:overview-split}. The remaining space represents all excluded branches together, without grouping their modes into individual blocks.
Using the dominant branches recovered in Step~\ref{step:overview-split}, define
\begin{align}
    K&:=\sum_{b:w_b>\theta,\ \sigma_b>t}\Gamma_{f_b}^{(1)},
    \qquad F:=\operatorname{ran}K.
\end{align}
Each $\Gamma_{f_b}^{(1)}$ projects onto the occupied one-particle space of a recovered dominant branch. Since these spaces are mutually orthogonal, $K$ is the projector onto their combined occupied space $F$.
The space $S_\theta$ contains the occupied directions of both the recovered and excluded dominant branches. Removing $F$ therefore leaves exactly the occupied space of the excluded branches, which we define as
\begin{align}
    S_{\mathrm{core}}:=S_\theta\cap F^\perp.
\end{align}
Occupying every mode in an orthonormal basis of $S_{\mathrm{core}}$ yields the state $\ket{\sigma_{\mathrm{core}}}$, which represents the product of the excluded dominant branches up to a global phase. Use the vacuum when the core space is empty.

\item \textbf{Construct a block-product input and a passive Gaussian unitary.}
\label{step:overview-assemble}
\label{step:overview-preparation}
Let $\ket{\omega_1},\ldots,\ket{\omega_d}$ denote all full block states recovered in Steps~\ref{step:overview-split} and~\ref{step:overview-low}. We now find their branch modes and combine them with the core.

For $p=2$, use the pair decomposition from the guiding example. For $p\geq3$, an arbitrary eigenvector in a degenerate $1$-RDM eigenspace can mix branches without yielding a valid branch decomposition, so we first use a random contraction to separate their supports.

For each recovered block with $p\geq3$, sample $g\in\mathbb C^m$ with independent standard complex Gaussian components and form $\hat c[g]\ket{\omega_b}$. Each contracted branch occupies $p-1$ orthonormal modes within its original support. Since $p-1\geq2$, cross-branch terms vanish in the $1$-RDM of this unnormalized vector. Each contracted branch contributes its squared norm as an eigenvalue with multiplicity $p-1$. These eigenvalues are distinct with probability one, so diagonalizing this classically computed RDM identifies the occupied space $E_\ell$ of each contracted branch.

Choose an orthonormal basis $e_{\ell,1},\ldots,e_{\ell,p-1}$ of $E_\ell$ and compute
\begin{align}
    \ket{w_{b,\ell}}&:=\hat c[e_{\ell,p-1}]\cdots\hat c[e_{\ell,1}]\ket{\omega_b}.
    \label{eq:overview-branch-completion}
\end{align}
By branch orthogonality, this leaves a nonzero one-particle vector orthogonal to $E_\ell$. Normalize it to obtain the missing mode, and take the overlap of the resulting ordered branch state with $\ket{\omega_b}$ to recover its coefficient $a_{b,\ell}$, including its phase.

For any $p\geq2$, collect the modes of the $s_b$ recovered branches as columns of $V_b$, ordered by branch. In the corresponding local input order, set
\begin{align}
    \ket{\Xi_{b,\mathrm{in}}}&:=\sum_{\ell=1}^{s_b}a_{b,\ell}\ket{p(\ell-1)+1,\ldots,p\ell}_b.
\end{align}
As in the example, combine these columns with an orthonormal basis $V_0$ of the core and an orthonormal completion $V_{\mathrm{vac}}$. With $q:=\dim S_{\mathrm{core}}$ and $v:=m-q-p\sum_b s_b$, define
\begin{align}
    U_{\mathrm{out}}&:=\bigl(V_0,V_1,\ldots,V_d,V_{\mathrm{vac}}\bigr),\qquad \ket{\Xi_{\mathrm{in}}}:=\ket{1}^{\otimes q}\otimes\bigotimes_{b=1}^{d}\ket{\Xi_{b,\mathrm{in}}}\otimes\ket{0}^{\otimes v}.
    \label{eq:overview-preparation-input}
\end{align}
With probability one, the output $\ket{\widetilde\Psi}:=\hat U_{\mathrm{out}}\ket{\Xi_{\mathrm{in}}}$ equals
\begin{align}
    \ket{\Phi}:=\left(\prod_{j=1}^{d}\hat c^\dagger[\omega_j]\right)\ket{\sigma_{\mathrm{core}}},
    \label{eq:overview-numbered-output}
\end{align}
up to a global phase, where the factors are taken in a fixed order. Its target infidelity is at most the sum of $1-w_b$ over the excluded blocks, and hence at most $nt^2/\theta$. This construction introduces no additional approximation error or target copies.

\end{enumerate}

\section{Learning homogeneous blocks}
\label{sec:equal-reconstruction}

\begin{algorithm}[tbp]
\caption{Finite-copy learning of homogeneous blocks}
\label{alg:homogeneous-estimated}

\KwIn{Copies of an $m$-mode state $\ket{\Psi}$ in the homogeneous setting;
$p\geq2$ and target infidelity $\varepsilon_{\mathrm{fid}}$.}
\KwOut{A classical description of $\ket{\widetilde\Psi}$ satisfying
$\abs{\braket{\widetilde\Psi}{\Psi}}^2\geq1-\varepsilon_{\mathrm{fid}}$.}

\BlankLine
\tcp{Estimate the RDMs and separate the modes.}
$s_0\gets\sqrt{\varepsilon_{\mathrm{fid}}/(48m)}$;
choose $\mu$ by Eq.~\eqref{eq:equal-rdm-accuracy}\;

Obtain the RDM estimates $\widetilde\Gamma^{(1)}$ and
$\widetilde\Gamma^{(p)}$ using fermionic partial tomography
of Ref.~\cite{zhao2021fermionic}, each with operator-norm error
at most $\mu$\;

Choose $\theta\in[2/3,3/4]$ by Eq.~\eqref{eq:noisy-one-rdm-threshold};
let $\widetilde S_\theta$ be the spectral subspace of
$\widetilde\Gamma^{(1)}$ above $\theta$\;

Let $\widetilde Q_{\mathrm{hi}}$ and $\widetilde Q_{\mathrm{lo}}$
project onto $p$-particle states with all particles in
$\widetilde S_\theta$ and $\widetilde S_\theta^\perp$, respectively\;

$\widetilde C_p\gets
\widetilde Q_{\mathrm{lo}}\widetilde\Gamma^{(p)}\widetilde Q_{\mathrm{hi}},
\quad
\widetilde A_p\gets
\widetilde Q_{\mathrm{lo}}\widetilde\Gamma^{(p)}\widetilde Q_{\mathrm{lo}}$\;

\BlankLine
\tcp{Recover the branches and full blocks.}
Choose $t\in[s_0,2s_0]$ and the retained right singular space
$\widetilde R_p$ by Eq.~\eqref{eq:noisy-retained-right-singular-space}\;

Let $\widetilde W_p^{\mathrm H}$ be the spectral subspace of
$\widetilde A_p$ above $(1+\theta)/2$\;

$\widetilde d_p^{\mathrm N}\gets\dim\widetilde R_p$,
\quad
$\widetilde d_p^{\mathrm H}\gets\dim\widetilde W_p^{\mathrm H}$,
\quad
$D\gets \widetilde d_p^{\mathrm N}+\widetilde d_p^{\mathrm H}$\;

Apply robust Gram splitting (Lem.~\ref{lem:main-robust-gram-splitting})
to $\widetilde R_p$ and $\widetilde W_p^{\mathrm H}$;
store the respective outputs as
$\{\ket{\widetilde f_j}\}_{j=1}^{\widetilde d_p^{\mathrm N}}$ and
$\{\ket{\widetilde\omega_j}\}_{j=\widetilde d_p^{\mathrm N}+1}^D$\;

\For{$j=1,2,\dots,\widetilde d_p^{\mathrm N}$}{
    $\ket{\widetilde y_j}\gets
    \widetilde\Gamma^{(p)}\ket{\widetilde f_j}$,
    \quad
    $\ket{\widetilde\omega_j}\gets
    \frac{\ket{\widetilde y_j}}{\norm{\widetilde y_j}}$\;
}

\BlankLine
\tcp{Construct the core and prepare the output.}
$\widetilde K\gets\sum_{j=1}^{\widetilde d_p^{\mathrm N}}\Gamma_{\widetilde f_j}^{(1)}$;
let $\widetilde F$ be its spectral subspace above $1/2$\;

$\widetilde S_{\mathrm{core}}\gets
\widetilde S_\theta\cap\widetilde F^\perp$;
choose an orthonormal basis $V_0$ of $\widetilde S_{\mathrm{core}}$\;

$\eta\gets\mathsf A_p m^5\mu/s_0$, where $\mathsf A_p$ is the constant in
Prop.~\ref{prop:equal-p-input-accuracy}\;

$(\ket{\Xi_{\mathrm{in}}},\hat U_{\mathrm{out}}) \gets \textnormal{\textsc{Prepare}}\bigl(
V_0,\{(\ket{\widetilde\omega_b},p,\eta)\}_{b=1}^{D}
\bigr)$\;

\KwRet{$\ket{\widetilde\Psi}=\hat U_{\mathrm{out}}\ket{\Xi_{\mathrm{in}}}$}\;
\end{algorithm}

%
We now implement the reconstruction of Sec.~\ref{sec:homogeneous-overview} using RDM estimates obtained from finitely many copies. The steps remain the same, but the occupation and singular-value thresholds must be chosen to make subspace selection stable, and Gram splitting and core construction must account for errors in the recovered spaces. The \textsc{Prepare} procedure in Alg.~\ref{alg:prepare} combines branch recovery from App.~\ref{app:iu-estimated-branches} with the joint mode orthogonalization and input-and-unitary construction of App.~\ref{app:iu-input-and-unitary}.


\subsection{Reconstruction from estimated RDMs}
\label{subsec:one-rdm-equal}
\label{subsec:p-rdm-equal}

We first specify how accurately the RDMs must be estimated to achieve the target infidelity $\varepsilon_{\mathrm{fid}}$. Our stability analysis shows that it is sufficient to obtain Hermitian estimates $\widetilde{\Gamma}^{(1)}$ and $\widetilde{\Gamma}^{(p)}$ satisfying
\begin{align}
    \max\!\left\{ \norm{ \widetilde\Gamma^{(1)} - \Gamma_{\Psi}^{(1)} }, \norm{ \widetilde\Gamma^{(p)} - \Gamma_{\Psi}^{(p)} } \right\} \leq \mu,
    \qquad \mu \leq c_p \frac{ \varepsilon_{\mathrm{fid}} }{ m^{\lceil(p+13)/2\rceil} }, 
    \label{eq:equal-rdm-accuracy}
\end{align}
where $c_p>0$ is a sufficiently small constant depending only on the fixed particle number $p$. Prop.~\ref{prop:equal-p-input-accuracy} establishes this sufficient accuracy. For fixed $p$, the required RDM accuracy is inverse-polynomial in $m$.
The operator-norm guarantee in Eq.~\eqref{eq:equal-rdm-accuracy} can be achieved simultaneously for both RDMs with probability at least $1-\delta/3$ using
\begin{align}
    O_p
    \left(
        m^{3p+2\lceil(p+13)/2\rceil}
        \varepsilon_{\mathrm{fid}}^{-2}
        \log\frac{2m}{\delta}
    \right)
    \label{eq:eps-two-copy-scaling}
\end{align}
copies of the target state. The conversion from the required RDM accuracy to this sample bound is derived in App.~\ref{app:rdm-estimation}.

The algorithm proceeds as follows.

\begin{enumerate}
\setcounter{enumi}{-1}
\item \textbf{Estimate the RDMs.}
\label{step:equal-estimate}
Using independent copies of $\ket{\Psi}$, estimate $\Gamma_{\Psi}^{(1)}$ and $\Gamma_{\Psi}^{(p)}$ to the accuracy in Eq.~\eqref{eq:equal-rdm-accuracy} with the fermionic partial-tomography protocol of Ref.~\cite{zhao2021fermionic}, based on classical shadows~\cite{huang2020classical}. Allocate failure probability $\delta/6$ to each estimate, so that both satisfy the required accuracy with probability at least $1-\delta/3$.

\item \textbf{Separate the modes using the $1$-RDM estimate.}
\label{step:equal-high}
In Step~\ref{step:overview-high}, we used an occupation threshold $\theta$ to define the high-occupation space. With estimated RDMs, an eigenvalue close to $\theta$ can cross the threshold and change the dimension of the selected space. We therefore choose $\theta$ away from the estimated spectrum. Writing $\widetilde\Gamma^{(1)}=\sum_{j=1}^{m}\widetilde\lambda_j\ket{\widetilde v_j}\bra{\widetilde v_j}$, choose $\theta\in[2/3,3/4]$ such that
\begin{align}
    \min_{j\in[m]}\abs{\theta-\widetilde\lambda_j}
    &\geq \frac{1}{192m}.
    \label{eq:noisy-one-rdm-threshold}
\end{align}
Such a threshold always exists because the intervals of radius $1/(192m)$ around the $m$ estimated eigenvalues have total length at most $1/96$, which is smaller than the length $1/12$ of $[2/3,3/4]$.

We construct the high-occupation subspace from the estimate of $1$-RDM, denoted as $\widetilde S_\theta:=\operatorname{span}\{\ket{\widetilde v_j}:\widetilde\lambda_j>\theta\}$. Let $S_\theta$ denote the subspace obtained from the exact $1$-RDM using the same threshold. Under the accuracy condition in Eq.~\eqref{eq:equal-rdm-accuracy}, these subspaces have the same dimension, and spectral perturbation bounds imply
\begin{align}
    \norm{P_{\widetilde S_\theta}-P_{S_\theta}} &=O(m\mu).
    \label{eq:noisy-high-space-error-main}
\end{align}
Following Step~\ref{step:overview-high}, we define $\widetilde Q_{\mathrm{hi}}$ as the projector onto the $p$-particle states whose particles all lie in $\widetilde S_\theta$, and $\widetilde Q_{\mathrm{lo}}$ as the projector onto those whose particles all lie in $\widetilde S_\theta^\perp$. Then, the error in $\widetilde S_\theta$ bounds the errors in both projectors,
\begin{align}
    \max\left\{
        \norm{\widetilde Q_{\mathrm{hi}}-Q_{\mathrm{hi}}},
        \norm{\widetilde Q_{\mathrm{lo}}-Q_{\mathrm{lo}}}
    \right\} &=O(pm\mu).
    \label{eq:noisy-p-particle-projector-error}
\end{align}

Next, we use these projectors to compress the estimated $p$-RDM. Specifically, we define
\begin{align}
\widetilde C_p
&:=\widetilde Q_{\mathrm{lo}}\widetilde\Gamma^{(p)}\widetilde Q_{\mathrm{hi}},
\qquad
\widetilde A_p
:=\widetilde Q_{\mathrm{lo}}\widetilde\Gamma^{(p)}\widetilde Q_{\mathrm{lo}}.
\label{eq:noisy-p-rdm-compressions}
\end{align}
As before, $\widetilde C_p$ is used for blocks with a dominant branch, while $\widetilde A_p$ is used for blocks without one. Both the RDM-estimation error and the projector errors affect these operators, but their combined effect satisfies
\begin{align}
    \max\left\{
        \norm{\widetilde C_p-C_p},
        \norm{\widetilde A_p-A_p}
    \right\} &=O(pm\mu),
    \label{eq:noisy-p-rdm-compression-error}
\end{align}
where $C_p$ and $A_p$ are the exact compressions defined in Eq.~\eqref{eq:overview-two-compressions}, constructed using the same threshold $\theta$.

\item \textbf{Recover the retained dominant blocks using $\widetilde C_p$.}
\label{step:equal-split}
In Sec.~\ref{sec:homogeneous-overview}, the retained dominant branches are recovered from the right singular space of $C_p$ above a threshold $t$. With the estimate $\widetilde C_p$, we choose $t$ small enough that replacing the excluded blocks by their dominant branches incurs only a small total error. At the same time, we keep $t$ sufficiently far from the estimated singular values to prevent estimation errors from changing which singular values are retained.
Set $s_0:=\sqrt{\varepsilon_{\mathrm{fid}}/(48m)}$, and let $\widetilde\sigma_j$ and $\ket{\widetilde g_j}$ be the singular values and corresponding right singular vectors of $\widetilde C_p$. Choose $t\in[s_0,2s_0]$ such that
\begin{align}
    \min_j\abs{t-\widetilde\sigma_j}\geq\frac{s_0}{16m},
    \qquad \widetilde R_p:=\operatorname{span}\{\ket{\widetilde g_j}:\widetilde\sigma_j>t\}.
\label{eq:noisy-retained-right-singular-space}
\end{align}
The exact $C_p$ has rank at most the number of dominant blocks, which is at most $m$. Under Eq.~\eqref{eq:equal-rdm-accuracy}, all remaining estimated singular values lie below $s_0/2$, so only at most $m$ values can constrain the choice of $t$ in $[s_0,2s_0]$; their excluded intervals have total length at most $s_0/8$. Thus, such a threshold exists. The separation condition and the compression-error bound in Eq.~\eqref{eq:noisy-p-rdm-compression-error} ensure that $\widetilde R_p$ has the same dimension as the exact space $R_p$ constructed with the same threshold $t$ and remains close to it. 

Next, we recover the retained dominant branches by applying Gram splitting to $\widetilde R_p$. However, small eigenvalue gaps in the Gram-splitting matrix, defined in Eq.~\eqref{eq:main-gram-splitting-matrix}, can make its eigenvectors sensitive to errors in $\widetilde R_p$. Therefore, we add an additional step to check whether the observed gaps are sufficiently large relative to the error in this matrix. Otherwise, we repeat the random trial. Every accepted trial yields normalized vectors $\ket{\widetilde f_j}$ approximating the retained dominant branches, up to phases and relabeling. The gap check and repetition procedure are described in Sec.~\ref{subsec:block-separation-assembly}.

We now reconstruct the full block states from the recovered dominant branches. In the exact setting, $\Gamma_\Psi^{(p)}\ket{f_b}=\sqrt{w_b}\ket{\omega_b}$, as noted after Eq.~\eqref{eq:overview-full-block-reconstruction}. We therefore apply the estimated $p$-RDM to each recovered branch and normalize its image:
\begin{align}
    \ket{\widetilde y_j}
    &:=\widetilde\Gamma^{(p)}\ket{\widetilde f_j},
    \qquad
    \ket{\widetilde\omega_j}
    :=\frac{\ket{\widetilde y_j}}{\norm{\widetilde y_j}}.
    \label{eq:noisy-block-reconstruction}
\end{align}
If $\ket{\widetilde y_j}=0$, the algorithm returns \textsc{Fail}. Under Eq.~\eqref{eq:equal-rdm-accuracy} and a successful Gram-splitting call, Lem.~\ref{lem:df-near-block-stability} gives $\norm{\widetilde y_j}>\sqrt{\theta}/2$. Since $\theta\geq2/3$, normalization amplifies the error by at most a constant factor. The same lemma bounds the reconstructed block error by a constant multiple of the sum of the RDM-estimation and recovered-branch errors. The dominant branches $\ket{\widetilde f_j}$ are kept separately for the core construction.

\item \textbf{Recover the blocks without a dominant branch using $\widetilde A_p$.}
\label{step:equal-low}
For blocks with $w_b\leq\theta$, we instead use $\widetilde A_p$. In the exact setting, their joint span is the eigenvalue-one space of $A_p$, while every remaining eigenvalue is at most $\theta$. Since estimation errors can shift the desired eigenvalues away from one, we select them using the midpoint of this gap. Let $\widetilde W_p^{\mathrm H}$ be the spectral subspace of $\widetilde A_p$ corresponding to eigenvalues above $(1+\theta)/2$.
Because $1-\theta\geq1/4$, the error bound in Eq.~\eqref{eq:noisy-p-rdm-compression-error} ensures that this subspace remains close to $\operatorname{span}\{\ket{\omega_b}:w_b\leq\theta\}$. Then, we apply the modified Gram-splitting procedure to $\widetilde W_p^{\mathrm H}$, which approximately recovers these block states.

\item \textbf{Approximate the excluded dominant blocks.}
\label{step:equal-core}
Having recovered the retained dominant blocks and the complementary blocks, we now construct the remaining occupied core. As in Sec.~\ref{sec:homogeneous-overview}, this core represents the excluded dominant branches together through their occupied one-particle subspace. Because small errors in the recovered branches can enlarge their one-particle supports, we estimate the combined occupied support of the retained branches by a spectral threshold before removing it from $\widetilde S_\theta$.

Since the $1$-RDMs of the recovered dominant branches $\ket{\widetilde f_j}$ approximate the individual occupied-space projectors used in Step~\ref{step:overview-core}, their sum approximates the projector onto the combined support of the retained dominant branches. Let $\widetilde F$ be the span of the eigenvectors of this sum with eigenvalues above $1/2$. Every recovered branch lies in $\wedge^p\widetilde S_\theta$, so $\widetilde F\subseteq\widetilde S_\theta$. The remaining occupied subspace is therefore
\begin{align}
    \widetilde S_{\mathrm{core}}
    &:=\widetilde S_\theta\cap\widetilde F^\perp,
    \qquad
    P_{\widetilde S_{\mathrm{core}}}
    =P_{\widetilde S_\theta}-P_{\widetilde F}.
    \label{eq:main-noisy-core-projector}
\end{align}
The sum of the branch $1$-RDMs is independent of the orthonormal basis chosen for $\widetilde R_p$. Its accuracy, and hence the core error, is therefore controlled directly by the error in this space before Gram splitting, as formalized in Prop.~\ref{prop:df-core-stability}. Let $\ket{\widetilde\sigma_{\mathrm{core}}}$ be the state obtained by occupying every mode of an orthonormal basis of $\widetilde S_{\mathrm{core}}$, using the vacuum when this space is empty.

The error associated with this construction comes from replacing the excluded blocks by their dominant branches and estimating their combined occupied subspace. Our choice $t\leq2s_0$ bounds the infidelity from the first replacement by $\varepsilon_{\mathrm{fid}}/16$. 

\item \textbf{Construct a block-product input and a passive Gaussian unitary.}
\label{step:equal-assemble}
\label{step:equal-preparation}
Let $\ket{\widetilde\omega_1},\ldots,\ket{\widetilde\omega_D}$ denote all block states reconstructed in Steps~\ref{step:equal-split} and~\ref{step:equal-low}. With estimated RDMs, these blocks need not admit the required branch decomposition, and different blocks need not have mutually orthogonal one-particle supports. We first approximate each block by a superposition of branches on orthogonal modes, then make the modes from all blocks and the core jointly orthogonal. Controlling the branch counts ensures that these modes fit within the available $m$ modes.

Under Eq.~\eqref{eq:equal-rdm-accuracy}, successful Gram splitting gives, by Prop.~\ref{prop:equal-p-input-accuracy},
\begin{align}
    \min_{\phi\in\mathbb R}\left\|\ket{\widetilde\omega_b}-e^{i\phi}\ket{\omega_b}\right\|&\leq\eta,\qquad \eta:=\mathsf A_pm^5\frac{\mu}{s_0},
\end{align}
where $\mathsf A_p$ is the constant in Prop.~\ref{prop:equal-p-input-accuracy}. This bound sets the local approximation tolerances. For $p=2$, use the pair decomposition from the guiding example and retain its largest coefficients until the normalized partial sum is within $2\eta$ of the estimate. This uses no more pairs than the exact block and gives error at most $3\eta$, up to phase. For $p\geq3$, stabilize the random-contraction construction of Step~\ref{step:overview-preparation} by grouping nearby spectral values, discarding weak groups, and completing the retained branches. A search over increasing branch-count bounds accepts a candidate when its distance from the estimated block meets the corresponding tolerance. App.~\ref{app:iu-estimated-branches} specifies these tests and proves that the searches succeed without exceeding the exact branch counts, with combined failure probability at most $\delta/3$ when each is assigned budget $\delta/(3\max\{1,D\})$.

Finally, combine the recovered branch modes and estimated core modes, rejecting if their total number exceeds $m$. Make them jointly orthonormal while minimizing the weighted sum of squared changes: use weight $1$ for each core mode and $|a_{b,\ell}|^2$ for each mode of branch $\ell$ in block $b$, with the recovered coefficients normalized within each block. These occupation weights limit the effect of poorly determined modes in weak branches. Keeping the coefficients fixed and occupying every core mode defines $\ket{\Xi_{\mathrm{in}}}$; the adjusted columns and an orthonormal completion define $U_{\mathrm{out}}$, with the additional input modes empty. The final estimate is
\begin{align}
    \ket{\widetilde\Psi}&:=\hat U_{\mathrm{out}}\ket{\Xi_{\mathrm{in}}}.
\end{align}
Let $\ket{\Phi}$ be the exact reconstructed state in Eq.~\eqref{eq:overview-numbered-output}, using the same thresholds $\theta,t$. Prop.~\ref{prop:iu-postprocessing}, together with the block and core error bounds, gives on the joint success event
\begin{align}
    \sqrt{1-|\langle\Phi|\widetilde\Psi\rangle|^2}&\leq\mathsf C_pm^{15/2}\frac{\mu}{s_0}\leq\frac{\sqrt{\varepsilon_{\mathrm{fid}}}}{4},
\end{align}
where $\mathsf C_p$ depends only on $p$, and Eq.~\eqref{eq:equal-rdm-accuracy} ensures the second inequality for sufficiently small $c_p$. Combining this with the truncation bound gives target infidelity at most $\varepsilon_{\mathrm{fid}}$. All processing uses the recovered classical data, with no additional target copies.

\end{enumerate}

\subsection{Gram splitting with an estimated subspace}
\label{subsec:block-separation-assembly}

The reconstruction above applies Gram splitting to $\widetilde R_p$ and $\widetilde W_p^{\mathrm H}$. Unlike the exact setting of Lem.~\ref{lem:main-exact-gram-splitting}, an estimated subspace perturbs the splitting matrix, and small eigenvalue gaps can make its eigenvectors unstable. We therefore use the same matrix construction on the estimated space, but accept a trial only when its observed minimum eigenvalue gap exceeds a threshold determined by the supplied subspace-error bound. Otherwise, we sample a new random vector and repeat. The resulting guarantee is as follows.

\begin{lemma}[Robust Gram splitting]\label{lem:main-robust-gram-splitting}
Fix $k\geq2$. There are constants $c_k,C_k>0$, depending only on $k$, with the following property. Let $T\subseteq\wedge^k\mathbb C^m$ be spanned by normalized states $\ket{\psi_1},\ldots,\ket{\psi_d}$ with mutually orthogonal one-particle supports. Given an orthonormal basis of $\widetilde T$, a bound $\eta$ satisfying
\begin{align}
    \dim\widetilde T=\dim T=d,\qquad
    \norm{P_{\widetilde T}-P_T}\leq\eta\leq c_km^{-3},
    \label{eq:main-robust-gram-input}
\end{align}
and a failure budget $\beta_{\mathrm{loc}}\in(0,1)$, the certified Gram-splitting algorithm returns either \textnormal{\textsc{Fail}} or an orthonormal basis $\ket{\widetilde\psi_1},\ldots,\ket{\widetilde\psi_d}$ of $\widetilde T$. Every nonempty returned list satisfies, for some permutation $\pi$ and phases $\varphi_j$,
\begin{align}
    \max_{j\in[d]}\norm{\ket{\widetilde\psi_j}-e^{i\varphi_j}\ket{\psi_{\pi(j)}}}
    \leq C_km^3\eta.
    \label{eq:main-robust-gram-output}
\end{align}
The probability of returning \textnormal{\textsc{Fail}} is at most $\beta_{\mathrm{loc}}$. For fixed $k$, the algorithm uses polynomially many arithmetic operations in $m$ and $\log(1/\beta_{\mathrm{loc}})$ and requires no additional copies of the target state.
\end{lemma}

Prop.~\ref{prop:robust-gram-splitting} proves this guarantee with $c_k=1/(64k)$ and $C_k=8k+2$. Its constructive proof in App.~\ref{app:certified-gram-splitting-subroutine} specifies the acceptance rule and computes all transition RDMs from the supplied classical vectors.

Under Eq.~\eqref{eq:equal-rdm-accuracy}, the subspace-error bounds for both calls satisfy the hypotheses of Lem.~\ref{lem:main-robust-gram-splitting}, as shown in App.~\ref{app:homogeneous-noisy}. In particular, Prop.~\ref{prop:equal-p-input-accuracy} bounds the error in each recovered dominant branch or complementary block by $O_p(m^5\mu/s_0)$, up to phases and relabeling. Assigning failure probability $\delta/6$ to each Gram-splitting call bounds their total failure probability by $\delta/3$. RDM estimation and the construction in Step~\ref{step:equal-preparation} of Sec.~\ref{subsec:one-rdm-equal} each contribute at most $\delta/3$, giving overall success probability at least $1-\delta$. Each call requires $O(\log(4/\delta))$ trials.

\subsection{Classical post-processing time}
\label{sec:homogeneous-time-complexity}

We account separately for the classical cost of constructing the RDM estimates and reconstructing the state from them. 
Let $D_p=\binom{m}{p}$ be the dimension of the $p$-particle space. As detailed in App.~\ref{app:rdm-estimation}, the protocol of Ref.~\cite{zhao2021fermionic} can estimate all entries of the $p$-RDM to additive error $\varepsilon_{\mathrm{entry}}$, with failure probability at most $\delta/6$, in $O_p(D_p^2\varepsilon_{\mathrm{entry}}^{-2}\log(2m/\delta))$ classical time. Since the $p$-RDM is a $D_p\times D_p$ matrix, the operator-norm error of its estimate is at most $D_p$ times the maximum entrywise error. Therefore, $\varepsilon_{\mathrm{entry}}=\mu/D_p$ ensures an
operator-norm error of at most $\mu$ required in
Eq.~\eqref{eq:equal-rdm-accuracy}. Moreover, estimation of $1$-RDM with the same operator-norm accuracy and failure budget $\delta/6$ does not change the asymptotic cost. Thus, the total classical cost of constructing the required RDM estimates is
\begin{align}
     O_p\left( D_p^4\frac{\log(2m/\delta)}{\mu^2}\right) = O_p\left( m^{4p}\frac{\log(2m/\delta)}{\mu^2} \right).
    \label{eq:equal-rdm-estimation-runtime}
\end{align}

Once these estimates are available, the dense matrix operations on the RDMs, including the construction of the projectors, compressions, and their spectral decompositions, cost $O(D_p^3)=O_p(m^{3p})$. Furthermore, the construction of the remaining occupied subspace and the output factor list is also included in this bound.
For each of the two Gram-splitting calls, we precompute the transition $1$-RDMs once and reuse them across all trials. This costs $O_p(m^{p+3})$, which is included in $O_p(m^{3p})$ for $p\geq2$. Using these stored matrices, each trial costs $O_p(m^4)$, and we allow $O(\log(4/\delta))$ trials per call.

The construction in Step~\ref{step:equal-preparation} of Sec.~\ref{subsec:one-rdm-equal} costs $O_p(m^{p+1})$ per trial for $p\geq3$, including branch completion, coefficient evaluation, and orthogonalization. There are $O(m)$ reconstructed blocks, each testing $O(m)$ branch-count bounds with $O(\log(2m/\delta))$ trials per bound. Thus, the total cost is $O_p(m^{p+3}\log(2m/\delta))$. This bound also covers the deterministic $p=2$ construction and the final joint orthogonalization and vacuum completion. At the prescribed accuracy, it is dominated by the RDM-estimation cost in Eq.~\eqref{eq:equal-rdm-estimation-runtime}.

Combining the RDM-estimation cost with the reconstruction cost yields
\begin{align}
    O_p\left( m^{4p}\frac{\log(2m/\delta)}{\mu^2} +m^{3p} +m^4\log\frac{4}{\delta}  \right).
    \label{eq:equal-p-total-cost-before-substitution}
\end{align}
Finally, choose $\mu$ as the upper bound in Eq.~\eqref{eq:equal-rdm-accuracy}. Then, the RDM-estimation term dominates, and we write the total classical post-processing time
\begin{align}
    O_p\left( m^{4p+2\lceil(p+13)/2\rceil} \varepsilon_{\mathrm{fid}}^{-2} \log\frac{2m}{\delta}\right).
    \label{eq:equal-p-classical-runtime}
\end{align}
This implies the classical time bound in Thm.~\ref{thm:equal-blocks}, including both RDM estimation and state reconstruction.

\section{Learning heterogeneous blocks}
\label{sec:bounded-reconstruction}

\begin{algorithm}[tbp]
\caption{Finite-copy learning of heterogeneous blocks}
\label{alg:heterogeneous-estimated}

\KwIn{Copies of an $m$-mode state $\ket{\Psi}$ in the heterogeneous setting;
$2\leq r\leq m$ and target infidelity $\varepsilon_{\mathrm{fid}}$.
Individual block particle numbers are unknown.}
\KwOut{A classical description of $\ket{\widetilde\Psi}$ satisfying
$\abs{\braket{\widetilde\Psi}{\Psi}}^2\geq1-\varepsilon_{\mathrm{fid}}$.}

\BlankLine
\tcp{Estimate the RDMs and separate the modes.}
$s_0\gets\sqrt{\varepsilon_{\mathrm{fid}}/(48m)}$;
choose $\nu$ by Eq.~\eqref{eq:bounded-rdm-accuracy-main}\;

Obtain the RDM estimates
$\widetilde\Gamma^{(1)},\ldots,\widetilde\Gamma^{(r)}$
using fermionic partial tomography of Ref.~\cite{zhao2021fermionic}
from the same measurement records,
each with operator-norm error at most $\nu$\;

Choose $\theta\in[2/3,3/4]$ by Eq.~\eqref{eq:noisy-one-rdm-threshold};
let $\widetilde S_\theta$ be the spectral subspace of
$\widetilde\Gamma^{(1)}$ above $\theta$\;

\BlankLine
\tcp{Recover the branches and full blocks.}
Initialize empty lists $\mathcal F$, $\mathcal H$, and $\mathcal W$
for dominant branches, complementary blocks, and all full blocks,
respectively\;

\For{$k=2,\ldots,r$}{
    Let $\widetilde Q_{\mathrm{hi},k}$ and $\widetilde Q_{\mathrm{lo},k}$
    project onto $k$-particle states with all particles in
    $\widetilde S_\theta$ and $\widetilde S_\theta^\perp$, respectively\;

    $\widetilde C_k\gets
    \widetilde Q_{\mathrm{lo},k}\widetilde\Gamma^{(k)}
    \widetilde Q_{\mathrm{hi},k}$,
    \quad
    $\widetilde A_k\gets
    \widetilde Q_{\mathrm{lo},k}\widetilde\Gamma^{(k)}
    \widetilde Q_{\mathrm{lo},k}$\;

    Choose $t_k\in[s_0,2s_0]$
    by Eq.~\eqref{eq:bounded-singular-threshold-main};
    let $\widetilde R_k$ be the right singular space of
    $\widetilde C_k$ above $t_k$\;

    Let $\widetilde W_k^{\mathrm H}$ be the spectral subspace of
    $\widetilde A_k$ above $(1+\theta)/2$\;

    Form $\hat O_k^{\mathrm N}$ and $\hat O_k^{\mathrm H}$
    from factors in $\mathcal F$ and $\mathcal H$, respectively:
    span products of at least two distinct factors of total
    particle number $k$, using creation polynomials in stored order;
    use the zero space if no such product exists\;

    Form $\widetilde B_k^{\mathrm N}$ and $\widetilde B_k^{\mathrm H}$
    by Eq.~\eqref{eq:bounded-overlap-operators};
    let $\widetilde F_k^{\mathrm N}$ and $\widetilde T_k^{\mathrm H}$
    be their respective spectral subspaces above $1/2$\;

    $\widetilde d_k^{\mathrm N}\gets\dim\widetilde F_k^{\mathrm N}$,
    \quad
    $\widetilde d_k^{\mathrm H}\gets\dim\widetilde T_k^{\mathrm H}$\;

    Apply robust Gram splitting
    (Lem.~\ref{lem:main-robust-gram-splitting})
    to $\widetilde F_k^{\mathrm N}$ and $\widetilde T_k^{\mathrm H}$
    with the respective error bounds in
    Eqs.~\eqref{eq:heterogeneous-near-recursive-certificates}
    and~\eqref{eq:heterogeneous-strong-recursive-certificates};
    store the respective outputs as
    $\{\ket{\widetilde f_{k,j}}\}_{j=1}^{\widetilde d_k^{\mathrm N}}$
    and
    $\{\ket{\widetilde\omega_{k,j}^{\mathrm H}}\}_{j=1}^{\widetilde d_k^{\mathrm H}}$\;

    \For{$j=1,2,\dots,\widetilde d_k^{\mathrm N}$}{
        $\ket{\widetilde y_{k,j}}\gets
        \widetilde\Gamma^{(k)}\ket{\widetilde f_{k,j}}$,
        \quad
        $\ket{\widetilde\omega_{k,j}^{\mathrm N}}\gets
        \frac{\ket{\widetilde y_{k,j}}}{\norm{\widetilde y_{k,j}}}$\;
    }

    Compute $\xi_{f,k}$, $\xi_{\mathrm N,k}$, and $\xi_{\mathrm H,k}$
    by Eq.~\eqref{eq:df-factor-certificates}\;

    Append
    $\{(\ket{\widetilde f_{k,j}},k,\xi_{f,k})\}_{j=1}^{\widetilde d_k^{\mathrm N}}$
    to $\mathcal F$\;

    Append
    $\{(\ket{\widetilde\omega_{k,j}^{\mathrm H}},k,\xi_{\mathrm H,k})
    \}_{j=1}^{\widetilde d_k^{\mathrm H}}$
    to $\mathcal H$\;

    Append
    $\{(\ket{\widetilde\omega_{k,j}^{\mathrm N}},k,\xi_{\mathrm N,k})
    \}_{j=1}^{\widetilde d_k^{\mathrm N}}$
    and
    $\{(\ket{\widetilde\omega_{k,j}^{\mathrm H}},k,\xi_{\mathrm H,k})
    \}_{j=1}^{\widetilde d_k^{\mathrm H}}$
    to $\mathcal W$\;
}

\BlankLine
\tcp{Construct the core and prepare the output.}
$\widetilde K\gets
\sum_{k=2}^{r}\sum_{j=1}^{\widetilde d_k^{\mathrm N}}
\Gamma_{\widetilde f_{k,j}}^{(1)}$;
let $\widetilde F$ be its spectral subspace above $1/2$\;

$\widetilde S_{\mathrm{core}}\gets
\widetilde S_\theta\cap\widetilde F^\perp$;
choose an orthonormal basis $V_0$ of $\widetilde S_{\mathrm{core}}$\;

$D\gets|\mathcal W|$;
relabel its entries as
$\{(\ket{\widetilde\omega_b},p_b,\eta_b)\}_{b=1}^{D}$\;

$(\ket{\Xi_{\mathrm{in}}},\hat U_{\mathrm{out}})
\gets\textnormal{\textsc{Prepare}}\bigl(
V_0,\{(\ket{\widetilde\omega_b},p_b,\eta_b)\}_{b=1}^{D}
\bigr)$\;

\KwRet{$\ket{\widetilde\Psi}
=\hat U_{\mathrm{out}}\ket{\Xi_{\mathrm{in}}}$}\;
\end{algorithm}

We now consider the heterogeneous setting of Thm.~\ref{thm:heterogeneous-blocks}, with unknown block particle numbers $2\leq p_b\leq r$ for a known constant $r$. Unlike the homogeneous case, the candidate spaces can contain products of smaller blocks. We therefore process the RDMs in increasing order and remove product directions formed from previously recovered factors before applying Gram splitting, accounting for propagated errors when using estimated RDMs.

Always-occupied and vacuum modes are also allowed; the former enter the final core, while the latter contribute only zero sectors to the exact RDMs. We first describe the exact recursion and then its implementation with RDM estimates.

\subsection{Reconstruction with exact RDMs}
\label{subsec:bounded-order-k-recovery}


\begin{enumerate}
\setcounter{enumi}{-1}
\item \textbf{Block structure of the $k$-RDM.}
\label{step:bounded-exact-structure}
As in the homogeneous setting, the $k$-RDM separates into sectors according to how many particles are selected from each block and the always-occupied component. Each sector is given by the tensor product of the corresponding component RDMs.

In the homogeneous $p$-RDM, selecting particles from several blocks necessarily selects only part of each block, so the mix-block contribution has no coherence between distinct branches. With different block particle numbers, however, several complete blocks can together supply exactly $k$ particles. For example, let $a$ and $b$ be two two-particle blocks. Identifying the sector that selects two particles from each with the tensor product of their two-particle spaces gives
\begin{align}
    \left.\Gamma_\Psi^{(4)}\right|_{\text{two particles from each of }a,b}
    &\cong \Gamma_{\omega_a}^{(2)}\otimes\Gamma_{\omega_b}^{(2)}
    =\ket{\omega_a}\bra{\omega_a}\otimes\ket{\omega_b}\bra{\omega_b}.
\end{align}
This sector preserves the full coherence of both blocks and has eigenvalue one on their product state, just as a complete four-particle block contributes eigenvalue one. Thus, the eigenvalue alone does not distinguish a new block from a product of smaller blocks. We account for these product contributions in the recursive reconstruction below.

\item \textbf{Separate the modes using the $1$-RDM.}
\label{step:bounded-exact-high}
As in the homogeneous setting, construct $S_\theta$ from the $1$-RDM with $\theta\in[2/3,3/4]$, and use it at every order. Here, $S_\theta$ also includes the transformed always-occupied modes, with occupation one, while the transformed vacuum modes have occupation zero and lie in $S_\theta^\perp$.

\item \textbf{Recover the blocks recursively.}
\label{step:bounded-exact-recursion}
We process the RDMs in increasing order, $k=2,\ldots,r$, using the common subspace $S_\theta$. As in the homogeneous setting, a singular-value threshold determines which dominant branches are recovered individually. To coordinate these choices across orders, we set $s_0:=\sqrt{\varepsilon_{\mathrm{fid}}/(48m)}$ and choose each threshold $t_k\in[s_0,2s_0]$. This common interval ensures that previously recovered factors suffice for product removal, as explained in step~(b) below. At each order, we perform the following three steps.

\begin{enumerate}
\renewcommand{\theenumii}{(\alph{enumii})}
\renewcommand{\labelenumii}{\theenumii}
\item \textbf{Extract the candidate spaces.}
\label{step:bounded-exact-extract}
Let $Q_{\mathrm{hi},k}$ and $Q_{\mathrm{lo},k}$ project onto states whose $k$ particles lie entirely in $S_\theta$ and $S_\theta^\perp$, respectively. Form $C_k$ and $A_k$ using Eq.~\eqref{eq:overview-two-compressions} with $p$ replaced by $k$. Since the transformed modes corresponding to $I_0$ lie in $S_\theta$, the projector $Q_{\mathrm{lo},k}$ removes contributions involving these modes from both compressions. Sectors containing transformed vacuum modes contribute only zeros.

As in the homogeneous case, $C_k$ retains coherence between dominant branches and the remaining branches of complete blocks, while contributions selecting only part of a block vanish. Here, several smaller blocks can also contribute all their particles with total particle number $k$. Thus, the right singular space of $C_k$ is spanned by dominant branches of individual $k$-particle blocks and products of dominant branches of smaller blocks. To retain only directions with sufficiently large coherence, let $R_k$ be the span of the right singular vectors of $C_k$ with singular values larger than $t_k$.

For $A_k$, each complete $k$-particle block without a dominant branch contributes its full state with eigenvalue one. Products of smaller complete blocks without a dominant branch also contribute eigenvalue one when their particle numbers sum to $k$. Let $W_k^{\mathrm H}$ denote this eigenvalue-one space. 

Each remaining sector contains either a low-space contribution from a dominant block, with norm at most $1-w_b<1-\theta\leq\theta$, or a proper lower-order RDM of a block without a dominant branch, with norm at most $\theta$. Since each sector is a tensor product and all other factors have norm at most one, its eigenvalues are at most $\theta$, leaving a spectral gap of at least $1-\theta$.

\item \textbf{Remove products of previously recovered factors.}
\label{step:bounded-exact-products}
We now use factors recovered at lower orders to construct the product directions that must be removed from the candidate spaces. For $R_k$, we form products of recovered dominant branches; for $W_k^{\mathrm H}$, we form products of recovered full block states without a dominant branch. In each case, we combine at least two distinct factors whose particle numbers sum to $k$, using the creation polynomials as in Eq.~\eqref{eq:overview-numbered-output}. Let $O_k^{\mathrm N}$ and $O_k^{\mathrm H}$ denote the spans of these two families of products, respectively. For $k=2,3$, no such products exist because every factor contains at least two particles, so both spaces are zero.

For this removal to work, every product direction retained in $R_k$ must be generated by branches already recovered at lower orders. This is ensured by choosing all thresholds from $[s_0,2s_0]$. Indeed, a product of dominant branches has a singular value equal to the product of their individual singular values $\sigma_b=\sqrt{w_b(1-w_b)}$. If one branch was excluded at a lower order, its singular value is at most $2s_0$. Since each additional branch has singular value at most $1/2$, the product has singular value at most $s_0\leq t_k$ and is also excluded from $R_k$.

A new $k$-particle block has a one-particle support orthogonal to those of all smaller blocks. Its dominant branch or full block state is therefore orthogonal to every product formed from the previously recovered factors. Thus, removing the product directions preserves the new factors. Define $F_k^{\mathrm N}:=R_k\cap(O_k^{\mathrm N})^\perp$ and $T_k^{\mathrm H}:=W_k^{\mathrm H}\cap(O_k^{\mathrm H})^\perp$.

Some products in $O_k^{\mathrm N}$ may have singular values at or below $t_k$ and therefore lie outside $R_k$. The part removed from $R_k$ is consequently $R_k\cap O_k^{\mathrm N}$. In contrast, every product in $O_k^{\mathrm H}$ lies in $W_k^{\mathrm H}$. Hence, the block structure yields
\begin{align}
    R_k &=(R_k\cap O_k^{\mathrm N})\oplus F_k^{\mathrm N},
    \qquad W_k^{\mathrm H}=O_k^{\mathrm H}\oplus T_k^{\mathrm H},
    \label{eq:bounded-exact-old-new-decomposition}
\end{align}
where the remaining spaces contain precisely the new branches and block states,
\begin{align}
    F_k^{\mathrm N} =\operatorname{span}\left\{\ket{f_b}:p_b=k,\ w_b>\theta,\ \sigma_b>t_k\right\},
    \qquad T_k^{\mathrm H} =\operatorname{span}\left\{\ket{\omega_b}:p_b=k,\ w_b\leq\theta\right\}.
\end{align}
Different products can share constituent blocks, so they need not satisfy the orthogonality condition required by Gram splitting. Removing these product directions leaves each remaining space spanned by individual factors with mutually orthogonal one-particle supports, making the subroutine applicable.

\item \textbf{Separate the new factors and reconstruct their blocks.}
\label{step:bounded-exact-split}
We now apply the Gram-splitting subroutine, described in Lem.~\ref{lem:main-exact-gram-splitting}. Applied to $F_k^{\mathrm N}$, it separates the dominant branches, from which the full block states are reconstructed using Eq.~\eqref{eq:overview-full-block-reconstruction} with $p$ replaced by $k$. Applied to $T_k^{\mathrm H}$, it returns the complementary block states directly. After product removal, the combined dimension of these spaces is at most $m/k$, because their factors correspond to distinct $k$-particle blocks with mutually orthogonal one-particle supports.

The recovered dominant branches are retained for constructing $O_\ell^{\mathrm N}$ at later orders, while the block states obtained from $T_k^{\mathrm H}$ are retained for constructing $O_\ell^{\mathrm H}$. We also retain all reconstructed full block states for final assembly. We then proceed to order $k+1$, unless $k=r$.
\end{enumerate}

\item \textbf{Construct the core, block-product input, and passive Gaussian unitary.}
\label{step:bounded-exact-assemble}
\label{step:bounded-exact-preparation}
After completing all orders, construct the core by removing the combined one-particle support of all recovered dominant branches from $S_\theta$, as in Sec.~\ref{sec:homogeneous-overview}. The remaining core contains both the excluded dominant branch supports and the transformed modes corresponding to $I_0$. Occupying every mode in this core represents these contributions without identifying them individually.

Using this core and the recovered blocks, construct $\ket{\Xi_{\mathrm{in}}}$ and $U_{\mathrm{out}}$ as in Step~\ref{step:overview-preparation}, using each block's particle number in place of $p$. With probability one, the output $\ket{\widetilde\Psi}:=\hat U_{\mathrm{out}}\ket{\Xi_{\mathrm{in}}}$ equals the reconstructed state $\ket{\Phi}$, defined by the product expression in Eq.~\eqref{eq:overview-numbered-output}, up to a global phase. Thus, replacing the excluded blocks by their dominant branches remains the only approximation.

\end{enumerate}

\begin{algorithm}[tbp]
\caption{\textsc{Prepare}: construction of the input and unitary}
\label{alg:prepare}

\KwIn{An orthonormal basis $V_0$ of
$\widetilde S_{\mathrm{core}}\subseteq\mathbb C^m$;
normalized blocks $\{(\ket{\widetilde\omega_b},p_b,\eta_b)\}_{b=1}^{D}$
with particle numbers $p_b\geq2$ and phase-aligned error bounds $\eta_b$.}
\KwOut{A classical description of
$\ket{\widetilde\Psi}=\hat U_{\mathrm{out}}\ket{\Xi_{\mathrm{in}}}$,
specified by $\ket{\Xi_{\mathrm{in}}}$ and $U_{\mathrm{out}}$.}

Set $L$ as in Sec.~\ref{app:iu-estimated-higher-particle},
with the failure allocation of Sec.~\ref{app:iu-input-and-unitary}\;

\BlankLine
\tcp{Recover the branch modes and coefficients.}
\For{$b=1,\ldots,D$}{
    $k\gets p_b$, \quad $\eta\gets\eta_b$,
    \quad $\ket{\widetilde\omega}\gets\ket{\widetilde\omega_b}$\;

    \eIf{$k=2$}{
        Form antisymmetric $M$ with
        $M_{ij}=\braket{ij}{\widetilde\omega}$ for $i<j$;
        $t\gets\operatorname{rank}M/2$, $z\gets0$\;

        \For{$S=1,\ldots,t$, stopping at acceptance}{
            Choose a normalized eigenvector $u_S$ of $MM^\dagger$
            with largest eigenvalue $a_S^2>0$, taking $a_S>0$\;

            $v_S\gets-M\overline{u_S}/a_S$,
            \quad $M\gets M-a_S(u_Sv_S^T-v_Su_S^T)$,
            \quad $z\gets z+a_S^2$\;

            \lIf{$\sqrt{2-2\sqrt z}\leq2\eta$}{
                accept $((u_1,v_1,\ldots,u_S,v_S),
                (a_\ell/\sqrt z)_{\ell=1}^{S})$
            }
        }
    }{
        \For{$S=1,\ldots,\lfloor m/k\rfloor$, stopping at acceptance}{
    $T_S\gets K_kS^2\eta$,
    \quad $R_S\gets H_kS^{5/2}\eta$
    by Eqs.~\eqref{eq:iu-count-dependent-thresholds} and~\eqref{eq:iu-local-residual-test}\;

    \For{up to $L$ independent trials, stopping at acceptance}{
        $g\sim\mathcal N_{\mathbb C}(0,I_m)$,
        \quad $\ket{\widetilde h}\gets
        \hat c[g]\ket{\widetilde\omega}$\;

        Diagonalize $\Gamma_{\widetilde h}^{(1)}$,
        with eigenvalues $\lambda_1\geq\cdots\geq\lambda_m$\;

        Group consecutive $\sqrt{\lambda_i}$ across gaps
        $\leq16\sqrt k\,\eta$; retain groups $G_1,\ldots,G_d$
        with $\min_{i\in G_\ell}\sqrt{\lambda_i}>T_S$\;

        Reject unless $1\leq d\leq S$ and
        $|G_\ell|=k-1$ for every $\ell$\;

        For each $G_\ell$, let
        $(\widetilde e_{\ell,j})_{j=1}^{k-1}$
        be its orthonormal eigenvectors and set
        $\ket{w_\ell}\gets
        \hat c[\widetilde e_{\ell,k-1}]\cdots
        \hat c[\widetilde e_{\ell,1}]
        \ket{\widetilde\omega}$\;

        Discard zero $w_\ell$; relabel survivors $1,\ldots,d$;
        reject if $d=0$\;

        $V_{b,\ell}\gets
        (\widetilde e_{\ell,1},\ldots,\widetilde e_{\ell,k-1},
         w_\ell/\norm{w_\ell})$
        for $\ell=1,\ldots,d$\;

        $V\gets(V_{b,1},\ldots,V_{b,d})$,
        \quad $D_\alpha\gets
        \operatorname{diag}(\norm{w_1}I_k,\ldots,\norm{w_d}I_k)$\;

        Compute $VD_\alpha^2=U\Sigma Z^\dagger$;
        $Q\gets UZ^\dagger$
        by Lem.~\ref{lem:iu-weighted-orthogonalization}\;

        Write $Q=(Q_1,\ldots,Q_d)$ with
        $Q_\ell=(q_{\ell,1},\ldots,q_{\ell,k})$\;

        $\ket{q_\ell}\gets
        \hat c^\dagger[q_{\ell,1}]\cdots
        \hat c^\dagger[q_{\ell,k}]\ket{\vac}$,
        \quad
        $a\gets(\braket{q_\ell}{\widetilde\omega})_{\ell=1}^{d}$\;

        \lIf{$a\neq0$ and $\sqrt{2-2\norm{a}_2}\leq R_S$}{
            accept $(Q,a/\norm{a}_2)$
        }
    }
}
    }

    \lIf{no candidate was accepted}{\KwRet{\textnormal{\textsc{Fail}}}}

    Discard zero coefficients and their branch modes; store the
    accepted data as $V_b,(a_{b,\ell})_{\ell=1}^{s'_b}$\;

    $D_b\gets\operatorname{diag}
    (|a_{b,1}|I_{p_b},\ldots,|a_{b,s'_b}|I_{p_b})$\;
}

\BlankLine
\tcp{Orthonormalize all modes and construct the output.}
$q\gets\dim\widetilde S_{\mathrm{core}}$,
\quad $t\gets q+\sum_{b=1}^{D}p_bs'_b$\;

\lIf{$t>m$}{\KwRet{\textnormal{\textsc{Fail}}}}

$V\gets(V_0,V_1,\ldots,V_D)$,
\quad $D_{\mathrm{all}}\gets\operatorname{diag}(I_q,D_1,\ldots,D_D)$\;

Compute $VD_{\mathrm{all}}^2=U\Sigma Z^\dagger$;
$Q\gets UZ^\dagger$ by Lem.~\ref{lem:iu-weighted-orthogonalization};
keep all coefficients fixed\;

Complete $Q$ to a unitary $U_{\mathrm{out}}\gets[Q,Q_\perp]$;
form $\ket{\Xi_{\mathrm{in}}}$ by Eq.~\eqref{eq:iu-explicit-input}
in stored core and branch order, leaving the $Q_\perp$ modes empty\;

\KwRet{$\ket{\Xi_{\mathrm{in}}}$ and $U_{\mathrm{out}}$}\;
\end{algorithm}
\subsection{Reconstruction from estimated RDMs}
\label{subsec:bounded-noisy-recovery}

We implement the exact recursion using the same threshold and Gram-splitting procedures as in the reconstruction from estimated RDMs in the homogeneous setting. The additional issue is that errors in previously recovered factors perturb the product spaces removed at later orders and thereby affect the recovery of new factors. To control this propagation and achieve the target infidelity, our stability analysis shows that it suffices to obtain Hermitian estimates $\widetilde\Gamma^{(1)},\ldots,\widetilde\Gamma^{(r)}$ satisfying
\begin{align}
    \max_{1\leq k\leq r} \norm{\widetilde\Gamma^{(k)}-\Gamma_\Psi^{(k)}} \leq\nu,
    \qquad \nu\leq c_r\frac{\varepsilon_{\mathrm{fid}}}
    {m^{\lceil(7\lfloor r/2\rfloor+r+6)/2\rceil}},
    \label{eq:bounded-rdm-accuracy-main}
\end{align}
where $c_r>0$ is a sufficiently small constant depending only on $r$. Prop.~\ref{prop:delta-free-explicit-accuracy} establishes this sufficient accuracy by controlling the accumulation of errors across orders. For fixed $r$, the required RDM accuracy is inverse-polynomial in $m$.
The operator-norm guarantee in Eq.~\eqref{eq:bounded-rdm-accuracy-main} can be achieved simultaneously for all required RDMs with probability at least $1-\delta/3$ using
\begin{align}
    O_r\left( m^{O(r)} \varepsilon_{\mathrm{fid}}^{-2} \log\frac{2mr}{\delta}\right)
    \label{eq:bounded-copy-complexity-main}
\end{align}
copies of the target state, as derived in App.~\ref{app:rdm-estimation}. The algorithm proceeds as follows.

\begin{enumerate}
\setcounter{enumi}{-1}
\item \textbf{Estimate the RDMs.}
\label{step:bounded-estimate}
Using independent copies of $\ket{\Psi}$, estimate $\Gamma_\Psi^{(1)},\ldots,\Gamma_\Psi^{(r)}$ to the accuracy in Eq.~\eqref{eq:bounded-rdm-accuracy-main} with the fermionic partial-tomography protocol of Ref.~\cite{zhao2021fermionic}. Use the same measurement records for all orders and allocate failure probability $\delta/(3r)$ to each RDM estimate. All estimates then satisfy the required accuracy with probability at least $1-\delta/3$; the subsequent reconstruction requires no additional copies. We condition the following analysis on this RDM-accuracy event.

\item \textbf{Separate the modes using the $1$-RDM estimate.}
\label{step:bounded-high}
We first construct $\widetilde S_\theta$ from $\widetilde\Gamma^{(1)}$ using the threshold rule in Eq.~\eqref{eq:noisy-one-rdm-threshold}. The same argument as in Sec.~\ref{subsec:one-rdm-equal}, with $\mu$ replaced by $\nu$, bounds the subspace error by $O(m\nu)$. As in the exact recursion, we use this common subspace at every order. It also approximates the occupied directions of the transformed always-occupied component. Throughout the analysis below, the exact comparison spaces and operators use the same $\theta$ and $t_k$ chosen from the estimates.

\item \textbf{Recover the blocks recursively.}
\label{step:bounded-recursion}
At each order $k=2,\ldots,r$, we perform the following three steps before proceeding to the next order.

\begin{enumerate}
\renewcommand{\theenumii}{(\alph{enumii})}
\renewcommand{\labelenumii}{\theenumii}
\item \textbf{Extract the candidate spaces.}
\label{step:bounded-noisy-extract}
Define $\widetilde Q_{\mathrm{hi},k}$ and $\widetilde Q_{\mathrm{lo},k}$ from $\widetilde S_\theta$ in the same way as their exact counterparts, and form
\begin{align}
    \widetilde C_k
    &:=\widetilde Q_{\mathrm{lo},k}\widetilde\Gamma^{(k)}\widetilde Q_{\mathrm{hi},k},
    \qquad
    \widetilde A_k
    :=\widetilde Q_{\mathrm{lo},k}\widetilde\Gamma^{(k)}\widetilde Q_{\mathrm{lo},k}.
    \label{eq:bounded-noisy-compressions}
\end{align}
As in Sec.~\ref{subsec:p-rdm-equal}, the projector errors are bounded by $O(km\nu)$. Together with the RDM-estimation error, this yields, for every $k\leq r$,
\begin{align}
    \max\left\{\norm{\widetilde C_k-C_k},\norm{\widetilde A_k-A_k}\right\}
    &=O_r(m\nu).
    \label{eq:bounded-noisy-compression-error}
\end{align}

The singular-value threshold must again be separated from the estimated spectrum. Since $C_k$ can contain directions associated with products of several blocks, its rank can exceed the number of individual blocks. The minimum of two particles per block limits each such product to at most $\lfloor k/2\rfloor$ blocks, giving $\operatorname{rank}C_k\leq r m^{\lfloor k/2\rfloor}$. Hence, we replace the homogeneous separation $s_0/(16m)$ by a separation that accounts for this larger rank. Let $\widetilde\sigma_{k,j}$ denote the singular values of $\widetilde C_k$, and choose $t_k\in[s_0,2s_0]$ such that
\begin{align}
    \min_j\abs{t_k-\widetilde\sigma_{k,j}}
    \geq\frac{s_0}{16r m^{\lfloor k/2\rfloor}}.
    \label{eq:bounded-singular-threshold-main}
\end{align}
Under Eq.~\eqref{eq:bounded-rdm-accuracy-main}, the additional singular values created by the perturbation remain close to zero, and such a threshold exists. Let $\widetilde R_k$ be the span of the right singular vectors of $\widetilde C_k$ with singular values larger than $t_k$. The separation condition and the compression-error bound ensure that this space has the same dimension as $R_k$ and remains close to it.

For $\widetilde A_k$, the exact gap remains at least $1-\theta\geq1/4$, even though the eigenvalue-one space now includes products of smaller blocks. We therefore define $\widetilde W_k^{\mathrm H}$ as the spectral subspace of $\widetilde A_k$ above the same midpoint threshold $(1+\theta)/2$. This space likewise has the same dimension as $W_k^{\mathrm H}$ and remains close to it under the assumed accuracy.

\item \textbf{Remove products of previously recovered factors.}
\label{step:bounded-noisy-products}
Construct the estimated product spans $\hat O_k^{\mathrm N}$ and $\hat O_k^{\mathrm H}$ using the product rule in Sec.~\ref{subsec:bounded-order-k-recovery}. Use previously recovered dominant branches for $\hat O_k^{\mathrm N}$ and complementary block states for $\hat O_k^{\mathrm H}$. Since both spaces are zero for $k=2,3$, no product removal is needed: set $\widetilde F_k^{\mathrm N}:=\widetilde R_k$ and $\widetilde T_k^{\mathrm H}:=\widetilde W_k^{\mathrm H}$.

At higher orders, both the candidate spaces and the product spans are approximate: the former inherit errors from the RDM estimates and the latter from previously recovered factors. Directly intersecting these estimated subspaces need not be stable: two subspaces that coincide in the exact setting can have a smaller intersection after an arbitrarily small perturbation. We instead compress the orthogonal complement of each estimated product span to the corresponding candidate space:
\begin{align}
    \widetilde B_k^{\mathrm N}
    &:=P_{\widetilde R_k}(I-P_{\hat O_k^{\mathrm N}})P_{\widetilde R_k},
    \qquad
    \widetilde B_k^{\mathrm H}
    :=P_{\widetilde W_k^{\mathrm H}}(I-P_{\hat O_k^{\mathrm H}})P_{\widetilde W_k^{\mathrm H}}.
    \label{eq:bounded-overlap-operators}
\end{align}
The corresponding exact operators are the projectors onto $F_k^{\mathrm N}$ and $T_k^{\mathrm H}$: they have eigenvalue one on the new branch or block directions and zero on the old-product directions. We use this unit gap and retain the eigenspaces above $1/2$:
\begin{align}
    \widetilde F_k^{\mathrm N}
    &:=\operatorname{ran}\mathbf1_{(1/2,1]}(\widetilde B_k^{\mathrm N}),
    \qquad
    \widetilde T_k^{\mathrm H}
    :=\operatorname{ran}\mathbf1_{(1/2,1]}(\widetilde B_k^{\mathrm H}).
    \label{eq:bounded-noisy-new-factor-spaces}
\end{align}
These spaces lie inside $\widetilde R_k$ and $\widetilde W_k^{\mathrm H}$, respectively. Lem.~\ref{lem:df-stable-intersection} controls their errors, including the errors in the previously recovered factors used to construct the product spans.

\item \textbf{Separate the new factors and reconstruct their blocks.}
\label{step:bounded-noisy-split}
To continue the recursion, the errors in the new-factor spaces must be small enough for robust Gram splitting. Under Eq.~\eqref{eq:bounded-rdm-accuracy-main}, and provided that all Gram-splitting calls at orders below $k$ have succeeded, these errors satisfy
\begin{align}
    \max\left\{
        \norm{P_{\widetilde F_k^{\mathrm N}}-P_{F_k^{\mathrm N}}},
        \norm{P_{\widetilde T_k^{\mathrm H}}-P_{T_k^{\mathrm H}}}
    \right\}
    =O_r\left(m^{O(r)}\frac{\nu}{s_0}\right).
\end{align}
Thus, the projector-error bounds remain linear in $\nu$, including the errors propagated from previously recovered factors. The accuracy requirement in Eq.~\eqref{eq:bounded-rdm-accuracy-main} also ensures that these spaces have the correct dimensions and satisfy the input conditions of Lem.~\ref{lem:main-robust-gram-splitting}. Therefore, we apply the robust Gram-splitting procedure to $\widetilde F_k^{\mathrm N}$ and $\widetilde T_k^{\mathrm H}$, using the corresponding computable error bounds. On success, these calls return approximations to the individual retained dominant branches and the full block states without a dominant branch, respectively.

Each returned dominant branch $\ket{\widetilde f_{k,j}}$ is converted to a full block estimate using Eq.~\eqref{eq:noisy-block-reconstruction} with $p$ replaced by $k$: apply $\widetilde\Gamma^{(k)}$ and normalize the resulting vector, returning \textsc{Fail} if it vanishes. Under the assumed accuracy and successful splitting calls, Lem.~\ref{lem:df-near-block-stability} bounds its norm below by $\sqrt\theta/2$ and controls the reconstructed block error. The outputs from $\widetilde T_k^{\mathrm H}$ are already full block estimates.

We retain the recovered dominant branches and complementary block states together with their computable error bounds so that we can control the product-span errors at later orders and verify the input conditions of subsequent Gram-splitting calls. We also retain all reconstructed full block states and their error bounds for final assembly. We then proceed to order $k+1$, unless $k=r$.
\end{enumerate}

Conditional on the RDM-accuracy event and successful earlier calls, each Gram-splitting call satisfies the required input conditions. We use fresh random trials and assign each call failure probability $\delta/[6(r-1)]$. Since there are at most two calls at each of the $r-1$ orders, summing the bounds for the first failed call yields a total failure probability of at most $\delta/3$.

\item \textbf{Construct the core, block-product input, and passive Gaussian unitary.}
\label{step:bounded-assemble}
\label{step:bounded-preparation}
After completing all orders, construct the core as in Sec.~\ref{subsec:p-rdm-equal}, using the dominant branches recovered across all orders. Sum their $1$-RDMs and let $\widetilde F$ be the spectral subspace above $1/2$. At each order, this sum is independent of the orthonormal basis chosen for the recovered branch space. Prop.~\ref{prop:df-core-stability} therefore controls the core directly from the errors of these spaces before Gram splitting. Since these branches lie in exterior powers of $\widetilde S_\theta$, we have $\widetilde F\subseteq\widetilde S_\theta$, and the estimated core has projector $P_{\widetilde S_\theta}-P_{\widetilde F}$, as in Eq.~\eqref{eq:main-noisy-core-projector}.

The exact core contains both the excluded dominant branch supports and the transformed always-occupied subspace associated with $I_0$. Occupying every core mode represents these contributions without identifying them individually. Each excluded block satisfies $\sigma_b\leq t_{p_b}\leq2s_0$ and $w_b>\theta\geq2/3$, so replacing it by its dominant branch incurs infidelity $1-w_b=\sigma_b^2/w_b\leq6s_0^2$. Since there are at most $m/2$ blocks, the total truncation infidelity is at most $\varepsilon_{\mathrm{fid}}/16$.

Let $\ket{\widetilde\omega_1},\ldots,\ket{\widetilde\omega_D}$ denote all full block states recovered across orders $2,\ldots,r$. Apply the block approximation in Step~\ref{step:equal-preparation} to each block, using its particle number and propagated error bound in place of $p$ and $\eta$, with constants chosen uniformly over particle numbers $2,\ldots,r$. Then jointly orthogonalize the recovered branch modes and the estimated core modes, and construct $\ket{\Xi_{\mathrm{in}}}$ and $U_{\mathrm{out}}$ as in that step. The final estimate is $\ket{\widetilde\Psi}:=\hat U_{\mathrm{out}}\ket{\Xi_{\mathrm{in}}}$. Assign each randomized block search failure budget $\delta/(3\max\{1,D\})$, so their combined failure probability is at most $\delta/3$.

Let $\ket{\Phi}$ be the exact comparison state using the same thresholds $\theta,t_2,\ldots,t_r$. Conditioned on successful reconstruction, the propagated block and core bounds under Eq.~\eqref{eq:bounded-rdm-accuracy-main}, together with Prop.~\ref{prop:iu-postprocessing}, ensure with probability at least $1-\delta/3$ that
\begin{align}
    \sqrt{1-|\langle\Phi|\widetilde\Psi\rangle|^2}&\leq\frac{\sqrt{\varepsilon_{\mathrm{fid}}}}{4}.
\end{align}
Combining this bound with the truncation error yields
\begin{align}
    |\langle\Psi|\widetilde\Psi\rangle|^2&\geq1-\varepsilon_{\mathrm{fid}}.
    \label{eq:bounded-final-fidelity-main}
\end{align}

\end{enumerate}

RDM estimation, Gram splitting, and the final input-and-unitary construction each contribute failure probability at most $\delta/3$. The overall success probability is therefore at least $1-\delta$.

\subsection{Classical post-processing time}
\label{sec:heterogeneous-time-complexity}

We account separately for RDM estimation and reconstruction, as in Sec.~\ref{sec:homogeneous-time-complexity}. The additional cost comes from constructing and removing products of previously recovered factors at each order.

Let $D_k=\binom{m}{k}$ be the dimension of the $k$-particle space. Using the estimation-cost bound in Sec.~\ref{sec:homogeneous-time-complexity}, we choose entrywise accuracy $\nu/D_k$ and failure probability $\delta/(3r)$ at each order $k=1,\ldots,r$. Therefore, the total classical cost of obtaining the RDM estimates required by Eq.~\eqref{eq:bounded-rdm-accuracy-main} is
\begin{align}
    O_r\left( \frac{\log(2mr/\delta)}{\nu^2} \sum_{k=1}^{r}D_k^4 \right) = O_r\left( m^{4r}\frac{\log(2mr/\delta)}{\nu^2} \right).
    \label{eq:bounded-rdm-estimation-runtime}
\end{align}

Once these estimates are available, constructing the projectors and compressed matrices and performing their spectral decompositions costs $O_r(\sum_{k=2}^{r}D_k^3)=O_r(m^{3r})$. We next bound the additional cost of constructing and removing products of previously recovered factors. At order $k$, each product contains at most $\lfloor k/2\rfloor$ of the $O_r(m)$ previously recovered factors, since each factor contains at least two particles. Thus, there are $O_r(m^{\lfloor k/2\rfloor})$ product vectors, each with $D_k=O_r(m^k)$ coefficients. Constructing these vectors, orthonormalizing their spans, and forming and diagonalizing the compressed operators in Eq.~\eqref{eq:bounded-overlap-operators} costs at most $O_r(m^{3k})$. Summing over $k=2,\ldots,r$ keeps these operations within $O_r(m^{3r})$. Additionally, the construction of $\widetilde S_{\mathrm{core}}$ and the final factor list is also included in this bound.

For Gram splitting, the operation counts in Sec.~\ref{sec:homogeneous-time-complexity} apply at each order $k$. Precomputing the transition $1$-RDMs costs $O_r(m^{k+3})$ per call, which is contained in $O_r(m^{3k})$ for $k\geq2$. Each subsequent trial costs $O_r(m^4)$. There are at most $2(r-1)$ calls, each assigned failure probability $\delta/[6(r-1)]$. Hence, we allow $O(\log(2r/\delta))$ trials per call, giving a total trial cost of $O_r(m^4\log(2r/\delta))$.

The construction in Step~\ref{step:bounded-preparation} of Sec.~\ref{subsec:bounded-noisy-recovery} follows the cost analysis in Sec.~\ref{sec:homogeneous-time-complexity}, using each block's particle number. Since all block particle numbers are at most $r$, there are $O(m)$ blocks, and each randomized search tests $O(m)$ branch-count bounds, the total cost is $O_r(m^{r+3}\log(2m/\delta))$. This bound includes the two-particle blocks and the final joint orthogonalization and vacuum completion, and is dominated by the RDM-estimation cost in Eq.~\eqref{eq:bounded-rdm-estimation-runtime}.

Combining the RDM-estimation cost with the reconstruction cost yields
\begin{align}
    O_r\left( m^{4r}\frac{\log(2mr/\delta)}{\nu^2} +m^{3r} +m^4\log\frac{2r}{\delta} \right).
    \label{eq:bounded-total-cost-before-substitution}
\end{align}
Finally, choose $\nu$ equal to the upper bound in Eq.~\eqref{eq:bounded-rdm-accuracy-main}. The RDM-estimation term dominates, yielding the total classical post-processing time
\begin{align}
    O_r\left( m^{O(r)} \varepsilon_{\mathrm{fid}}^{-2} \log\frac{2mr}{\delta}  \right).
    \label{eq:bounded-classical-runtime-main}
\end{align}
This yields the classical time bound in Thm.~\ref{thm:heterogeneous-blocks}, including both RDM estimation and state reconstruction.

\section{Necessity of the highest RDM order}
\label{sec:rdm-order-necessity}

Sec.~\ref{sec:bounded-reconstruction} shows that RDMs through order $r$ suffice when every block contains at most $r$ particles. To prove the matching necessity statement in Thm.~\ref{thm:rdm-order-necessary}, we construct two orthogonal states in the family whose RDMs of all orders below $r$ coincide.

\begin{proof}[Proof of Thm.~\ref{thm:rdm-order-necessary}]
Choose the disjoint mode sets $I_1:=\{1,\ldots,r\}$ and $I_2:=\{r+1,\ldots,2r\}$, and define
\begin{align}
    \ket{\Omega_+}:=\frac{\ket{I_1}+\ket{I_2}}{\sqrt{2}},
    \qquad \ket{\Omega_-}:=\frac{\ket{I_1}-\ket{I_2}}{\sqrt{2}}.
    \label{eq:rdm-order-witnesses}
\end{align}
Both states belong to the prescribed family with $n=1$, $p_1=r$, $s_1=2$, $I_0=\varnothing$, and $\hat U=\hat I$; any remaining modes are unoccupied. Since the two Fock states are orthonormal, $\ket{\Omega_+}$ and $\ket{\Omega_-}$ are normalized and satisfy $\braket{\Omega_+}{\Omega_-}=0$.

Fix $1\leq k<r$. Each entry of the $k$-RDM is an expectation of an operator $\hat c_A^\dagger\hat c_B$ with $\abs{A}=\abs{B}=k$. Consider the cross term $\bra{I_1}\hat c_A^\dagger\hat c_B\ket{I_2}$. If $\hat c_B\ket{I_2}=0$, this term vanishes. Otherwise, annihilating $k$ particles leaves $r-k\geq1$ occupied modes in $I_2$. Since applying $\hat c_A^\dagger$ cannot remove these remaining occupations, any nonzero resulting state is orthogonal to $\ket{I_1}$. Thus, the cross term is zero. The same argument applies with $I_1$ and $I_2$ exchanged. Expanding the expectation for either state therefore gives
\begin{align}
    \Gamma_{\Omega_\pm}^{(k)} =\frac12\left( \Gamma_{I_1}^{(k)}+\Gamma_{I_2}^{(k)} \right),
    \qquad 1\leq k<r,
    \label{eq:lower-rdm-branch-mixture}
\end{align}
which is independent of the relative sign and proves \eqref{eq:lower-rdms-identical}.
\end{proof}

The two states differ only in the relative sign between their branches, and this information is absent from every lower-order RDM. It first appears at order $r$, where $\bra{I_1}\Gamma_{\Omega_\pm}^{(r)}\ket{I_2}=\pm1/2$. Importantly, this indistinguishability already occurs in the simplest setting of a single $r$-particle block with two branches, even when the RDMs are known exactly and no Gaussian unitary is applied. Thus, the need for the $r$-RDM is an intrinsic limitation of lower-order RDMs, rather than a consequence of varying block sizes or RDM-estimation errors.

This indistinguishability also prevents high-fidelity reconstruction from lower-order RDMs. For any normalized output state $\ket{\widetilde\Psi}$, orthogonality of the two targets implies
\begin{align}
    \abs{\braket{\widetilde\Psi}{\Omega_+}}^2 +\abs{\braket{\widetilde\Psi}{\Omega_-}}^2 \leq1.
    \label{eq:rdm-order-fidelity-obstruction}
\end{align}
Hence, when $\varepsilon_{\mathrm{fid}}<1/2$, no output can satisfy the fidelity requirement for both targets. An algorithm given only the RDMs of orders below $r$ receives identical input for the two states and therefore has the same output distribution for both. Since the sets of successful outputs are disjoint, their two success probabilities sum to at most one. At least one target consequently has a success probability at most $1/2$. Thus, no algorithm using only these RDMs can solve Problem~\ref{prob:learning} uniformly over the family when $\varepsilon_{\mathrm{fid}}<1/2$ and $\delta<1/2$, even with exact input data and unlimited classical computation. Together with Thm.~\ref{thm:heterogeneous-blocks}, this establishes that $r$ is the necessary highest RDM order in the worst case for reconstruction from RDM data. Taking $r=p$ and $m=2p$ gives the corresponding necessity statement for the homogeneous setting of Thm.~\ref{thm:equal-blocks}, with no additional vacuum modes.

\section{Discussion}\label{Sec:Discussion}


We have shown that efficient learning remains possible with an extensive number of non-Gaussian input blocks, even when an unknown particle-number-preserving free-fermion evolution mixes modes across all input blocks. Within the family studied here, a fixed upper bound on the particle number per block suffices to keep both the sample and classical computational costs polynomial in the system size and in the inverse of the target infidelity.

Our results also determine the RDM order needed to reconstruct the hidden block structure. RDMs of orders up to the maximal particle number per block suffice to construct a block-product input and a passive Gaussian unitary that prepares the final estimate. This highest order is necessary in the worst case: orthogonal states within the same family can have identical RDMs of every lower order.

Several questions remain open. First, the optimal sample and classical computational complexities for learning this family remain unknown. A dominant cost in our algorithm comes from estimating and processing large RDMs at the high precision required to control error propagation throughout reconstruction. Our bounds may be improved by reducing this overhead or by directly estimating the information needed to recover the blocks, without reconstructing the full RDMs.

Another direction is to establish efficient learnability for broader fermionic families. One extension is to allow branches within a block to share occupied modes. Such blocks are used to model strongly correlated fermionic systems~\cite{jimenez2015cluster}. Another is to allow Gaussian evolution that does not preserve particle number and mixes creation and annihilation operators. Related progress in bosonic systems includes efficient learning of states obtained by applying arbitrary Gaussian unitaries to Fock inputs~\cite{iosue2025higher}.

Finally, it would be interesting to extend our approach to bosonic inputs whose blocks contain superpositions of Fock configurations. The key question is whether higher-order correlations can recover the hidden block structure after an unknown passive Gaussian evolution, enabling efficient learning beyond the Fock-input setting.

\section*{Statement on the Use of Artificial Intelligence}
The authors developed the learning protocol and the overall proof framework, including the principal theorems and lemmas. Generative artificial intelligence (AI) tools (ChatGPT 5 and 6) assisted with exploring reconstruction and post-processing ideas, deriving and refining some proofs, and drafting and revising portions of the manuscript. The authors reviewed and verified all AI-assisted arguments and text and take full responsibility for the content of this work.

\begin{acknowledgments}
This work was supported by the National Research Foundation of Korea (NRF) Grants (No. RS-2024-00431768 and No. RS-2025-00515456) funded by the Korean government (Ministry of Science and ICT (MSIT)) and the Institute of Information \& Communications Technology Planning \& Evaluation (IITP) Grants funded by the Korean government (MSIT) (No. RS-2024-00437284, No. IITP-2025-RS-2025-02283189 and No. IITP-2025-RS-2025-02263264) and by the Global Partnership Program of Leading Universities in Quantum Science and Technology (RS-2025-08542968) through the NRF funded by the Korean government (MSIT).
\end{acknowledgments}

\bibliography{References.bib}

@book{fetter2003quantum,
  title     = {Quantum Theory of Many-Particle Systems},
  author    = {Fetter, Alexander L. and Walecka, John Dirk},
  publisher = {Dover Publications},
  address   = {Mineola, New York},
  year      = {2003},
  isbn      = {9780486428277}
}

@article{DavisKahan1970,
  author = {Davis, Chandler and Kahan, W. M.},
  title = {The Rotation of Eigenvectors by a Perturbation. {III}},
  journal = {SIAM Journal on Numerical Analysis},
  volume = {7},
  number = {1},
  pages = {1--46},
  year = {1970},
  doi = {10.1137/0707001}
}

@book{KatoPerturbation,
  author = {Kato, Tosio},
  title = {Perturbation Theory for Linear Operators},
  publisher = {Springer},
  edition = {2},
  year = {1995},
  doi = {10.1007/978-3-642-66282-9},
  series = {Classics in Mathematics},
  address = {Berlin, Heidelberg}
}

@article{Wedin1972,
  author  = {Wedin, Per-{\AA}ke},
  title   = {Perturbation Bounds in Connection with Singular Value Decomposition},
  journal = {BIT Numerical Mathematics},
  volume  = {12},
  pages   = {99--111},
  year    = {1972},
  doi     = {10.1007/BF01932678}
}

@article{Coleman1963,
  author  = {Coleman, A. J.},
  title   = {Structure of Fermion Density Matrices},
  journal = {Reviews of Modern Physics},
  volume  = {35},
  number  = {3},
  pages   = {668--686},
  year    = {1963},
  doi     = {10.1103/RevModPhys.35.668}
}

@book{Bhatia1997MatrixAnalysis,
  author = {Bhatia, Rajendra},
  title = {Matrix Analysis},
  publisher = {Springer},
  year = {1997},
  doi = {10.1007/978-1-4612-0653-8},
  series = {Graduate Texts in Mathematics},
  volume = {169},
  address = {New York}
}

@article{zyczkowski-sommers,
    author  = {{\.Z}yczkowski, Karol and Sommers, Hans-J{\"u}rgen},
    title   = {Induced Measures in the Space of Mixed Quantum States},
    journal = {Journal of Physics A: Mathematical and General},
    volume  = {34},
    number  = {35},
    pages   = {7111--7125},
    year    = {2001},
    doi     = {10.1088/0305-4470/34/35/335}
}

@article{haah2017sample,
  title = {Sample-optimal tomography of quantum states},
  author = {Haah, Jeongwan and Harrow, Aram W. and Ji, Zhengfeng and Wu, Xiaodi and Yu, Nengkun},
  journal = {IEEE Transactions on Information Theory},
  volume = {63},
  number = {9},
  pages = {5628--5641},
  year = {2017},
  doi = {10.1109/tit.2017.2719044},
  eprint = {1508.01797},
  archiveprefix = {arXiv},
  primaryclass = {quant-ph}
}

@article{huang2020classical,
  title = {Predicting many properties of a quantum system from very few measurements},
  author = {Huang, Hsin-Yuan and Kueng, Richard and Preskill, John},
  journal = {Nature Physics},
  volume = {16},
  pages = {1050--1057},
  year = {2020},
  doi = {10.1038/s41567-020-0932-7},
  eprint = {2002.08953},
  archiveprefix = {arXiv},
  primaryclass = {quant-ph},
  number = {10}
}

@article{zhao2021fermionic,
  title = {Fermionic partial tomography via classical shadows},
  author = {Zhao, Andrew and Rubin, Nicholas C. and Miyake, Akimasa},
  journal = {Physical Review Letters},
  volume = {127},
  pages = {110504},
  year = {2021},
  doi = {10.1103/physrevlett.127.110504},
  eprint = {2010.16094},
  archiveprefix = {arXiv},
  primaryclass = {quant-ph},
  number = {11}
}

@inproceedings{aaronson2023fermion,
  title        = {Efficient Tomography of Non-Interacting-Fermion States},
  author       = {Aaronson, Scott and Grewal, Sabee},
  booktitle    = {18th Conference on the Theory of Quantum Computation, Communication and Cryptography (TQC 2023)},
  series       = {Leibniz International Proceedings in Informatics},
  volume       = {266},
  pages        = {12:1--12:18},
  year         = {2023},
  doi          = {10.4230/LIPIcs.TQC.2023.12},
  eprint       = {2102.10458},
  archivePrefix= {arXiv},
  primaryClass = {quant-ph}
}

@article{terhal2002classical,
  title        = {Classical simulation of noninteracting-fermion quantum circuits},
  author       = {Terhal, Barbara M. and DiVincenzo, David P.},
  journal      = {Physical Review A},
  volume       = {65},
  pages        = {032325},
  year         = {2002},
  doi          = {10.1103/PhysRevA.65.032325},
  eprint       = {quant-ph/0108010},
  archivePrefix= {arXiv}
}

@article{oszmaniec2022fermion,
  title = {Fermion Sampling: A robust quantum computational advantage scheme using fermionic linear optics and magic input states},
  author = {Oszmaniec, Micha{\l} and Dangniam, Ninnat and Morales, Mauro E. S. and Zimbor{\'a}s, Zolt{\'a}n},
  journal = {PRX Quantum},
  volume = {3},
  pages = {020328},
  year = {2022},
  doi = {10.1103/prxquantum.3.020328},
  eprint = {2012.15825},
  archiveprefix = {arXiv},
  primaryclass = {quant-ph},
  number = {2}
}

@article{mele2025few,
  title = {Efficient learning of quantum states prepared with few fermionic non-{Gaussian} gates},
  author = {Mele, Antonio Anna and Herasymenko, Yaroslav},
  journal = {PRX Quantum},
  year = {2025},
  eprint = {2402.18665},
  archiveprefix = {arXiv},
  primaryclass = {quant-ph},
  doi = {10.1103/prxquantum.6.010319},
  volume = {6},
  number = {1},
  pages = {010319}
}

@misc{chen2026optimal,
  title        = {Optimal tomography of bosonic and fermionic {Gaussian} states},
  author       = {Chen, Senrui and Fanizza, Marco and Girardi, Filippo and Lami, Ludovico and Mele, Francesco Anna and Walter, Michael and Witteveen, Freek},
  year         = {2026},
  eprint       = {2607.11847},
  archivePrefix= {arXiv},
  primaryClass = {quant-ph}
}

@misc{alam2025dynamics,
  title = {Fermionic dynamics on a trapped-ion quantum computer beyond exact classical simulation},
  author = {Alam, Faisal and others},
  author_full = {Alam, Faisal and Bosse, Jan Lukas and {\v{C}}epait{\.e}, Ieva and Chapman, Adrian and Clinton, Laura and Crichigno, Marcos and Crosson, Elizabeth and Cubitt, Toby and Derby, Charles and Dowinton, Oliver and Eassa, Norhan and Faehrmann, Paul K. and Flammia, Steve and Flynn, Brian and Gambetta, Filippo Maria and Garc{\'i}a-Patr{\'o}n, Ra{\'u}l and Hunter-Gordon, Max and Jones, Glenn and Khedkar, Abhishek and Klassen, Joel and Kreshchuk, Michael and McMullan, Edward Harry and Mineh, Lana and Montanaro, Ashley and Mora, Caterina and Morton, John J. L. and Nocera, Alberto and Patel, Dhrumil and Rolph, Pete and Santos, Raul A. and Seddon, James R. and Sheridan, Evan and Somogyi, Wilfrid and Svensson, Marika and Vaishnav, Niam and Wang, Sabrina Yue and Wright, Gethin and Chertkov, Eli and Dreyer, Henrik and Foss-Feig, Michael},
  year = {2025},
  eprint = {2510.26300},
  archiveprefix = {arXiv},
  primaryclass = {quant-ph},
  doi = {10.48550/arXiv.2510.26300}
}

@misc{bako2025fermionic,
  title = {Fermionic {Born} Machines: Classical training of quantum generative models based on {Fermion Sampling}},
  author = {Bak{\'o}, Bence and Kolarovszki, Zolt{\'a}n and Zimbor{\'a}s, Zolt{\'a}n},
  year = {2025},
  eprint = {2511.13844},
  archiveprefix = {arXiv},
  primaryclass = {quant-ph},
  doi = {10.48550/arXiv.2511.13844}
}

@misc{kerenidis2026scalable,
  title = {Scalable Quantum Machine Learning: Trainability, Expressivity and Efficiency},
  author = {Kerenidis, Iordanis},
  year = {2026},
  eprint = {2607.24014},
  archiveprefix = {arXiv},
  primaryclass = {quant-ph},
  doi = {10.48550/arXiv.2607.24014}
}

@article{huggins2022unbiasing,
  title = {Unbiasing fermionic quantum {Monte Carlo} with a quantum computer},
  author = {Huggins, William J. and O'Gorman, Bryan A. and Rubin, Nicholas C. and Reichman, David R. and Babbush, Ryan and Lee, Joonho},
  journal = {Nature},
  volume = {603},
  number = {7901},
  pages = {416--420},
  year = {2022},
  doi = {10.1038/s41586-021-04351-z}
}

@article{hebenstreit2019all,
  title = {All pure fermionic non-{Gaussian} states are magic states for matchgate computations},
  author = {Hebenstreit, M. and Jozsa, R. and Kraus, B. and Strelchuk, S. and Yoganathan, M.},
  journal = {Physical Review Letters},
  volume = {123},
  number = {8},
  pages = {080503},
  year = {2019},
  doi = {10.1103/physrevlett.123.080503}
}

@misc{oh2026classical,
  title = {Classical simulation of free-fermionic dynamics and quantum chemistry with magic input},
  author = {Oh, Changhun and Oszmaniec, Micha{\l} and Reardon-Smith, Oliver and Zimbor{\'a}s, Zolt{\'a}n},
  year = {2026},
  eprint = {2604.26813},
  archiveprefix = {arXiv},
  primaryclass = {quant-ph},
  doi = {10.48550/arXiv.2604.26813}
}

@article{ivanov2016computational,
  title = {Computational complexity of exterior products and multiparticle amplitudes of noninteracting fermions in entangled states},
  author = {Ivanov, Dmitri A.},
  journal = {Physical Review A},
  year = {2017},
  doi = {10.1103/physreva.96.012322},
  volume = {96},
  number = {1},
  pages = {012322},
  eprint = {1603.02724},
  archiveprefix = {arXiv},
  primaryclass = {quant-ph}
}

@misc{koizumi2026provably,
  title = {Provably Efficient Learning of Fermionic Correlations under Particle-Number Symmetry},
  author = {Koizumi, Yuki and Wada, Kaito and Takama, Toshinori P. and Yoshioka, Nobuyuki},
  year = {2026},
  eprint = {2606.30601},
  archiveprefix = {arXiv},
  primaryclass = {quant-ph},
  doi = {10.48550/arXiv.2606.30601}
}

@article{heyraud2025unified,
  title = {Unified framework for matchgate classical shadows},
  author = {Heyraud, Valentin and Chomet, H{\'e}loise and Tilly, Jules},
  journal = {npj Quantum Information},
  volume = {11},
  number = {1},
  pages = {65},
  year = {2025},
  doi = {10.1038/s41534-025-01015-y}
}

@misc{low2022classical,
  title = {Classical shadows of fermions with particle number symmetry},
  author = {Low, Guang Hao},
  year = {2022},
  eprint = {2208.08964},
  archiveprefix = {arXiv},
  primaryclass = {quant-ph},
  doi = {10.48550/arXiv.2208.08964}
}

@article{kottmann2022optimized,
  author = {Kottmann, Jakob S. and Aspuru-Guzik, Al{\'a}n},
  title = {Optimized low-depth quantum circuits for molecular electronic structure using a separable-pair approximation},
  journal = {Physical Review A},
  volume = {105},
  number = {3},
  pages = {032449},
  year = {2022},
  doi = {10.1103/physreva.105.032449},
  eprint = {2105.03836},
  archiveprefix = {arXiv},
  primaryclass = {quant-ph}
}

@misc{montanaro2017learning,
  title = {Learning stabilizer states by {Bell} sampling},
  author = {Montanaro, Ashley},
  year = {2017},
  eprint = {1707.04012},
  archiveprefix = {arXiv},
  primaryclass = {quant-ph},
  doi = {10.48550/arXiv.1707.04012}
}

@article{bittel2025optimal,
  title = {Optimal trace-distance bounds for free-fermionic states: Testing and improved tomography},
  author = {Bittel, Lennart and Mele, Antonio Anna and Eisert, Jens and Leone, Lorenzo},
  journal = {PRX Quantum},
  volume = {6},
  number = {3},
  pages = {030341},
  year = {2025},
  doi = {10.1103/pzx6-nkfb}
}

@article{mele2025learning,
  title = {Learning quantum states of continuous-variable systems},
  author = {Mele, Francesco A. and Mele, Antonio A. and Bittel, Lennart and Eisert, Jens and Giovannetti, Vittorio and Lami, Ludovico and Leone, Lorenzo and Oliviero, Salvatore F. E.},
  journal = {Nature Physics},
  volume = {21},
  number = {12},
  pages = {2002--2008},
  year = {2025},
  doi = {10.1038/s41567-025-03086-2}
}

@article{grewal2025efficient,
  title = {Efficient learning of quantum states prepared with few non-{Clifford} gates},
  author = {Grewal, Sabee and Iyer, Vishnu and Kretschmer, William and Liang, Daniel},
  journal = {Quantum},
  volume = {9},
  pages = {1907},
  year = {2025},
  doi = {10.22331/q-2025-11-06-1907}
}

@inproceedings{arunachalam2022optimal,
  author = {Arunachalam, Srinivasan and Bravyi, Sergey and Dutt, Arkopal and Yoder, Theodore J.},
  title = {Optimal Algorithms for Learning Quantum Phase States},
  booktitle = {18th Conference on the Theory of Quantum Computation, Communication and Cryptography (TQC 2023)},
  series = {Leibniz International Proceedings in Informatics (LIPIcs)},
  volume = {266},
  editor = {Fawzi, Omar and Walter, Michael},
  pages = {3:1--3:24},
  year = {2023},
  publisher = {Schloss Dagstuhl -- Leibniz-Zentrum f{\"u}r Informatik},
  address = {Dagstuhl, Germany},
  doi = {10.4230/LIPIcs.TQC.2023.3},
  eprint = {2208.07851},
  archiveprefix = {arXiv},
  primaryclass = {quant-ph}
}

@inproceedings{huang2024learning,
  title = {Learning shallow quantum circuits},
  author = {Huang, Hsin-Yuan and Liu, Yunchao and Broughton, Michael and Kim, Isaac and Anshu, Anurag and Landau, Zeph and McClean, Jarrod R.},
  booktitle = {Proceedings of the 56th Annual ACM Symposium on Theory of Computing},
  pages = {1343--1351},
  year = {2024},
  doi = {10.1145/3618260.3649722}
}

@inproceedings{landau2025learning,
  title = {Learning quantum states prepared by shallow circuits in polynomial time},
  author = {Landau, Zeph and Liu, Yunchao},
  booktitle = {Proceedings of the 57th Annual ACM Symposium on Theory of Computing},
  pages = {1828--1838},
  year = {2025},
  doi = {10.1145/3717823.3718311}
}

@misc{iosue2025higher,
  title = {Higher moment theory and learnability of bosonic states},
  author = {Iosue, Joseph T. and Wang, Yu-Xin and Datta, Ishaun and Ghosh, Soumik and Oh, Changhun and Fefferman, Bill and Gorshkov, Alexey V.},
  year = {2025},
  eprint = {2510.01610},
  archiveprefix = {arXiv},
  primaryclass = {quant-ph},
  doi = {10.48550/arXiv.2510.01610}
}

@article{gigena2021many,
  title = {Many-body entanglement in fermion systems},
  author = {Gigena, N. and Di Tullio, M. and Rossignoli, R.},
  journal = {Physical Review A},
  volume = {103},
  number = {5},
  pages = {052424},
  year = {2021},
  doi = {10.1103/physreva.103.052424}
}

@article{cianciulli2024bipartite,
  title = {Bipartite representations and many-body entanglement of pure states of {$N$} indistinguishable particles},
  author = {Cianciulli, J. A. and Rossignoli, R. and Di Tullio, M. and Gigena, N. and Petrovich, Federico},
  journal = {Physical Review A},
  volume = {110},
  number = {3},
  pages = {032414},
  year = {2024},
  doi = {10.1103/physreva.110.032414}
}

@article{lanyon2017efficient,
  title = {Efficient tomography of a quantum many-body system},
  author = {Lanyon, B. P. and others},
  author_full = {Lanyon, B. P. and Maier, C. and Holz{\"a}pfel, M. and Baumgratz, T. and Hempel, C. and Jurcevic, P. and Dhand, I. and Buyskikh, A. S. and Daley, A. J. and Cramer, M. and Plenio, M. B. and Blatt, R. and Roos, C. F.},
  journal = {Nature Physics},
  volume = {13},
  number = {12},
  pages = {1158--1162},
  year = {2017},
  doi = {10.1038/nphys4244}
}

@article{lange2023adaptive,
  title = {Adaptive quantum state tomography with active learning},
  author = {Lange, Hannah and Kebri{\v{c}}, Matja{\v{z}} and Buser, Maximilian and Schollw{\"o}ck, Ulrich and Grusdt, Fabian and Bohrdt, Annabelle},
  journal = {Quantum},
  volume = {7},
  pages = {1129},
  year = {2023},
  doi = {10.22331/q-2023-10-09-1129}
}

@incollection{d2003quantum,
  author = {D'Ariano, G. Mauro and Paris, Matteo G. A. and Sacchi, Massimiliano F.},
  title = {Quantum Tomography},
  booktitle = {Advances in Imaging and Electron Physics},
  volume = {128},
  pages = {205--308},
  publisher = {Elsevier},
  year = {2003},
  doi = {10.1016/S1076-5670(03)80065-4},
  eprint = {quant-ph/0302028},
  archiveprefix = {arXiv}
}

@book{paris2004quantum,
  editor = {Paris, Matteo and {\v{R}}eh{\'a}{\v{c}}ek, Jaroslav},
  title = {Quantum State Estimation},
  series = {Lecture Notes in Physics},
  volume = {649},
  publisher = {Springer},
  address = {Berlin, Heidelberg},
  year = {2004},
  doi = {10.1007/b98673}
}

@article{anshu2024survey,
  title = {A survey on the complexity of learning quantum states},
  author = {Anshu, Anurag and Arunachalam, Srinivasan},
  journal = {Nature Reviews Physics},
  volume = {6},
  number = {1},
  pages = {59--69},
  year = {2024},
  doi = {10.1038/s42254-023-00662-4}
}

@article{eisert2020quantum,
  title = {Quantum certification and benchmarking},
  author = {Eisert, Jens and Hangleiter, Dominik and Walk, Nathan and Roth, Ingo and Markham, Damian and Parekh, Rhea and Chabaud, Ulysse and Kashefi, Elham},
  journal = {Nature Reviews Physics},
  volume = {2},
  number = {7},
  pages = {382--390},
  year = {2020},
  doi = {10.1038/s42254-020-0186-4}
}

@misc{qin2026statistical,
  title = {Statistical and Algorithmic Foundations of Probing Quantum Systems with Compressive Measurements: A Review},
  author = {Qin, Zhen and Wakin, Michael B and Zhu, Zhihui},
  year = {2026},
  eprint = {2605.27191},
  archiveprefix = {arXiv},
  primaryclass = {quant-ph},
  doi = {10.48550/arXiv.2605.27191}
}

@book{derezinski2013mathematics,
  title={Mathematics of quantization and quantum fields},
  author={Derezi{\'n}ski, Jan and G{\'e}rard, Christian},
  year={2013},
  publisher={Cambridge University Press}
}

@article{jimenez2015cluster,
  title={Cluster-based mean-field and perturbative description of strongly correlated fermion systems: Application to the one- and two-dimensional {Hubbard} model},
  author={Jim{\'e}nez-Hoyos, Carlos A and Scuseria, Gustavo E},
  journal={Physical Review B},
  volume={92},
  number={8},
  pages={085101},
  year={2015},
  publisher={APS}
}

\newpage

\begingroup
\renewcommand{\tocname}{Appendix Contents}
\tableofcontents
\endgroup

\appendix
\addtocontents{toc}{\protect\let\protect\contentsline\protect\savedcontentsline}
\crefalias{section}{appendix}

\section{Mathematical preliminaries and notation}
\label{app:mathematical-preliminaries}

We first relate the Fock-space notation of the main text to the \textit{exterior algebra} representation used in the proofs, then collect the linear algebra notation and perturbation bounds needed below. State symbols are unchanged; we generally omit ket notation for exterior-algebra vectors. App.~\ref{app:basic-rdm-identities} uses this notation to establish the RDM identities.

\subsection{Exterior algebra and fermionic system}
\label{app:basic-exterior-algebra}

We recall the standard fermionic Fock-space description~\cite{fetter2003quantum, derezinski2013mathematics} in exterior-algebra notation. Let $H\simeq\mathbb C^m$ be an $m$-dimensional complex Hilbert space with fixed orthonormal basis $e_1,\ldots,e_m$. For $0\leq k\leq m$, set $\mathcal I_k:=\{I\subseteq[m]:|I|=k\}$. For $I=\{i_1<\cdots<i_k\}\in\mathcal I_k$, define $e_I:=e_{i_1}\wedge\cdots\wedge e_{i_k}$, where $\wedge$ denotes the antisymmetric product satisfying $e_i\wedge e_j=-e_j\wedge e_i$ and $e_i\wedge e_i=0$. The $k$-th exterior power of $H$ is
\begin{align}
    \wedge^kH
    &:=
    \operatorname{span}\{e_I:I\in\mathcal I_k\},
    \qquad
    D_k:=\dim\wedge^kH=\binom{m}{k}.
    \label{eq:basic-qth-exterior-power}
\end{align}
We use the conventions $\wedge^0H:=\mathbb C$, $e_\varnothing:=1$, and $\wedge^1H:=H$, and equip $\wedge^kH$ with the inner product for which $\{e_I:I\in\mathcal I_k\}$ is orthonormal.
The wedge product extends bilinearly to vectors of different degrees.
For $x\in\wedge^kH$ and $y\in\wedge^lH$,
$x\wedge y\in\wedge^{k+l}H$ and $x\wedge y=(-1)^{kl}y\wedge x$. For a subspace $E\subseteq H$, we write $\wedge^kE:=\operatorname{span}\{v_1\wedge\cdots\wedge v_k:v_1,\ldots,v_k\in E\}\subseteq\wedge^kH$. The one-particle support of a nonzero vector $x\in\wedge^kH$, $k\geq1$, is the smallest subspace $B\subseteq H$ such that $x\in\wedge^kB$.

To connect with the Fock-space notation of Sec.~\ref{Sec:State}, we identify a $k$-particle state with its exterior-algebra representation by
\begin{align}
    \ket{\psi}&=\sum_{I\in\mathcal I_k}\psi_I\ket I
    \quad\longleftrightarrow\quad
    \psi=\sum_{I\in\mathcal I_k}\psi_I e_I.
    \label{eq:basic-fock-exterior-identification}
\end{align}
Here, $\ket I=\hat c_{i_1}^\dagger\cdots\hat c_{i_k}^\dagger\ket{\vac}$ for $I=\{i_1<\cdots<i_k\}\in\mathcal I_k$.
With this identification,
the $m$-mode fermionic Fock space is represented as
$\bigoplus_{k=0}^m\wedge^kH$.
For $x=\sum_{I\in\mathcal I_k}x_Ie_I$, the creation polynomial $\hat c^\dagger[x]:=\sum_Ix_I\hat c_I^\dagger$ defined in Sec.~\ref{Sec:State} satisfies $x\wedge y=\hat c^\dagger[x]\,y$ for $y\in\wedge^lH$. We also write $\hat c[v]:=(\hat c^\dagger[v])^\dagger$ for a general exterior vector $v$.

A passive Gaussian unitary $\hat U$ associated with $U\in\operatorname{U}(m)$ applies the same single-particle transformation $U$ to each particle. With the global phase fixed by $\hat U\ket{\vac}=\ket{\vac}$, its action on the $k$-particle sector is represented by $\wedge^kU$:
\begin{align}
    (\wedge^kU)(v_1\wedge\cdots\wedge v_k)&=Uv_1\wedge\cdots\wedge Uv_k,\qquad \hat U\ket{\psi}\longleftrightarrow(\wedge^kU)\psi.
    \label{eq:basic-exterior-unitary-action}
\end{align}
The first relation extends linearly to a unitary operator on $\wedge^kH$.

\subsection{Linear-algebra notation}
\label{app:linear-algebra-notation}

For a linear map $A:V\to W$ between finite-dimensional Hilbert spaces, we write $\operatorname{ran}A$ for its range, $\ker A$ for its kernel, and $A^\dagger$ for its adjoint. The notation $\|\cdot\|$ denotes the Hilbert-space norm on vectors and the induced operator norm on linear maps, $\|A\|:=\sup_{\|x\|=1}\|Ax\|$. The trace and Hilbert--Schmidt norms are $\|A\|_1:=\operatorname{Tr}\sqrt{A^\dagger A}$ and $\|A\|_{\mathrm{HS}}:=\sqrt{\operatorname{Tr}(A^\dagger A)}$, respectively.

For a subspace $E\subseteq V$, let $P_E$ be the orthogonal projector onto $E$ and $I_E$ the identity on $E$. We write $I$ when the underlying space is clear. For a normalized vector $x\in V$, let $\Pi_x:=|x\rangle\langle x|$.

Under the standard isometric identification of $\wedge^kH$ with the antisymmetric subspace of $H^{\otimes k}$, the projector onto states with all $k$ particles in $E\subseteq H$ is
\begin{align}
    P_{\wedge^kE}&=\left.P_E^{\otimes k}\right|_{\wedge^kH}.
    \label{eq:lifted-projector-notation}
\end{align}

For a Hermitian operator $A$, let $\operatorname{spec}(A)$ denote its set of eigenvalues. For a linear map $B:V\to W$, let $\operatorname{sing}(B):=\{\sqrt\lambda:\lambda\in\operatorname{spec}(B^\dagger B)\}$ denote its set of singular values. If $A=\sum_{\lambda\in\operatorname{spec}(A)}\lambda P_\lambda$ is the spectral decomposition of $A$, then, for $J\subseteq\mathbb R$, we write
\begin{align}
    \mathbf 1_J(A)
    &:={}
    \sum_{\lambda\in \operatorname{spec}(A)\cap J} P_\lambda.
    \label{eq:spectral-projector-notation}
\end{align}
Thus, $\operatorname{ran}\mathbf 1_{(\theta,\infty)}(A)$ is the subspace spanned by eigenvectors of $A$ with eigenvalues greater than $\theta$.

For a point $x$ in a normed space and a nonempty subset $\mathcal S$, define
\begin{align}
    \operatorname{dist}\left(x,\mathcal S\right)
    &:={}
    \inf_{y\in\mathcal S}\|x-y\|.
    \label{eq:distance-to-set-notation}
\end{align}
In particular, for a real threshold $t$,
\begin{align*}
    \operatorname{dist}(t,\operatorname{spec}(A))
    =\min_{\lambda\in\operatorname{spec}(A)}|t-\lambda|,\qquad
    \operatorname{dist}(t,\operatorname{sing}(B))
    =\min_{\sigma\in\operatorname{sing}(B)}|t-\sigma|.
\end{align*}

\subsection{Eigenvalue perturbation and subspace alignment}
\label{app:basic-perturbation-bounds}

We recall Weyl's eigenvalue bound and a unitary alignment of nearby subspaces.

\begin{lemma}[Weyl's inequality~\cite{Bhatia1997MatrixAnalysis}]\label{lem:weyl-inequality}
Let \(A\) and \(\widetilde A\) be Hermitian operators on a \(D\)-dimensional Hilbert space, and suppose $ \|\widetilde{A} - A\|\leq \epsilon$. Let $\lambda_1 \geq\cdots \geq \lambda_D$ and $\widetilde\lambda_1 \geq \cdots \geq \widetilde\lambda_D$ be the eigenvalues of \(A\) and \(\widetilde A\), respectively. Then, $|\widetilde{\lambda}_j - \lambda_j|\leq \epsilon$ for every \(j\in[D]\).
\end{lemma}

The following lemma identifies nearby subspaces by a unitary map with a controlled displacement of each unit vector.

\begin{lemma}[Canonical unitary alignment]\label{lem:canonical-unitary-alignment}
Let $E,F$ be subspaces of a finite-dimensional Hilbert space with $\dim E=\dim F=d$ and $\|P_E-P_F\|\leq\eta<1$. The map $Q:=P_FP_E(P_EP_FP_E|_E)^{-1/2}:E\to F$ is unitary and, for every unit vector $x\in E$,
\begin{align}
    \|Qx-x\|&\leq\sqrt2\eta.
    \label{eq:canonical-unitary-vector-bound}
\end{align}
The same bound holds for $\|Q^\dagger y-y\|$ when $y\in F$ is normalized. In particular, $Q$ maps the Haar measure on the unit sphere of $E$ to that on the unit sphere of $F$ when $d\geq1$.
\end{lemma}
\begin{proof}
For $d\geq1$, choose orthonormal principal-vector bases $e_j$ of $E$ and $f_j$ of $F$, with $\langle e_i,f_j\rangle=\delta_{ij}\cos\theta_j$~\cite{KatoPerturbation}. Since $\sin\theta_{\max}=\|P_E-P_F\|\leq\eta<1$, all $\cos\theta_j$ are positive. The identities $P_Fe_j=\cos\theta_j f_j$ and $(P_EP_FP_E)e_j=\cos^2\theta_j e_j$ therefore show that the inverse square root is well defined and $Qe_j=f_j$. Thus $Q$ is unitary. The vectors $f_j-e_j$ are mutually orthogonal, with squared norms $2(1-\cos\theta_j)$, so
\begin{align}
    \sup_{\substack{x\in E\\\|x\|=1}}\|Qx-x\|^2&=2(1-\cos\theta_{\max})\leq2\sin^2\theta_{\max}\leq2\eta^2.
    \label{eq:canonical-alignment-principal-calculation}
\end{align}
The adjoint bound follows by taking $x=Q^\dagger y$, and unitarity preserves Haar measure. The case $d=0$ is immediate.
\end{proof}

\section{Particle RDMs and identities}
\label{app:basic-rdm-identities}
\label{app:exact-structural-tools}
\label{app:exact-rdm-structure}

We recall the standard particle-RDM formalism~\cite{Coleman1963} in the exterior-algebra notation of App.~\ref{app:basic-exterior-algebra}, using the normalization of Sec.~\ref{Sec:Main results}. We then derive the identities for individual branches, blocks, and block-product states used in reconstruction. The short proofs are included to make the normalization and fermionic signs explicit.

\subsection{RDM definition}
\label{app:rdm-definition-and-covariance}

For $N$-particle vectors $x,y$, define the transition $k$-RDM by
\begin{align}
    (\Gamma_{x,y}^{(k)})_{I,J}
    &:=
    \langle y|\hat c_J^\dagger\hat c_I|x\rangle
    =
    \langle \hat c_Jy,\hat c_Ix\rangle,
    \qquad I,J\in\mathcal I_k.
    \label{eq:exterior-transition-rdm-definition}
\end{align}
We likewise regard $\Gamma_{x,y}^{(k)}$ as an operator on $\wedge^k H$ through its matrix in the basis $\{e_I\}_{I\in\mathcal I_k}$. In particular, $\Gamma_x^{(k)}:=\Gamma_{x,x}^{(k)}$ is the usual
$k$-RDM of $x$. For $N\geq1$ and $x\neq0$, its one-particle support is $\operatorname{ran}\Gamma_x^{(1)}$. The second expression in
Eq.~\eqref{eq:exterior-transition-rdm-definition} identifies each RDM entry
with the overlap of the states left after removing the particles in $I$
and $J$, respectively. By comparing these residual states, we can determine
which parts of the RDM retain coherence and which are orthogonal.

We use the unnormalized particle-RDM convention introduced in
Sec.~\ref{Sec:Main results}. With this convention,
\begin{align}
    \operatorname{Tr}\Gamma_{x,y}^{(k)}
    &=
    \binom Nk\langle y,x\rangle.
    \label{eq:transition-rdm-trace}
\end{align}
Indeed,
$\sum_{I\in\mathcal I_k}\hat c_I^\dagger\hat c_I
=\binom Nk I$
on the $N$-particle sector. Hence, if $x$ is normalized,
$\operatorname{Tr}\Gamma_x^{(k)}=\binom Nk$. We also have
$\Gamma_{x,y}^{(0)}=\langle y,x\rangle$, while
$\Gamma_{x,y}^{(k)}=0$ for $k>N$.

To transfer the block structure from the input basis, where the occupied
mode sets are explicitly disjoint, to the unknown output basis, we use
the covariance of RDMs under passive Gaussian evolution.
If $\hat U$ is the passive Gaussian unitary associated with
$U\in\operatorname{U}(m)$, then
\begin{align}
    \Gamma_{\hat Ux,\hat Uy}^{(k)}
    &=
    (\wedge^kU)\Gamma_{x,y}^{(k)}(\wedge^kU)^\dagger.
    \label{eq:rdm-unitary-covariance}
\end{align}
To see this, for $|I|=k$ we expand the annihilation operators after
the single-particle transformation as
$\hat U^\dagger\hat c_I\hat U
=\sum_{J\in\mathcal I_k}\det(U_{I,J})\,\hat c_J$,
where $U_{I,J}$ is the submatrix with rows indexed
by $I$ and columns indexed by $J$, both in increasing order.
Substituting this expression into
Eq.~\eqref{eq:exterior-transition-rdm-definition} gives the stated
conjugation. The matrix with entries $\det(U_{I,J})$ is exactly
the action of $U$ on the $k$-particle sector, denoted by
$\wedge^kU$. Thus, spectral and orthogonality properties proved
for the input RDM remain valid after the unknown passive Gaussian
unitary.

\subsection{RDMs of individual branches and blocks}
\label{app:individual-block-rdms}

We next record properties of the RDM that will be used repeatedly
for the branch states. Let $F\subseteq H$ be a $p$-dimensional
one-particle subspace. Since $\wedge^pF$ is one-dimensional, any
normalized $f\in\wedge^pF\subseteq\wedge^pH$ is, up to a global
phase, the $p$-particle state that occupies all modes in $F$.
In particular, since $f$ is a $p$-particle vector, its $q$-RDM
$\Gamma_f^{(q)}$ is defined by
Eq.~\eqref{eq:exterior-transition-rdm-definition} with $x=y=f$.
The following proposition shows that these RDMs are simply the
projectors onto the corresponding exterior powers of $F$.

\begin{proposition}[RDMs of a fully occupied subspace]\label{prop:occupied-subspace-rdm}
Let $F\subseteq H$ be a $p$-dimensional subspace and let
$f\in\wedge^pF$ be normalized. Then, for every $0\leq q\leq p$,
\begin{align}
    \Gamma_f^{(q)}&=P_{\wedge^qF},
    \label{eq:occupied-subspace-rdm}
\end{align}
where $P_{\wedge^qF}$ denotes the orthogonal projector onto
$\wedge^qF\subseteq\wedge^qH$.
\end{proposition}

\begin{proof}
By the covariance in Eq.~\eqref{eq:rdm-unitary-covariance}, it suffices to take $F=\operatorname{span}\{e_1,\ldots,e_p\}$ and $f=e_1\wedge\cdots\wedge e_p$; a global phase does not affect the RDM. For $1\leq q\leq p$, the contraction $\hat c_If$ vanishes unless $I\subseteq[p]$. Otherwise, it is, up to a sign, the normalized state occupying $[p]\setminus I$. Distinct sets $I$ leave orthogonal residual states, so
\begin{align}
    (\Gamma_f^{(q)})_{I,J}
    &=\langle\hat c_Jf,\hat c_If\rangle
    =\begin{cases}
        1,&I=J\subseteq[p],\\
        0,&\text{otherwise}.
    \end{cases}
\end{align}
These are precisely the matrix elements of $P_{\wedge^qF}$. For $q=0$, both sides are the identity on $\wedge^0H\simeq\mathbb C$.
\end{proof}

The preceding proposition provides the RDMs of a single branch.
A block in Eq.~\eqref{eq:block-source} is a superposition of
such branches on mutually orthogonal one-particle subspaces.
We next compute its RDMs. The distinction is whether fewer
than $p$ particles or all $p$ particles are removed: in the
first case, the remaining particles distinguish the branches,
whereas in the second case every branch leaves the vacuum.

\begin{proposition}[RDMs of a single block]\label{prop:correlated-block-rdm}
Let $E\subseteq H$, let $p\geq2$, and let
$\omega:=\sum_{l=1}^s\omega_l f_l\in\wedge^pE$,
where $F_1,\ldots,F_s\subseteq E$ are mutually orthogonal
$p$-dimensional subspaces, each $f_l\in\wedge^pF_l$ is
normalized, and $\sum_l|\omega_l|^2=1$. Then,
$\Gamma_\omega^{(0)}=1$ and $\Gamma_\omega^{(p)}=\Pi_\omega$,
while
\begin{align}
    \Gamma_\omega^{(q)}
    &=
    \sum_{l=1}^s|\omega_l|^2P_{\wedge^qF_l},
    \qquad 1\leq q<p.
    \label{eq:single-block-rdm-orders}
\end{align}
In particular,
$\|\Gamma_\omega^{(q)}\|=\max_l|\omega_l|^2$
for $1\leq q<p$.
\end{proposition}

\begin{proof}
The branch vectors are orthonormal, so $\|\omega\|=1$ and $\Gamma_\omega^{(0)}=1$. For $1\leq q<p$, expansion in the branch basis gives
\begin{align}
    \Gamma_\omega^{(q)}
    &=\sum_{l,l'}\omega_l\overline{\omega_{l'}}
      \Gamma_{f_l,f_{l'}}^{(q)}.
\end{align}
The diagonal terms are $|\omega_l|^2P_{\wedge^qF_l}$ by Prop.~\ref{prop:occupied-subspace-rdm}. For $l\neq l'$, every nonzero residual state $\hat c_If_l$ lies in $\wedge^{p-q}F_l$, which is orthogonal to $\wedge^{p-q}F_{l'}$ because $p-q\geq1$. Hence $\Gamma_{f_l,f_{l'}}^{(q)}=0$. This proves Eq.~\eqref{eq:single-block-rdm-orders}; the mutually orthogonal ranges of its projectors give the stated operator norm.

For $q=p$, write $\omega=\sum_{I\in\mathcal I_p}\omega_Ie_I$. All particles are removed, so $\hat c_I\omega=\omega_I|\mathrm{vac}\rangle$ and
\begin{align}
    (\Gamma_\omega^{(p)})_{I,J}
    &=\omega_I\overline{\omega_J}.
\end{align}
Thus $\Gamma_\omega^{(p)}=\Pi_\omega$, retaining the coherence between branches.
\end{proof}

\subsection{RDMs of block-product states}
\label{app:block-product-rdms}

To obtain the RDM of a product of blocks, we group the removed
particles according to how many come from each block. Since different
choices leave different residual particle numbers in at least one block,
the resulting sectors are orthogonal. Within each sector, the RDM
factorizes into the tensor product of the corresponding block RDMs. We now
introduce the notation needed to state this factorization.

Let $H\simeq\mathbb C^m$ be the single-particle space, and let
$H=H_0\oplus H_1\oplus\cdots\oplus H_n$
be an orthogonal decomposition, where $H_0$ is reserved for unoccupied modes. For $a=1,\ldots,n$, let
$p_a\geq0$ and let $\psi_a\in\wedge^{p_a}H_a$ be normalized. Define
$\Psi:=\psi_1\wedge\cdots\wedge\psi_n\in\wedge^NH$,
where $N:=p_1+\cdots+p_n$. For $\mathbf q=(q_0,\ldots,q_n)\in\mathbb Z_{\geq0}^{n+1}$
satisfying $\sum_{a=0}^n q_a=k$, define
\begin{align}\label{eq:exterior-ortho}
\mathcal H_{\mathbf q}
&:=
(\wedge^{q_0}H_0)\wedge\cdots\wedge(\wedge^{q_n}H_n)
\subseteq\wedge^kH.
\end{align}
Thus, $\mathcal H_{\mathbf q}$ is the sector containing exactly
$q_a$ particles in $H_a$ for each $a$, and
\begin{align}
\wedge^kH
&=
\bigoplus_{\substack{q_0,\ldots,q_n\geq0\\q_0+\cdots+q_n=k}}
\mathcal H_{\mathbf q}.
\end{align}
We identify each sector with the corresponding tensor product as
\begin{align}
\mathcal H_{\mathbf q}
&\simeq
\wedge^{q_0}H_0\otimes\cdots\otimes\wedge^{q_n}H_n,
\qquad
\eta_0\wedge\cdots\wedge\eta_n
\longmapsto
\eta_0\otimes\cdots\otimes\eta_n,
\end{align}
where $\eta_a\in\wedge^{q_a}H_a$ and the factors are ordered by
increasing block index. Since the spaces $H_a$ are mutually
orthogonal, this identification preserves inner products. We use
$\wedge^0H_a\simeq\mathbb C$ and
$\wedge^{q_a}H_a=\{0\}$ whenever $q_a>\dim H_a$. With this notation, the RDM of $\Psi$ has the following blockwise factorization.

\begin{proposition}[RDM factorization by sector]\label{prop:blockwise-rdm-factorization}
For $0\leq k\leq N$, the RDM $\Gamma_\Psi^{(k)}$ is block diagonal with respect to the decomposition in Eq.~\eqref{eq:exterior-ortho}. On a sector $\mathcal H_{\mathbf q}$ with $q_0=0$ and $0\leq q_a\leq p_a$ for every $a\geq1$, the identification above gives
\begin{align}
    \Gamma_\Psi^{(k)}\big|_{\mathcal H_{\mathbf q}}&=\bigotimes_{a=1}^n\Gamma_{\psi_a}^{(q_a)}.
\end{align}
The restriction is zero on all other sectors.
\end{proposition}

\begin{proof}
If $q_0>0$, then any corresponding contraction in $\mathcal H_{\mathbf q}$ vanishes because $\Psi$ has no particles in $H_0$. Likewise, if $q_a>p_a$ for some $a\geq1$, the contraction vanishes because $\psi_a$ contains only $p_a$ particles. Hence the corresponding restriction of $\Gamma_\Psi^{(k)}$ is zero.
For two remaining sectors, if $\mathbf q\neq\mathbf q'$, then
$q_a\neq q'_a$ for some block $a$. The corresponding contractions
leave $p_a-q_a$ and $p_a-q'_a$ particles in $H_a$, respectively,
and hence lie in orthogonal particle-number sectors. Therefore all
matrix elements between $\mathcal H_{\mathbf q}$ and
$\mathcal H_{\mathbf q'}$ vanish, proving block diagonality.

Now fix $\mathbf q$ with $q_0=0$ and $q_a\leq p_a$ for every $a$.
Take $\eta=\eta_1\wedge\cdots\wedge\eta_n$ and
$\xi=\xi_1\wedge\cdots\wedge\xi_n$ in
$\mathcal H_{\mathbf q}$, where
$\eta_a,\xi_a\in\wedge^{q_a}H_a$. By the RDM definition,
\begin{align}
\langle\eta,\Gamma_\Psi^{(k)}\xi\rangle
&=
\langle\hat c[\xi]\Psi,\hat c[\eta]\Psi\rangle.
\end{align}
Since $\Psi=\psi_1\wedge\cdots\wedge\psi_n$, the contractions
factor blockwise as
\begin{align}
\hat c[\eta]\Psi
&=
(-1)^{\sum_{a<b}(p_a-q_a)q_b}
(\hat c[\eta_1]\psi_1)\wedge\cdots\wedge
(\hat c[\eta_n]\psi_n),
\end{align}
and similarly for $\xi$. The sign is the same for $\eta$ and
$\xi$, since both remove $q_a$ particles from each block, and
therefore cancels in the inner product. Since the spaces $H_a$
are mutually orthogonal, the inner product of the remaining
states factorizes, giving
\begin{align}
\langle\eta,\Gamma_\Psi^{(k)}\xi\rangle
&=
\prod_{a=1}^n
\langle\hat c[\xi_a]\psi_a,\hat c[\eta_a]\psi_a\rangle
=
\prod_{a=1}^n
\langle\eta_a,\Gamma_{\psi_a}^{(q_a)}\xi_a\rangle.
\end{align}
Under the identification of $\mathcal H_{\mathbf q}$ with
$\wedge^{q_1}H_1\otimes\cdots\otimes\wedge^{q_n}H_n$, this is
precisely the matrix element of
$\Gamma_{\psi_1}^{(q_1)}\otimes\cdots\otimes
\Gamma_{\psi_n}^{(q_n)}$.
\end{proof}

As a direct consequence of Prop.~\ref{prop:blockwise-rdm-factorization}, the RDMs of our target states satisfy a uniform operator-norm bound. Each block in Eq.~\eqref{eq:block-source} has RDMs of operator norm at most one by Prop.~\ref{prop:correlated-block-rdm}. The state occupying all modes in $I_0$ likewise has RDMs of operator norm at most one by Prop.~\ref{prop:occupied-subspace-rdm}. Applying Prop.~\ref{prop:blockwise-rdm-factorization} to the complete pre-unitary state, with the $I_0$ state treated as an additional block when present, shows that each tensor-product restriction of its $k$-RDM has operator norm at most one. Since these restrictions act on mutually orthogonal sectors, the norm of the full RDM is their maximum and is therefore at most one. Finally, the passive Gaussian unitary transforms the RDM by unitary conjugation and hence preserves its operator norm. Thus, for every relevant order $k$,
\begin{align}
&\|\Gamma_\Psi^{(k)}\|\leq1.
\label{eq:structured-rdm-norm}
\end{align}
This bound relies on the prescribed block structure and need not hold for arbitrary higher-order fermionic RDMs. We will use it below to control the effect of RDM and projector errors on the compressed operators.

\subsection{Summary of RDM definitions and properties}
\label{app:rdm-identities-summary}
We collect the definitions, conventions, and properties needed to use the RDMs in the reconstruction. Throughout, $H\simeq\mathbb C^m$ is the one-particle space, $P_E$ is the orthogonal projector onto a subspace $E$, and $\Pi_x:=|x\rangle\langle x|$ for normalized $x$.

\begin{enumerate}

\item \textbf{Definition and transition RDMs.} For $x,y\in\wedge^NH$, not necessarily normalized, and $k$-element subsets $I,J\subseteq[m]$, define
\begin{align*}
    (\Gamma_{x,y}^{(k)})_{I,J}&:=\langle y|\hat c_J^\dagger\hat c_I|x\rangle=\langle\hat c_Jy,\hat c_Ix\rangle,\qquad \Gamma_x^{(k)}:=\Gamma_{x,x}^{(k)}.
\end{align*}
Here $\hat c_I:=\hat c_{i_k}\cdots\hat c_{i_1}$ for $I=\{i_1<\cdots<i_k\}$, with $\hat c_\varnothing=I$. These entries give the matrix of an operator on $\wedge^kH$ in the Fock basis, of dimension $\binom mk$. The transition RDM is linear in $x$ and conjugate-linear in $y$, and $(\Gamma_{x,y}^{(k)})^\dagger=\Gamma_{y,x}^{(k)}$. In particular, writing $x=\sum_\alpha v_\alpha u_\alpha$ yields
\begin{align*}
    \Gamma_x^{(k)}&=\sum_{\alpha,\beta}\overline{v_\alpha}v_\beta\Gamma_{u_\beta,u_\alpha}^{(k)}.
\end{align*}

\item \textbf{Normalization and limiting orders.} We use the particle-RDM convention, without trace-one normalization (App.~\ref{app:rdm-definition-and-covariance}):
\begin{align*}
    \operatorname{Tr}\Gamma_{x,y}^{(k)}&=\binom Nk\langle y,x\rangle,\qquad \Gamma_{x,y}^{(0)}=\langle y,x\rangle,\qquad \Gamma_{x,y}^{(N)}=|x\rangle\langle y|.
\end{align*}
For $N<k\leq m$, the RDM is zero. For normalized $x$, $\Gamma_x^{(k)}$ is positive semidefinite with trace $\binom Nk$; in particular, $\Gamma_x^{(0)}=1$ and $\Gamma_x^{(N)}=\Pi_x$.

\item \textbf{One-particle support and stability.} For nonzero $x\in\wedge^NH$ with $N\geq1$, the smallest subspace $B\subseteq H$ satisfying $x\in\wedge^NB$ is $\operatorname{ran}\Gamma_x^{(1)}$. For normalized $x,y\in\wedge^NH$, the occupation bound and continuity estimate are
\begin{align*}
    0\leq\Gamma_x^{(1)}&\leq I,\qquad \|\Gamma_x^{(1)}-\Gamma_y^{(1)}\|\leq2\|x-y\|.
\end{align*}
These bounds hold for arbitrary fixed-particle-number states; their proof is given in Lem.~\ref{lem:df-one-rdm-lipschitz}.

\item \textbf{Passive Gaussian evolution.} For the passive Gaussian unitary $\hat U$ associated with $U\in\operatorname{U}(m)$, Eq.~\eqref{eq:rdm-unitary-covariance} implies
\begin{align*}
    \Gamma_{\hat Ux,\hat Uy}^{(k)}&=(\wedge^kU)\Gamma_{x,y}^{(k)}(\wedge^kU)^\dagger.
\end{align*}
Here $(\wedge^kU)(v_1\wedge\cdots\wedge v_k)=Uv_1\wedge\cdots\wedge Uv_k$. Thus the RDM spectrum and operator norm are preserved.

\item \textbf{Single branch} (Prop.~\ref{prop:occupied-subspace-rdm}). A normalized $f\in\wedge^pF$ with $\dim F=p$ occupies every mode in $F$, and
\begin{align*}
    \Gamma_f^{(q)}&=P_{\wedge^qF},\qquad 0\leq q\leq p.
\end{align*}
The projector selects states with all $q$ particles in $F$; in particular, $\Gamma_f^{(1)}=P_F$ when $p\geq1$.

\item \textbf{Single block} (Prop.~\ref{prop:correlated-block-rdm}). Let $p\geq2$ and $\omega=\sum_\ell\omega_\ell f_\ell$, where the $p$-dimensional spaces $F_\ell$ are mutually orthogonal, each $f_\ell\in\wedge^pF_\ell$ is normalized, and $\sum_\ell|\omega_\ell|^2=1$. For $1\leq q<p$,
\begin{align*}
    \Gamma_\omega^{(q)}&=\sum_\ell|\omega_\ell|^2P_{\wedge^qF_\ell},\qquad \|\Gamma_\omega^{(q)}\|=\max_\ell|\omega_\ell|^2,\qquad \Gamma_\omega^{(p)}=\Pi_\omega.
\end{align*}
Coherence between distinct branches vanishes at orders below $p$, whereas the $p$-RDM retains the full block state.

\item \textbf{Product of blocks} (Prop.~\ref{prop:blockwise-rdm-factorization}). Let $H=H_0\oplus H_1\oplus\cdots\oplus H_n$ be an orthogonal decomposition, with $H_0$ unoccupied. For normalized $\psi_a\in\wedge^{p_a}H_a$, $p_a\geq0$, set $\Psi=\psi_1\wedge\cdots\wedge\psi_n$. The $k$-RDM is block diagonal in the sectors $\mathcal H_{\mathbf q}:=(\wedge^{q_0}H_0)\wedge\cdots\wedge(\wedge^{q_n}H_n)$, where $q_a\geq0$ and $\sum_{a=0}^nq_a=k$. For $q_0=0$ and $q_a\leq p_a$ for all $a\geq1$, identifying wedge factors with tensor factors in increasing order of $a$ implies
\begin{align*}
    \Gamma_\Psi^{(k)}\big|_{\mathcal H_{\mathbf q}}&=\bigotimes_{a=1}^n\Gamma_{\psi_a}^{(q_a)}.
\end{align*}
All other sectors give zero. An always-occupied component is included as another factor and obeys the single-branch formula; factors with $q_a=0$ contribute the scalar one.

\item \textbf{Norm bound for the target family.} For products of the blocks above, together with always-occupied and vacuum modes, and after any passive Gaussian evolution, Eq.~\eqref{eq:structured-rdm-norm} induces
\begin{align*}
    \|\Gamma_\Psi^{(k)}\|&\leq1,\qquad 0\leq k\leq m.
\end{align*}
Unlike the one-particle occupation bound, this bound at higher orders relies on the prescribed block structure and need not hold for arbitrary fermionic states.
\end{enumerate}

\section{Gram-splitting subroutine}
\label{app:gram-splitting}

We prove the exact Gram-splitting guarantee of Lem.~\ref{lem:main-exact-gram-splitting} and its extension to estimated subspaces. The latter follows by combining a random eigenvalue-gap bound with a perturbation bound for the splitting operator.

\subsection{Exact Gram splitting}
\label{app:exact-gram-splitting}
\label{app:gram-splitting-subroutine}

\begin{proposition}[Exact Gram splitting]\label{prop:exact-gram-splitting}
Let $k\geq2$. Suppose a $d$-dimensional subspace $T\subseteq\wedge^kH$ has the form
\begin{align}
    T&=\operatorname{span}\{\omega_1,\ldots,\omega_d\},
    \label{eq:gram-splitting-ideal-space}
\end{align}
where the normalized vectors $\omega_b\in\wedge^kB_b$, $b\in[d]$, have pairwise orthogonal one-particle supports $B_b\subseteq H$. Given $k$ and a classical description of any exact orthonormal basis of $T$, a randomized classical algorithm returns normalized vectors $\hat\omega_1,\ldots,\hat\omega_d$ such that, with probability one, there exist a permutation $\pi$ of $[d]$ and phases $\varphi_j\in\mathbb R$ satisfying
\begin{align}
    \hat\omega_j&=e^{i\varphi_j}\omega_{\pi(j)},\qquad j\in[d].
\end{align}
Neither the factors $\omega_b$ nor their supports $B_b$ are required as input. For fixed $k$, the algorithm uses polynomially many arithmetic operations in $m$ and no additional copies of the target state.
\end{proposition}

\begin{proof}[Proof of Prop.~\ref{prop:exact-gram-splitting}]
The cases $d=0,1$ are immediate, so assume $d\geq2$. We first show that the overlap of the $1$-RDMs defines a quadratic form diagonal in the hidden block basis. The orthogonality of the $B_b$ makes $\{\omega_b\}$ an orthonormal basis of $T$. For $b\neq c$, the residual vectors $\hat c_i\omega_b\in\wedge^{k-1}B_b$ and $\hat c_j\omega_c\in\wedge^{k-1}B_c$ are orthogonal because $k\geq2$. Hence, for $z=\sum_bz_b\omega_b$,
\begin{align}
    \Gamma_{\omega_b,\omega_c}^{(1)}&=0\quad(b\neq c),\qquad \Gamma_z^{(1)}=\sum_b|z_b|^2\Gamma_{\omega_b}^{(1)}.
    \label{eq:gram-cross-transition-vanishing}
\end{align}
The same decomposition holds for $x=\sum_bx_b\omega_b$. Since the block RDMs have mutually orthogonal supports,
\begin{align}
    \operatorname{Tr}\!\left[\Gamma_x^{(1)}\Gamma_z^{(1)}\right]&=\sum_b\tau_b|z_b|^2|x_b|^2,\qquad \tau_b:=\operatorname{Tr}[(\Gamma_{\omega_b}^{(1)})^2]>0.
\end{align}
Here $\tau_b>0$ follows from $\operatorname{Tr}\Gamma_{\omega_b}^{(1)}=k>0$. Thus distinct coefficients $\tau_b|z_b|^2$ allow the block vectors to be recovered by diagonalizing this form.

To carry out the diagonalization in the supplied basis, define the Hermitian operator $\mathcal S_z$ on $T$ by
\begin{align}
    \operatorname{Tr}\!\left[\Gamma_{x,y}^{(1)}\Gamma_z^{(1)}\right]&=\langle y,\mathcal S_zx\rangle,\qquad x,y\in T.
    \label{eq:operator-gram-splitting-matrix}
\end{align}
This definition applies to any subspace $T\subseteq\wedge^kH$ and normalized $z\in T$, and will also be used for estimated subspaces. In an orthonormal basis $u_1,\ldots,u_d$ of $T$, its matrix $S_z$ satisfies, for $x=\sum_\alpha v_\alpha u_\alpha$,
\begin{align}
    \operatorname{Tr}\!\left[\Gamma_x^{(1)}\Gamma_z^{(1)}\right]&=v^\dagger S_zv,\qquad (S_z)_{\alpha,\beta}:=\operatorname{Tr}\!\left[\Gamma_{u_\beta,u_\alpha}^{(1)}\Gamma_z^{(1)}\right].
    \label{eq:main-gram-splitting-matrix}
\end{align}
Thus $S_z$ acts on coefficient vectors, while $\mathcal S_z$ acts on the corresponding states. The quadratic-form identity above gives
\begin{align}
    \mathcal S_z&=\sum_{b=1}^d\tau_b|z_b|^2\Pi_{\omega_b}.
    \label{eq:gram-splitting-spectral-decomposition}
\end{align}

Sample a Haar-random normalized vector $z\in T$, form $S_z$, and diagonalize it. Return $\hat\omega_j:=\sum_{\alpha=1}^d r_\alpha^{(j)}u_\alpha$ for its normalized eigenvectors $r^{(j)}$; if an eigenvalue is repeated, return \textsc{Fail}. The vector $(|z_1|^2,\ldots,|z_d|^2)$ has a density on the probability simplex, where each equality $\tau_b|z_b|^2=\tau_c|z_c|^2$ defines a set of measure zero. The eigenvalues are therefore distinct with probability one, and Eq.~\eqref{eq:gram-splitting-spectral-decomposition} identifies the returned vectors with the hidden blocks up to phases and a permutation.

For fixed $k$, each supplied vector has $\binom{m}{k}$ coefficients. Computing the transition RDMs, forming $S_z$, and diagonalizing it require polynomially many arithmetic operations in $m$, using only these classical vectors and no additional state copies.
\end{proof}

\subsection{Gram splitting with an estimated subspace}
\label{app:robust-gram-splitting}

When only an estimated factor space is available, Gram splitting must control the error in each recovered factor. The following result supplies this guarantee directly from the subspace error and a prescribed failure probability.

\begin{proposition}[Robust Gram splitting]\label{prop:robust-gram-splitting}
Let $k\geq2$ and let
\begin{align}
    T&:=\operatorname{span}\{\omega_1,\ldots,\omega_d\}\subseteq\wedge^kH,
    \label{eq:certified-gram-splitting-target-space}
\end{align}
where $\omega_b\in\wedge^kB_b$, $b\in[d]$, are normalized and the one-particle subspaces $B_b\subseteq H$ are pairwise orthogonal. Given $k$, an orthonormal basis of $\widetilde T$, a supplied error bound $\eta$ satisfying
\begin{align}
    \dim\widetilde T&=\dim T=d,\qquad \|P_{\widetilde T}-P_T\|\leq\eta\leq\frac{1}{64km^3},
    \label{eq:robust-gram-input-bound}
\end{align}
and $\beta_{\mathrm{loc}}\in(0,1)$, a randomized classical algorithm returns either \textnormal{\textsc{Fail}} or an orthonormal basis $\widetilde\omega_1,\ldots,\widetilde\omega_d$ of $\widetilde T$. Its failure probability is at most $\beta_{\mathrm{loc}}$, and every returned basis satisfies, for some permutation $\pi$ and phases $\varphi_j\in\mathbb R$,
\begin{align}
    \|\widetilde\omega_j-e^{i\varphi_j}\omega_{\pi(j)}\|&\leq(8k+2)m^3\eta,\qquad j\in[d].
    \label{eq:split-factorwise-certificate}
\end{align}
For fixed $k$, the algorithm uses polynomially many arithmetic operations in $m$ and $\log(1/\beta_{\mathrm{loc}})$ and no additional copies of the target state.
\end{proposition}

The proof uses two estimates: a bound on how often the exact eigenvalue gap is small, and a bound on the splitting-operator perturbation caused by estimating the factor space. We establish these first, then use them to choose the acceptance threshold and number of trials.

\subsubsection{A quantitative random-gap bound}
We first quantify the separation of the eigenvalues in Eq.~\eqref{eq:gram-splitting-spectral-decomposition}.

\begin{lemma}[Eigenvalue gaps in Gram splitting]\label{lem:gram-gap-anticoncentration}
Assume the setting of Prop.~\ref{prop:exact-gram-splitting}, with $d\geq2$, and set $\tau_b:=\operatorname{Tr}[(\Gamma_{\omega_b}^{(1)})^2]$. Let $z$ be a Haar-random normalized vector in $T$, written as $z=\sum_bz_b\omega_b$. Its minimum eigenvalue gap obeys the following bound for every $t\geq0$:
\begin{align}
    g(z)&:=\min_{1\leq b<c\leq d}\left|\tau_b|z_b|^2-\tau_c|z_c|^2\right|,\qquad \Pr[g(z)\leq t]\leq\frac{(d-1)^2m}{2k^2}\,t.
    \label{eq:gram-gap-anticoncentration}
\end{align}
\end{lemma}

\begin{proof}
For each $b\in[d]$, set $X_b:=|z_b|^2$. The vector $(X_1,\ldots,X_d)$ is uniformly distributed over the probability simplex $X_b\geq0$, $\sum_{b=1}^dX_b=1$~\cite{zyczkowski-sommers}. Fix $1\leq b<c\leq d$ and condition on $R:=X_b+X_c$. Since the conditional distribution is uniform along the segment $X_b+X_c=R$, the difference $\tau_bX_b-\tau_cX_c=(\tau_b+\tau_c)X_b-\tau_cR$ is uniform on $[-\tau_cR,\tau_bR]$. Its probability of lying in $[-t,t]$ is therefore at most $2t/[R(\tau_b+\tau_c)]$.

For $d=2$, we have $R=1$, while for $d>2$ its density is $(d-1)(d-2)s(1-s)^{d-3}$ on $0<s<1$. In either case, $\mathbb E[R^{-1}]=d-1$, so averaging the conditional bound gives
\begin{align}
    \Pr\!\left[|\tau_bX_b-\tau_cX_c|\leq t\right]&\leq\frac{2t\,\mathbb E[R^{-1}]}{\tau_b+\tau_c}=\frac{2(d-1)t}{\tau_b+\tau_c}.
    \label{eq:gram-pair-gap-bound}
\end{align}

To sum this bound over pairs, let $m_b:=\dim B_b$. Positivity and $\operatorname{Tr}\Gamma_{\omega_b}^{(1)}=k$ give $\tau_b\geq k^2/m_b$ by Cauchy--Schwarz, and hence $1/(\tau_b+\tau_c)\leq m_bm_c/[k^2(m_b+m_c)]\leq(m_b+m_c)/(4k^2)$. Since each $m_b$ occurs in exactly $d-1$ unordered pairs and $\sum_bm_b\leq m$ by orthogonality, a union bound yields
\begin{align}
    \Pr[g(z)\leq t]&\leq\frac{(d-1)t}{2k^2}\sum_{b<c}(m_b+m_c)=\frac{(d-1)^2t}{2k^2}\sum_{b=1}^dm_b\leq\frac{(d-1)^2m}{2k^2}\,t.
\end{align}
\end{proof}

\subsubsection{Perturbation of the splitting operator}

We next control the splitting-operator perturbation after aligning the two subspaces. The proof uses the following continuity bound for the $1$-RDM.

\begin{lemma}[Continuity of the $1$-RDM]\label{lem:transition-rdm-perturbation}\label{lem:df-one-rdm-lipschitz}
For normalized $x,y\in\wedge^kH$, the $1$-RDMs satisfy
\begin{align}
    0\leq\Gamma_x^{(1)}\leq I,
    \qquad
    \|\Gamma_x^{(1)}-\Gamma_y^{(1)}\|\leq2\|x-y\|.
    \label{eq:df-one-rdm-lipschitz}
\end{align}
\end{lemma}
\begin{proof}
For a unit vector $h\in H$, the occupation operator $n_h:=\hat c[h]^\dagger\hat c[h]$ satisfies $0\leq n_h\leq I$ by the anticommutation relations. Since $\langle h,\Gamma_x^{(1)}h\rangle=\langle x,n_hx\rangle$, the first claim follows. For the second, write the difference of expectations as $\langle x-y,n_hx\rangle+\langle y,n_h(x-y)\rangle$ and bound each term by $\|x-y\|$. Taking the supremum over unit $h$ proves the assertion.
\end{proof}

\begin{lemma}[Perturbation of the splitting operator]\label{lem:gram-operator-perturbation}
Let $E,F\subseteq\wedge^kH$ have the same dimension. Let $J:E\hookrightarrow\wedge^kH$ be the inclusion and $Q:E\to F$ a unitary with
\begin{align}
    \|Q-J\|\leq\rho.
    \label{eq:generic-unitary-alignment-error}
\end{align}
For a normalized $\widetilde z\in E$, set $z:=Q\widetilde z$ and define the splitting operators on $E$ and $F$ by Eq.~\eqref{eq:operator-gram-splitting-matrix}. Then
\begin{align}
    \bigl\|\widetilde{\mathcal S}_{\widetilde z}-Q^\dagger\mathcal S_zQ\bigr\|
    \leq4k\rho.
    \label{eq:gram-operator-perturbation}
\end{align}
\end{lemma}
\begin{proof}
For any normalized $x\in E$, the alignment bound gives $\|Qx-x\|\leq\rho$ and $\|z-\widetilde z\|\leq\rho$. By Lem.~\ref{lem:df-one-rdm-lipschitz}, both corresponding $1$-RDM differences have norm at most $2\rho$. Each $1$-RDM is positive with trace $k$. Adding and subtracting $\operatorname{Tr}[\Gamma_{Qx}^{(1)}\Gamma_{\widetilde z}^{(1)}]$ therefore gives
\begin{align}
\left|\langle x,(\widetilde{\mathcal S}_{\widetilde z}-Q^\dagger\mathcal S_zQ)x\rangle\right|
=\left|\operatorname{Tr}[\Gamma_x^{(1)}\Gamma_{\widetilde z}^{(1)}]-\operatorname{Tr}[\Gamma_{Qx}^{(1)}\Gamma_z^{(1)}]\right|
\leq k\|\Gamma_x^{(1)}-\Gamma_{Qx}^{(1)}\|+k\|\Gamma_{\widetilde z}^{(1)}-\Gamma_z^{(1)}\|\leq4k\rho.
\end{align}
Taking the supremum over normalized $x\in E$ proves the claimed operator-norm bound.
\end{proof}

\subsubsection{Proof of the recovery guarantee}
\label{app:certified-gram-splitting-subroutine}

\begin{proof}[Proof of Prop.~\ref{prop:robust-gram-splitting}]
We construct one trial and bound its perturbation, then use this bound to decide when its output is reliable. The acceptance probability will determine how many trials are needed.

Set $d:=\dim\widetilde T$. Return \textsc{Fail} if $\eta>1/(64km^3)$ or $d>\lfloor m/k\rfloor$. Under the proposition's assumptions, these checks pass because $dk\leq\sum_b\dim B_b\leq m$. If $d=0$, return the empty list. If $d=1$, return a normalized spanning vector of $\widetilde T$; Lem.~\ref{lem:canonical-unitary-alignment} bounds its phase-aligned error by $\sqrt2\eta$, which is smaller than the claimed bound. Henceforth, assume $d\geq2$.

\emph{One trial and its perturbation.}
In the supplied orthonormal basis $\widetilde u_1,\ldots,\widetilde u_d$ of $\widetilde T$, sample a Haar-random normalized vector $\widetilde z$ and form its splitting matrix:
\begin{align}
    \widetilde z&:=\sum_{\alpha=1}^d z_\alpha\widetilde u_\alpha,\qquad (\widetilde S_{\widetilde z})_{\alpha,\beta}:=\operatorname{Tr}\!\left[\Gamma_{\widetilde u_\beta,\widetilde u_\alpha}^{(1)}\Gamma_{\widetilde z}^{(1)}\right].
    \label{eq:uniform-certified-splitting-matrix}
\end{align}
All entries are computed from the supplied classical vectors. Denote the corresponding operator on $\widetilde T$ by $\widetilde{\mathcal S}_{\widetilde z}$.

To assess the accuracy of its eigenvectors, let $Q:\widetilde T\to T$ be the canonical unitary of Lem.~\ref{lem:canonical-unitary-alignment} and set $z:=Q\widetilde z$. This alignment is used only in the analysis. It preserves the Haar distribution and allows us to compare the estimated and exact splitting operators on the same space. Lem.~\ref{lem:gram-operator-perturbation} gives
\begin{align}
    \bigl\|\widetilde{\mathcal S}_{\widetilde z}-Q^\dagger\mathcal S_zQ\bigr\|&\leq\kappa:=4\sqrt2 k\eta.
    \label{eq:robust-split-operator-error}
\end{align}
We therefore need an eigenvalue gap large compared with $\kappa$.

\emph{Choosing the acceptance rule.}
Lem.~\ref{lem:gram-gap-anticoncentration} supplies a gap scale reached with probability at least one half:
\begin{align*}
    t&:=\frac{k^2}{(d-1)^2m},\qquad \Pr[g(z)\geq t]\geq\frac12.
\end{align*}
Using $d\leq m$, $k\geq2$, and the input bound on $\eta$, we obtain
\begin{align}
    t&\geq\frac4{m^3},\qquad \kappa\leq4\sqrt2 k\eta\leq\frac{\sqrt2}{16m^3},\qquad t>6\kappa.
    \label{eq:robust-gram-gap-comparison}
\end{align}
The exact gap is unknown, so the algorithm must test the observed spectrum. Write the eigenvalues of $\widetilde S_{\widetilde z}$ as $\widetilde\lambda_1\leq\cdots\leq\widetilde\lambda_d$ and compute
\begin{align}
    \hat g(\widetilde z)&:=\min_{j\in[d-1]}(\widetilde\lambda_{j+1}-\widetilde\lambda_j).
    \label{eq:uniform-certified-observed-splitting-gap}
\end{align}
Weyl's inequality bounds each eigenvalue shift by $\kappa$, and hence each gap shift by $2\kappa$. A cutoff at $t/2$ therefore accepts every trial with exact gap at least $t$, while ensuring that every accepted observed gap exceeds $3\kappa$. Accordingly, accept if $\hat g(\widetilde z)\geq t/2$ and return
\begin{align}
    \widetilde\omega_j&:=\sum_{\alpha=1}^d r_\alpha^{(j)}\widetilde u_\alpha,\qquad j\in[d],
    \label{eq:uniform-certified-splitting-lift}
\end{align}
where $r^{(j)}$ is a normalized eigenvector associated with $\widetilde\lambda_j$.

We now verify the accuracy certified by this test. Match the exact and observed eigenvalues in increasing order. Since $\hat g\geq t/2>3\kappa$, Weyl's inequality makes the exact eigenvalues distinct. The aligned exact operator has eigenvectors $Q^\dagger\omega_b$ by Eq.~\eqref{eq:gram-splitting-spectral-decomposition}, and each returned vector satisfies $\|(Q^\dagger\mathcal S_zQ-\widetilde\lambda_j I)\widetilde\omega_j\|\leq\kappa$. Every other exact eigenvalue is at distance at least $\hat g-\kappa$ from $\widetilde\lambda_j$. Expanding in the exact eigenbasis therefore gives, for a permutation $\pi$,
\begin{align}
    \|\Pi_{\widetilde\omega_j}-Q^\dagger\Pi_{\omega_{\pi(j)}}Q\|&\leq\frac{\kappa}{\hat g-\kappa}\leq\frac{2\kappa}{t-2\kappa}\leq\frac{4\kappa}{t}<1.
    \label{eq:robust-split-rank-one-projector}
\end{align}
The rank-one case of Lem.~\ref{lem:canonical-unitary-alignment} bounds the corresponding vector error by $4\sqrt2\kappa/t$. Adding the alignment error and using Eq.~\eqref{eq:robust-gram-gap-comparison} yields
\begin{align}
    \|\widetilde\omega_j-e^{i\varphi_j}\omega_{\pi(j)}\|&\leq\frac{4\sqrt2\kappa}{t}+\sqrt2\eta\leq8km^3\eta+\sqrt2\eta\leq(8k+2)m^3\eta.
\end{align}
The returned vectors form an orthonormal basis of $\widetilde T$ because the coefficient vectors $r^{(j)}$ are orthonormal. Thus every accepted trial has the required output accuracy.

\emph{Repetition and failure probability.}
A trial is accepted with probability at least $1/2$, so repeat independent trials until one is accepted, allowing at most $\lceil\log_2(1/\beta_{\mathrm{loc}})\rceil$ trials. If none is accepted, return \textsc{Fail}. Independence bounds this probability by $2^{-\lceil\log_2(1/\beta_{\mathrm{loc}})\rceil}\leq\beta_{\mathrm{loc}}$. For fixed $k$, computing the transition RDMs, splitting matrix, and eigendecomposition requires polynomially many arithmetic operations in $m$ per trial. The resulting total cost is polynomial in $m$ and $\log(1/\beta_{\mathrm{loc}})$, with no additional state copies.
\end{proof}

\section{Exact reconstruction for homogeneous blocks}\label{app:homogeneous-blocks}

We establish efficient reconstruction from exact RDMs in the homogeneous setting of Sec.~\ref{sec:homogeneous-overview}. After fixing the exterior-algebra notation and recording the relevant RDM decomposition, we state the reconstruction guarantee and give its constructive proof. Estimated RDMs are treated in App.~\ref{app:homogeneous-noisy}.

\subsection{Problem reformulation and notation}\label{app:uniform-problem-reformulation}

Use the homogeneous setting of Problem~\ref{prob:learning}: $I_0=\varnothing$, there are no additional vacuum modes, and every block contains the same known number $p\geq2$ of particles. We express the state family of Eq.~\eqref{eq:block-source} in exterior-algebra notation.

For each block $b\in[n]$ and branch $l\in[s_b]$, let $I_{b,l}\subseteq[m]$ be the mutually disjoint sets from Eq.~\eqref{eq:block-source}, with $|I_{b,l}|=p$, $s_b\geq2$, and $\bigcup_{b,l}I_{b,l}=[m]$. Under the identification of App.~\ref{app:basic-exterior-algebra}, the input block states are
\begin{align}
    \omega_b^{\mathrm{in}}
    &:=\sum_{l=1}^{s_b}\omega_{b,l}e_{I_{b,l}},
    \qquad \sum_{l=1}^{s_b}|\omega_{b,l}|^2=1.
    \label{eq:uniform-block-notation}
\end{align}
The coefficient normalization and disjoint branch sets give normalized block states on mutually orthogonal one-particle supports. The full input state is
\begin{align}
    \Psi_{\mathrm{in}}
    &:=\omega_1^{\mathrm{in}}\wedge\cdots\wedge\omega_n^{\mathrm{in}}
    \in\wedge^{np}H.
    \label{eq:uniform-exterior-target}
\end{align}
For the single-particle unitary $U$ associated with $\hat U$, define the transformed branch spaces and vectors by
\begin{align}
    F_{b,l}&:=\operatorname{span}\{Ue_i:i\in I_{b,l}\},
    \qquad f_{b,l}:=(\wedge^pU)e_{I_{b,l}}.
    \label{eq:uniform-output-branch-notation}
\end{align}
Each $f_{b,l}$ occupies all modes in $F_{b,l}$. For each block, set
\begin{align}
    B_b&:=\bigoplus_{l=1}^{s_b}F_{b,l},
    \qquad \omega_b:=\sum_{l=1}^{s_b}\omega_{b,l}f_{b,l}
    =(\wedge^pU)\omega_b^{\mathrm{in}}\in\wedge^pB_b.
    \label{eq:uniform-output-block-notation}
\end{align}
Because $U$ preserves inner products, the branch spaces remain mutually orthogonal, as do the block spaces $B_1,\ldots,B_n$. Finally,
\begin{align}
    \Psi&:=(\wedge^{np}U)\Psi_{\mathrm{in}}
    =\omega_1\wedge\cdots\wedge\omega_n.
    \label{eq:uniform-output-target}
\end{align}
This is the exterior-algebra representation of the target state in Eq.~\eqref{eq:output-state}.

\subsection{The eigenvalue-one space of the \texorpdfstring{$p$}{p}-RDM}
\label{app:Exact RDM}

The sector factorization in App.~\ref{app:block-product-rdms} separates the $p$-RDM into contributions from a single complete block and from several partially selected blocks. The former give eigenvalue one, whereas the latter have operator norm strictly smaller than one. To quantify this distinction, define the largest branch weight of block $b$ by
\begin{align}
    w_b &:= \max_{l\in[s_b]} |\omega_{b,l}|^2.
    \label{eq:uniform-largest-branch-weight}
\end{align}
Since $s_b\geq2$ and all branch coefficients are nonzero, $0<w_b<1$.

\begin{lemma}[Block space at eigenvalue one]\label{lem:equal-p-full-rdm-one-space}
In the input Fock basis, the $p$-RDM has the decomposition
\begin{align}
    \Gamma_{\Psi_{\mathrm{in}}}^{(p)}&=\sum_b\Pi_{\omega_b^{\mathrm{in}}}+D_{\mathrm{mix}},\qquad D_{\mathrm{mix}}\geq0,\quad\|D_{\mathrm{mix}}\|<1,
    \label{eq:homogeneous-input-rdm-decomposition}
\end{align}
where $D_{\mathrm{mix}}$ is diagonal and supported only on sectors selecting particles from at least two blocks. In particular, the eigenvalue-one space of the output RDM is
\begin{align}
    \operatorname{ran}\mathbf 1_{\{1\}}\left(\Gamma_\Psi^{(p)}\right)&=\operatorname{span}\{\omega_1,\ldots,\omega_n\}.
    \label{eq:equal-p-full-rdm-one-space}
\end{align}
\end{lemma}
\begin{proof}
By Prop.~\ref{prop:blockwise-rdm-factorization}, each input sector with $\sum_bq_b=p$ carries $\bigotimes_b\Gamma_{\omega_b^{\mathrm{in}}}^{(q_b)}$. A complete-block sector $q_b=p$ contributes $\Pi_{\omega_b^{\mathrm{in}}}$, whereas every mixed sector has $q_b<p$ for all $b$ and hence, by Prop.~\ref{prop:correlated-block-rdm},
\begin{align*}
    \left\|\bigotimes_b\Gamma_{\omega_b^{\mathrm{in}}}^{(q_b)}\right\|&=\prod_{b:q_b>0}w_b<1.
\end{align*}
Taking $D_{\mathrm{mix}}$ to be the direct sum of the mixed-sector restrictions gives Eq.~\eqref{eq:homogeneous-input-rdm-decomposition}, with $D_{\mathrm{mix}}=0$ if there are no mixed sectors. Each proper block RDM is diagonal in the input Fock basis by Eq.~\eqref{eq:single-block-rdm-orders}, since its branch supports are coordinate subspaces; hence $D_{\mathrm{mix}}$ is diagonal as well.

Because the mixed sectors are orthogonal to the complete-block sectors and $\|D_{\mathrm{mix}}\|<1$, the input RDM's eigenvalue-one space is exactly the span of the orthonormal block states $\omega_b^{\mathrm{in}}$. Covariance under $\wedge^pU$ in Eq.~\eqref{eq:rdm-unitary-covariance} maps this space to $\operatorname{span}\{\omega_1,\ldots,\omega_n\}$.
\end{proof}

\subsection{Reconstruction from exact RDMs}
\label{app:same-particle-algorithm}
\label{app:same-particle-correctness}

The eigenvalue-one space above need not have a uniform spectral gap. We therefore construct the state using the $1$-RDM to separate highly occupied modes and the $p$-RDM to recover the correlated blocks. The construction allows a controlled truncation so that the same operations can later be applied to estimated RDMs.

\begin{theorem}[Reconstruction from exact RDMs]\label{thm:exact-same-particle-reconstruction}
Let $\Psi$ be an $m$-mode state in the homogeneous setting with block particle number $p\geq2$. Given $p$, $\varepsilon_{\mathrm{fid}}\in(0,1)$, and exact classical descriptions of $\Gamma_\Psi^{(1)}$ and $\Gamma_\Psi^{(p)}$, a randomized classical algorithm returns, with probability one, a compact classical description of a normalized state $\Phi$ satisfying
\begin{align}
    1-|\langle\Phi,\Psi\rangle|^2&\leq\frac{\varepsilon_{\mathrm{fid}}}{16}.
    \label{eq:equal-p-df-truncation-bound}
\end{align}
The description consists of an orthonormal basis for a fully occupied subspace and a list of normalized $p$-particle block vectors, all with mutually orthogonal one-particle supports. For fixed $p$, the algorithm uses polynomially many arithmetic operations in $m$ and no additional copies of the target state.
\end{theorem}

We construct the algorithm and prove the theorem in four stages. First, the $1$-RDM identifies a subspace containing the dominant branches. We then recover the dominant and complementary blocks from two compressions of the $p$-RDM, before assembling the remaining occupied core and bounding the approximation error.

\subsubsection{Identifying the high-occupation space}

Choose $\theta\in[2/3,3/4]$ and compute
\begin{align}
    S_\theta&:=\operatorname{ran}\mathbf1_{(\theta,\infty)}\left(\Gamma_\Psi^{(1)}\right).
    \label{eq:uniform-ideal-high-occupation-space}
\end{align}
To identify the branches selected by this threshold, set $\mathcal N_\theta:=\{b\in[n]:w_b>\theta\}$. Since $\theta>1/2$, each such block has a unique branch of weight above $\theta$; denote its occupied support by $F_b$.

\begin{lemma}[High-occupation space]\label{lem:equal-p-df-high-space}
The $1$-RDM and its high-occupation space satisfy
\begin{align}
    \Gamma_\Psi^{(1)}&=\sum_{b=1}^n\sum_{l=1}^{s_b}|\omega_{b,l}|^2P_{F_{b,l}},\qquad S_\theta=\bigoplus_{b\in\mathcal N_\theta}F_b.
    \label{eq:equal-p-df-one-body-formula}
\end{align}
Every other branch support lies in $S_\theta^\perp$.
\end{lemma}
\begin{proof}
At order one, the blockwise factorization in Prop.~\ref{prop:blockwise-rdm-factorization} selects one particle from one block, while all other blocks contribute their order-zero RDMs. Prop.~\ref{prop:correlated-block-rdm} therefore yields Eq.~\eqref{eq:equal-p-df-one-body-formula}. Since the branch supports are mutually orthogonal, spectral thresholding selects exactly those with weight above $\theta$, proving the remaining claims.
\end{proof}

Thus each dominant block has one branch in $S_\theta$ and its remaining branches in $S_\theta^\perp$, while every complementary block lies entirely in $S_\theta^\perp$. We use this separation to recover the two types of blocks in turn.

\subsubsection{Recovering the dominant blocks}

To access the coherence between the dominant branch and the rest of its block, let $Q_{\mathrm{hi}}:=P_{\wedge^pS_\theta}$ and $Q_{\mathrm{lo}}:=P_{\wedge^p(S_\theta^\perp)}$, and form the cross compression
\begin{align}
    C_p&:=Q_{\mathrm{lo}}\Gamma_\Psi^{(p)}Q_{\mathrm{hi}}.
    \label{eq:homogeneous-compressed-operators}
\end{align}
These projectors select states with all $p$ particles in the high or low space, respectively. States occupying both spaces lie in neither range, so $Q_{\mathrm{lo}}\neq I-Q_{\mathrm{hi}}$ in general.

For $b\in\mathcal N_\theta$, let $f_b$ be the normalized dominant branch with its coefficient's phase absorbed into it. Writing $G_b:=B_b\cap F_b^\perp$ for the combined support of the other branches, decompose the block and define its coherence scale by
\begin{align}
    \omega_b&=\sqrt{w_b}f_b+r_b,\qquad \|r_b\|^2=1-w_b,\qquad \sigma_b:=\sqrt{w_b}\|r_b\|=\sqrt{w_b(1-w_b)}.
    \label{eq:uniform-near-block-decomposition}
\end{align}
Here $r_b\in\wedge^pG_b$ and $F_b\perp G_b$.
The following lemma shows how these branches appear in $C_p$ and how each branch determines its full block.

\begin{lemma}[Structure of the cross compression]\label{lem:equal-p-df-coherence}
For every $x\in\wedge^pS_\theta$,
\begin{align}
    C_px&=\sum_{b\in\mathcal N_\theta}\sqrt{w_b}r_b\langle f_b,x\rangle.
    \label{eq:equal-p-df-coherence-formula}
\end{align}
Its nonzero singular values are the $\sigma_b$ for $b\in\mathcal N_\theta$, with corresponding right singular vectors $f_b$. In particular, $\operatorname{rank}C_p\leq|\mathcal N_\theta|\leq n$. Moreover, for every $b\in\mathcal N_\theta$,
\begin{align}
    C_pf_b&=\sqrt{w_b}r_b,\qquad \Gamma_\Psi^{(p)}f_b=\sqrt{w_b}\omega_b.
    \label{eq:equal-p-df-near-identities}
\end{align}
\end{lemma}
\begin{proof}
By covariance, we may evaluate the compression in the input Fock basis. The high- and low-space projectors are then coordinate projectors, so the diagonal mixed part of Lem.~\ref{lem:equal-p-full-rdm-one-space} has zero cross compression. Returning to the output basis, a dominant block contributes through $Q_{\mathrm{hi}}\omega_b=\sqrt{w_b}f_b$ and $Q_{\mathrm{lo}}\omega_b=r_b$, whereas a complementary block has no high component. Orthogonality of different block supports therefore gives
\begin{align}
    C_p&=\sum_{b\in\mathcal N_\theta}\sqrt{w_b}|r_b\rangle\langle f_b|,\qquad C_p^\dagger C_p=\sum_{b\in\mathcal N_\theta}\sigma_b^2|f_b\rangle\langle f_b|.
    \label{eq:uniform-cross-compression-operator-form}
\end{align}
These identities prove the formula and singular-value claims. The full RDM restricts to $\Pi_{\omega_b}$ on the complete-block sector, so $\Gamma_\Psi^{(p)}f_b=\Pi_{\omega_b}f_b=\sqrt{w_b}\omega_b$.
\end{proof}

Since $1-w_b=\sigma_b^2/w_b$, a dominant block with small coherence is close to its fully occupied dominant branch. We therefore choose a threshold $t\in[s_0,2s_0]$ at the resolution below and retain the right singular space above it:
\begin{align}
    s_0^2&:=\frac{\varepsilon_{\mathrm{fid}}}{48m},\qquad R_p:=\operatorname{ran}\mathbf1_{(t^2,\infty)}(C_p^\dagger C_p),\qquad d_{\mathrm N}:=\dim R_p.
    \label{eq:equal-p-df-s0}
\end{align}
Writing $\mathcal R:=\{b\in\mathcal N_\theta:\sigma_b>t\}$ for the retained blocks, Lem.~\ref{lem:equal-p-df-coherence} gives
\begin{align}
    R_p&=\operatorname{span}\{f_b:b\in\mathcal R\}.
    \label{eq:equal-p-df-right-space}
\end{align}
The branches in this span have mutually orthogonal one-particle supports, so apply the exact Gram-splitting algorithm of Prop.~\ref{prop:exact-gram-splitting} to an orthonormal basis of $R_p$. If it returns \textsc{Fail}, return \textsc{Fail}; otherwise denote the returned factors by $\hat f_1,\ldots,\hat f_{d_{\mathrm N}}$ and reconstruct
\begin{align}
    \hat\omega_j^{\mathrm N}&:=\frac{\Gamma_\Psi^{(p)}\hat f_j}{\|\Gamma_\Psi^{(p)}\hat f_j\|},\qquad j\in[d_{\mathrm N}].
    \label{eq:uniform-main-near-factor}
\end{align}
With probability one, Gram splitting recovers the retained branches up to phases and a permutation. After matching and phase alignment, we may index the outputs by $b\in\mathcal R$; the lemma then gives
\begin{align}
    \hat f_b&=f_b,\qquad \hat\omega_b^{\mathrm N}=\omega_b,\qquad \|\Gamma_\Psi^{(p)}\hat f_b\|=\sqrt{w_b}>\sqrt\theta.
    \label{eq:equal-p-exact-near-block-recovery}
\end{align}
Thus the normalization is well defined and returns every retained dominant block. Keep both the branch list and the full-block list: the branches will identify which occupied modes to remove when constructing the core.

\subsubsection{Recovering the complementary blocks}

Complementary blocks lie entirely in $S_\theta^\perp$. To recover them, form the all-low compression
\begin{align}
    A_p&:=Q_{\mathrm{lo}}\Gamma_\Psi^{(p)}Q_{\mathrm{lo}}.
    \label{eq:uniform-ideal-all-low-compression}
\end{align}
Let $\mathcal H_\theta:=[n]\setminus\mathcal N_\theta=\{b:w_b\leq\theta\}$ and compute $W_p^{\mathrm H}:=\operatorname{ran}\mathbf1_{\{1\}}(A_p)$, with $d_{\mathrm H}:=\dim W_p^{\mathrm H}$. The next lemma identifies this space and separates it from the dominant-block remainders and mixed-sector contributions.

\begin{lemma}[Gap of the all-low compression]\label{lem:equal-p-df-strong-gap}
The eigenvalue-one space of $A_p$ is
\begin{align}
    W_p^{\mathrm H}&=\operatorname{span}\{\omega_b:b\in\mathcal H_\theta\}.
    \label{eq:uniform-exact-eigenvalue-one-space}
\end{align}
On the all-low space,
\begin{align}
    0\leq A_p-P_{W_p^{\mathrm H}}&\leq\theta\left(I_{\wedge^p(S_\theta^\perp)}-P_{W_p^{\mathrm H}}\right).
    \label{eq:equal-p-df-strong-remainder-bound}
\end{align}
In particular, every other eigenvalue is at most $\theta$, giving a gap of at least $1-\theta\geq1/4$.
\end{lemma}
\begin{proof}
Conjugating the decomposition of Lem.~\ref{lem:equal-p-full-rdm-one-space} by $\wedge^pU$ gives the complete-block projectors $\Pi_{\omega_b}$ and a mixed part diagonal in the transformed Fock basis. A complementary block lies entirely in the low space, so its projector is unchanged by the compression. For a dominant block, only the remainder survives, giving $|r_b\rangle\langle r_b|$ with norm $1-w_b<1-\theta<\theta$.

The all-low projector is a coordinate projector in the same transformed Fock basis. Every branch selected in a nonzero mixed-sector entry has weight at most $\theta$, so the product giving that entry is also at most $\theta$. The dominant remainders and mixed contributions act on mutually orthogonal subspaces, all orthogonal to the complementary block vectors. Thus $A_p$ is the identity on the span of the complementary block vectors and a positive operator of norm at most $\theta$ on its orthogonal complement, proving both claims.
\end{proof}

The lemma verifies the input structure required by Prop.~\ref{prop:exact-gram-splitting}. Apply that algorithm to an orthonormal basis of $W_p^{\mathrm H}$, again returning \textsc{Fail} if the call does so. With probability one, its outputs $\hat\omega_1^{\mathrm H},\ldots,\hat\omega_{d_{\mathrm H}}^{\mathrm H}$ are precisely the complementary blocks, up to phases and a permutation.

\subsubsection{Assembling the state and bounding the error}

We have recovered the complementary blocks and the dominant blocks whose coherence exceeds $t$. The remaining high-occupation modes belong to the unresolved dominant branches. Compute
\begin{align}
    K&:=\sum_{j=1}^{d_{\mathrm N}}\Gamma_{\hat f_j}^{(1)},\qquad F:=\operatorname{ran}K,\qquad S_{\mathrm{core}}:=S_\theta\cap F^\perp.
    \label{eq:uniform-main-retained-support-operator}
\end{align}
For an empty branch list, use $K=0$ and $F=\{0\}$. Choose an orthonormal basis of $S_{\mathrm{core}}$ and let $\hat\sigma_{\mathrm{core}}$ be the state occupying all its basis modes, using the vacuum when the core is zero-dimensional. The output is the factorized description
\begin{align}
    \Phi&:=\hat\sigma_{\mathrm{core}}\wedge\bigwedge_{j=1}^{d_{\mathrm N}}\hat\omega_j^{\mathrm N}\wedge\bigwedge_{j=1}^{d_{\mathrm H}}\hat\omega_j^{\mathrm H},
    \label{eq:uniform-main-output-vector}
\end{align}
with factors in a fixed order. Return the core basis and the two ordered block lists. The following lemma identifies this output and bounds the loss from replacing unresolved blocks by their dominant branches.

\begin{lemma}[Output and truncation]\label{lem:equal-p-exact-output}
On the probability-one event that both Gram-splitting calls recover their factors, the core is
\begin{align}
    S_{\mathrm{core}}&=\bigoplus_{b\in\mathcal N_\theta\setminus\mathcal R}F_b.
    \label{eq:equal-p-exact-core-decomposition}
\end{align}
For the analysis, define the truncated target $\Psi_{\mathrm{tr}}$ using the exact retained blocks and any normalized occupied state $\sigma_{\mathrm{core}}$ of this space. The exact-RDM output $\Phi$ agrees with this comparison state up to an overall phase:
\begin{align}
    \Psi_{\mathrm{tr}}&:=\sigma_{\mathrm{core}}\wedge\bigwedge_{b\in\mathcal R}\omega_b\wedge\bigwedge_{b\in\mathcal H_\theta}\omega_b.
    \label{eq:equal-p-truncated-target}
\end{align}
These factors have mutually orthogonal one-particle supports, so $\Phi$ is normalized. Its infidelity satisfies $1-|\langle\Phi,\Psi\rangle|^2\leq3ms_0^2$.
\end{lemma}
\begin{proof}
The occupied-subspace identity in Eq.~\eqref{eq:occupied-subspace-rdm} and the recovery of the retained branches give $K=\sum_{b\in\mathcal R}P_{F_b}$. Removing its range from the high space identified in Lem.~\ref{lem:equal-p-df-high-space} proves Eq.~\eqref{eq:equal-p-exact-core-decomposition}. The wedge product of the unresolved dominant branches occupies every mode in this core and therefore agrees with $\hat\sigma_{\mathrm{core}}$ up to phase. Since the other blocks are recovered exactly up to phases and permutations, the returned state is $\Psi_{\mathrm{tr}}$ up to an overall phase.

For every unresolved block, $\sigma_b\leq t\leq2s_0$ and $w_b>\theta\geq2/3$, giving $1-w_b=\sigma_b^2/w_b\leq6s_0^2$. Only these blocks change in $\Psi_{\mathrm{tr}}$, and their supports are mutually orthogonal. Hence
\begin{align}
    |\langle\Psi_{\mathrm{tr}},\Psi\rangle|^2&=\prod_{b\in\mathcal N_\theta\setminus\mathcal R}w_b,\qquad 1-|\langle\Psi_{\mathrm{tr}},\Psi\rangle|^2\leq\sum_{b\in\mathcal N_\theta\setminus\mathcal R}(1-w_b)\leq6ns_0^2\leq3ms_0^2.
\end{align}
The last inequality uses $p\geq2$ and $np\leq m$.
\end{proof}

\begin{proof}[Completion of the proof of Thm.~\ref{thm:exact-same-particle-reconstruction}]
Each Gram-splitting call succeeds with probability one by Prop.~\ref{prop:exact-gram-splitting}, so both do so with probability one. Lem.~\ref{lem:equal-p-exact-output} and the choice $s_0^2=\varepsilon_{\mathrm{fid}}/(48m)$ then imply the claimed fidelity.

For fixed $p$, the projectors, compressions, and spectral decompositions act on spaces of dimension at most $\binom{m}{p}$ or $m$, and both Gram-splitting calls have polynomial arithmetic cost. Each recovered block has $\binom{m}{p}$ coefficients, there are at most $n\leq m/p$ such blocks, and the core is stored through at most $m$ orthonormal one-particle vectors. Thus constructing and storing the returned description has polynomial cost in $m$. Every operation uses the supplied RDMs and classically recovered vectors, so no additional copies of the target are required.
\end{proof}

\section{Exact reconstruction for heterogeneous blocks}
\label{app:heterogeneous-blocks}

We extend the exact reconstruction guarantee to the heterogeneous setting, allowing different block particle numbers, an always-occupied component, and additional vacuum modes. After fixing the state notation, we state the guarantee and construct the algorithm by proving the spectral identities and product-removal property required at each order.

\subsection{Problem reformulation and notation}
\label{app:heterogeneous-problem-reformulation}

Fix a known integer $2\leq r\leq m$. Let $I_0\subseteq[m]$. For each block $b\in[n]$, fix integers $2\leq p_b\leq r$ and $s_b\geq2$, and let $I_{b,l}\subseteq[m]$ satisfy $|I_{b,l}|=p_b$ for $l\in[s_b]$. Assume that $I_0$ and all sets $I_{b,l}$ are mutually disjoint, and that all coefficients $\omega_{b,l}$ are nonzero and satisfy $\sum_{l=1}^{s_b}|\omega_{b,l}|^2=1$.

Modes outside $I_0\cup\bigcup_{b,l}I_{b,l}$ are in the vacuum state. Use the input block states and transformed branch and block notation of App.~\ref{app:uniform-problem-reformulation}, with $p$ replaced by $p_b$ for block $b$. In particular, $\omega_b=(\wedge^{p_b}U)\omega_b^{\mathrm{in}}\in\wedge^{p_b}B_b$ is normalized. For the always-occupied component, define
\begin{align*}
    S:=\operatorname{span}\{Ue_i:i\in I_0\}, \qquad \sigma_S:=(\wedge^{|I_0|}U)e_{I_0}.
\end{align*}
When $I_0=\varnothing$, use $S=\{0\}$ and $\sigma_S=1$. The full input state and its particle number are
\begin{align}
    \Psi_{\mathrm{in}}:=e_{I_0}\wedge\omega_1^{\mathrm{in}}\wedge\cdots\wedge\omega_n^{\mathrm{in}}, \qquad N:=|I_0|+\sum_{b=1}^{n}p_b.
    \label{eq:heterogeneous-input-target}
\end{align}
The spaces $S,B_1,\ldots,B_n$ are mutually orthogonal, and the target is
\begin{align}
    \Psi :=(\wedge^NU)\Psi_{\mathrm{in}} =\sigma_S\wedge\omega_1\wedge\cdots\wedge\omega_n \in\wedge^NH.
    \label{eq:heterogeneous-output-target}
\end{align}

\subsection{Reconstruction from exact RDMs}
\label{app:different-particle-algorithm}
\label{app:different-particle-correctness}

At order $k$, the RDM can contain both individual $k$-particle blocks and products of smaller blocks. We therefore process the orders in increasing sequence, using the blocks recovered at earlier orders to remove these product directions before applying Gram splitting. As in the homogeneous case, we allow a controlled truncation of dominant blocks so that the construction extends to estimated RDMs.

\begin{samepage}
\begin{theorem}[Reconstruction from exact RDMs]\label{thm:ideal-delta-free-correctness}
Let $\Psi$ be an $m$-mode state in the heterogeneous setting, with block particle numbers at most a known $2\leq r\leq m$. Given $r$, $\varepsilon_{\mathrm{fid}}\in(0,1)$, and exact classical descriptions of $\Gamma_\Psi^{(1)},\ldots,\Gamma_\Psi^{(r)}$, a randomized classical algorithm returns, with probability one, a compact classical description of a normalized state $\Phi$ satisfying
\begin{align}
    1-|\langle\Phi,\Psi\rangle|^2&\leq\frac{\varepsilon_{\mathrm{fid}}}{16}.
    \label{eq:heterogeneous-exact-reconstruction-guarantee}
\end{align}
The description consists of an orthonormal basis for a fully occupied subspace and a list of normalized block vectors, each with particle number between $2$ and $r$, all with mutually orthogonal one-particle supports. The individual particle numbers, block decomposition, and passive Gaussian unitary are not required as input. For fixed $r$, the algorithm uses polynomially many arithmetic operations in $m$ and no additional copies of the target state.
\end{theorem}
\end{samepage}

We first identify the common high-occupation space. We then process $k=2,\ldots,r$ in increasing order, performing the dominant- and complementary-block recovery steps below at each order, and finally assemble the recovered blocks with the remaining occupied core. Both recovery steps use only outputs from smaller orders. We use the convention $\Gamma_\Psi^{(k)}=0$ for $k>N$.

\subsubsection{Identifying the high-occupation space}

Choose $\theta\in[2/3,3/4]$ and compute
\begin{align}
    S_\theta&:=\operatorname{ran}\mathbf1_{(\theta,\infty)}\left(\Gamma_\Psi^{(1)}\right).
    \label{eq:ideal-high-occupation-space}
\end{align}
To identify the branches selected by this threshold, set
\begin{align}
    w_b&:=\max_{l\in[s_b]}|\omega_{b,l}|^2,\qquad \mathcal N_\theta:=\{b\in[n]:w_b>\theta\}.
    \label{eq:heterogeneous-near-block-set}
\end{align}
We call the blocks in $\mathcal N_\theta$ dominant and the others complementary. Since $\theta>1/2$, each dominant block has a unique branch of weight above $\theta$; denote its occupied one-particle support by $F_b$.

\begin{lemma}[High-occupation space]\label{lem:heterogeneous-high-space}
The $1$-RDM and its high-occupation space satisfy
\begin{align}
    \Gamma_\Psi^{(1)}&=P_S+\sum_{b=1}^n\sum_{l=1}^{s_b}|\omega_{b,l}|^2P_{F_{b,l}},\qquad S_\theta=S\oplus\bigoplus_{b\in\mathcal N_\theta}F_b.
    \label{eq:exact-high-occupation-decomposition}
\end{align}
Every other branch support lies in $S_\theta^\perp$.
\end{lemma}
\begin{proof}
The blockwise factorization in Prop.~\ref{prop:blockwise-rdm-factorization} and the individual-block formula in Prop.~\ref{prop:correlated-block-rdm} give the block contributions, while the always-occupied component contributes $P_S$ by Eq.~\eqref{eq:occupied-subspace-rdm}. These operators have mutually orthogonal supports. Spectral thresholding therefore retains $S$ and precisely the branches of weight above $\theta$.
\end{proof}

For $k=2,\ldots,r$, form the cross and all-low compressions
\begin{align}
    Q_{\mathrm{hi},k}&:=P_{\wedge^kS_\theta},\qquad Q_{\mathrm{lo},k}:=P_{\wedge^k(S_\theta^\perp)},\qquad C_k:=Q_{\mathrm{lo},k}\Gamma_\Psi^{(k)}Q_{\mathrm{hi},k},\qquad A_k:=Q_{\mathrm{lo},k}\Gamma_\Psi^{(k)}Q_{\mathrm{lo},k}.
    \label{eq:ideal-compressed-operators}
\end{align}
The cross compression accesses dominant-block coherence, while the all-low compression contains the complementary blocks. Since $S\subseteq S_\theta$, the low projection removes any sector selecting particles from the always-occupied component.

\subsubsection{Recovering the dominant blocks}

For $b\in\mathcal N_\theta$, write $\omega_b=\sqrt{w_b}f_b+r_b$ with the phase convention of Eq.~\eqref{eq:uniform-near-block-decomposition}, now with $p$ replaced by $p_b$. Thus $f_b$ occupies $F_b$, $r_b\in\wedge^{p_b}G_b$ for $G_b:=B_b\cap F_b^\perp$, and $\sigma_b:=\sqrt{w_b(1-w_b)}\leq1/2$ is the block's coherence scale. To describe products of complete blocks, set $p(J):=\sum_{b\in J}p_b$ for $J\subseteq[n]$, and define
\begin{align}
    f_J&:=\bigwedge_{b\in J}f_b,\qquad r_J:=\bigwedge_{b\in J}r_b,\qquad J\subseteq\mathcal N_\theta,
    \label{eq:heterogeneous-subset-products}
\end{align}
with factors ordered by increasing block label. The following lemma identifies both the individual branches and the product directions in $C_k$.

\begin{lemma}[Structure of the cross compression]\label{lem:heterogeneous-coherence-decomposition}
For every $k=2,\ldots,r$,
\begin{align}
    C_k&=\sum_{\substack{J\subseteq\mathcal N_\theta\\p(J)=k}}\left(\prod_{b\in J}\sqrt{w_b}\right)|r_J\rangle\langle f_J|.
    \label{eq:heterogeneous-coherence-decomposition}
\end{align}
All products use the same fixed block order on the bra and ket sides. The nonzero singular values are
\begin{align}
    \prod_{b\in J}\sigma_b, \qquad J\subseteq\mathcal N_\theta, \qquad p(J)=k,
    \label{eq:heterogeneous-coherence-singular-values}
\end{align}
and the corresponding right singular vectors may be chosen from the orthonormal family $\{f_J\}$. In particular, for every
$b\in\mathcal N_\theta$,
\begin{align}
    C_{p_b}f_b&=\sqrt{w_b}r_b, \qquad \Gamma_\Psi^{(p_b)}f_b=\sqrt{w_b}\omega_b.
    \label{eq:exact-near-identities}
\end{align}
In particular, the rank is bounded independently of the branch weights:
\begin{align}
    \operatorname{rank}C_k&\leq\sum_{j=1}^{\lfloor k/2\rfloor}\binom nj\leq r m^{\lfloor k/2\rfloor}.
    \label{eq:df-coherence-rank-count}
\end{align}
\end{lemma}

\begin{proof}
A matrix element of $C_k$ compares removal from all-high modes with removal from all-low modes. Selecting a core particle is impossible on the low side, and selecting a complementary-block particle is impossible on the high side. A dominant block can contribute only if all its particles are removed: otherwise a residual particle distinguishes its high branch from every low branch and makes the overlap zero. By Prop.~\ref{prop:correlated-block-rdm}, the local cross compression is
\begin{align}
    P_{\wedge^qG_b}\Gamma_{\omega_b}^{(q)}P_{\wedge^qF_b}&=\begin{cases}
        1,&q=0,\\
        0,&1\leq q<p_b,\\
        \sqrt{w_b}|r_b\rangle\langle f_b|,&q=p_b.
    \end{cases}
    \label{eq:heterogeneous-local-cross-compression}
\end{align}
Thus the surviving contributions are exactly the subsets $J$ of complete dominant blocks with $p(J)=k$. Prop.~\ref{prop:blockwise-rdm-factorization} multiplies the local contributions, with the same fixed factor order on the bra and ket sides. Its sign cancellation gives the coefficient $\prod_{b\in J}\sqrt{w_b}$. This proves the formula for $C_k$.

Distinct subsets use different collections of orthogonal block supports, so the $f_J$ are orthonormal, as are the normalized $r_J$. Since $\|r_J\|=\prod_{b\in J}\|r_b\|$,
\begin{align}
    C_k^\dagger C_k&=\sum_{\substack{J\subseteq\mathcal N_\theta\\p(J)=k}} \left(\prod_{b\in J}\sigma_b^2\right)|f_J\rangle\langle f_J|.
\end{align}
This proves the singular-value statement. For a singleton $J=\{b\}$ it gives $C_{p_b}f_b=\sqrt{w_b}r_b$. The column $f_b$ belongs to the sector selecting all $p_b$ particles from block $b$ and none from the other factors. Products of other blocks therefore cannot contribute to this column of the uncompressed RDM. Its restriction is $\Pi_{\omega_b}$, giving $\Gamma_\Psi^{(p_b)}f_b=|\omega_b\rangle\langle\omega_b|f_b=\sqrt{w_b}\omega_b$.
Finally, each contributing subset has at most $\lfloor k/2\rfloor$ blocks because every block contains at least two particles. Counting such subsets proves Eq.~\eqref{eq:df-coherence-rank-count}.
\end{proof}

Since a block with small $\sigma_b$ is close to its dominant branch, use the resolution scale $s_0^2=\varepsilon_{\mathrm{fid}}/(48m)$ from Eq.~\eqref{eq:equal-p-df-s0} and choose thresholds $t_k\in[s_0,2s_0]$ for $k=2,\ldots,r$. At order $k$, compute
\begin{align}
    R_k&:=\operatorname{ran}\mathbf1_{(t_k^2,\infty)}(C_k^\dagger C_k).
    \label{eq:ideal-thresholded-right-space}
\end{align}
For the proof, denote the retained dominant blocks by
\begin{align}
    \mathcal R&:=\{b\in\mathcal N_\theta:\sigma_b>t_{p_b}\}.
    \label{eq:retained-near-set}
\end{align}
Lem.~\ref{lem:heterogeneous-coherence-decomposition} gives
\begin{align}
    R_k&=\operatorname{span}\left\{f_J:J\subseteq\mathcal N_\theta,\ p(J)=k,\ \prod_{b\in J}\sigma_b>t_k\right\}.
    \label{eq:exact-thresholded-right-space}
\end{align}
This space can contain products of several branches, whose one-particle supports need not be disjoint from those of other candidates. We must remove these products before invoking Gram splitting.

Let $O_k^{\mathrm N}$ be the span of exterior products of at least two distinct dominant branches recovered at smaller orders, each used at most once, with total degree $k$. Use $O_k^{\mathrm N}=\{0\}$ when the list is empty. In particular, no such product exists at orders $k=2,3$. Compute
\begin{align}
    F_k^{\mathrm N}&:=R_k\cap(O_k^{\mathrm N})^\perp,\qquad d_k^{\mathrm N}:=\dim F_k^{\mathrm N}.
    \label{eq:ideal-new-near-space}
\end{align}
Under the induction hypothesis that all retained dominant branches at smaller orders have been recovered up to phases and a permutation, the constructed old-product space is
\begin{align}
    O_k^{\mathrm N}&=\operatorname{span}\{f_J:J\subseteq\mathcal R,\ |J|\geq2,\ p(J)=k\}.
    \label{eq:near-old-new-spaces}
\end{align}
The common threshold interval ensures that these earlier outputs suffice to remove every retained product.

\begin{lemma}[Removal of earlier branch products]\label{lem:near-hierarchy-closure}
Let $J\subseteq\mathcal N_\theta$ satisfy $|J|\geq2$, $p(J)=k$, and $\prod_{b\in J}\sigma_b>t_k$. Then, $J\subseteq\mathcal R$. If all retained dominant branches of degree below $k$ have been recovered up to phases and a permutation,
\begin{align}
    R_k&=(R_k\cap O_k^{\mathrm N})\oplus F_k^{\mathrm N}, \qquad F_k^{\mathrm N} =R_k\cap(O_k^{\mathrm N})^\perp =\operatorname{span}\{f_b:b\in\mathcal R,\ p_b=k\}.
    \label{eq:near-old-new-decomposition}
\end{align}
Moreover, $P_{R_k}$ and $P_{O_k^{\mathrm N}}$ commute.
\end{lemma}

\begin{proof}
A retained product cannot contain an unresolved constituent. Indeed, if $b\in J$ were unresolved, then $\sigma_b\leq t_{p_b}\leq2s_0$. Since $|J|\geq2$ and every other coherence scale is at most $1/2$,
\begin{align}
    \prod_{c\in J}\sigma_c&\leq2s_0\cdot\frac12\leq s_0\leq t_k,
\end{align}
contradicting retention of that product. Every nonsingleton direction in $R_k$ is therefore generated by already recovered lower-order branches.

Since both $R_k$ and $O_k^{\mathrm N}$ are spanned by subsets of the same orthonormal family $\{f_J:p(J)=k\}$, their projectors commute, and removing the old-product coordinates from $R_k$ leaves precisely the retained singleton coordinates. This proves the stated decomposition. Some old products can lie below the current threshold, so $O_k^{\mathrm N}\subseteq R_k$ is not required.
\end{proof}

The remaining space $F_k^{\mathrm N}$ is therefore spanned by the new retained dominant branches, which have mutually orthogonal one-particle supports. Apply Prop.~\ref{prop:exact-gram-splitting} to an orthonormal basis of this space. If the call returns \textsc{Fail}, return \textsc{Fail}; otherwise denote its outputs by $\hat f_{k,1},\ldots,\hat f_{k,d_k^{\mathrm N}}$ and reconstruct
\begin{align}
    \hat\omega_{k,j}^{\mathrm N}&:=\frac{\Gamma_\Psi^{(k)}\hat f_{k,j}}{\|\Gamma_\Psi^{(k)}\hat f_{k,j}\|},\qquad j\in[d_k^{\mathrm N}].
    \label{eq:heterogeneous-exact-near-block-recovery}
\end{align}
With probability one, the recovered branches agree with the $f_b$ for $b\in\mathcal R$, $p_b=k$, up to phases and a permutation. Eq.~\eqref{eq:exact-near-identities} then shows that the normalization denominator is $\sqrt{w_b}>\sqrt\theta$ and that each output is the corresponding full block. Store the branch list for constructing $O_\ell^{\mathrm N}$ at later orders and the full-block list for final assembly.

\subsubsection{Recovering the complementary blocks}

For the complementary blocks, set
\begin{align}
    \mathcal H_\theta&:=[n]\setminus\mathcal N_\theta=\{b\in[n]:w_b\leq\theta\},
    \label{eq:heterogeneous-strong-block-set}
\end{align}
and write $\omega_J:=\bigwedge_{b\in J}\omega_b$ for $J\subseteq\mathcal H_\theta$, with factors in increasing block order. At order $k$, compute the eigenvalue-one space of the all-low compression:
\begin{align}
    W_k^{\mathrm H}&:=\operatorname{ran}\mathbf1_{\{1\}}(A_k).
    \label{eq:ideal-strong-all-space}
\end{align}
As with the dominant branches, we remove products of already recovered factors. Let $O_k^{\mathrm H}$ be the span of exterior products of at least two distinct complementary block states recovered at smaller orders, each used at most once, with total degree $k$, and use the zero space if there is no such product. Compute
\begin{align}
    T_k^{\mathrm H}&:=W_k^{\mathrm H}\cap(O_k^{\mathrm H})^\perp,\qquad d_k^{\mathrm H}:=\dim T_k^{\mathrm H}.
    \label{eq:ideal-new-strong-space}
\end{align}
Under the induction hypothesis that all complementary blocks at smaller orders have been recovered up to phases and a permutation,
\begin{align}
    O_k^{\mathrm H}&=\operatorname{span}\{\omega_J:J\subseteq\mathcal H_\theta,\ |J|\geq2,\ p(J)=k\}.
    \label{eq:strong-all-old-new-spaces}
\end{align}
The next lemma identifies the surviving blocks and gives the spectral gap needed when the RDMs are estimated.

\begin{lemma}[Gap of the all-low compression]\label{lem:constant-gap-strong-hierarchy}
The eigenvalue-one space of $A_k$ is
$W_k^{\mathrm H}=\operatorname{span}\{\omega_J:J\subseteq\mathcal H_\theta,\ p(J)=k\}$,
and, on the all-low space,
\begin{align}
    0\leq A_k-P_{W_k^{\mathrm H}}&\leq\theta\left(I_{\wedge^k(S_\theta^\perp)}-P_{W_k^{\mathrm H}}\right).
    \label{eq:strong-constant-gap-remainder}
\end{align}
Thus every remaining eigenvalue is at most $\theta$, with spectral gap at least $1-\theta\geq1/4$. If all complementary blocks of degree below $k$ have been recovered up to phases and a permutation,
\begin{align}
    W_k^{\mathrm H}&=O_k^{\mathrm H}\oplus T_k^{\mathrm H}, \qquad T_k^{\mathrm H}=W_k^{\mathrm H}\cap(O_k^{\mathrm H})^\perp =\operatorname{span}\{\omega_b:b\in\mathcal H_\theta,\ p_b=k\}.
    \label{eq:strong-old-new-decomposition}
\end{align}
\end{lemma}

\begin{proof}
Apply Prop.~\ref{prop:blockwise-rdm-factorization} to the orthogonal decomposition into $S$, the block supports, and the unused modes. The low projection removes any sector selecting particles from $S$, while unused modes contribute zero. By Prop.~\ref{prop:correlated-block-rdm}, a complete complementary block contributes $\Pi_{\omega_b}$ and a partial complementary block contributes an operator of norm $w_b\leq\theta$. For a dominant block, a partial selection has norm at most the largest nondominant branch weight, hence at most $1-w_b<\theta$; a complete selection gives $|r_b\rangle\langle r_b|$ with norm $1-w_b$.

The factors multiply within each fixed block-selection sector, and different sectors are orthogonal. The tensor-product eigenvectors with eigenvalue one are precisely the products $\omega_J$ of complete complementary blocks with $p(J)=k$, and they span $W_k^{\mathrm H}$. The orthogonal directions within those sectors have eigenvalue zero. In every other nonzero sector, at least one factor has norm at most $\theta$ and all others have norm at most one, so the sector norm is at most $\theta$.

The sum of the eigenvalue-one projectors is $P_{W_k^{\mathrm H}}$. The remaining part $A_k-P_{W_k^{\mathrm H}}$ is positive, vanishes on $W_k^{\mathrm H}$, and is bounded by $\theta$ on its orthogonal complement. Finally, products with $|J|\geq2$ belong to $O_k^{\mathrm H}$, while singletons are the new $k$-particle blocks. These orthogonal subsets give Eq.~\eqref{eq:strong-old-new-decomposition}.
\end{proof}

The lemma verifies the input structure of Prop.~\ref{prop:exact-gram-splitting} for $T_k^{\mathrm H}$. Apply that algorithm to an orthonormal basis of this space, returning \textsc{Fail} if the call does so. With probability one, its outputs $\hat\omega_{k,1}^{\mathrm H},\ldots,\hat\omega_{k,d_k^{\mathrm H}}^{\mathrm H}$ are precisely the complementary blocks of particle number $k$, up to phases and a permutation. Store them both for constructing $O_\ell^{\mathrm H}$ at later orders and for final assembly. The old-product lists thus contain dominant branches for the cross compression and complementary full blocks for the all-low compression.

To verify the recursion, start at $k=2$, where both old-product spaces are zero. At a general order $k$, every product of at least two blocks has constituents of particle number at most $k-2$. If the earlier calls have succeeded, Lems.~\ref{lem:near-hierarchy-closure} and~\ref{lem:constant-gap-strong-hierarchy} therefore justify both recovery steps at order $k$. Induction shows that all blocks in $\mathcal R\cup\mathcal H_\theta$ are recovered, together with the dominant branches indexed by $\mathcal R$. Each call succeeds with conditional probability one by Prop.~\ref{prop:exact-gram-splitting}, and there are at most $2(r-1)$ calls, so the entire recursion succeeds with probability one.

\subsubsection{Assembling the state and bounding the error}

After level $r$, the high-occupation modes not used by recovered dominant branches belong to the original occupied component or to unresolved dominant branches. Compute
\begin{align}
    K&:=\sum_{k=2}^r\sum_{j=1}^{d_k^{\mathrm N}}\Gamma_{\hat f_{k,j}}^{(1)},\qquad F:=\operatorname{ran}K,\qquad S_{\mathrm{core}}:=S_\theta\cap F^\perp.
    \label{eq:ideal-core-space}
\end{align}
For an empty branch list, use $K=0$ and $F=\{0\}$. Choose an orthonormal basis of $S_{\mathrm{core}}$ and let $\hat\sigma_{\mathrm{core}}$ be the normalized state occupying all its basis modes, using the vacuum when the core is zero-dimensional. Form
\begin{align}
    \Phi&:=\hat\sigma_{\mathrm{core}}\wedge\bigwedge_{k=2}^r\bigwedge_{j=1}^{d_k^{\mathrm N}}\hat\omega_{k,j}^{\mathrm N}\wedge\bigwedge_{k=2}^r\bigwedge_{j=1}^{d_k^{\mathrm H}}\hat\omega_{k,j}^{\mathrm H}.
    \label{eq:ideal-delta-free-output}
\end{align}
Return the core basis and the two ordered block lists, with all factors in a fixed deterministic order. The following lemma identifies the returned state and bounds the loss from replacing unresolved blocks by their dominant branches.

\begin{lemma}[Output and truncation]\label{lem:heterogeneous-exact-output}
On the probability-one event that all Gram-splitting calls recover their factors, the core is
\begin{align}
    S_{\mathrm{core}}&=S\oplus\bigoplus_{b\in\mathcal N_\theta\setminus\mathcal R}F_b.
    \label{eq:exact-core-decomposition}
\end{align}
For the analysis, define the truncated target $\Psi_{\mathrm{tr}}$ using the exact retained blocks and any normalized occupied state $\sigma_{\mathrm{core}}$ of this space. The exact-RDM output $\Phi$ agrees with this comparison state up to an overall phase:
\begin{align}
    \Psi_{\mathrm{tr}}&:=\sigma_{\mathrm{core}}\wedge\bigwedge_{b\in\mathcal R}\omega_b\wedge\bigwedge_{b\in\mathcal H_\theta}\omega_b.
    \label{eq:truncated-target}
\end{align}
These factors have mutually orthogonal one-particle supports, so $\Phi$ is normalized. Its infidelity satisfies
\begin{align}
    1-|\langle\Psi_{\mathrm{tr}},\Psi\rangle|^2&\leq3ms_0^2=\frac{\varepsilon_{\mathrm{fid}}}{16}.
    \label{eq:ideal-truncation-bound}
\end{align}
\end{lemma}
\begin{proof}
After matching factors and phases, Eq.~\eqref{eq:occupied-subspace-rdm} gives $K=\sum_{b\in\mathcal R}P_{F_b}$. Removing its range from the high space in Lem.~\ref{lem:heterogeneous-high-space} proves Eq.~\eqref{eq:exact-core-decomposition}. The wedge product of $\sigma_S$ and the unresolved dominant branches occupies every mode in this core and therefore agrees with $\hat\sigma_{\mathrm{core}}$ up to phase. Since all remaining blocks have been recovered up to phases and permutations, the returned state is $\Psi_{\mathrm{tr}}$ up to an overall phase.

For each unresolved block, $\sigma_b\leq t_{p_b}\leq2s_0$ and $w_b>\theta\geq2/3$, giving $1-w_b=\sigma_b^2/w_b\leq6s_0^2$. Only these blocks change in $\Psi_{\mathrm{tr}}$, so orthogonality of the block supports gives
\begin{align}
    |\langle\Psi_{\mathrm{tr}},\Psi\rangle|^2&=\prod_{b\in\mathcal N_\theta\setminus\mathcal R}w_b,\qquad 1-|\langle\Psi_{\mathrm{tr}},\Psi\rangle|^2\leq\sum_{b\in\mathcal N_\theta\setminus\mathcal R}(1-w_b)\leq6ns_0^2\leq3ms_0^2.
\end{align}
The last inequality uses $p_b\geq2$ and $|I_0|+\sum_bp_b=N\leq m$.
\end{proof}

\begin{proof}[Completion of the proof of Thm.~\ref{thm:ideal-delta-free-correctness}]
The induction above establishes probability-one recovery at every order. Lem.~\ref{lem:heterogeneous-exact-output} then proves normalization, the claimed form of the output, and the fidelity guarantee.

For fixed $r$, all projectors, compressions, and spectral decompositions act on spaces of dimension at most $m^r$. At order $k$, each old product uses at most $\lfloor k/2\rfloor$ factors from lists containing at most $m/2$ blocks, so enumerating the products, computing their spans, and removing those spans require polynomially many arithmetic operations in $m$. The at most $2(r-1)$ Gram-splitting calls have polynomial arithmetic cost by Prop.~\ref{prop:exact-gram-splitting}. Finally, at most $m/2$ recovered blocks are stored, each with at most $m^r$ coefficients, and the core is stored through at most $m$ orthonormal one-particle vectors. Thus the returned description and its construction have polynomial cost in $m$. Every operation uses the supplied RDMs and classically recovered vectors, so no additional target-state copies are required.
\end{proof}

\section{Stability of the reconstruction steps}
\label{app:perturbation-tools}
\label{app:mixed-hierarchy-stability}

We now bound the errors in the reconstruction steps, following their computational order: selecting subspaces from the RDM estimates, removing previously recovered products, recovering blocks, constructing the remaining occupied subspace, and assembling the final state. Throughout this appendix, exact RDMs belong to the state family of App.~\ref{app:heterogeneous-problem-reformulation}, and their Hermitian estimates have operator-norm errors at most $\nu$. The homogeneous case uses only orders $1$ and $p$, with $\nu$ replaced by $\mu$. All thresholds used to define exact comparison spaces are the same thresholds chosen from the estimates.

\subsection{Thresholds and compressed RDMs}
\label{app:one-level-reconstruction-stability}

We first state the spectral-projector bound and then apply it to thresholds chosen from the estimated spectra. This controls the spaces extracted from the RDMs and their compressions.

\begin{lemma}[Spectral-projector perturbation~\cite{DavisKahan1970}]\label{lem:projector-perturbation}
Let $A$ and $\widetilde A$ be Hermitian operators on the same finite-dimensional Hilbert space, and suppose that $\|\widetilde A-A\|\leq\epsilon$. Let $P$ be a spectral projector of $A$, with $P\neq0,I$. Assume that either the eigenvalues selected by $P$ or those selected by $I-P$ occupy consecutive positions in the ordered spectrum of $A$, counting multiplicities. Suppose that the selected and unselected eigenvalues are separated by a gap of at least $g>0$.

If $\epsilon<g/2$, let $\widetilde P$ select the eigenvalues of $\widetilde A$ in the same ordered positions as those selected by $P$. Then, $\widetilde P$ has the same rank as $P$, and
$\|\widetilde P-P\| \leq\frac{\epsilon}{g-\epsilon} \leq\frac{2\epsilon}{g}$.
The cases $P=0$ and $P=I$ hold trivially with $\widetilde P=P$.
\end{lemma}

\begin{proof}
Suppose first that $P$ selects a consecutive group. Weyl's inequality moves each ordered eigenvalue by at most $\epsilon$, so $2\epsilon<g$ keeps this group separated from the remaining perturbed eigenvalues. Thus $\widetilde P$ is well defined and has the same rank as $P$. The exact selected spectral interval is separated from the unselected spectrum of $\widetilde A$ by at least $g-\epsilon$. The operator-norm Davis--Kahan theorem~\cite{DavisKahan1970} therefore gives $\|(I-\widetilde P)P\|\leq\epsilon/(g-\epsilon)$. For equal-rank orthogonal projectors, $\|\widetilde P-P\|=\|(I-\widetilde P)P\|$, yielding the claimed bound since $g-\epsilon>g/2$. If instead $I-P$ selects a consecutive group, apply the same argument to the complementary projectors, whose difference has the same norm.
\end{proof}

\begin{lemma}[Stability at a separated threshold]\label{lem:df-certified-threshold}
Let $A$ and $\widetilde A$ be Hermitian operators and suppose $\|\widetilde A-A\|\leq\varepsilon_{\mathrm{op}}$. Suppose that the learner chooses a threshold $\tau$ satisfying $\operatorname{dist}(\tau,\operatorname{spec}\widetilde A)\geq h/2$. If $\varepsilon_{\mathrm{op}}<h/4$, then we have $\operatorname{dist}(\tau,\operatorname{spec}A)\geq h/4$. The exact and noisy spectral projectors onto eigenvalues larger than $\tau$ have the same rank and satisfy
\begin{align}
    \|P_{\widetilde E}-P_E\| &\leq \frac{8\varepsilon_{\mathrm{op}}}{h}.
    \label{eq:df-threshold-projector-error}
\end{align}
For square operators $C$ and $\widetilde C$ on the same space and a positive threshold $\tau$, the same conclusions hold for the right singular spaces, assuming $\|\widetilde C-C\|\leq\varepsilon_{\mathrm{op}}$ and $\operatorname{dist}(\tau,\operatorname{sing}\widetilde C)\geq h/2$.
\end{lemma}

\begin{proof}
Write $e=\varepsilon_{\mathrm{op}}$. Weyl's inequality gives $\operatorname{dist}(\tau,\operatorname{spec}A)\geq h/2-e>h/4$. The eigenvalues of $A$ and $\widetilde A$ in matching ordered positions are on the same side of $\tau$, so the selected ranks agree. If both spectral groups are nonempty, the observed selected and unselected eigenvalues are separated by at least $h$. Apply Lem.~\ref{lem:projector-perturbation} with $\widetilde A$ as the reference operator. It gives $\|P_{\widetilde E}-P_E\|\leq e/(h-e)\leq8e/h$. For an empty or full selected group, equal ranks make the two projectors identical.

For singular spaces, the singular-value perturbation bound gives the same distance from $\tau$ and equality of the selected ranks. Let $E$ and $\widetilde E$ denote the right singular spaces above $\tau$. If the selected rank is neither zero nor full, the observed singular values above and below $\tau$ are separated by at least $h$. Applying the operator-norm form of Wedin's $\sin\Theta$ theorem~\cite{Wedin1972}, with $\widetilde C$ as the reference matrix, gives
\begin{align*}
    \|P_{\widetilde E}-P_E\|
    &\leq\frac{\sqrt{2}\,e}{h-e}
    \leq\frac{8e}{h}.
\end{align*}
The empty and full cases again follow from equality of the ranks.

\end{proof}

\begin{lemma}[Selection of separated thresholds]\label{lem:df-threshold-selection}
Given a Hermitian matrix $\widetilde A$ of dimension $D\geq1$ and a closed interval $J$ of length $L>0$, one can choose $\tau\in J$ from its observed spectrum such that $\operatorname{dist}(\tau,\operatorname{spec}\widetilde A)\geq h/2$, with $h=L/(8D)$.

For singular values, suppose $s>0$, $\operatorname{rank}C\leq M$ for a supplied integer $M\geq1$, and $\|\widetilde C-C\|\leq e<s/2$. Given $\widetilde C$, $s$, and $M$, one can choose $\tau\in[s,2s]$ such that $\operatorname{dist}(\tau,\operatorname{sing}\widetilde C)\geq h/2$, with $h=s/(8M)$. Neither $C$ nor its singular vectors are required as input. Both selections use only a spectral decomposition and sorting.
\end{lemma}

\begin{proof}
In the Hermitian case, exclude intervals of radius $h/2$ around the observed eigenvalues. Their total length is at most $Dh=L/8$, so they cannot cover $J$. Sorting their endpoints finds a point in the complement.

In the singular-value case, perturbation bounds place every singular value of $\widetilde C$ after the first $M$ at most $e<s/2$. These values are more than $s/2$ away from $[s,2s]$. The exclusion intervals around the remaining at most $M$ singular values have total length at most $Mh=s/8$, leaving a valid threshold in $[s,2s]$. This argument also applies when $\widetilde C$ has full rank; if $C=0$, one may use $M=1$.
\end{proof}

\begin{lemma}[Stability of RDM compressions]\label{lem:df-compressed-operator-error}
Fix $k\geq2$ and suppose $\|\Gamma_\Psi^{(k)}\|\leq1$, as holds for the block-product states considered here by Eq.~\eqref{eq:structured-rdm-norm}. Suppose that the Hermitian RDM estimate and the estimated high-occupation space satisfy
$\|\widetilde\Gamma^{(k)}-\Gamma_\Psi^{(k)}\| \leq \nu, \|P_{\widetilde S_\theta}-P_{S_\theta}\| \leq \eta_S$.
Let $Q_{\mathrm{hi},k}:=P_{\wedge^kS_\theta}$ and $Q_{\mathrm{lo},k}:=P_{\wedge^k(S_\theta^\perp)}$, and define $C_k:=Q_{\mathrm{lo},k}\Gamma_\Psi^{(k)}Q_{\mathrm{hi},k}$ and $A_k:=Q_{\mathrm{lo},k}\Gamma_\Psi^{(k)}Q_{\mathrm{lo},k}$. Define their estimated counterparts using $\widetilde S_\theta$ and $\widetilde\Gamma^{(k)}$.

Then, the lifted projectors satisfy
\begin{align}
    \|P_{\wedge^k\widetilde S_\theta}-P_{\wedge^kS_\theta}\|&\leq k\eta_S, \qquad \|P_{\wedge^k(\widetilde S_\theta^\perp)} -P_{\wedge^k(S_\theta^\perp)}\|\leq k\eta_S.
    \label{eq:df-exterior-projector-bound}
\end{align}
Moreover, by setting $e_k:=\nu+2k\eta_S$, the compressed operators
satisfy
\begin{align}
    \max\left\{ \|\widetilde C_k-C_k\|, \|\widetilde A_k-A_k\| \right\} &\leq e_k.
    \label{eq:df-compressed-error}
\end{align}
\end{lemma}

\begin{proof}
Set $P:=P_{S_\theta}$ and $\widetilde P:=P_{\widetilde S_\theta}$. By Eq.~\eqref{eq:lifted-projector-notation}, the lifted projectors are the restrictions of $P^{\otimes k}$ and $\widetilde P^{\otimes k}$ to the antisymmetric subspace, respectively. Therefore, the identity
$\widetilde P^{\otimes k}-P^{\otimes k} = \sum_{j=1}^{k} \widetilde P^{\otimes(j-1)} \otimes(\widetilde P-P) \otimes P^{\otimes(k-j)}$
yields the first bound in Eq.~\eqref{eq:df-exterior-projector-bound}. The second follows in the same way because
$\|P_{\widetilde S_\theta^\perp}-P_{S_\theta^\perp}\| = \|P_{\widetilde S_\theta}-P_{S_\theta}\| \leq\eta_S$.
The assumed bound $\|\Gamma_\Psi^{(k)}\|\leq1$ controls both compressions.

For the cross compression, write
\begin{align}
    \widetilde C_k-C_k &= \widetilde Q_{\mathrm{lo},k} (\widetilde\Gamma^{(k)}-\Gamma_\Psi^{(k)}) \widetilde Q_{\mathrm{hi},k} + (\widetilde Q_{\mathrm{lo},k}-Q_{\mathrm{lo},k}) \Gamma_\Psi^{(k)} \widetilde Q_{\mathrm{hi},k} + Q_{\mathrm{lo},k}\Gamma_\Psi^{(k)} (\widetilde Q_{\mathrm{hi},k}-Q_{\mathrm{hi},k}).
\end{align}
All orthogonal projectors have norm at most one, so Eq.~\eqref{eq:df-exterior-projector-bound} implies
$\|\widetilde C_k-C_k\| \leq\nu+k\eta_S+k\eta_S =e_k$.
Replacing the high-sector projectors with the corresponding low-sector projectors implies the same bound for $\widetilde A_k-A_k$.
\end{proof}

The following corollary collects the conditions and error bounds for the three candidate spaces used in the reconstruction.

\begin{corollary}[Stability of the selected subspaces]\label{cor:df-threshold-space-errors}
Consider the heterogeneous setting with $2\leq k\leq r$, or the homogeneous setting with $k=p$, and fix $\theta\in[2/3,3/4]$. Suppose the Hermitian estimates of the $1$- and $k$-RDMs have operator-norm errors at most $\nu$. Let $\widetilde S_\theta$ select eigenvectors of $\widetilde\Gamma^{(1)}$ above $\theta$, and form $\widetilde C_k$ and $\widetilde A_k$ as in Lem.~\ref{lem:df-compressed-operator-error}. Suppose that, for $t_k,h_1,h_k>0$,
\begin{align}
    \operatorname{dist}(\theta,\operatorname{spec}\widetilde\Gamma^{(1)})&\geq h_1/2,\qquad \operatorname{dist}(t_k,\operatorname{sing}\widetilde C_k)\geq h_k/2.
\end{align}
Define the error bounds
\begin{align}
    \eta_S&:=8\nu/h_1,\qquad e_k:=\nu+2k\eta_S,\qquad \eta_{R,k}:=8e_k/h_k,\qquad \eta_{W,k}^{\mathrm H}:=4e_k/(1-\theta).
\end{align}
Let $\widetilde R_k$ select right singular vectors of $\widetilde C_k$ above $t_k$, and let $\widetilde W_k^{\mathrm H}$ select eigenvectors of $\widetilde A_k$ above $(1+\theta)/2$. If $\nu<h_1/4$ and $e_k<\min\{h_k/4,(1-\theta)/4\}$, these spaces have the dimensions of $S_\theta$, $R_k$, and $W_k^{\mathrm H}$ selected from the exact operators at the same thresholds, and satisfy
\begin{align}
    \|P_{\widetilde S_\theta}-P_{S_\theta}\|&\leq\eta_S,\qquad \|P_{\widetilde R_k}-P_{R_k}\|\leq\eta_{R,k},\qquad \|P_{\widetilde W_k^{\mathrm H}}-P_{W_k^{\mathrm H}}\|\leq\eta_{W,k}^{\mathrm H}.
    \label{eq:df-threshold-space-errors}
\end{align}
\end{corollary}
\begin{proof}
Apply Lem.~\ref{lem:df-certified-threshold} to the $1$-RDM at $\theta$ to obtain the dimension and error bound for $\widetilde S_\theta$. Lem.~\ref{lem:df-compressed-operator-error} then bounds both compression errors by $e_k$. A second application of Lem.~\ref{lem:df-certified-threshold} yields the dimension and error bound for $\widetilde R_k$. For $\widetilde W_k^{\mathrm H}$, the exact gap in Lem.~\ref{lem:constant-gap-strong-hierarchy} separates eigenvalue one from the remaining spectrum by at least $1-\theta$. Since $e_k<(1-\theta)/4$, the midpoint threshold selects the corresponding perturbed cluster, and Lem.~\ref{lem:projector-perturbation} gives the stated bound.
\end{proof}

For homogeneous blocks, use the exact spaces from App.~\ref{app:same-particle-correctness} and the same calculation with $k=p$ and $\nu=\mu$.

\subsection{Products of recovered factors and their removal}
\label{app:product-removal-stability}

In the heterogeneous recursion, products of previously recovered factors must be removed from the candidate spaces. We first bound their exterior products, since the estimated factors need not have exactly orthogonal one-particle supports. The total degree of each product used to remove earlier factors is at most $r$, so the factorial factor in the following lemma depends only on $r$ in that application. For the final product of many blocks, App.~\ref{app:assembly-explicit-accuracy} uses a separate expansion around the exact orthogonal factors; it does not apply this factorial bound to all blocks at once.

\begin{lemma}[Norm and stability of exterior products]\label{lem:exterior-product-stability}
Let $H\simeq\mathbb{C}^m$ be a finite-dimensional single-particle Hilbert space. For $j\in[n]$, let $p_j\geq0$ and $\alpha_j\in\wedge^{p_j}H$. Set
$P:=\sum_{j=1}^n p_j$.
Then, we have
\begin{align}
    \left\| \alpha_1\wedge\cdots\wedge\alpha_n \right\| &\leq \sqrt{ \frac{P!}{p_1!\cdots p_n!} } \prod_{j=1}^n\|\alpha_j\|.
    \label{eq:general-exterior-product-norm}
\end{align}
Consequently, let $\psi_j,\widetilde\psi_j\in \wedge^{p_j}H$ be normalized vectors satisfying, for $j\in [n]$, $\|\widetilde\psi_j-\psi_j\|  \leq \epsilon_j$. Then, we obtain
\begin{align}
    \left\| \widetilde\psi_1\wedge\cdots\wedge\widetilde\psi_n - \psi_1\wedge\cdots\wedge\psi_n \right\| &\leq \sqrt{ \frac{P!}{p_1!\cdots p_n!}} \sum_{j=1}^n\epsilon_j.
    \label{eq:general-exterior-product-stability}
\end{align}
The prefactor in Eq.~\eqref{eq:general-exterior-product-norm} also bounds the operator norm of the linear exterior-product map $\bigotimes_j\wedge^{p_j}H\to\wedge^PH$.
\end{lemma}

\begin{proof}
If $P>m$, both exterior products vanish. Otherwise, identify each $\wedge^qH$ isometrically with the antisymmetric subspace of $H^{\otimes q}$. Let $\mathsf A_P$ be the orthogonal projector onto that subspace at degree $P$. With the normalization of App.~\ref{app:basic-exterior-algebra},
\begin{align}
    \alpha_1\wedge\cdots\wedge\alpha_n
    &=\sqrt{\frac{P!}{p_1!\cdots p_n!}}\,
    \mathsf A_P(\alpha_1\otimes\cdots\otimes\alpha_n).
    \label{eq:exterior-product-antisymmetrization}
\end{align}
The identity follows on the orthonormal occupation basis. By linearity, the exterior-product map on the full tensor-product space has operator norm at most $\sqrt{P!/(p_1!\cdots p_n!)}$, not merely this bound on simple tensors. Since $\|\mathsf A_P\|\leq1$, this also proves \eqref{eq:general-exterior-product-norm}. To obtain \eqref{eq:general-exterior-product-stability}, replace the factors one at a time. Each difference contains one factor $\widetilde\psi_j-\psi_j$ and otherwise normalized factors. Apply the norm bound to each term and sum the resulting errors.
\end{proof}

To remove products of previously recovered factors, we need a bound on their span in terms of the errors already attached to the factors. The following lemma gives this bound directly from the largest factor error, with a coefficient depending only on the maximal degree and the number of modes.

\begin{lemma}[Stability of product spaces]\label{lem:df-product-span-stability}
For each degree $\ell\geq2$, let $v_{\ell,1},\ldots,v_{\ell,d_\ell}$ and $\widetilde v_{\ell,1},\ldots,\widetilde v_{\ell,d_\ell}$ be orthonormal lists in $\wedge^\ell H$, with $\sum_\ell d_\ell\leq m/2$. Suppose that, after matching and phase alignment, every factor error satisfies $\|\widetilde v_{\ell,a}-v_{\ell,a}\|\leq x$. Choose any set of configurations of distinct factors of total degree $2\leq k\leq r$, and assume that the corresponding exact products are orthonormal. Let $O$ and $\widetilde O$ be the spans of the exact and estimated products, formed in the same fixed order.

Define the explicit error bound
\begin{align}
    \rho&:=r\sqrt{2^r r!\,m}\,x.
    \label{eq:df-product-span-rho}
\end{align}
If $\rho<1$, then
\begin{align}
    \dim\widetilde O&=\dim O,\qquad \|P_{\widetilde O}-P_O\|\leq\rho.
    \label{eq:df-product-span-projector}
\end{align}
For an empty configuration set, both spaces are zero and we take $\rho=0$. Thus, for fixed $r$, the projector-error bound is $O_r(\sqrt m\,x)$.
\end{lemma}
\begin{proof}
Let $V_\ell$ and $\widetilde V_\ell$ have the respective factor lists as columns. Since the lists are orthonormal, $\|V_\ell\|,\|\widetilde V_\ell\|\leq1$ and $\|\widetilde V_\ell-V_\ell\|\leq\sqrt{d_\ell}\,x$.

Group the product columns by their nondecreasing degree tuples $\boldsymbol\ell=(\ell_1,\ldots,\ell_q)$, and write $W_{\boldsymbol\ell}$ and $\widetilde W_{\boldsymbol\ell}$ for the corresponding column matrices. Each matrix is obtained by selecting columns of the tensor product of the factor matrices and then taking their exterior products. Ordering the factors by degree and then by index makes the selected coordinate columns distinct; the signs restoring the fixed product order are the same for both matrices. Column selection and these signs therefore do not increase the norm of a difference. Replacing the tensor factors one at a time and using the operator-norm bound in Lem.~\ref{lem:exterior-product-stability} gives
\begin{align}
    \|\widetilde W_{\boldsymbol\ell}-W_{\boldsymbol\ell}\|&\leq\sqrt{\frac{k!}{\ell_1!\cdots\ell_q!}}\sum_{j=1}^q\sqrt{d_{\ell_j}}\,x\leq r\sqrt{r!\,m}\,x.
    \label{eq:product-column-factorization}
\end{align}
There are at most $2^r$ such degree tuples. Concatenating their column matrices into $W$ and $\widetilde W$ and applying Cauchy--Schwarz therefore yields $\|\widetilde W-W\|\leq\rho$. Grouping the columns in this way avoids a factor depending on the total number of products.

Because $W$ has orthonormal columns, $\rho<1$ makes $\widetilde W$ injective. The two spans consequently have the same dimension, and every unit vector $Wz$ is within $\rho$ of $\widetilde Wz$. The equal-dimension projector identity now gives Eq.~\eqref{eq:df-product-span-projector}.
\end{proof}

To apply the lemma, take $x$ to be the largest stored error among the constituent factors. The estimated products need not be orthogonal: only each factor list at a fixed degree must be orthonormal, as guaranteed by Gram splitting for both the dominant branches and the complementary blocks. The full dominant blocks reconstructed afterwards are not used to form these product spaces. We next control the error in removing the estimated product spaces from the candidates.

\begin{lemma}[Stable product removal]\label{lem:df-stable-intersection}
Let $P,Q$ be commuting orthogonal projectors, and set $F:=\operatorname{ran}P\cap\ker Q$, so that $P_F=P(I-Q)P$. Let $\widetilde P,\widetilde Q$ be orthogonal projectors with errors at most $\eta_P,\eta_Q$, respectively. Define $\widetilde B:=\widetilde P(I-\widetilde Q)\widetilde P$ and $\delta_\cap:=2\eta_P+\eta_Q$. Then
\begin{align}
    \|\widetilde B-P_{F}\|&\leq\delta_\cap.
    \label{eq:df-intersection-operator-error}
\end{align}
If $\delta_\cap\leq1/4$, the space $\widetilde F:=\operatorname{ran}\mathbf1_{(1/2,1]}(\widetilde B)$ lies in $\operatorname{ran}\widetilde P$, has dimension $\dim F$, and satisfies
\begin{align}
    &\|P_{\widetilde F}-P_{F}\|\leq4\delta_\cap.
    \label{eq:df-intersection-projector-error}
\end{align}
\end{lemma}
\begin{proof}
Commutativity gives $P_F=P(I-Q)P$. Replacing the three factors one at a time bounds $\|\widetilde B-P_F\|$ by $2\eta_P+\eta_Q=\delta_\cap$. Since $0\leq\widetilde B\leq I$, Weyl's inequality places its eigenvalues in $[0,\delta_\cap]\cup[1-\delta_\cap,1]$. Thus the threshold $1/2$ selects exactly $\operatorname{rank}P_F$ eigenvalues when $\delta_\cap\leq1/4$. Lem.~\ref{lem:projector-perturbation} gives $\|P_{\widetilde F}-P_F\|\leq\delta_\cap/(1-\delta_\cap)\leq4\delta_\cap$; the zero- and full-rank cases follow from the same spectral localization. Finally, if $\widetilde Bv=\lambda v$ with $\lambda>0$, then $v=\lambda^{-1}\widetilde P(I-\widetilde Q)\widetilde Pv\in\operatorname{ran}\widetilde P$.
\end{proof}

\subsection{Errors in the recovered block states}
\label{app:recovered-block-stability}

Prop.~\ref{prop:robust-gram-splitting} converts a degree-$k$ subspace error $\eta\leq1/(64km^3)$ into a factor error at most $(8k+2)m^3\eta$. This applies directly to the retained dominant branches and complementary block states. A recovered dominant branch must still be converted into its full block state; the following lemma controls that normalization.

\begin{lemma}[Full-block reconstruction]\label{lem:df-near-block-stability}
Let $\theta\in[2/3,3/4]$, let $f,\omega\in\wedge^kH$ be normalized, and suppose that $\|\Gamma^{(k)}\|\leq1$ and $\Gamma^{(k)}f=a\omega$ for some $a\in[\sqrt\theta,1]$. Let $\widetilde f$ be normalized and assume, after phase alignment, that $\|\widetilde f-f\|\leq\xi$ and $\|\widetilde\Gamma^{(k)}-\Gamma^{(k)}\|\leq\nu$. Set $\widetilde y:=\widetilde\Gamma^{(k)}\widetilde f$. If $\nu+\xi<\sqrt\theta/2$, then $\|\widetilde y\|>\sqrt\theta/2$ and
\begin{align}
    \left\|\frac{\widetilde y}{\|\widetilde y\|}-\omega\right\|
    &\leq\frac{2(\nu+\xi)}{\sqrt\theta}.
    \label{eq:df-near-block-error}
\end{align}
\end{lemma}
\begin{proof}
Write $\widetilde y-a\omega=(\widetilde\Gamma^{(k)}-\Gamma^{(k)})\widetilde f+\Gamma^{(k)}(\widetilde f-f)$, whose norm is at most $\nu+\xi$. Hence $\|\widetilde y\|\geq a-(\nu+\xi)>\sqrt\theta/2$. For any nonzero $y$ and unit vector $\omega$, the triangle inequality gives $\|y/\|y\|-\omega\|\leq2\|y-a\omega\|/a$ when $a>0$. Applying this to $\widetilde y$ proves the bound. A phase multiplying $\widetilde f$ multiplies the reconstructed block by the same phase.
\end{proof}

For selected spaces with the correct dimensions and valid projector-error bounds $\eta_{F,k}^{\mathrm N},\eta_{T,k}^{\mathrm H}\leq1/(64km^3)$, Prop.~\ref{prop:robust-gram-splitting} gives the computable factor-error bounds
\begin{align}
    \xi_{f,k}&:=(8k+2)m^3\eta_{F,k}^{\mathrm N},\qquad \xi_{\mathrm H,k}:=(8k+2)m^3\eta_{T,k}^{\mathrm H},\qquad \xi_{\mathrm N,k}:=\frac{2(\nu+\xi_{f,k})}{\sqrt\theta}.
    \label{eq:df-factor-certificates}
\end{align}
Set $\xi_{f,k}=0$ or $\xi_{\mathrm H,k}=0$ when the corresponding output list is empty. The first two bounds apply to the returned dominant branches and complementary blocks; the last applies after full-block reconstruction whenever $\nu+\xi_{f,k}<\sqrt\theta/2$.

\subsection{Stability of the remaining occupied subspace}
\label{app:effective-core-stability}

Once the retained branch spaces have been recovered, the remaining occupied core can be estimated directly from their subspace errors. The following proposition collects the inputs and the resulting guarantees; its construction and proof follow below.

\begin{proposition}[Recovery of the occupied core]\label{prop:df-core-stability}
Let $S_\theta\subseteq H$, and for each degree $\ell\geq2$ let $T_\ell$ be spanned by $d_\ell$ normalized $\ell$-particle Gaussian pure states. Assume that the one-particle supports of all these states are pairwise orthogonal and contained in $S_\theta$. Write $F$ for their combined one-particle support and $S_{\mathrm{core}}:=S_\theta\cap F^\perp$ for the remaining occupied subspace.

The inputs are orthonormal bases of $\widetilde S_\theta$ and $d_\ell$-dimensional spaces $\widetilde T_\ell\subseteq\wedge^\ell\widetilde S_\theta$, together with bounds
\begin{align}
    \|P_{\widetilde S_\theta}-P_{S_\theta}\|&\leq\eta_S,\qquad \|P_{\widetilde T_\ell}-P_{T_\ell}\|\leq\eta_\ell.
    \label{eq:core-recovery-input}
\end{align}
If $\eta_S<1$ and $\sum_\ell d_\ell\eta_\ell\leq1/4$, a deterministic procedure returns a subspace $\widetilde S_{\mathrm{core}}\subseteq\widetilde S_\theta$ of the correct dimension and its normalized occupied state $\widetilde\sigma_{\mathrm{core}}$, satisfying
\begin{align}
    \|P_{\widetilde S_{\mathrm{core}}}-P_{S_{\mathrm{core}}}\|&\leq\eta_S+4\sum_\ell d_\ell\eta_\ell.
    \label{eq:df-core-projector-error}
\end{align}
After phase alignment with the normalized occupied state $\sigma_{\mathrm{core}}$ of $S_{\mathrm{core}}$,
\begin{align}
    \|\widetilde\sigma_{\mathrm{core}}-\sigma_{\mathrm{core}}\|&\leq\sqrt{2m}\left(\eta_S+4\sum_\ell d_\ell\eta_\ell\right).
    \label{eq:df-core-volume-error}
\end{align}
For fixed maximal degree, the procedure uses polynomially many arithmetic operations in $m$, using only the supplied classical descriptions and no additional state copies. Empty branch lists contribute zero to the sum, and a zero-dimensional core gives the vacuum.
\end{proposition}

In the homogeneous case, apply the proposition with $T_p=R_p$, $\widetilde T_p=\widetilde R_p$, and $\eta_p=\eta_{R,p}$. In the heterogeneous case, use $T_\ell=F_\ell^{\mathrm N}$, $\widetilde T_\ell=\widetilde F_\ell^{\mathrm N}$, and $\eta_\ell=\eta_{F,\ell}^{\mathrm N}$. These are the spaces and error bounds available before Gram splitting.

The proof uses two elementary estimates. The first controls the sum of the $1$-RDMs over any orthonormal basis of a subspace.

\begin{lemma}[Basis-independent sum of $1$-RDMs]
\label{lem:one-rdm-subspace-sum}
For a $d$-dimensional subspace $T\subseteq\wedge^kH$, set $K(T):=\sum_{a=1}^d\Gamma_{u_a}^{(1)}$, where $u_1,\ldots,u_d$ is any orthonormal basis. This matrix is independent of the basis. If $\widetilde T$ has the same dimension and $\|P_{\widetilde T}-P_T\|\leq\eta$, then
\begin{align}
    \|K(\widetilde T)-K(T)\|&\leq d\eta.
    \label{eq:one-rdm-subspace-sum-error}
\end{align}
We set $K(\{0\})=0$.
\end{lemma}
\begin{proof}
For each unit $h\in H$, $\langle h,K(T)h\rangle=\operatorname{Tr}(n_hP_T)$, where $n_h=\hat c[h]^\dagger\hat c[h]$. This proves basis independence. Set $X:=P_{\widetilde T}-P_T$. It is Hermitian, has trace zero and rank at most $2d$, and satisfies $\|X\|\leq\eta$. Its positive and negative parts therefore have equal trace $\|X\|_1/2\leq d\eta$. Since $0\leq n_h\leq I$, $|\operatorname{Tr}(n_hX)|\leq d\eta$. Taking the supremum over $h$ proves Eq.~\eqref{eq:one-rdm-subspace-sum-error}.
\end{proof}

The second estimate converts the core's subspace error into its state error. Its attachment bound will also be used in the final assembly.

\begin{lemma}[Stability of an occupied subspace]\label{lem:exterior-canonical-alignment}
Let $E,F\subseteq H$ have the same dimension $d$, and let $\widetilde\sigma\in\wedge^dE$ and $\sigma\in\wedge^dF$ be normalized. A choice of their relative phase gives
\begin{align}
    \|\widetilde\sigma-\sigma\|&\leq\sqrt{2d}\,\|P_E-P_F\|.
    \label{eq:exterior-canonical-alignment}
\end{align}
For $0\leq t\leq\dim F^\perp$, any choice of phases, any normalized $\Omega\in\wedge^t(F^\perp)$, and any $\widetilde\Omega\in\wedge^tH$, we also have
\begin{align}
    \|(\widetilde\sigma-\sigma)\wedge\Omega\|&\leq\|\widetilde\sigma-\sigma\|,\qquad \|\widetilde\sigma\wedge\widetilde\Omega-\sigma\wedge\Omega\|\leq\|\widetilde\sigma-\sigma\|+\|\widetilde\Omega-\Omega\|.
    \label{eq:core-attachment-stability}
\end{align}
The degree-zero case uses the vacuum convention.
\end{lemma}
\begin{proof}
Let $\theta_1,\ldots,\theta_d$ be the principal angles between $E$ and $F$. Choosing the relative phase makes $\langle\sigma,\widetilde\sigma\rangle=\prod_{j=1}^d\cos\theta_j\geq0$. Since $\sin\theta_j\leq\|P_E-P_F\|$ and $1-\cos\theta_j\leq\sin^2\theta_j$,
\begin{align}
    \|\widetilde\sigma-\sigma\|^2&=2\left(1-\prod_{j=1}^d\cos\theta_j\right)\leq2\sum_{j=1}^d\sin^2\theta_j\leq2d\,\|P_E-P_F\|^2.
    \label{eq:core-volume-principal-angle-bound}
\end{align}
For the attachment bounds, wedging with the unit volume $\widetilde\sigma$ is a product of creation operators for orthonormal modes, each of norm one. Hence $\|\widetilde\sigma\wedge X\|\leq\|X\|$ for every exterior vector $X$. Also, $\|\sigma\wedge\Omega\|=1$. Because $\Omega$ has no particles in $F$, only the component of $\widetilde\sigma$ in $\wedge^dF$ contributes to the overlap with $\sigma\wedge\Omega$. Thus $\langle\sigma\wedge\Omega,\widetilde\sigma\wedge\Omega\rangle=\langle\sigma,\widetilde\sigma\rangle$, and
\begin{align}
    \|(\widetilde\sigma-\sigma)\wedge\Omega\|^2&\leq2-2\operatorname{Re}\langle\sigma,\widetilde\sigma\rangle=\|\widetilde\sigma-\sigma\|^2.
    \label{eq:core-attachment-overlap-bound}
\end{align}
Finally, write the difference in the second bound of Eq.~\eqref{eq:core-attachment-stability} as $\widetilde\sigma\wedge(\widetilde\Omega-\Omega)+(\widetilde\sigma-\sigma)\wedge\Omega$ and apply the two preceding bounds.
\end{proof}

\begin{proof}[Proof of Prop.~\ref{prop:df-core-stability}]
For each supplied basis of $\widetilde T_\ell$, sum its vectors' $1$-RDMs to form $\widetilde K:=\sum_\ell K(\widetilde T_\ell)$. Each exact branch occupies all modes of its one-particle support, so its $1$-RDM projects onto that support. Since these supports are mutually orthogonal, Lem.~\ref{lem:one-rdm-subspace-sum} gives
\begin{align}
    \sum_\ell K(T_\ell)&=P_F,\qquad \|\widetilde K-P_F\|\leq\sum_\ell d_\ell\eta_\ell.
    \label{eq:core-subspace-certificate}
\end{align}
Set $\widetilde F:=\operatorname{ran}\mathbf1_{(1/2,\infty)}(\widetilde K)$. The last bound is at most $1/4$, so thresholding this approximate projector at $1/2$ and applying Lem.~\ref{lem:projector-perturbation} yields
\begin{align}
    \dim\widetilde F&=\dim F,\qquad \|P_{\widetilde F}-P_F\|\leq4\sum_\ell d_\ell\eta_\ell.
    \label{eq:df-retained-support-error}
\end{align}

For $h\perp\widetilde S_\theta$, the operator $\hat c[h]$ vanishes on every vector in every $\widetilde T_\ell$. Their $1$-RDMs, and hence $\widetilde K$, are therefore supported on $\widetilde S_\theta$. Thus $\widetilde F\subseteq\widetilde S_\theta$, and we return $\widetilde S_{\mathrm{core}}:=\widetilde S_\theta\cap\widetilde F^\perp$. The projector identities
\begin{align}
    P_{\widetilde S_{\mathrm{core}}}&=P_{\widetilde S_\theta}-P_{\widetilde F},\qquad P_{S_{\mathrm{core}}}=P_{S_\theta}-P_F
    \label{eq:core-projector-difference}
\end{align}
give Eq.~\eqref{eq:df-core-projector-error} by the triangle inequality. The high-occupation spaces have equal dimensions because $\eta_S<1$; subtracting the equal retained-support dimensions proves the core dimension claim. Finally, occupy every mode of an orthonormal basis of $\widetilde S_{\mathrm{core}}$ to obtain $\widetilde\sigma_{\mathrm{core}}$. Lem.~\ref{lem:exterior-canonical-alignment} gives Eq.~\eqref{eq:df-core-volume-error}.
\end{proof}

\subsection{Assembly of approximate blocks and an occupied core}
\label{app:assembly-explicit-accuracy}

We now bound the final normalized state in terms of errors in the individual factors and the occupied core. Since the estimated factors need not have mutually orthogonal one-particle supports, their exterior product need not be normalized. The following proposition implies the common fidelity bound used in both settings, including truncation and normalization.

\begin{proposition}[Assembly with an occupied core]\label{prop:equal-p-assembly}
Let $\Psi_{\mathrm{tr}}=\sigma\wedge\bigwedge_{j=1}^n\psi_j$ be a normalized state on $m$ modes, where $\sigma$ is a normalized exterior volume of the core and the normalized $p_j$-particle factors $\psi_j$, $p_j\geq1$, have mutually orthogonal one-particle supports outside the core. Suppose the normalized estimates $\widetilde\psi_j$ have the same particle numbers, and $\widetilde\sigma$ is a normalized exterior volume of an estimated core of the same dimension. After phase and permutation alignment, assume $\|\widetilde\psi_j-\psi_j\|\leq\epsilon_j$ and $\|\widetilde\sigma-\sigma\|\leq\xi_{\mathrm{core}}$.

Define the total assembly-error bound
\begin{align}
    \mathfrak A&:=\xi_{\mathrm{core}}+\prod_{j=1}^n\left(1+\frac{m^{p_j/2}}{\sqrt{p_j!}}\epsilon_j\right)-1.
    \label{eq:common-core-assembly-error}
\end{align}
If $\mathfrak A<1$, the exterior product of $\widetilde\sigma$ and the $\widetilde\psi_j$ is nonzero, and its normalization $\widetilde\Phi$ satisfies
\begin{align}
    1-|\langle\widetilde\Phi,\Psi_{\mathrm{tr}}\rangle|^2&\leq\mathfrak A^2.
    \label{eq:equal-p-assembly-to-truncated}
\end{align}
If a normalized target $\Psi$ additionally satisfies $1-|\langle\Psi_{\mathrm{tr}},\Psi\rangle|^2\leq3ms_0^2$, then
\begin{align}
    1-|\langle\widetilde\Phi,\Psi\rangle|^2&\leq\left(\mathfrak A+\sqrt{3ms_0^2}\right)^2.
    \label{eq:equal-p-assembly-final}
\end{align}
An empty block list contributes the vacuum and gives $\mathfrak A=\xi_{\mathrm{core}}$; a zero-dimensional core also uses the vacuum convention.
\end{proposition}

To prove the proposition, we first control the unnormalized product of the approximate blocks. We expand around the exact factors, keeping all unchanged factors together in each term. Their product has norm one because their supports are orthogonal, which avoids applying a factorial bound to all blocks at once.

\begin{lemma}[Assembly of approximate blocks]\label{lem:structured-assembly}
Let $\psi_j$ and $\widetilde\psi_j$ be normalized $p_j$-particle states, $p_j\geq1$, with $\|\widetilde\psi_j-\psi_j\|\leq\epsilon_j$ after matching and phase alignment. If the exact factors have mutually orthogonal one-particle supports, then, with $P:=\sum_jp_j$,
\begin{align}
    \left\|\bigwedge_{j=1}^n\widetilde\psi_j-\bigwedge_{j=1}^n\psi_j\right\|&\leq\prod_{j=1}^n\left(1+\frac{P^{p_j/2}}{\sqrt{p_j!}}\epsilon_j\right)-1.
    \label{eq:structured-assembly-vector-bound}
\end{align}
For an empty factor list, both exterior products are the vacuum and the bound is zero.
\end{lemma}

\begin{proof}
The empty case is immediate. Align the factors as in the hypothesis and set $P:=\sum_jp_j$.

Set $e_j:=\widetilde\psi_j-\psi_j$ and, for $A\subseteq[n]$, let $p(A):=\sum_{j\in A}p_j$. Expand by multilinearity and group each term according to the nonempty set $A$ of error factors. Reordering changes only a sign, so
\begin{align}
    \left\|\bigwedge_j\widetilde\psi_j-\bigwedge_j\psi_j\right\| &\leq\sum_{\varnothing\neq A\subseteq[n]} \left\| \left(\bigwedge_{j\in A}e_j\right) \wedge\left(\bigwedge_{j\notin A}\psi_j\right) \right\|.
    \label{eq:structured-assembly-expansion}
\end{align}
The unchanged exact factors in each term combine into one normalized vector of degree $P-p(A)$, because their supports are mutually orthogonal. Apply Eq.~\eqref{eq:general-exterior-product-norm} to this vector and the error factors. Using $P!/(P-p(A))!\leq P^{p(A)}$, each summand is at most
\begin{align}
    \sqrt{\frac{P!}{(P-p(A))!\prod_{j\in A}p_j!}} \prod_{j\in A}\epsilon_j &\leq\prod_{j\in A}\frac{P^{p_j/2}\epsilon_j}{\sqrt{p_j!}}.
    \label{eq:structured-assembly-combinatorial-bound}
\end{align}
Summing over the nonempty subsets gives the product on the right-hand side of Eq.~\eqref{eq:structured-assembly-vector-bound}, completing the proof.
\end{proof}

\begin{proof}[Proof of Prop.~\ref{prop:equal-p-assembly}]
After aligning phases and factors, write $\Omega:=\bigwedge_j\psi_j$, $\widetilde\Omega:=\bigwedge_j\widetilde\psi_j$, and $Y:=\widetilde\sigma\wedge\widetilde\Omega$. The exact product $\Omega$ is normalized and supported outside the exact core. Since $\sum_jp_j\leq m$, Lem.~\ref{lem:structured-assembly} gives $\|\widetilde\Omega-\Omega\|\leq\mathfrak A-\xi_{\mathrm{core}}$. The attachment bound in Eq.~\eqref{eq:core-attachment-stability} therefore yields $\|Y-\Psi_{\mathrm{tr}}\|\leq\mathfrak A$. Thus $\|Y\|\geq1-\mathfrak A>0$. Since $(I-\Pi_{\widetilde\Phi})Y=0$, projecting $\Psi_{\mathrm{tr}}-Y$ onto the orthogonal complement of $\widetilde\Phi$ proves Eq.~\eqref{eq:equal-p-assembly-to-truncated}.

For normalized pure states $x,y$, $\tfrac12\|\Pi_x-\Pi_y\|_1=\sqrt{1-|\langle x,y\rangle|^2}$. The trace-norm triangle inequality gives
\begin{align}
    \sqrt{1-|\langle\widetilde\Phi,\Psi\rangle|^2}&\leq\sqrt{1-|\langle\widetilde\Phi,\Psi_{\mathrm{tr}}\rangle|^2}+\sqrt{1-|\langle\Psi_{\mathrm{tr}},\Psi\rangle|^2}\leq\mathfrak A+\sqrt{3ms_0^2}.
    \label{eq:appendix-fidelity-triangle}
\end{align}
Squaring proves Eq.~\eqref{eq:equal-p-assembly-final}.
\end{proof}

\section{Reconstruction from estimated RDMs in the homogeneous setting}
\label{app:homogeneous-noisy}

We establish the reconstruction guarantee from estimated RDMs in the homogeneous setting of Sec.~\ref{sec:equal-reconstruction}.

\begin{theorem}[Reconstruction from estimated RDMs]\label{thm:equal-p-noisy-compact-direct}
Fix $p\geq2$ and $\varepsilon_{\mathrm{fid}},\beta\in(0,1)$. There is a constant $c_p>0$, depending only on $p$, such that the following holds for every target $\Psi$ in the homogeneous setting. Given $p$, $\varepsilon_{\mathrm{fid}}$, $\beta$, and Hermitian RDM estimates with a supplied error bound
\begin{align}
    \max\!\left\{\|\widetilde\Gamma^{(1)}-\Gamma_\Psi^{(1)}\|,\|\widetilde\Gamma^{(p)}-\Gamma_\Psi^{(p)}\|\right\}&\leq\mu\leq c_p\frac{\varepsilon_{\mathrm{fid}}}{m^{\lceil(p+13)/2\rceil}},
    \label{eq:equal-p-df-input-accuracy}
\end{align}
a randomized classical algorithm returns, with probability at least $1-\beta$, a factorized description of a normalized state $\widetilde\Phi$ satisfying
\begin{align}
    |\langle\widetilde\Phi,\Psi\rangle|^2&\geq1-\varepsilon_{\mathrm{fid}}.
    \label{eq:equal-p-df-final-fidelity}
\end{align}
The description consists of an orthonormal basis for an occupied core and an ordered list of normalized $p$-particle vectors. Their exterior product is nonzero, with normalization implicit in the description. For fixed $p$, the algorithm uses polynomially many arithmetic operations in $m$ and $\log(1/\beta)$ and no additional copies of the target state.
\end{theorem}

We construct the algorithm below, then verify the required subspace and block accuracies and apply the assembly bound of Prop.~\ref{prop:equal-p-assembly}. The final proof combines these estimates with the two Gram-splitting failure budgets.

\subsection{Reconstruction from estimated RDMs}
\label{app:uniform-certified-noisy-algorithm}

We follow the construction of App.~\ref{app:same-particle-algorithm}, with thresholds chosen from the estimated spectra. The formulas below specify the operations and their computable error bounds. Prop.~\ref{prop:equal-p-input-accuracy} verifies these bounds, the selected dimensions, and all local stability conditions needed to prove Thm.~\ref{thm:equal-p-noisy-compact-direct}. Every \textsc{Fail} condition is checked from the estimates and their error certificates.

\paragraph*{Parameters and error bounds.}
Use the inputs and the supplied error bound $\mu>0$ of Thm.~\ref{thm:equal-p-noisy-compact-direct}, together with the resolution scale $s_0$ from Eq.~\eqref{eq:equal-p-df-s0}. For the accuracy analysis, write
\begin{align}
    \mathfrak a_p^{\mathrm{hom}}&:=\left\lceil\frac{p+13}{2}\right\rceil.
    \label{eq:equal-p-direct-rdm-exponent}
\end{align}
All error bounds below are computed from $\mu$ and known parameters. Assign failure probability $\beta/2$ to each of the two Gram-splitting calls. RDM estimation is treated separately in App.~\ref{app:rdm-estimation}; here the probability is over the internal Gram-splitting randomness.

\paragraph*{Step 1: Find the high-occupation space.}
Set $h_1:=1/(96m)$ and choose $\theta\in[2/3,3/4]$ away from the observed spectrum. Define the resulting high-occupation space and its error bound by
\begin{align}
    \operatorname{dist}\!\left(\theta,\operatorname{spec}(\widetilde\Gamma^{(1)})\right)&\geq\frac{h_1}{2},\qquad \widetilde S_\theta:=\operatorname{ran}\mathbf1_{(\theta,\infty)}(\widetilde\Gamma^{(1)}),\qquad \eta_S:=\frac{8\mu}{h_1}.
    \label{eq:uniform-noisy-high-occupation-space}
\end{align}
All exact comparison spaces and operators, including $S_\theta$, $C_p$, and $A_p$, use the thresholds chosen from the estimates.

\paragraph*{Step 2: Form the two compressed RDMs.}
After constructing the high-occupation space, we form the all-high and all-low projectors and apply them to the $p$-RDM estimate:
\begin{align}
    \widetilde Q_{\mathrm{hi}} &:=P_{\wedge^p\widetilde S_\theta},
    \qquad \widetilde Q_{\mathrm{lo}} :=P_{\wedge^p(\widetilde S_\theta^\perp)},
    \qquad \widetilde C_p :=\widetilde Q_{\mathrm{lo}}\widetilde\Gamma^{(p)} \widetilde Q_{\mathrm{hi}},
    \qquad \widetilde A_p :=\widetilde Q_{\mathrm{lo}}\widetilde\Gamma^{(p)} \widetilde Q_{\mathrm{lo}}.
    \label{eq:uniform-noisy-compressions}
\end{align}
Set $e_p:=\mu+2p\eta_S$ to bound both compression errors. This includes the RDM error and the two projector errors, as in Lem.~\ref{lem:df-compressed-operator-error}.

\paragraph*{Step 3: Recover the retained dominant blocks.}
Select the right singular space of $\widetilde C_p$ using a threshold $t\in[s_0,2s_0]$ separated from the observed singular values by
\begin{align}
    h_p&:=\frac{s_0}{8m},\qquad \operatorname{dist}\!\left(t,\operatorname{sing}(\widetilde C_p)\right)\geq\frac{h_p}{2}=\frac{s_0}{16m}.
    \label{eq:uniform-certified-singular-threshold}
\end{align}
If no such threshold exists, return \textsc{Fail}.

Define the retained right singular space and its error bound by
\begin{align}
    \widetilde R_p &:=\operatorname{ran}\mathbf1_{(t^2,\infty)} \left(\widetilde C_p^\dagger\widetilde C_p\right),
    \qquad \widetilde d_p^{\mathrm N}:=\dim\widetilde R_p,
    \qquad \eta_{R,p}:=\frac{8e_p}{h_p}.
    \label{eq:uniform-near-space-certificate}
\end{align}
Here $\eta_{R,p}$ is the projector-error bound relative to the exact space $R_p$ selected at the same threshold $t$.

Apply the Gram-splitting algorithm of Prop.~\ref{prop:robust-gram-splitting} to $(\widetilde R_p,p,\eta_{R,p},\beta/2)$. If it fails, return \textsc{Fail}; otherwise denote its orthonormal branch estimates by $\widetilde f_1,\ldots,\widetilde f_{\widetilde d_p^{\mathrm N}}$. For $j\in[\widetilde d_p^{\mathrm N}]$, reconstruct the full block by
\begin{align}
    \widetilde y_j&:=\widetilde\Gamma^{(p)}\widetilde f_j,
    \qquad
    \widetilde\omega_j^{\mathrm N}
    :=\frac{\widetilde y_j}{\|\widetilde y_j\|}.
    \label{eq:uniform-noisy-near-reconstruction}
\end{align}
Return \textsc{Fail} if $\widetilde y_j=0$. Keep the dominant branches $\widetilde f_j$ separately for the core construction; the normalized full blocks enter the final product.

\paragraph*{Step 4: Recover the complementary blocks.}
Use the midpoint of the exact gap from Lem.~\ref{lem:equal-p-df-strong-gap} to define
\begin{align}
    \widetilde W_p^{\mathrm H} &:=\operatorname{ran}\mathbf1_{((1+\theta)/2,\infty)} \left(\widetilde A_p\right),
    \qquad \widetilde d_p^{\mathrm H}:=\dim\widetilde W_p^{\mathrm H},
    \qquad \eta_{W,p}^{\mathrm H}:=\frac{4e_p}{1-\theta}.
    \label{eq:uniform-strong-space-certificate}
\end{align}
The certificate $\eta_{W,p}^{\mathrm H}$ bounds the projector error relative to the exact complementary-block space $W_p^{\mathrm H}$. Both $\widetilde d_p^{\mathrm N}$ and $\widetilde d_p^{\mathrm H}$ are computed from the estimated spaces.

Apply the same algorithm to $(\widetilde W_p^{\mathrm H},p,\eta_{W,p}^{\mathrm H},\beta/2)$. If it fails, return \textsc{Fail}; otherwise write the output as $\widetilde\omega_1^{\mathrm H},\ldots, \widetilde\omega_{\widetilde d_p^{\mathrm H}}^{\mathrm H}$. These vectors approximate the complementary block states directly.

\paragraph*{Step 5: Assemble the unresolved core and recovered blocks.}
To obtain the unresolved occupied core, remove the retained branch supports from the high-occupation space. The branches form an orthonormal basis of $\widetilde R_p$, so apply the core-recovery procedure of Prop.~\ref{prop:df-core-stability} to $\widetilde S_\theta$ and $\widetilde R_p$ by setting
\begin{align}
    \widetilde K&:=\sum_{j=1}^{\widetilde d_p^{\mathrm N}}\Gamma_{\widetilde f_j}^{(1)},\qquad \widetilde F:=\operatorname{ran}\mathbf1_{(1/2,\infty)}(\widetilde K),
    \label{eq:uniform-noisy-retained-support-space}
\end{align}
with $\widetilde K=0$ when $\widetilde d_p^{\mathrm N}=0$.
Because $\widetilde f_j\in\wedge^p\widetilde S_\theta$, the operator $\widetilde K$ is supported on $\widetilde S_\theta$, and $\widetilde F\subseteq\widetilde S_\theta$. We therefore define the remaining occupied subspace directly by
\begin{align}
    \widetilde S_{\mathrm{core}}&:=\widetilde S_\theta\cap\widetilde F^\perp,
    \qquad
    P_{\widetilde S_{\mathrm{core}}}
    =P_{\widetilde S_\theta}-P_{\widetilde F}.
    \label{eq:uniform-noisy-core-space}
\end{align}
Prop.~\ref{prop:df-core-stability} bounds the resulting core error using $\eta_S$ and $\eta_{R,p}$, the errors before Gram splitting.

Choose an orthonormal basis of $\widetilde S_{\mathrm{core}}$ and let $\widetilde\sigma_{\mathrm{core}}$ denote its unit exterior volume, using the vacuum convention when $\dim\widetilde S_{\mathrm{core}}=0$. With the factors in a fixed order, define the assembled vector and, whenever it is nonzero, its normalized state by
\begin{align}
    \widetilde Y&:=\widetilde\sigma_{\mathrm{core}}\wedge\bigwedge_{j=1}^{\widetilde d_p^{\mathrm N}}\widetilde\omega_j^{\mathrm N}\wedge\bigwedge_{j=1}^{\widetilde d_p^{\mathrm H}}\widetilde\omega_j^{\mathrm H},\qquad \widetilde\Phi:=\frac{\widetilde Y}{\|\widetilde Y\|}.
    \label{eq:uniform-noisy-output-state}
\end{align}
Return the chosen core basis and the ordered block lists, which give a factorized description of $\widetilde\Phi$.
Prop.~\ref{prop:equal-p-input-accuracy} proves nonvanishing and controls the assembly error when both splitting calls succeed. Combining this bound with the two splitting failure budgets proves Thm.~\ref{thm:equal-p-noisy-compact-direct}.

\subsection{Verification of the required accuracy}
\label{app:homogeneous-fidelity-guarantee}

Lem.~\ref{lem:equal-p-exact-output} identifies the ideal output as a product of the exact core and recovered blocks on mutually orthogonal one-particle supports, with truncation error at most $3ms_0^2$. We can therefore apply the common assembly bound of Prop.~\ref{prop:equal-p-assembly} once the block and core errors have been controlled. The following proposition verifies the reconstruction conditions and bounds the resulting assembly error $\mathfrak A$ defined there.

\begin{proposition}[Accuracy of the recovered blocks and core]\label{prop:equal-p-input-accuracy}
Fix $p\geq2$, $m\geq p$, and $0<\varepsilon_{\mathrm{fid}}\leq1$ in the homogeneous setting, and set $s_0:=\sqrt{\varepsilon_{\mathrm{fid}}/(48m)}$. There are constants $c_p,\mathsf A_p>0$, depending only on $p$, such that the following holds. Suppose the supplied Hermitian $1$- and $p$-RDM estimates have operator-norm errors at most $\mu$, with
\begin{align}
    \mu&\leq c_p\frac{\varepsilon_{\mathrm{fid}}}{m^{\lceil(p+13)/2\rceil}}.
    \label{eq:equal-p-df-explicit-mu}
\end{align}
Then the algorithm of App.~\ref{app:uniform-certified-noisy-algorithm}, run with error bound $\mu$, admits the required thresholds and satisfies both Gram-splitting input conditions.

If both calls succeed, the recovered lists match the exact retained dominant branches and full blocks up to phases and a permutation, and the occupied core has the correct dimension. Each recovered branch or full block $\widetilde v$, matched with its exact counterpart $v$, and the phase-aligned occupied-core state satisfy
\begin{align}
    \|\widetilde v-v\|&\leq \mathsf A_pm^5\frac{\mu}{s_0},\qquad \|\widetilde\sigma_{\mathrm{core}}-\sigma_{\mathrm{core}}\|\leq \mathsf A_pm^{7/2}\frac{\mu}{s_0}.
    \label{eq:homogeneous-block-core-error-contract}
\end{align}
The full-block normalizations are well defined. Inserting these bounds into Prop.~\ref{prop:equal-p-assembly} gives
\begin{align}
    \mathfrak A&\leq\frac{\sqrt{\varepsilon_{\mathrm{fid}}}}{4}.
    \label{eq:equal-p-df-explicit-assembly}
\end{align}
\end{proposition}

\begin{proof}
We follow the reconstruction order, verifying each stability condition before using it. All constants below depend only on $p$. The assumed accuracy and the definition of $s_0$ give
\begin{align}
    \frac{\mu}{s_0}&\leq\sqrt{48}\,c_p\sqrt{\varepsilon_{\mathrm{fid}}}m^{1/2-\mathfrak a_p^{\mathrm{hom}}}.
    \label{eq:equal-p-df-mu-over-s0}
\end{align}
Since $\mathfrak a_p^{\mathrm{hom}}\geq8$ and $\varepsilon_{\mathrm{fid}}\leq1$, this implies $m^5\mu/s_0\leq\sqrt{48}\,c_p$. We choose $c_p$ sufficiently small for the bounds below.

\emph{Thresholds and subspaces.} Use $h_1=1/(96m)$ and $h_p=s_0/(8m)$. Lem.~\ref{lem:df-threshold-selection} supplies the one-body threshold $\theta\in[2/3,3/4]$. The certificate definitions in App.~\ref{app:uniform-certified-noisy-algorithm} give
\begin{align}
    \eta_S&\leq\mathsf C_pm\mu,\qquad e_p\leq\mathsf C_pm\mu,\qquad \eta_{R,p}\leq\mathsf C_pm^2\frac{\mu}{s_0},\qquad \eta_{W,p}^{\mathrm H}\leq\mathsf C_pm\mu.
    \label{eq:equal-p-df-basic-power}
\end{align}
The ratios $\mu/h_1$, $e_p/h_p$, and $e_p/(1-\theta)$ are at most $\mathsf C_pm^2\mu/s_0\leq\mathsf C_pc_p$. Thus $\mu<h_1/4$ and $e_p<\min\{h_p/4,(1-\theta)/4\}$. Since $\operatorname{rank}C_p\leq n\leq m$ by Lem.~\ref{lem:equal-p-df-coherence} and $e_p<h_p/4<s_0/2$, the singular-value part of Lem.~\ref{lem:df-threshold-selection} supplies the singular threshold as well. Cor.~\ref{cor:df-threshold-space-errors}, with $k=p$ and $\nu=\mu$, now verifies the selected-space dimensions and projector-error bounds. In particular, $\widetilde d_p^{\mathrm N}+\widetilde d_p^{\mathrm H}\leq n\leq m/p$.

\emph{Recovered blocks.} The subspace certificates $\eta_{R,p}$ and $\eta_{W,p}^{\mathrm H}$ are at most $1/(64pm^3)$: multiplying either by $64pm^3$ gives at most $\mathsf C_pm^5\mu/s_0\leq\mathsf C_pc_p$. Both calls to Prop.~\ref{prop:robust-gram-splitting} therefore satisfy their input conditions. On successful calls, and after phase and permutation matching, that proposition implies the dominant-branch and complementary-block error bounds
\begin{align}
    \xi_f&:=(8p+2)m^3\eta_{R,p}\leq\mathsf C_pm^5\frac{\mu}{s_0},\qquad \xi_{\mathrm H}:=(8p+2)m^3\eta_{W,p}^{\mathrm H}\leq\mathsf C_pm^4\mu\leq\mathsf C_pm^5\frac{\mu}{s_0}.
    \label{eq:equal-p-df-splitting-power}
\end{align}
Set the corresponding bound to zero for an empty output list. We condition the remaining argument on both calls succeeding.

For full dominant-block reconstruction, these bounds give $\mu+\xi_f\leq\mathsf C_pm^5\mu/s_0\leq\mathsf C_pc_p$. Since $\theta\geq2/3$, choosing $c_p$ sufficiently small ensures $\mu+\xi_f<\sqrt\theta/2$. Lem.~\ref{lem:df-near-block-stability} then gives $\|\widetilde y_j\|>\sqrt\theta/2$, so normalization is well defined, and bounds the reconstructed block error by
\begin{align}
    \frac{2(\mu+\xi_f)}{\sqrt\theta}&\leq\mathsf C_pm^5\frac{\mu}{s_0}.
    \label{eq:equal-p-df-full-near-power}
\end{align}
Thus every recovered full block has the same error bound, regardless of its type.

\emph{Occupied core.} Core recovery uses the branch-space error $\eta_{R,p}$ before Gram splitting, since the returned branches form an orthonormal basis of $\widetilde R_p\subseteq\wedge^p\widetilde S_\theta$. Eq.~\eqref{eq:equal-p-df-basic-power} gives $\eta_S\leq\mathsf C_pm\mu<1$ and $\widetilde d_p^{\mathrm N}\eta_{R,p}\leq\mathsf C_pm^3\mu/s_0\leq1/4$ for sufficiently small $c_p$. Prop.~\ref{prop:df-core-stability} therefore applies to this single space of dimension $\widetilde d_p^{\mathrm N}$. It gives the correct core dimension and bounds the phase-aligned core error by
\begin{align}
    \sqrt{2m}\bigl(\eta_S+4\widetilde d_p^{\mathrm N}\eta_{R,p}\bigr)&\leq\mathsf C_pm^{7/2}\frac{\mu}{s_0}.
    \label{eq:equal-p-df-core-power}
\end{align}
Together with the recovered-block bounds, this proves Eq.~\eqref{eq:homogeneous-block-core-error-contract}.

\emph{Assembly.} We now insert the block and core bounds into Prop.~\ref{prop:equal-p-assembly}. There are at most $m/p$ recovered blocks, each of degree $p$, containing at most $m$ particles in total. Each first-order term $u_j:=m^{p/2}\epsilon_j/\sqrt{p!}$ in the assembly product is bounded by $\mathsf C_pm^{p/2+5}\mu/s_0$. Summing over the blocks gives
\begin{align}
    \sum_j u_j&\leq\mathsf C_pm^{p/2+6}\frac{\mu}{s_0}.
    \label{eq:equal-p-df-correlated-power}
\end{align}
By Eq.~\eqref{eq:equal-p-df-mu-over-s0} and $\mathfrak a_p^{\mathrm{hom}}\geq(p+13)/2$, this sum is at most $\mathsf C_pc_p\sqrt{\varepsilon_{\mathrm{fid}}}$, hence at most one. We may therefore use $\prod_j(1+u_j)-1\leq e^{\sum_j u_j}-1\leq2\sum_j u_j$. Adding the core bound in Eq.~\eqref{eq:equal-p-df-core-power} yields
\begin{align}
    \mathfrak A&\leq\mathsf C_pm^{p/2+6}\frac{\mu}{s_0}.
    \label{eq:equal-p-df-total-assembly-power}
\end{align}
Since $s_0=\sqrt{\varepsilon_{\mathrm{fid}}/(48m)}$, this bound scales as $\mu m^{(p+13)/2}/\sqrt{\varepsilon_{\mathrm{fid}}}$. Thus an RDM error of order $\varepsilon_{\mathrm{fid}}m^{-(p+13)/2}$ suffices, explaining the choice $\mathfrak a_p^{\mathrm{hom}}=\lceil(p+13)/2\rceil$. The assumed accuracy gives $\mathfrak A\leq\mathsf C_pc_p\sqrt{\varepsilon_{\mathrm{fid}}}$; decreasing $c_p$ if necessary proves Eq.~\eqref{eq:equal-p-df-explicit-assembly}.
\end{proof}

\subsection{Completion of the reconstruction guarantee}
\label{app:same-particle-noise-analysis}

\begin{proof}[Proof of Thm.~\ref{thm:equal-p-noisy-compact-direct}]
Choose $c_p$ as in Prop.~\ref{prop:equal-p-input-accuracy}. That proposition verifies the thresholds, selected-space dimensions, and the input conditions of both splitting calls. Each call has failure probability at most $\beta/2$ directly by Prop.~\ref{prop:robust-gram-splitting}. Hence both calls succeed with probability at least $1-\beta$.

On this event the same accuracy proposition gives $\mathfrak A\leq\sqrt{\varepsilon_{\mathrm{fid}}}/4$. Applying Prop.~\ref{prop:equal-p-assembly} and $3ms_0^2=\varepsilon_{\mathrm{fid}}/16$ shows that the final product is nonzero and its infidelity is at most $(\sqrt{\varepsilon_{\mathrm{fid}}}/4+\sqrt{\varepsilon_{\mathrm{fid}}}/4)^2\leq\varepsilon_{\mathrm{fid}}$.

All operations use the supplied RDMs and recovered classical vectors. The cost analysis in Sec.~\ref{sec:homogeneous-time-complexity} gives polynomial arithmetic cost for fixed $p$, with $O(\log(1/\beta))$ trials per splitting call, and no additional state copies.
\end{proof}

\section{Reconstruction from estimated RDMs in the heterogeneous setting}
\label{app:heterogeneous-noisy}

At each particle number, reconstruction in the heterogeneous setting inherits errors from the earlier factors through the product spaces that must be removed. The following theorem gives a sufficient input accuracy for the full recursion.

\begin{theorem}[Reconstruction from estimated RDMs]\label{thm:mixed-noisy-compact}
Fix $2\leq r\leq m$ and $\varepsilon_{\mathrm{fid}},\beta\in(0,1)$. There is a constant $c_r>0$, depending only on $r$, such that the following holds for every target $\Psi$ in the heterogeneous setting with block particle numbers at most $r$. Given $r$, $\varepsilon_{\mathrm{fid}}$, $\beta$, and Hermitian RDM estimates with a supplied error bound
\begin{align}
    \max_{1\leq k\leq r}\|\widetilde\Gamma^{(k)}-\Gamma_\Psi^{(k)}\|&\leq\nu\leq c_r\frac{\varepsilon_{\mathrm{fid}}}{m^{\lceil(7\lfloor r/2\rfloor+r+6)/2\rceil}},
    \label{eq:delta-free-common-rdm-error}
\end{align}
a randomized classical algorithm returns, with probability at least $1-\beta$, a factorized description of a normalized state $\widetilde\Phi$ satisfying
\begin{align}
    |\langle\widetilde\Phi,\Psi\rangle|^2&\geq1-\varepsilon_{\mathrm{fid}}.
    \label{eq:delta-free-final-fidelity}
\end{align}
The description consists of an orthonormal basis for an occupied core and an ordered list of normalized vectors of degrees between $2$ and $r$. Their exterior product is nonzero, with normalization implicit in the description. The individual block particle numbers, block decomposition, and passive Gaussian unitary are not required as input. For fixed $r$, the algorithm uses polynomially many arithmetic operations in $m$ and $\log(1/\beta)$ and no additional copies of the target state.
\end{theorem}

We first specify the algorithm and its computable error bounds, then control their propagation with a single recurrence in App.~\ref{app:heterogeneous-recursive-stability}. The final assembly and conditional failure estimates complete the proof.

\subsection{Recursive reconstruction from estimated RDMs}
\label{app:heterogeneous-certified-noisy-algorithm}

We use the exact recursion of App.~\ref{app:different-particle-algorithm} with estimated RDMs and the full-block and core constructions of App.~\ref{app:uniform-certified-noisy-algorithm}. The error bounds below are computed from the RDM accuracy and earlier outputs. Prop.~\ref{prop:df-certified-recursion} verifies one level, and Prop.~\ref{prop:df-simultaneous-induction} verifies its conditions throughout the recursion. All \textsc{Fail} conditions are checked from the estimates and these propagated bounds.

\paragraph*{Parameters and error bounds.}
Use the inputs and the supplied error bound $\nu>0$ of Thm.~\ref{thm:mixed-noisy-compact}, together with the resolution scale $s_0$ from Eq.~\eqref{eq:equal-p-df-s0}. For the accuracy analysis, write
\begin{align}
    \mathfrak a_r^{\mathrm{het}}&:=\left\lceil\frac{7\lfloor r/2\rfloor+r+6}{2}\right\rceil.
    \label{eq:explicit-delta-free-rdm-exponent}
\end{align}
Assign each of the at most $2(r-1)$ Gram-splitting calls the local failure budget
\begin{align}
    \beta_{\mathrm{loc}}&:=\frac{\beta}{2(r-1)}.
    \label{eq:heterogeneous-local-failure-budget}
\end{align}
All error bounds below are computed from $\nu$, known parameters, and the stored bounds from earlier levels. RDM estimation is treated separately in App.~\ref{app:rdm-estimation}; here the probability is over the internal Gram-splitting randomness.

\paragraph*{Step 1: Construct the common high-occupation space.}
Use the homogeneous initialization with $\mu$ replaced by $\nu$: set $h_1:=1/(96m)$, choose $\theta\in[2/3,3/4]$ with the spectral separation below, and define the high-occupation space and its error bound by
\begin{align}
    \operatorname{dist}\!\left(\theta,\operatorname{spec}\widetilde\Gamma^{(1)}\right)&\geq\frac{h_1}{2},\qquad \widetilde S_\theta:=\operatorname{ran}\mathbf1_{(\theta,\infty)}(\widetilde\Gamma^{(1)}),\qquad \eta_S:=\frac{8\nu}{h_1}.
    \label{eq:noisy-high-occupation-space}
\end{align}
The certificate $\eta_S$ bounds the projector error relative to $S_\theta$, which includes the always-occupied component $S$ as in Eq.~\eqref{eq:exact-high-occupation-decomposition}. All exact comparison spaces and operators use the thresholds $\theta$ and $t_k$ chosen from the estimates.

After Step~1, process $k=2,\ldots,r$ in increasing order, carrying out Steps~2--4 at each order. Step~5 is performed after all orders are completed.

At each completed order $\ell$, store the dominant branches $\widetilde f_{\ell,j}$, their reconstructed full blocks $\widetilde\omega_{\ell,j}^{\mathrm N}$, and the complementary blocks $\widetilde\omega_{\ell,j}^{\mathrm H}$, together with their error bounds $\xi_{f,\ell}$, $\xi_{\mathrm N,\ell}$, and $\xi_{\mathrm H,\ell}$ from Eq.~\eqref{eq:df-factor-certificates}. Use the orthonormal branch list for later dominant products and the orthonormal complementary-block list for later complementary products. The branches also determine the occupied core, while the reconstructed dominant blocks are used only in final assembly. Since each old product contains at least two factors of degree at least two, its constituents at order $k$ have degrees at most $k-2$; the old-product spaces and their error bounds are therefore zero at $k=2,3$.

\paragraph*{Step 2: Form the two compressed RDMs at order $k$.}
Using the same estimated high-occupation space at every level, define
\begin{align}
    \widetilde Q_{\mathrm{hi},k} &:=P_{\wedge^k\widetilde S_\theta},
    \qquad\widetilde Q_{\mathrm{lo},k} :=P_{\wedge^k(\widetilde S_\theta^\perp)},
    \qquad \widetilde C_k :=\widetilde Q_{\mathrm{lo},k}\widetilde\Gamma^{(k)} \widetilde Q_{\mathrm{hi},k},
    \qquad \widetilde A_k :=\widetilde Q_{\mathrm{lo},k}\widetilde\Gamma^{(k)} \widetilde Q_{\mathrm{lo},k},
    \label{eq:heterogeneous-noisy-compressions}
\end{align}
and set $e_k:=\nu+2k\eta_S$ to bound both compression errors, as in Lem.~\ref{lem:df-compressed-operator-error}.

\paragraph*{Step 3: Remove old dominant products and recover new blocks.}
Choose the singular threshold in $[s_0,2s_0]$ using a gap parameter that accounts for products of smaller blocks. Set
\begin{align}
    q_k&:=\left\lfloor\frac{k}{2}\right\rfloor,
    \qquad h_k:=\frac{s_0}{8r m^{q_k}}.
    \label{eq:certified-singular-gap}
\end{align}
A degree-$k$ product contains at most $q_k$ blocks, giving the rank bound $\operatorname{rank}C_k\leq r m^{q_k}$ in Lem.~\ref{lem:heterogeneous-coherence-decomposition}. Choose $t_k\in[s_0,2s_0]$ satisfying
\begin{align}
    \operatorname{dist}\left( t_k,\operatorname{sing}\widetilde C_k \right)&\geq\frac{h_k}{2}.
    \label{eq:certified-singular-threshold}
\end{align}
If no such threshold exists, return \textsc{Fail}. Define
\begin{align}
    \widetilde R_k&:=\operatorname{ran}\mathbf1_{(t_k^2,\infty)}\left(\widetilde C_k^\dagger\widetilde C_k\right),
    \qquad \eta_{R,k}:=\frac{8e_k}{h_k}.
    \label{eq:heterogeneous-near-all-space-certificate}
\end{align}
Here $\eta_{R,k}$ bounds the projector error relative to the exact candidate space $R_k$ selected at $t_k$.

To remove the earlier branch products from $\widetilde R_k$, form their span $\hat O_k^{\mathrm N}$. Each product uses at least two distinct previously recovered dominant branches, each at most once, with total degree $k$. Since every constituent has degree at most $k-2$, Lem.~\ref{lem:df-product-span-stability} gives the projector-error bound
\begin{align}
    \rho_k^{\mathrm N}&:=r\sqrt{2^r r!\,m}\max_{2\leq\ell\leq k-2}\xi_{f,\ell}.
    \label{eq:df-old-branch-error-bound}
\end{align}
If no such product exists, set $\hat O_k^{\mathrm N}=\{0\}$ and $\rho_k^{\mathrm N}=0$.

Lem.~\ref{lem:df-stable-intersection} bounds the product-removal operator error by $2\eta_{R,k}+\rho_k^{\mathrm N}$ and gives the subspace-error certificate
\begin{align}
    \eta_{F,k}^{\mathrm N}&:=4\bigl(2\eta_{R,k}+\rho_k^{\mathrm N}\bigr).
    \label{eq:heterogeneous-near-recursive-certificates}
\end{align}
If $\rho_k^{\mathrm N}\geq1$ or $2\eta_{R,k}+\rho_k^{\mathrm N}>1/4$, return \textsc{Fail}. Otherwise, define
\begin{align}
    \widetilde B_k^{\mathrm N}&:=P_{\widetilde R_k}(I-P_{\hat O_k^{\mathrm N}})P_{\widetilde R_k},\qquad \widetilde F_k^{\mathrm N}:=\operatorname{ran}\mathbf1_{(1/2,1]}(\widetilde B_k^{\mathrm N}),\qquad \widetilde d_k^{\mathrm N}:=\dim\widetilde F_k^{\mathrm N}.
    \label{eq:noisy-new-near-space}
\end{align}
The new space has projector-error certificate $\eta_{F,k}^{\mathrm N}$ and lies in $\widetilde R_k\subseteq\wedge^k\widetilde S_\theta$. If the old-product space is zero, the candidate space is retained unchanged.

Apply the Gram-splitting algorithm of Prop.~\ref{prop:robust-gram-splitting} to $(\widetilde F_k^{\mathrm N},k,\eta_{F,k}^{\mathrm N},\beta_{\mathrm{loc}})$. If it returns \textsc{Fail}, return \textsc{Fail}; otherwise denote the branch estimates by $\widetilde f_{k,1},\ldots,\widetilde f_{k,\widetilde d_k^{\mathrm N}}$. Apply $\widetilde\Gamma^{(k)}$ to each branch and normalize, as in Eq.~\eqref{eq:uniform-noisy-near-reconstruction}, to obtain $\widetilde\omega_{k,j}^{\mathrm N}$. Return \textsc{Fail} if any image is zero. Keep both lists: the branches are needed for subsequent product removal and the core, and the full blocks for final assembly.

\paragraph*{Step 4: Remove old complementary products and recover new blocks.}
Using the exact gap in Lem.~\ref{lem:constant-gap-strong-hierarchy}, define
\begin{align}
    \widetilde W_k^{\mathrm H}
    &:=\operatorname{ran}\mathbf1_{((1+\theta)/2,\infty)}
    \left(\widetilde A_k\right),
    \qquad
    \eta_{W,k}^{\mathrm H}:=\frac{4e_k}{1-\theta}.
    \label{eq:noisy-strong-all-space}
\end{align}
Here $\eta_{W,k}^{\mathrm H}$ bounds the error of the candidate space. Remove the complementary block products already accounted for at lower levels as follows.

Construct $\hat O_k^{\mathrm H}$ from products of at least two distinct complementary block states returned from the lower-order $\widetilde T_\ell^{\mathrm H}$ spaces, each used at most once, with total degree $k$. Apply Lem.~\ref{lem:df-product-span-stability} to their stored errors and set
\begin{align}
    \rho_k^{\mathrm H}&:=r\sqrt{2^r r!\,m}\max_{2\leq\ell\leq k-2}\xi_{\mathrm H,\ell}.
    \label{eq:df-old-complementary-error-bound}
\end{align}
If no such product exists, set $\hat O_k^{\mathrm H}=\{0\}$ and $\rho_k^{\mathrm H}=0$. Set
\begin{align}
    \eta_{T,k}^{\mathrm H}&:=4\bigl(2\eta_{W,k}^{\mathrm H}+\rho_k^{\mathrm H}\bigr).
    \label{eq:heterogeneous-strong-recursive-certificates}
\end{align}
If $\rho_k^{\mathrm H}\geq1$ or $2\eta_{W,k}^{\mathrm H}+\rho_k^{\mathrm H}>1/4$, return \textsc{Fail}. Define
\begin{align}
    \begin{aligned} \widetilde B_k^{\mathrm H} &:=P_{\widetilde W_k^{\mathrm H}}(I-P_{\hat O_k^{\mathrm H}})P_{\widetilde W_k^{\mathrm H}},\qquad \widetilde T_k^{\mathrm H} :=\operatorname{ran}\mathbf1_{(1/2,1]}(\widetilde B_k^{\mathrm H}), \qquad \widetilde d_k^{\mathrm H}:=\dim\widetilde T_k^{\mathrm H}. \end{aligned}
    \label{eq:noisy-new-strong-space}
\end{align}
The projector-error certificate is $\eta_{T,k}^{\mathrm H}$. A zero old-product space again leaves the candidate space unchanged.
Apply the same algorithm to $(\widetilde T_k^{\mathrm H},k,\eta_{T,k}^{\mathrm H},\beta_{\mathrm{loc}})$. If it returns \textsc{Fail}, return \textsc{Fail}; otherwise denote the outputs by $\widetilde\omega_{k,1}^{\mathrm H},\ldots,\widetilde\omega_{k,\widetilde d_k^{\mathrm H}}^{\mathrm H}$.
This call returns the full complementary blocks directly.

After both calls succeed, compute the three level-$k$ error bounds from Eq.~\eqref{eq:df-factor-certificates} and store them with their respective factor lists. If $k<r$, repeat Steps~2--4 at order $k+1$; otherwise, proceed to Step~5.

\paragraph*{Step 5: Construct the core and assemble the output.}
Sum the retained-branch contributions over all completed orders:
\begin{align}
    \widetilde K
    &:=\sum_{k=2}^{r}\sum_{j=1}^{\widetilde d_k^{\mathrm N}}
    \Gamma_{\widetilde f_{k,j}}^{(1)}.
    \label{eq:noisy-retained-support}
\end{align}
with $\widetilde K=0$ when no dominant branch is recovered. Since the branches at each order form an orthonormal basis of $\widetilde F_k^{\mathrm N}$, this is the sum used by the core-recovery procedure of Prop.~\ref{prop:df-core-stability}, applied to $\widetilde S_\theta$ and the branch spaces $\widetilde F_k^{\mathrm N}$. Its construction thresholds $\widetilde K$ as in Eq.~\eqref{eq:uniform-noisy-retained-support-space} and returns $\widetilde S_{\mathrm{core}}:=\widetilde S_\theta\cap\widetilde F^\perp$. The proposition bounds this space's error relative to the exact core in Lem.~\ref{lem:heterogeneous-exact-output} directly from $\eta_S$ and $\eta_{F,k}^{\mathrm N}$.

Choose an orthonormal basis of $\widetilde S_{\mathrm{core}}$ and let $\widetilde\sigma_{\mathrm{core}}$ denote its unit exterior volume, with the same vacuum convention as in the homogeneous algorithm. With all factors in a fixed deterministic order, define the assembled vector and, whenever it is nonzero, its normalized state by
\begin{align}
    \widetilde Y&:=\widetilde\sigma_{\mathrm{core}}\wedge\bigwedge_{k=2}^{r}\bigwedge_{j=1}^{\widetilde d_k^{\mathrm N}}\widetilde\omega_{k,j}^{\mathrm N}\wedge\bigwedge_{k=2}^{r}\bigwedge_{j=1}^{\widetilde d_k^{\mathrm H}}\widetilde\omega_{k,j}^{\mathrm H},\qquad \widetilde\Phi:=\frac{\widetilde Y}{\|\widetilde Y\|}.
    \label{eq:heterogeneous-noisy-output-state}
\end{align}
Return the chosen core basis and the ordered block lists from all levels, which give a factorized description of $\widetilde\Phi$.
Use the same implicit normalization convention as in App.~\ref{app:uniform-certified-noisy-algorithm}; the learner returns the factors without expanding the many-particle vector.

The conditional success probabilities and final fidelity are proved in Thm.~\ref{thm:mixed-noisy-compact}, using the propagated certificates below and the budget in Eq.~\eqref{eq:heterogeneous-local-failure-budget}.

\subsection{Stability of the heterogeneous recursion}
\label{app:heterogeneous-recursive-stability}

We first verify one reconstruction level under the assumption that the earlier outputs have been correctly matched. We then control the accumulated errors by induction, checking that these local conditions hold at every level reached after successful earlier calls.

\subsubsection{One-level guarantee}

At a new order, the input errors come from the RDM estimates and from earlier recovered factors. The following guarantee uses only these two error bounds.

\begin{proposition}[One reconstruction level]\label{prop:df-certified-recursion}
Fix $2\leq k\leq r\leq m$, $0<s_0\leq1$, and $\beta_{\mathrm{loc}}\in(0,1)$ in the heterogeneous setting. Suppose the supplied Hermitian $1$- and $k$-RDM estimates have operator-norm errors at most $\nu$. Use the common high-occupation space selected in Step~1 of App.~\ref{app:heterogeneous-certified-noisy-algorithm}, and suppose the earlier output lists match the exact retained branches and complementary blocks at their chosen thresholds. The branch and complementary-block lists are orthonormal at each degree, and their stored norm-error bounds are valid. Let $x\geq0$ bound these stored errors at degrees at most $k-2$; take $x=0$ for $k=2,3$.

There are constants $c_r,C_r>0$, depending only on $r$, such that, if
\begin{align}
    m^{\lfloor k/2\rfloor+4}\frac{\nu}{s_0}+m^{7/2}x&\leq c_r,
    \label{eq:df-recursion-smallness}
\end{align}
Steps~2--4 at order $k$ admit the required threshold, give new-factor spaces of the correct dimensions, and satisfy both Gram-splitting input conditions. Each splitting call fails with conditional probability at most $\beta_{\mathrm{loc}}$. Whenever both calls succeed, every returned dominant branch or full block $\widetilde v$ matches its exact counterpart $v$, up to phases and a permutation, with
\begin{align}
    \|\widetilde v-v\|&\leq C_r\left(m^{\lfloor k/2\rfloor+4}\frac{\nu}{s_0}+m^{7/2}x\right).
    \label{eq:df-recursive-factor-errors}
\end{align}
The full-block normalizations are well defined, and the returned factor-error certificates are valid and bounded by the same right-hand side.
\end{proposition}

\begin{proof}
We propagate the two supplied error bounds through one level, checking each construction before using it. Put $q:=\lfloor k/2\rfloor$; all constants below depend only on $r$.

\emph{Thresholds and candidate spaces.} Step~1 uses $h_1=1/(96m)$, so its certificate satisfies $\eta_S\leq C_rm\nu$. Lem.~\ref{lem:df-threshold-selection} provides the one-body threshold, and Eq.~\eqref{eq:df-recursion-smallness}, with $c_r$ sufficiently small, implies $\nu<h_1/4$. Hence Lem.~\ref{lem:df-certified-threshold} validates this certificate. Lem.~\ref{lem:df-compressed-operator-error} then gives $e_k\leq C_rm\nu$ for both compressed operators. The rank bound in Lem.~\ref{lem:heterogeneous-coherence-decomposition} and the gap parameter $h_k=s_0/(8rm^q)$ imply
\begin{align}
    \eta_{R,k}&\leq C_rm^{q+1}\frac{\nu}{s_0},\qquad \eta_{W,k}^{\mathrm H}\leq C_rm\nu.
    \label{eq:df-R-power}
\end{align}
The same smallness assumption gives $e_k<h_k/4<s_0/2$ and $e_k<(1-\theta)/4$. Thus Lem.~\ref{lem:df-threshold-selection} supplies the singular threshold, and Cor.~\ref{cor:df-threshold-space-errors} verifies the candidate-space dimensions and error bounds.

\emph{Removal of earlier products.} Every degree-$k$ old product uses factors of degree at most $k-2$, whose stored errors are at most $x$. Their exact counterparts have mutually orthogonal one-particle supports, and there are at most $m/2$ factors of either type. Lem.~\ref{lem:df-product-span-stability} bounds the two product-span certificates by $C_r\sqrt m\,x$. Eq.~\eqref{eq:df-recursion-smallness} makes them smaller than one, so
\begin{align}
    \|P_{\hat O_k^{\mathrm N}}-P_{O_k^{\mathrm N}}\|&\leq\rho_k^{\mathrm N},\qquad \|P_{\hat O_k^{\mathrm H}}-P_{O_k^{\mathrm H}}\|\leq\rho_k^{\mathrm H},\qquad \rho_k^{\mathrm N}+\rho_k^{\mathrm H}\leq C_r\sqrt m\,x.
    \label{eq:df-old-space-errors}
\end{align}
The exact candidate and old-product projectors commute by Lems.~\ref{lem:near-hierarchy-closure} and~\ref{lem:constant-gap-strong-hierarchy}. The certificate definitions therefore give
\begin{align}
    \max\{\eta_{F,k}^{\mathrm N},\eta_{T,k}^{\mathrm H}\}&\leq C_r\left(m^{q+1}\frac{\nu}{s_0}+\sqrt m\,x\right).
    \label{eq:df-presplit-one-step}
\end{align}
By decreasing $c_r$, Eq.~\eqref{eq:df-recursion-smallness} makes this bound at most $1/(64km^3)$. In particular, both product-removal errors are at most $1/4$, so Lem.~\ref{lem:df-stable-intersection} gives the correct new-space dimensions and
\begin{align}
    \|P_{\widetilde F_k^{\mathrm N}}-P_{F_k^{\mathrm N}}\|&\leq\eta_{F,k}^{\mathrm N},\qquad \|P_{\widetilde T_k^{\mathrm H}}-P_{T_k^{\mathrm H}}\|\leq\eta_{T,k}^{\mathrm H}.
    \label{eq:df-new-space-errors}
\end{align}
For $k=2,3$, the old-product spaces are zero and the same argument applies with $x=0$.

\emph{Recovery of branches and full blocks.} The verified dimensions and subspace errors permit both calls to Prop.~\ref{prop:robust-gram-splitting}. Each call fails with probability at most $\beta_{\mathrm{loc}}$, conditional on the supplied earlier outputs. On success, its factor errors are bounded by $(8k+2)m^3$ times the corresponding subspace certificate, giving the right-hand side of Eq.~\eqref{eq:df-recursive-factor-errors} for branches and complementary blocks. The smallness assumption also ensures $\nu+\xi_{f,k}<\sqrt\theta/2$. Lem.~\ref{lem:df-near-block-stability} therefore makes each dominant-block normalization well defined and gives the same bound, after increasing $C_r$, for the full dominant blocks. These are precisely the stored certificates in Eq.~\eqref{eq:df-factor-certificates}.
\end{proof}

\subsubsection{Control of the propagated errors}

Prop.~\ref{prop:df-certified-recursion} expresses one level's error directly in terms of the RDM accuracy and earlier factor errors. Iterating this guarantee gives the following bounds for the full reconstruction, without requiring the intermediate subspace certificates as inputs.

\begin{proposition}[Errors throughout the recursion]\label{prop:df-simultaneous-induction}
Fix $2\leq r\leq m$ and $0<s_0\leq1$ in the heterogeneous setting. There are constants $c_r,C_r>0$, depending only on $r$, such that the following holds. Suppose the supplied Hermitian RDM estimates of orders $1,\ldots,r$ have operator-norm errors at most $\nu$, with
\begin{align}
    \frac{\nu}{s_0}&\leq c_rm^{-(7\lfloor r/2\rfloor+r+5)/2}.
    \label{eq:df-bootstrap-smallness}
\end{align}
Run the algorithm of App.~\ref{app:heterogeneous-certified-noisy-algorithm} with resolution $s_0$ and error bound $\nu$. At every order reached after successful earlier Gram-splitting calls, the required thresholds exist, the selected new-factor spaces have the correct dimensions, and all reconstruction conditions hold.

If all calls through order $k$ succeed, the returned dominant branches and full blocks of degrees at most $k$ match their exact counterparts up to phases and a permutation. Each such factor $\widetilde v$, with exact counterpart $v$, satisfies
\begin{align}
    \|\widetilde v-v\|&\leq C_rm^{(7\lfloor k/2\rfloor+3)/2}\frac{\nu}{s_0}.
    \label{eq:df-recursive-power-bound}
\end{align}
Their stored error certificates are valid and bounded by the same right-hand side. If all calls through order $r$ succeed, the occupied core has the correct dimension and its phase-aligned state satisfies
\begin{align}
    \|\widetilde\sigma_{\mathrm{core}}-\sigma_{\mathrm{core}}\|&\leq C_rm^{(7\lfloor r/2\rfloor+6)/2}\frac{\nu}{s_0}.
    \label{eq:df-core-power-bound}
\end{align}
\end{proposition}

\begin{proof}
We apply Prop.~\ref{prop:df-certified-recursion} in increasing order. For this proof, let $x_k$ be the largest stored branch or full-block error certificate through order $k$, with $x_0=x_1=0$, and put $q_k:=\lfloor k/2\rfloor$. We show that these certificates are valid along every successful history.

\emph{Initial orders.} At $k=2,3$, there are no old products, so the one-level guarantee takes $x=0$ and its smallness condition reduces to $m^5\nu/s_0\leq c_r$. Eq.~\eqref{eq:df-bootstrap-smallness} ensures this condition after decreasing its constant. Each completed initial order therefore has valid certificates bounded by $C_rm^5\nu/s_0$.

\emph{Propagation to the next order.} An old product at order $k\geq4$ uses factors of degree at most $k-2$, so the next call to Prop.~\ref{prop:df-certified-recursion} takes $x=x_{k-2}$. Whenever its smallness condition holds, that proposition yields
\begin{align}
    x_k&\leq\max\left\{x_{k-1},C_r\left(m^{q_k+4}\frac{\nu}{s_0}+m^{7/2}x_{k-2}\right)\right\}.
    \label{eq:df-xk-recursion}
\end{align}
Thus the initial power $m^5$ increases by $7/2$ every two orders. Set $a_s:=(7s+3)/2$ for this induction; then $a_1=5$ and $a_s=a_{s-1}+7/2$. Assuming the asserted bounds through order $k-1$, the input to the next one-level guarantee satisfies
\begin{align}
    m^{q_k+4}\frac{\nu}{s_0}+m^{7/2}x_{k-2}&\leq C_rm^{a_{q_k}}\frac{\nu}{s_0}\leq C_rc_rm^{-1-r/2},
\end{align}
where $q_k+4\leq a_{q_k}$ and Eq.~\eqref{eq:df-bootstrap-smallness} were used. Choosing the input constant sufficiently small verifies the one-level hypothesis before invoking it. The proposition now gives valid new certificates and, through Eq.~\eqref{eq:df-xk-recursion}, $x_k\leq C_rm^{a_{q_k}}\nu/s_0$. This is Eq.~\eqref{eq:df-recursive-power-bound}, including all earlier orders because $a_{q_k}$ is nondecreasing. There are at most $r-1$ levels, so one may first bound the finite sequence of induction constants by a constant depending only on $r$, and then choose the input constant uniformly for all levels and successful histories.

\emph{Occupied core.} Suppose all calls through order $r$ succeed. Apply Prop.~\ref{prop:df-core-stability} to the recovered branch spaces $\widetilde F_k^{\mathrm N}$. Their dimensions are correct, and their valid subspace certificates satisfy $\eta_{F,k}^{\mathrm N}\leq\xi_{f,k}\leq x_r$ whenever the space is nonempty. Since there are at most $m/2$ retained branches,
\begin{align}
    \sum_{k=2}^r d_k^{\mathrm N}\eta_{F,k}^{\mathrm N}&\leq\frac m2 x_r\leq C_rc_rm^{-r/2}.
\end{align}
The input condition therefore makes this sum at most $1/4$ and also gives $\eta_S\leq C_rm\nu<1$. The core-recovery proposition applies and returns the correct dimension, with phase-aligned error at most $C_r\sqrt m(\eta_S+mx_r)$. Substituting the factor bound gives Eq.~\eqref{eq:df-core-power-bound}.
\end{proof}

\subsection{Verification of the required accuracy}
\label{app:heterogeneous-fidelity-guarantee}

The propagated-error bounds now connect the RDM accuracy to the final fidelity. We combine Prop.~\ref{prop:df-simultaneous-induction} with the assembly guarantee of Prop.~\ref{prop:equal-p-assembly} and the truncation bound in Lem.~\ref{lem:heterogeneous-exact-output}.

\begin{proposition}[Sufficient accuracy for final assembly]\label{prop:delta-free-explicit-accuracy}
Fix $2\leq r\leq m$ and $0<\varepsilon_{\mathrm{fid}}\leq1$ in the heterogeneous setting, and set $s_0:=\sqrt{\varepsilon_{\mathrm{fid}}/(48m)}$. There is a constant $c_r>0$, depending only on $r$, such that the following holds. Suppose the supplied Hermitian RDM estimates of orders $1,\ldots,r$ have operator-norm errors at most $\nu$, with
\begin{align}
    \nu&\leq c_r\frac{\varepsilon_{\mathrm{fid}}}{m^{\lceil(7\lfloor r/2\rfloor+r+6)/2\rceil}}.
    \label{eq:df-explicit-nu-choice}
\end{align}
Then the algorithm of App.~\ref{app:heterogeneous-certified-noisy-algorithm}, run with error bound $\nu$, satisfies every reconstruction condition at each order reached after successful earlier Gram-splitting calls.

If all calls succeed, the recovered block lists match the exact retained blocks, and the occupied core has the correct dimension. The block and phase-aligned core errors satisfy the bounds of Prop.~\ref{prop:df-simultaneous-induction}; inserting these bounds into Prop.~\ref{prop:equal-p-assembly} gives
\begin{align}
    \mathfrak A&\leq\frac{\sqrt{\varepsilon_{\mathrm{fid}}}}{4}.
    \label{eq:df-explicit-assembly-bound}
\end{align}
\end{proposition}

\begin{proof}
We first verify that the supplied RDM accuracy permits the recursion, then apply its block and core bounds to assembly. All constants below depend only on $r$.

\emph{Validity of the recursion.} By the stated accuracy and $s_0=\sqrt{\varepsilon_{\mathrm{fid}}/(48m)}$,
\begin{align}
    \frac{\nu}{s_0}&\leq\sqrt{48}\,c_r\sqrt{\varepsilon_{\mathrm{fid}}}m^{1/2-\mathfrak a_r^{\mathrm{het}}}\leq\sqrt{48}\,c_r\sqrt{\varepsilon_{\mathrm{fid}}}m^{-(7\lfloor r/2\rfloor+r+5)/2}.
    \label{eq:df-nu-over-s0}
\end{align}
Since $\varepsilon_{\mathrm{fid}}\leq1$, choosing $c_r$ sufficiently small verifies the input condition in Eq.~\eqref{eq:df-bootstrap-smallness}. Prop.~\ref{prop:df-simultaneous-induction} therefore validates every reconstruction step reached after successful earlier calls. Condition henceforth on all Gram-splitting calls succeeding. The proposition then gives the matching of recovered factors and the correct core dimension, with full-block errors at most $C_rm^{(7\lfloor r/2\rfloor+3)/2}\nu/s_0$ and phase-aligned core error at most $C_rm^{(7\lfloor r/2\rfloor+6)/2}\nu/s_0$.

\emph{Assembly.} There are at most $m/2$ recovered blocks, each of degree at most $r$. Thus, in the assembly product of Prop.~\ref{prop:equal-p-assembly}, the sum of the first-order terms $u_j:=m^{p_j/2}\epsilon_j/\sqrt{p_j!}$ satisfies
\begin{align}
    \sum_j u_j&\leq C_rm^{(7\lfloor r/2\rfloor+r+5)/2}\frac{\nu}{s_0}.
    \label{eq:df-explicit-correlated-sum}
\end{align}
The exponent contains the factor-error exponent plus $r/2$ for each term and one for the number of blocks. Eq.~\eqref{eq:df-nu-over-s0} bounds this sum by $C_rc_r\sqrt{\varepsilon_{\mathrm{fid}}}$, which is at most one for sufficiently small $c_r$. Hence $\prod_j(1+u_j)-1\leq2\sum_j u_j$. Since $r\geq2$, the core-error exponent is also smaller than the exponent in Eq.~\eqref{eq:df-explicit-correlated-sum}. Adding the core error, therefore, yields
\begin{align}
    \mathfrak A&\leq C_rm^{(7\lfloor r/2\rfloor+r+5)/2}\frac{\nu}{s_0}.
\end{align}
Because $s_0=\sqrt{\varepsilon_{\mathrm{fid}}/(48m)}$, it suffices to take an RDM error of order $\varepsilon_{\mathrm{fid}}m^{-(7\lfloor r/2\rfloor+r+6)/2}$ to make this bound of order $\sqrt{\varepsilon_{\mathrm{fid}}}$. This explains the exponent $\mathfrak a_r^{\mathrm{het}}$ in Eq.~\eqref{eq:explicit-delta-free-rdm-exponent}. Substituting Eq.~\eqref{eq:df-nu-over-s0} implies $\mathfrak A\leq C_rc_r\sqrt{\varepsilon_{\mathrm{fid}}}$; decreasing $c_r$ if necessary proves Eq.~\eqref{eq:df-explicit-assembly-bound}.
\end{proof}

\subsection{Completion of the reconstruction guarantee}
\label{app:different-particle-noise-analysis}

\begin{proof}[Proof of Thm.~\ref{thm:mixed-noisy-compact}]
Choose $c_r$ as in Prop.~\ref{prop:delta-free-explicit-accuracy}. For every history in which the preceding splitting calls have succeeded, that proposition verifies the next call's input conditions. Prop.~\ref{prop:robust-gram-splitting} directly bounds its conditional failure probability by $\beta/[2(r-1)]$. Summing over the at most $2(r-1)$ possible first failed calls gives total failure probability at most $\beta$; independence between calls is not required.

If all calls succeed, Prop.~\ref{prop:delta-free-explicit-accuracy} gives the correct core dimension and $\mathfrak A\leq\sqrt{\varepsilon_{\mathrm{fid}}}/4$. Prop.~\ref{prop:equal-p-assembly} then proves that the product is nonzero and bounds its infidelity by $(\sqrt{\varepsilon_{\mathrm{fid}}}/4+\sqrt{3ms_0^2})^2\leq\varepsilon_{\mathrm{fid}}$.

The algorithm uses only the supplied RDMs and the recovered classical vectors. For fixed $r$, the cost analysis in Sec.~\ref{sec:heterogeneous-time-complexity}, with $O(\log(2r/\beta))$ trials per call, gives polynomial arithmetic cost in $m$ and $\log(1/\beta)$ without additional copies.
\end{proof}

\section{RDM estimation and state-copy complexity}\label{app:rdm-estimation}

We convert the entrywise guarantee of the fermionic partial-tomography protocol~\cite{zhao2021fermionic} into an operator-norm guarantee, including the classical cost of constructing the estimates. Substituting the accuracies from Apps.~\ref{app:homogeneous-noisy} and~\ref{app:heterogeneous-noisy} then gives the resource bounds in the main text.

\subsection{Operator-norm RDM estimation}\label{app:operator-norm-rdm-estimation}

\begin{proposition}[RDM estimation in operator norm]\label{prop:shadow-operator-norm-rdm}
Fix an RDM order $k\geq1$, independent of $m$, with $k\leq m$, an operator-norm accuracy $\nu_k\in(0,1)$, and a failure probability $\alpha_k\in(0,1)$. Write $D_k:=\binom{m}{k}$. From independent copies of $\Psi$, the Gaussian Clifford protocol of Ref.~\cite{zhao2021fermionic} produces a Hermitian estimate $\widetilde\Gamma^{(k)}$ satisfying
\begin{align}
    \left\|\widetilde\Gamma^{(k)}-\Gamma_\Psi^{(k)}\right\|&\leq\nu_k
    \label{eq:shadow-fixed-order-operator-error}
\end{align}
with probability at least $1-\alpha_k$. The required numbers of copies and classical arithmetic operations are, respectively,
\begin{align}
    O_k\left(D_k^3\frac{\log(2m/\alpha_k)}{\nu_k^2}\right)&\qquad\text{and}\qquad O_k\left(D_k^4\frac{\log(2m/\alpha_k)}{\nu_k^2}\right).
    \label{eq:shadow-fixed-order-copy-complexity}
\end{align}
\end{proposition}

\begin{proof}
By the simultaneous sample-mean guarantee of Ref.~\cite[Theorem~1 and Supplemental Material, Sec.~B.5]{zhao2021fermionic}, for fixed $k$, the Gaussian Clifford protocol estimates all Majorana expectations of even degree at most $2k$ to accuracy $\varepsilon_{\mathrm{entry}}\in(0,1)$ with probability at least $1-\alpha_k$, using
\begin{align}
    &O_k\left(D_k\frac{\log(2m/\alpha_k)}{\varepsilon_{\mathrm{entry}}^2}\right)
    \label{eq:shadow-entrywise-copy-complexity}
\end{align}
independent measurement records. Here the degree-$2q$ factor in the cited bound is $\binom{2m}{2q}/\binom mq=O_k(D_k)$ for $q\leq k$, and the number of observables through degree $2k$ is $O_k(m^{2k})$, giving the logarithmic factor above. Expanding the $2k$ creation and annihilation operators in an RDM entry into Majorana operators gives at most $4^k$ terms, each with coefficient of magnitude $2^{-2k}$. Reducing repeated Majorana operators leaves only even degrees at most $2k$ and does not increase the total absolute coefficient. The triangle inequality therefore bounds the absolute error of every, possibly complex, RDM entry by $\varepsilon_{\mathrm{entry}}$ on the same simultaneous event.

To count the classical cost, use the estimator in Ref.~\cite[Eq.~(14)]{zhao2021fermionic}. For each $0\leq q\leq k$, exactly $\binom mq$ degree-$2q$ estimators are nonzero in each record, corresponding to choices of $q$ measured mode pairs. Enumerating these contributions, including their indices, signs, and rescaling factors, takes $O_k(\sum_{q=0}^k\binom mq)=O_k(D_k)$ operations per record. We accumulate only these contributions and divide by the total number of records, not by the number of nonzero occurrences of an individual estimator. Initialization and conversion to all RDM entries cost $O_k(D_k^2)$ additional operations. Hence, the accumulation cost is $O_k(D_k^2\varepsilon_{\mathrm{entry}}^{-2}\log(2m/\alpha_k))$.

Let $\hat\Gamma^{(k)}$ be the resulting entrywise estimate and set $\widetilde\Gamma^{(k)}:=(\hat\Gamma^{(k)}+\hat\Gamma^{(k)\dagger})/2$. Hermitianization costs $O(D_k^2)$ operations and does not increase the operator-norm error. On the simultaneous entrywise event,
\begin{align}
    \|\hat\Gamma^{(k)}-\Gamma_\Psi^{(k)}\|
    &\leq\|\hat\Gamma^{(k)}-\Gamma_\Psi^{(k)}\|_{\mathrm{HS}}
    \leq D_k\varepsilon_{\mathrm{entry}}.
    \label{eq:shadow-entrywise-to-operator}
\end{align}
Choosing $\varepsilon_{\mathrm{entry}}=\nu_k/D_k$ gives the claimed copy and classical-time bounds. No positivity projection or trace-one normalization is required.
\end{proof}

\subsection{Reconstruction in the homogeneous setting}\label{app:homogeneous-state-copy-complexity}

\begin{corollary}[State-copy complexity in the homogeneous setting]
\label{cor:homogeneous-rdm-copy-complexity}
Fix $p\geq2$, $m\geq p$, and $\varepsilon_{\mathrm{fid}},\alpha\in(0,1)$ in the homogeneous setting. From independent copies of the target state, Hermitian estimates of its $1$- and $p$-RDMs satisfying the input-accuracy condition of Thm.~\ref{thm:equal-p-noisy-compact-direct} can be produced with probability at least $1-\alpha$. The required numbers of copies and classical arithmetic operations are, respectively,
\begin{align}
    O_p\left(m^{3p+2\mathfrak a_p^{\mathrm{hom}}}\varepsilon_{\mathrm{fid}}^{-2}\log\frac{m}{\alpha}\right)&\qquad\text{and}\qquad O_p\left(m^{4p+2\mathfrak a_p^{\mathrm{hom}}}\varepsilon_{\mathrm{fid}}^{-2}\log\frac{m}{\alpha}\right).
    \label{eq:equal-p-df-copy-complexity}
\end{align}
\end{corollary}

\begin{proof}
To obtain the RDM accuracy required by Eq.~\eqref{eq:equal-p-df-input-accuracy}, apply Prop.~\ref{prop:shadow-operator-norm-rdm} at orders $1$ and $p$, assigning failure probability $\alpha/2$ to each estimate. Since the union bound does not require independent estimation events, we can use a common batch of measurement records whose size meets both fixed-order requirements and obtain simultaneous success probability at least $1-\alpha$. For fixed $p\geq2$, this gives the stated powers of $m$.
\end{proof}

\subsection{Reconstruction in the heterogeneous setting}\label{app:heterogeneous-state-copy-complexity}

\begin{corollary}[State-copy complexity in the heterogeneous setting]
\label{cor:heterogeneous-rdm-copy-complexity}
Fix $2\leq r\leq m$ and $\varepsilon_{\mathrm{fid}},\alpha\in(0,1)$ in the heterogeneous setting. From independent copies of the target state, Hermitian RDM estimates of orders $1,\ldots,r$ satisfying the input-accuracy condition of Thm.~\ref{thm:mixed-noisy-compact} can be produced with probability at least $1-\alpha$. The required numbers of copies and classical arithmetic operations are, respectively,
\begin{align}
    O_r\left(m^{3r+2\mathfrak a_r^{\mathrm{het}}}\varepsilon_{\mathrm{fid}}^{-2}\log\frac{mr}{\alpha}\right)&\qquad\text{and}\qquad O_r\left(m^{4r+2\mathfrak a_r^{\mathrm{het}}}\varepsilon_{\mathrm{fid}}^{-2}\log\frac{mr}{\alpha}\right).
    \label{eq:different-p-df-copy-complexity}
\end{align}
\end{corollary}

\begin{proof}
Set $\nu:=c_r\varepsilon_{\mathrm{fid}}/m^{\mathfrak a_r^{\mathrm{het}}}$. Apply Prop.~\ref{prop:shadow-operator-norm-rdm} at each order $1\leq k\leq r$ with error $\nu$ and failure probability $\alpha/r$. For fixed $r$, a common batch of $M=O_r(m^{3r}\nu^{-2}\log(mr/\alpha))$ records is at least as large as every fixed-order requirement. The union bound gives simultaneous success probability at least $1-\alpha$; independence between the RDM estimates is not required. Processing this batch through degree $2r$ costs $O_r(Mm^r+m^{2r})$ operations. Substituting $\nu$ gives both asserted bounds. Since $\mathfrak a_r^{\mathrm{het}}=\lceil(7\lfloor r/2\rfloor+r+6)/2\rceil=O(r)$, the powers of $m$ in both bounds are linear in $r$.
\end{proof}

Combining the RDM-estimation event with the corresponding reconstruction theorem gives an assembled state $\widetilde\Phi$ with success probability at least $1-\alpha-\beta$. The main text takes $\alpha=\beta=\delta/3$ and allocates the remaining $\delta/3$ to the subsequent input-and-unitary construction, giving overall success probability at least $1-\delta$.

\section{Explicit input states and passive Gaussian unitaries}
\label{app:explicit-preparation}
\label{app:input-unitary-reconstruction}

For the chosen truncation thresholds, $\Psi_{\mathrm{tr}}$ is the ideal state obtained by truncating the true target $\Psi$; it is defined from the exact blocks, not from their estimates. With exact RDMs, the preceding reconstruction assembles these retained blocks and the occupied core into a state $\Phi$ that agrees with $\Psi_{\mathrm{tr}}$ up to phase. With estimated RDMs, the preceding algorithms instead provide estimated block vectors and core modes; their normalized assembly is the intermediate state $\widetilde\Phi$. Here we turn those estimates into a block-product input $\Xi_{\mathrm{in}}$ and a passive Gaussian unitary $U_{\mathrm{out}}$, giving the final estimate $\widetilde\Psi=\hat U_{\mathrm{out}}\Xi_{\mathrm{in}}$. We compare $\widetilde\Psi$ with $\Psi_{\mathrm{tr}}$ and use the earlier truncation bound to compare $\Psi_{\mathrm{tr}}$ with the target $\Psi$. After stating the guarantee, we derive exact branch recovery, its stable version for estimated blocks, and the global mode construction.

\subsection{Post-processing guarantee}
\label{app:iu-input-data}

For the chosen thresholds, the exact retained blocks and occupied core define the truncated target state
\begin{align}
    \Psi_{\mathrm{tr}}&:=\sigma_{\mathrm{core}}\wedge\omega_1\wedge\cdots\wedge\omega_d,\qquad \omega_b:=\sum_{\ell=1}^{s_b}\omega_{b,\ell}f_{b,\ell}.
    \label{eq:iu-comparison-state}
\end{align}
Here $2\leq p_b\leq r$, each normalized $f_{b,\ell}$ occupies a $p_b$-dimensional space $F_{b,\ell}$, all coefficients are nonzero, and $\sum_\ell|\omega_{b,\ell}|^2=1$. The normalized core state $\sigma_{\mathrm{core}}$ occupies every mode of a $q$-dimensional space $S_{\mathrm{core}}$. This space and all $F_{b,\ell}$ are mutually orthogonal, so $q+\sum_b p_bs_b\leq m$. The core includes the dominant branches of excluded blocks and any original always-occupied component.

For normalized vectors, we use the phase-invariant distance
\begin{align}
    d_2(x,y)&:=\min_{\varphi\in\mathbb R}\|x-e^{i\varphi}y\|=\sqrt{2-2|\langle x,y\rangle|}.
    \label{eq:iu-phase-invariant-distance}
\end{align}

The next proposition states what can be constructed from the estimated blocks and core: an explicit block-product input and one passive Gaussian unitary whose output remains close to $\Psi_{\mathrm{tr}}$.

\begin{proposition}[Preparing a state from recovered blocks]
\label{prop:iu-postprocessing}
For fixed $r\geq2$, there are constants $c_r,C_r>0$ such that the following holds. Given normalized coefficient vectors $\widetilde\omega_b\in\wedge^{p_b}H$, their particle numbers $p_b$, and an orthonormal basis of a $q$-dimensional estimated core, let $\widetilde\sigma_{\mathrm{core}}$ be its normalized occupied state. Suppose the supplied error bounds $\eta>0$ and $\eta_0\geq0$ satisfy
\begin{align}
    d_2(\widetilde\omega_b,\omega_b)&\leq\eta\quad(b\in[d]),\qquad d_2(\widetilde\sigma_{\mathrm{core}},\sigma_{\mathrm{core}})\leq\eta_0,\qquad m^{5/2}\eta+\eta_0\leq c_r.
    \label{eq:iu-supplied-errors}
\end{align}
For any $\beta\in(0,1)$, a randomized classical algorithm returns a normalized block-product input $\Xi_{\mathrm{in}}$ and $U_{\mathrm{out}}\in\operatorname{U}(m)$ such that, with probability at least $1-\beta$,
\begin{align}
    \sqrt{1-|\langle\widetilde\Psi,\Psi_{\mathrm{tr}}\rangle|^2}&\leq C_r(m^{5/2}\eta+\eta_0),\qquad \widetilde\Psi:=\hat U_{\mathrm{out}}\Xi_{\mathrm{in}}.
    \label{eq:iu-postprocessing-guarantee}
\end{align}
It uses $O_r(m^{r+3}\log(2m/\beta))$ arithmetic operations and no additional copies of the target state. The exact blocks, their branch counts, and their supports need not be known.
\end{proposition}

Even with exact RDMs, the reconstructed Fock-basis coefficients do not identify the branch modes. We therefore begin with exact branch recovery and then account for estimation errors.

\subsection{Recovering branches from an exact block}
\label{app:iu-exact-branches}

An exact recovered block still needs an explicit list of branch modes. We obtain them by an orthogonal pair decomposition for $p_b=2$ and by a random one-particle removal for $p_b\geq3$.

\subsubsection{Two-particle blocks}
Every two-particle vector admits an orthogonal pair decomposition, so the same construction applies to exact and estimated blocks.

\begin{proposition}[Orthogonal pair decomposition]\label{prop:iu-two-particle-decomposition}
For a normalized $x=\sum_{i<j}x_{ij}e_i\wedge e_j$, one can compute
\begin{align}
    x&=\sum_{\ell=1}^{s'}a_\ell u_\ell\wedge v_\ell, \qquad a_\ell>0,\qquad\sum_\ell a_\ell^2=1,
    \label{eq:iu-two-particle-decomposition}
\end{align}
where all $u_\ell,v_\ell$ are mutually orthonormal. The computation costs $O(m^3)$ arithmetic operations.
\end{proposition}
\begin{proof}
Form the antisymmetric matrix $M$ with $M_{ij}=x_{ij}$ for $i<j$ and $M^T=-M$. Choose a unit eigenvector $u$ of $MM^\dagger$ with positive eigenvalue $a^2$, and set $v:=-M\overline u/a$. Antisymmetry gives $u^\dagger v=0$. Also, $M^\dagger M=\overline{MM^\dagger}$ and $M\overline M=-MM^\dagger$, so
\begin{align}
    \|v\|&=1,\qquad M\overline u=-av,\qquad M\overline v=au.
    \label{eq:iu-pair-identities}
\end{align}
For $w\perp\operatorname{span}\{u,v\}$, antisymmetry and these identities show that $M\overline w$ is also perpendicular to $u,v$. Thus $a(uv^T-vu^T)$ can be removed, and the same construction applied on the orthogonal complement. It terminates with Eq.~\eqref{eq:iu-two-particle-decomposition}. The corresponding exterior vectors are orthonormal, so their squared coefficients sum to $\|x\|^2=1$. One spectral decomposition followed by pairing within its eigenspaces and orthogonal basis updates gives the stated cost. Repeated eigenvalues do not require a unique choice of pairs.
\end{proof}
Applying Prop.~\ref{prop:iu-two-particle-decomposition} to an exact
block $\omega_b$ gives $\omega_b=\sum_{\ell=1}^{s_b'}a_{b,\ell}\,u_{b,\ell}\wedge v_{b,\ell}$. Hence, setting $f_{b,\ell}':=u_{b,\ell}\wedge v_{b,\ell}$ and $\omega_{b,\ell}':=a_{b,\ell}$ gives an explicit branch representation of $\omega_b$. For an estimated block, additional small pairs can appear. We control their number using the truncation test described in the next subsection.

\subsubsection{Higher-particle blocks}
For $k\geq3$, a general estimated block need not admit a decomposition into branches with mutually orthogonal supports. We therefore first recover the branches of an exact block; the next subsection analyzes stability under estimation errors. Consider an exact block $\omega=\sum_{\ell=1}^s\omega_\ell f_\ell\in\wedge^kH$ for $k\ge 3$ where $f_\ell$ occupies the $k$-dimensional space $F_\ell$, and the
spaces $F_\ell$ are mutually orthogonal. By
Prop.~\ref{prop:correlated-block-rdm},
\begin{align}
    \Gamma_\omega^{(1)}&= \sum_{\ell=1}^s|\omega_\ell|^2P_{F_\ell}.
    \label{eq:iu-block-one-rdm}
\end{align}
Thus the branch spaces are directly visible when the weights are
distinct. If some weights coincide, however, the corresponding
spaces appear only through their direct sum. We break this degeneracy
by removing one particle in a random direction.

\begin{proposition}[One-removal branch recovery]
\label{lem:iu-one-removal-recovery}
Let $k\geq3$ and $\omega=\sum_{\ell=1}^s\omega_\ell f_\ell\in\wedge^kH$, where every coefficient is nonzero and each normalized $f_\ell$ occupies a $k$-dimensional space $F_\ell$. Suppose these spaces are mutually orthogonal. Draw a standard complex Gaussian vector $g\in H$ and set $h:=\hat c[g]\omega$. Then
\begin{align}
    \Gamma_h^{(1)}&= \sum_{\ell=1}^s |\omega_\ell|^2R_\ell^2 P_{E_\ell}
    \label{eq:iu-exact-removal-rdm}
\end{align}
where $R_\ell:=\|P_{F_\ell}g\|$ and $E_\ell:=F_\ell\cap g^\perp$. With probability one, the spaces $E_\ell$ are the distinct positive eigenspaces of $\Gamma_h^{(1)}$. Having identified these eigenspaces, complete each branch separately: if $e_\ell$ is a normalized vector spanning
$\wedge^{k-1}E_\ell$, then
\begin{align}
    v_\ell&:=\hat c[e_\ell]\omega \qquad\text{satisfies}\qquad F_\ell = E_\ell\oplus\operatorname{span}\{v_\ell\}.
    \label{eq:iu-exact-completion}
\end{align}
Hence a single random annihilation suffices to recover every branch
space $F_\ell$ from the exact block $\omega$. An orthonormal basis of each recovered space determines a normalized branch vector, whose coefficient is its inner product with $\omega$.
\end{proposition}

\begin{proof}
Annihilating $g$ removes from $f_\ell$ the occupied direction parallel
to $P_{F_\ell}g$. Thus, up to phase, the contribution of this branch
to $h$ is a normalized occupied state on $E_\ell$ with magnitude
$|\omega_\ell|R_\ell$. Since $k-1\geq2$ and the spaces $E_\ell$ are
mutually orthogonal, Prop.~\ref{prop:correlated-block-rdm} gives
Eq.~\eqref{eq:iu-exact-removal-rdm}. The variables $R_\ell$ are
independent and continuous, so the positive values $|\omega_\ell|R_\ell$ are pairwise distinct with probability one. For the second claim, $\hat c[e_\ell]$ annihilates every $f_j$ with
$j\neq\ell$, since $F_j\perp E_\ell$. On $f_\ell$, it removes the
$k-1$ occupied modes of $E_\ell$ and leaves the unique remaining
direction in $F_\ell\cap E_\ell^\perp$. This proves
Eq.~\eqref{eq:iu-exact-completion}.
\end{proof}

Prop.~\ref{lem:iu-one-removal-recovery} gives all branch spaces
$F_\ell$. For each $\ell$, choose a normalized occupied vector $f_\ell'\in\wedge^kF_\ell$ and set $\omega_\ell':=\langle f_\ell',\omega\rangle$. Since $\wedge^kF_\ell$ is one-dimensional, $f_\ell'$ agrees with
$f_\ell$ up to phase, and therefore
\begin{align}
    \omega&= \sum_{\ell=1}^s\omega_\ell'f_\ell'.
    \label{eq:iu-recovered-exact-branch-decomposition}
\end{align}
This gives the desired explicit branch representation. Applying the above construction to every exact block recovers all branch spaces $F_{b,\ell}$ and their coefficients. No further orthogonalization is needed in the exact setting: by the structure of the comparison state, the spaces $S_{\mathrm{core}}$ and all $F_{b,\ell}$ are already mutually orthogonal. Hence orthonormal bases of these spaces can be concatenated directly and completed to a
single-particle unitary.

\subsection{Recovering branches from an estimated block}
\label{app:iu-estimated-branches}

We now turn from exact branch recovery to the estimated blocks supplied by the reconstruction algorithm. Let $\omega=\sum_{\ell=1}^s\omega_\ell f_\ell\in\wedge^kH$ be a
normalized block, where $k\geq2$, every coefficient is nonzero,
and each normalized $f_\ell$ occupies a $k$-dimensional space $F_\ell$.
The spaces $F_\ell$ are mutually orthogonal. We are given $k$, a
normalized estimate $\widetilde\omega$, a bound $\eta>0$ satisfying
$d_2(\widetilde\omega,\omega)\leq\eta$, and a local failure budget
$\beta_{\mathrm{loc}}\in(0,1)$, but neither $s$ nor the exact branch
supports.

We seek an accurate approximation with explicit coefficients and
mutually orthonormal branch modes, using at most $s$ branches.
Since each branch requires $k$ modes, this restriction keeps the
number of modes no larger than in the exact block. It will allow
us to combine the recovered modes across blocks and the occupied
core. As in the exact construction, we treat $k=2$ and $k\geq3$
separately.

\subsubsection{Two-particle blocks}
\label{app:iu-estimated-two-particle}

\paragraph*{Step 1: Extract the mode pairs.}
Apply Prop.~\ref{prop:iu-two-particle-decomposition} to $\widetilde\omega$ and order its pairs by decreasing coefficient magnitude:
\begin{align}
    \widetilde\omega&=\sum_{\ell=1}^{t}a_\ell u_\ell\wedge v_\ell,\qquad a_1\geq\cdots\geq a_t>0,\qquad \sum_{\ell=1}^{t}a_\ell^2=1.
    \label{eq:iu-estimated-pair-decomposition}
\end{align}
All $2t$ modes are mutually orthonormal, and $t=\operatorname{rank}M/2\leq\lfloor m/2\rfloor$, where $M$ is the antisymmetric coefficient matrix of $\widetilde\omega$. The construction remains valid when coefficients coincide.

\paragraph*{Step 2: Truncate and normalize the block.}
The estimate may contain more than $s$ pairs, so we retain the
largest coefficients first. For $S=1,\ldots,t$, define
\begin{align}
    \widetilde\omega_S&:= \frac{\sum_{\ell=1}^{S}a_\ell u_\ell\wedge v_\ell} {\left(\sum_{\ell=1}^{S}a_\ell^2\right)^{1/2}}.
    \label{eq:iu-estimated-pair-partial-sum}
\end{align}
Return the first $\omega':=\widetilde\omega_S$ satisfying
$d_2(\widetilde\omega_S,\widetilde\omega)\leq2\eta$, together with
its mode pairs and normalized coefficients. This test can be
performed after each pair is extracted, without computing the
remaining pairs after acceptance.

We now show that the test succeeds for some $S\leq s$.
After aligning the phase of $\omega$, its coefficient matrix
$M_\omega$ has rank $2s$ and
$\|M-M_\omega\|_{\mathrm{HS}}^2
=2\|\widetilde\omega-\omega\|^2\leq2\eta^2$.
If $t>s$, the first $s$ pairs give a best rank-$2s$ approximation
to $M$ in Hilbert--Schmidt norm. Their omitted squared
Hilbert--Schmidt norm is therefore at most $2\eta^2$.
Each pair contributes two equal singular values, so
$R:=\sum_{\ell>s}a_\ell^2\leq\eta^2$. Consequently,
\begin{align}
    d_2(\widetilde\omega_s,\widetilde\omega)&=\sqrt{2-2\sqrt{1-R}} \leq\sqrt{2R}\leq\sqrt2\,\eta.
    \label{eq:iu-estimated-pair-tail}
\end{align}
If $t\leq s$, taking $S=t$ gives zero truncation error.
In either case the procedure returns at most $s$ pairs, and
$d_2(\omega',\omega)\leq2\eta+\eta=3\eta$.
The two-particle reconstruction is therefore deterministic.

\subsubsection{Higher-particle blocks}
\label{app:iu-estimated-higher-particle}

For $k\geq3$, we first find the $(k-1)$-dimensional occupied spaces
obtained after removing one particle from each branch. We then
recover the missing modes, make the branch modes mutually
orthogonal, and determine the coefficients. Since $s$ is unknown, try $S=1,\ldots,\lfloor m/k\rfloor$ in
increasing order. At each $S$, seek a candidate with at most $S$
branches that is sufficiently close to $\widetilde\omega$.
Set $\epsilon_{\mathrm s}:=4\sqrt k\,\eta$ and
$L:=\max\{1,\lceil\log_2(1/\beta_{\mathrm{loc}})\rceil\}$.
Use the singular-value cutoff
\begin{align}
    T_S&:=K_kS^2\eta,
    \label{eq:iu-count-dependent-thresholds}
\end{align}
and accept a candidate only within distance $H_kS^{5/2}\eta$ of $\widetilde\omega$, where $K_k,H_k>0$ are sufficiently large constants depending only on $k$. Perform the four steps below for at most $L$ independent
trials at each $S$. Stop at the first acceptance; after $L$
rejections increase $S$, and return \textsc{Fail} if all counts
are exhausted.

We describe each step for arbitrary $S$ and analyze it at $S=s$,
assuming $s^{5/2}\eta\leq c_k$ for a sufficiently small constant
$c_k>0$. For this analysis, align the phase of $\omega$ so that
$\Delta:=\widetilde\omega-\omega$ satisfies $\|\Delta\|\leq\eta$.
Constants denoted by $C_k$ depend only on $k$
and may increase between estimates.

\paragraph*{Step 1: Find the occupied spaces after removing one particle.}
Draw $g\in H$ with independent standard complex Gaussian components
and form $\widetilde h:=\hat c[g]\widetilde\omega$, without
normalizing it. We identify its occupied spaces using a coefficient
matrix: for $x\in\wedge^{k-1}H$, define
$(M_x)_{i,J}:=\langle e_i\wedge e_J,x\rangle$, where $i\in[m]$
and $J\in\mathcal I_{k-2}$. The RDM definition gives
\begin{align}
    M_xM_x^\dagger&=\Gamma_x^{(1)}.
    \label{eq:iu-coefficient-matrix}
\end{align}
Thus the left singular spaces of $M_x$ are the eigenspaces of
$\Gamma_x^{(1)}$, and its singular values are the square roots
of the RDM eigenvalues.

Compute an SVD of $M_{\widetilde h}$. In the exact construction,
each contracted branch contributes $k-1$ equal positive singular
values. To identify these groups with estimated data, sort the
singular values and group adjacent values whenever their gap is
at most $4\epsilon_{\mathrm s}$. Discard groups whose smallest
value is at most $T_S$. Reject the trial unless between one and
$S$ groups remain, each containing exactly $k-1$ values.
For each retained group $\ell$, let $\widetilde E_\ell$ be the
span of its left singular vectors.

To justify this selection at $S=s$, we bound the perturbation of
the singular values and spaces, together with the discarded branch
weight. For $x,y\in\wedge^{k-1}H$,
\begin{align}
    \|M_x-M_y\|&\leq\|x-y\|, \qquad \|\hat c[g]\Delta\|\leq\|g\|\|\Delta\|.
    \label{eq:iu-coefficient-matrix-perturbation}
\end{align}
These bounds follow from
$\|M_x^\dagger u\|=\|\hat c[u]x\|$ and
$\hat c[u]^\dagger\hat c[u]\leq\|u\|^2I$.
In particular, the matrix error is at most $\|\hat c[g]\Delta\|$.
Since $\mathbb E\|\hat c[g]\Delta\|^2=k\|\Delta\|^2$, Markov's
inequality bounds the probability that it exceeds
$\epsilon_{\mathrm s}$ by $1/16$.

For the exact contraction $h:=\hat c[g]\omega$, set
$R_\ell:=\|P_{F_\ell}g\|$ and $t_\ell:=|\omega_\ell|R_\ell$.
The occupied space $E_\ell$ of $\hat c[g]f_\ell$ has dimension
$k-1$, lies in $F_\ell$, and is a left singular space of $M_h$
with singular value $t_\ell$.
Set $\gamma:=(256ks^2)^{-1}$. We use the event
\begin{align}
    \|M_{\widetilde h}-M_h\|&\leq\epsilon_{\mathrm s},\qquad \sum_{\ell=1}^{s}R_\ell^{-2}\leq\frac{16s}{k-1},\qquad |t_\ell-t_j|\geq\gamma\max\{t_\ell,t_j\}\quad(\ell\neq j).
    \label{eq:iu-good-trial}
\end{align}
The second condition controls the contribution of small
contractions to the error bounds; the third separates different
branches. These conditions hold together with probability greater
than $1/2$. The first failure probability was bounded above.
Orthogonality of the $F_\ell$ makes the $R_\ell^2$ independent,
with density $x^{k-1}e^{-x}/(k-1)!$. Thus
$\mathbb E R_\ell^{-2}=1/(k-1)$, and Markov's inequality bounds
the second failure probability by $1/16$.
The density of $\log R_\ell$ is
$2\exp(2ky-e^{2y})/(k-1)!\leq2k$.
Independence gives the same density bound for $\log(t_\ell/t_j)$,
regardless of the coefficients. A violation of the third condition
implies $|\log(t_\ell/t_j)|<2\gamma$, with probability at most
$8k\gamma$. Summing over pairs gives at most $1/64$.

Condition on Eq.~\eqref{eq:iu-good-trial} throughout the remaining
analysis at $S=s$. Lem.~\ref{lem:weyl-inequality}, applied to the
associated Hermitian matrices, bounds each singular-value shift
by $\epsilon_{\mathrm s}$. Values from one branch therefore have
adjacent gaps at most $2\epsilon_{\mathrm s}$. Choose $K_k$ so that
$\gamma T_s$ is a sufficiently large multiple of
$\epsilon_{\mathrm s}$. Above the cutoff, values from different
branches then remain separated by more than $4\epsilon_{\mathrm s}$.
Consequently, every retained group comes from one branch and
contains exactly $k-1$ values. For the analysis, index these groups by their matching exact
branches. Let $J$ be the retained index set and put
$\delta_\ell:=\|P_{\widetilde E_\ell}-P_{E_\ell}\|$.
To bound $\delta_\ell$, apply Lem.~\ref{lem:projector-perturbation}
to the positive eigenvalue group at $t_\ell$ of
$\left(\begin{smallmatrix}0&M_h\\M_h^\dagger&0\end{smallmatrix}\right)$
and its estimated counterpart. This group has gap at least
$\gamma t_\ell$, while the matrix perturbation is at most
$\epsilon_{\mathrm s}$. Its spectral projector has upper-left
block $P_{E_\ell}/2$, so the lemma gives
\begin{align}
    \delta_\ell&\leq \frac{C_k\epsilon_{\mathrm s}}{\gamma t_\ell} \leq\frac{C_ks^2\eta}{|\omega_\ell|R_\ell}.
    \label{eq:iu-singular-space-errors}
\end{align}
Every omitted branch has $t_\ell\leq T_s+\epsilon_{\mathrm s}$.
Hence its total weight
$\rho^2:=\sum_{\ell\notin J}|\omega_\ell|^2$ is at most
$(T_s+\epsilon_{\mathrm s})^2\sum_\ell R_\ell^{-2}$.
Squaring Eq.~\eqref{eq:iu-singular-space-errors}, multiplying by
$|\omega_\ell|^2$, and summing similarly yields
\begin{align}
    \rho^2+\sum_{\ell\in J}|\omega_\ell|^2\delta_\ell^2&\leq C_ks^5\eta^2.
    \label{eq:iu-local-tail-and-space-errors}
\end{align}
Increasing $K_k$ makes the retained $\delta_\ell$ small, and
choosing $c_k$ sufficiently small gives $\rho^2<1$.
Thus at least one group remains, and the selected spaces and
discarded weight satisfy the bounds needed for mode completion.

\paragraph*{Step 2: Recover the remaining mode of each branch.}
Each selected space supplies $k-1$ modes. Choose an orthonormal
basis matrix
$\widetilde V_\ell=(\widetilde e_{\ell,1},\ldots,
\widetilde e_{\ell,k-1})$ of $\widetilde E_\ell$, and set
$\widetilde e_\ell:=\widetilde e_{\ell,1}\wedge\cdots\wedge
\widetilde e_{\ell,k-1}$.
Removing these modes from the original estimate gives the
one-particle vector $w_\ell:=\hat c[\widetilde e_\ell]\widetilde\omega$.
Discard the group if $w_\ell=0$; otherwise set
\begin{align}
    \alpha_\ell&:=\|w_\ell\|,\qquad V_\ell:=\bigl(\widetilde V_\ell,w_\ell/\alpha_\ell\bigr).
    \label{eq:iu-estimated-completed-modes}
\end{align}
Anticommutation makes the appended mode perpendicular to
$\widetilde E_\ell$, so each $V_\ell$ has orthonormal columns.
Reject the trial if no group remains.

We now compare $w_\ell$ with the exact missing mode, accounting for
the errors in both $\widetilde\omega$ and $\widetilde E_\ell$.
Lem.~\ref{lem:canonical-unitary-alignment} allows us to choose a
basis matrix $U_\ell$ of $E_\ell$ close to the computed basis:
$\|U_\ell-\widetilde V_\ell\|_{\mathrm{HS}}\leq C_k\delta_\ell$.
Their occupied vectors $e_\ell,\widetilde e_\ell$ then satisfy
$\|e_\ell-\widetilde e_\ell\|\leq C_k\delta_\ell$ by multilinearity
of the exterior product. The exact completion
$d_\ell:=\hat c[e_\ell]f_\ell$ is a unit vector in
$F_\ell\cap E_\ell^\perp$.
Put $A_\ell:=\hat c[\widetilde e_\ell]$. The completion error separates into the input estimation error, contributions from other branches, and the error in the annihilated modes:
\begin{align*}
    r_\ell&:=w_\ell-\omega_\ell d_\ell =A_\ell\Delta +\sum_{j\neq\ell}\omega_j A_\ell f_j +\omega_\ell(A_\ell f_\ell-d_\ell).
\end{align*}
We bound the summed squared norms of these three terms in order.
For the first term, bounding each completion separately would count the same estimation error once per branch. Let $\Pi_\ell:=A_\ell^\dagger A_\ell$, which projects onto occupation of all $k-1$ modes of $\widetilde E_\ell$. Since the selected spaces are orthogonal and $2(k-1)>k$, a $k$-particle state cannot occupy two such spaces fully. Thus
\begin{align}
    \sum_{\ell\in J}\Pi_\ell&\leq I,\qquad \sum_{\ell\in J}\|A_\ell\Delta\|^2\leq\eta^2.
    \label{eq:iu-exclusive-completion-events}
\end{align}
To bound the second term, let $N_{\widetilde E_j}$ count particles in $\widetilde E_j$ for a retained branch $j$.
Occupying all modes of another $\widetilde E_\ell$ leaves at most
one particle in $\widetilde E_j$, so the commuting occupation
operators satisfy
$\sum_{\ell\in J,\,\ell\neq j}\Pi_\ell
\leq[(k-1)I-N_{\widetilde E_j}]/(k-2)$.
Taking the expectation in $f_j$ and using
$\Gamma_{f_j}^{(1)}=P_{F_j}$ gives
\begin{align}
    \sum_{\substack{\ell\in J\\\ell\neq j}}\|A_\ell f_j\|^2&\leq\frac{k-1-\operatorname{Tr} (P_{\widetilde E_j}P_{F_j})}{k-2} \leq C_k\delta_j^2.
    \label{eq:iu-completion-cross-sum}
\end{align}
Here $E_j\subseteq F_j$ and
$k-1-\operatorname{Tr}(P_{\widetilde E_j}P_{E_j})
\leq(k-1)\delta_j^2$ yield the last inequality.
For an omitted branch $j$,
Eq.~\eqref{eq:iu-exclusive-completion-events} instead bounds
$\sum_{\ell\in J}\|A_\ell f_j\|^2$ by one. For fixed $\ell$, the vectors $A_\ell f_j$ lie in mutually
orthogonal $F_j$. The squared norm of the middle term is therefore
$\sum_{j\neq\ell}|\omega_j|^2\|A_\ell f_j\|^2$.

For the third term, Prop.~\ref{prop:occupied-subspace-rdm} gives
$\Gamma_{f_\ell}^{(k-1)}=P_{\wedge^{k-1}F_\ell}$ and therefore
bounds the error caused by replacing the annihilated modes:
\begin{align}
    \|A_\ell f_\ell-d_\ell\|&\leq\|\widetilde e_\ell-e_\ell\| \leq C_k\delta_\ell.
    \label{eq:iu-direct-completion-comparison}
\end{align}
Combining these bounds with Eq.~\eqref{eq:iu-local-tail-and-space-errors} yields
\begin{align}
    \sum_{\ell\in J}\|r_\ell\|^2&\leq C_k\left(\eta^2+\rho^2+ \sum_{\ell\in J}|\omega_\ell|^2\delta_\ell^2\right) \leq C_ks^5\eta^2.
    \label{eq:iu-aggregate-completion-error}
\end{align}
For $w_\ell\neq0$, normalization and the reverse triangle
inequality give
\begin{align}
    |\omega_\ell| \left\|\frac{w_\ell}{\alpha_\ell} -\frac{\omega_\ell}{|\omega_\ell|}d_\ell\right\|&\leq2\|r_\ell\|,\qquad \bigl|\alpha_\ell-|\omega_\ell|\bigr|\leq\|r_\ell\|.
    \label{eq:iu-weighted-completion-direction}
\end{align}
The factor $|\omega_\ell|$ keeps the bound useful even for small
branch coefficients. If $w_\ell=0$, its discarded weight is
$|\omega_\ell|^2=\|r_\ell\|^2$. For the surviving indices $J_+:=\{\ell\in J:w_\ell\neq0\}$,
compare $V_\ell$ with
$W_\ell:=(U_\ell,(\omega_\ell/|\omega_\ell|)d_\ell)$.
These exact matrices have mutually orthonormal columns because
they lie in the orthogonal spaces $F_\ell$.
The basis and completion bounds give
\begin{align}
    \sum_{\ell\in J_+}|\omega_\ell|^2 \|V_\ell-W_\ell\|_{\mathrm{HS}}^2 +\sum_{\ell\in J_+}(\alpha_\ell-|\omega_\ell|)^2&\leq C_ks^5\eta^2.
    \label{eq:iu-completed-space-error}
\end{align}
The total discarded weight
$\rho_*^2:=\sum_{\ell\notin J_+}|\omega_\ell|^2$ obeys the same
bound, including zero completions. For sufficiently small $c_k$,
$\rho_*^2<1$, so this step leaves at least one completed branch.

\paragraph*{Step 3: Make the branch modes mutually orthogonal.}
The columns within each $V_\ell$ are orthonormal, while those from different branches need not be. Since weak branches may have poorly determined modes, we weight each change by its branch amplitude, as in Eq.~\eqref{eq:iu-completed-space-error}. Concatenate the surviving matrices
into $V$ in their stored order, giving $k|J_+|\leq kS\leq m$ columns. Let $D_\alpha$ repeat $\alpha_\ell$ on the
$k$ columns of branch $\ell$.
We seek $Q$ with orthonormal columns minimizing
$\|(Q-V)D_\alpha\|_{\mathrm{HS}}$, so changes to each mode of
branch $\ell$ have weight $\alpha_\ell^2$.
The following lemma gives the SVD construction and bounds its error
relative to any orthonormal comparison. It will also be used when
combining modes across blocks and the core.

\begin{lemma}[Weighted adjustment of mode columns]
\label{lem:iu-weighted-orthogonalization}
Let $A\in\mathbb C^{m\times t}$, $t\leq m$, and let $D$ be
nonnegative and diagonal. An SVD $AD^2=U\Sigma Z^\dagger$ gives a
minimizer $Q=UZ^\dagger$ of $\|(Q-A)D\|_{\mathrm{HS}}$ over
$Q^\dagger Q=I_t$, completing zero singular directions orthonormally.
For any $W^\dagger W=I_t$,
\begin{align}
    \|(Q-W)D\|_{\mathrm{HS}}&\leq2\|(A-W)D\|_{\mathrm{HS}}.
    \label{eq:iu-weighted-orthogonalization}
\end{align}
\end{lemma}
\begin{proof}
The squared objective depends on $Q$ only through
$-2\operatorname{Re}\operatorname{Tr}(Q^\dagger AD^2)$.
The SVD choice maximizes this trace at $\operatorname{Tr}\Sigma$.
Hence $\|(Q-A)D\|_{\mathrm{HS}}\leq\|(W-A)D\|_{\mathrm{HS}}$,
and the triangle inequality gives the stated bound.
\end{proof}

Apply Lem.~\ref{lem:iu-weighted-orthogonalization} with $A=V$ and $D=D_\alpha$, keeping the columns of the resulting $Q$ grouped by branch. For the error bound, concatenate the exact $W_\ell$ in the same
order to obtain $W$. The lemma gives
$\|(Q-W)D_\alpha\|_{\mathrm{HS}}
\leq2\|(V-W)D_\alpha\|_{\mathrm{HS}}$.
To use this bound for the state, replace the computed weights
$\alpha_\ell$ by the exact amplitudes $|\omega_\ell|$.
Let $D_\omega$ repeat $|\omega_\ell|$ on each branch's $k$ columns.
For matrices $X,Y$ with unit columns,
\begin{align*}
    \|(X-Y)(D_\alpha-D_\omega)\|_{\mathrm{HS}}&\leq2\sqrt{k\sum_{\ell\in J_+} (\alpha_\ell-|\omega_\ell|)^2}.
\end{align*}
Apply this inequality to $(Q,W)$ and $(V,W)$ in the lemma's bound.
Eq.~\eqref{eq:iu-completed-space-error} then gives
\begin{align}
    \|(Q-W)D_\omega\|_{\mathrm{HS}}&\leq C_ks^{5/2}\eta.
    \label{eq:iu-local-orthogonal-mode-error}
\end{align}
Thus making the modes mutually orthonormal preserves the required
amplitude-weighted error bound.

\paragraph*{Step 4: Compute the coefficients and check the approximation.}
For each branch $\ell$, write $q_{\ell,1},\ldots,q_{\ell,k}$ for
its columns in $Q$ and let
$q_\ell:=q_{\ell,1}\wedge\cdots\wedge q_{\ell,k}$.
These branch states are orthonormal.
Set $\mathcal L:=\operatorname{span}\{q_\ell\}$ and compute
$a_\ell:=\langle q_\ell,\widetilde\omega\rangle$.
Reject the trial if all overlaps vanish; otherwise form
\begin{align*}
    \omega'&:= \frac{\sum_\ell a_\ell q_\ell} {\left(\sum_\ell|a_\ell|^2\right)^{1/2}} =\frac{P_{\mathcal L}\widetilde\omega} {\|P_{\mathcal L}\widetilde\omega\|}.
\end{align*}
Accept only if
\begin{align}
    d_2(\omega',\widetilde\omega)&\leq H_kS^{5/2}\eta.
    \label{eq:iu-local-residual-test}
\end{align}
Because $\omega'$ is the normalized projection of $\widetilde\omega$, their overlap is positive real. Hence $d_2(\omega',\widetilde\omega)=\|\omega'-\widetilde\omega\|$, so the acceptance test uses the ordinary vector norm.
On acceptance, return the modes and normalized coefficients
$a_\ell/(\sum_j|a_j|^2)^{1/2}$, removing zero coefficients and
their modes. Otherwise, reject the trial.

We show that this test passes at $S=s$ by converting the mode
error in Eq.~\eqref{eq:iu-local-orthogonal-mode-error} into a
block-state error. The following lemma bounds the change in a
state when its mode vectors change but its coefficients stay fixed.

\begin{lemma}[Effect of changing the modes]
\label{lem:iu-fixed-particle-mode-error}
Let $V,W:\mathbb C^t\to H$ have orthonormal columns and let
$\chi\in\wedge^k\mathbb C^t$ be normalized. Then
\begin{align}
    \|(\wedge^kV)\chi-(\wedge^kW)\chi\|&\leq\sqrt k\, \|(V-W)(\Gamma_\chi^{(1)})^{1/2}\|_{\mathrm{HS}}.
    \label{eq:iu-fixed-particle-mode-error}
\end{align}
\end{lemma}
\begin{proof}
On the normalized antisymmetric tensor space, replace the $k$
copies of $W$ by $V$ one at a time. Each difference has squared norm
$k^{-1}\operatorname{Tr}[(V-W)^\dagger(V-W)\Gamma_\chi^{(1)}]$:
the other tensor factors preserve norms, and the reduced state of
one tensor factor is $\Gamma_\chi^{(1)}/k$.
Summing the $k$ norms proves the claim.
\end{proof}

Normalize the surviving part of the exact block:
\begin{align*}
    \omega_*&:= \frac{\sum_{\ell\in J_+}\omega_\ell f_\ell} {\sqrt{1-\rho_*^2}}.
\end{align*}
Its distance from $\omega$ is
$\sqrt{2-2\sqrt{1-\rho_*^2}}\leq\sqrt2\,\rho_*
\leq C_ks^{5/2}\eta$.
Since $e_\ell\wedge d_\ell=f_\ell$, occupying the columns of
$W_\ell$ gives $(\omega_\ell/|\omega_\ell|)f_\ell$.
Let $\chi$ be the normalized superposition of the corresponding
occupation states on the column labels, with coefficients
$|\omega_\ell|/\sqrt{1-\rho_*^2}$.
Then $(\wedge^kW)\chi=\omega_*$ and
$(\Gamma_\chi^{(1)})^{1/2}=D_\omega/\sqrt{1-\rho_*^2}$,
because the branches use disjoint sets of $k$ column labels. Apply Lem.~\ref{lem:iu-fixed-particle-mode-error} with the computed
columns $Q$ and exact columns $W$. Together with
Eq.~\eqref{eq:iu-local-orthogonal-mode-error}, it gives
\begin{align*}
    \|(\wedge^kQ)\chi-\omega_*\|&\leq\frac{\sqrt k}{\sqrt{1-\rho_*^2}} \|(Q-W)D_\omega\|_{\mathrm{HS}} \leq C_ks^{5/2}\eta.
\end{align*}
Thus $\mathcal L$ contains the unit vector $(\wedge^kQ)\chi$
within $C_ks^{5/2}\eta$ of $\omega$, and hence within
$C_ks^{5/2}\eta+\eta$ of $\widetilde\omega$. For every unit vector $z\in\mathcal L$,
$|\langle z,\widetilde\omega\rangle|
\leq\|P_{\mathcal L}\widetilde\omega\|$,
with equality for the normalized projection $\omega'$.
It therefore minimizes $d_2(z,\widetilde\omega)$ over these vectors,
and
\begin{align*}
    d_2(\omega',\widetilde\omega)&\leq C_ks^{5/2}\eta+\eta \leq H_ks^{5/2}\eta.
\end{align*}
Choose $K_k$ first to separate the groups, then $H_k$ to satisfy
this inequality, and finally $c_k$ small enough that the surviving
weight and projection are nonzero. All four steps then pass on
the event in Eq.~\eqref{eq:iu-good-trial}.

A trial at $S=s$ is consequently accepted with probability at least
$1/2$. Conditional on reaching this count, all $L$ independent
trials fail with probability at most $2^{-L}\leq\beta_{\mathrm{loc}}$.
An earlier acceptance at $S\leq s$ also uses at most $s$ branches
and, by Eq.~\eqref{eq:iu-local-residual-test}, satisfies
$d_2(\omega',\omega)\leq H_kS^{5/2}\eta+\eta$.
Therefore, under $s^{5/2}\eta\leq c_k$, the procedure returns,
with probability at least $1-\beta_{\mathrm{loc}}$, a normalized
block with explicit coefficients, mutually orthonormal branch
modes, at most $s$ branches, and
\begin{align}
    d_2(\omega',\omega)&\leq C_ks^{5/2}\eta.
    \label{eq:iu-local-recovery-bound}
\end{align}
On the exceptional event it may return \textsc{Fail} or accept a
larger count. For fixed $k$, each trial and the number of tested
counts have polynomial cost in $m$, while
$L=O(1+\log(1/\beta_{\mathrm{loc}}))$.
Thus the arithmetic cost is polynomial in $m$ and
$\log(1/\beta_{\mathrm{loc}})$.

\subsection{Constructing the input and the unitary}
\label{app:iu-input-and-unitary}
\label{app:iu-orthogonal-modes}

The recovered modes are orthonormal within each block. To combine them with the estimated core into a single unitary, we now make all mode columns mutually orthonormal while keeping the branch coefficients fixed. The local branch-count bounds ensure that these columns fit within the available $m$ modes.

For each recovered block, let $\omega'_b$ be the normalized local
output, with particle number $2\leq p_b\leq r$, branch count $s'_b$,
and normalized coefficients $a_{b,\ell}$. We first work on the
local success event, where $s'_b\leq s_b$ and
$d_2(\omega'_b,\omega_b)\leq e_b$ for every $b\in[d]$.
Assume also that the estimated core has the correct dimension and
$d_2(\widetilde\sigma_{\mathrm{core}},\sigma_{\mathrm{core}})
\leq\eta_0\leq1/2$.
Write $B_b:=\bigoplus_\ell F_{b,\ell}$ for the exact support of $\omega_b$, of dimension $p_bs_b$. The spaces $S_{\mathrm{core}},B_1,\ldots,B_d$ are mutually orthogonal by the comparison-state structure in Eq.~\eqref{eq:iu-comparison-state}. The errors $e_b$ refer to the outputs after local branch recovery. Constants denoted by
$C_r$ below depend only on $r$ and may increase between estimates.

\paragraph*{Step 1: Collect the core and branch modes.}
Let $V_0$ be the supplied orthonormal basis matrix of the estimated
core, with $q$ columns. For each block $b$, collect its recovered
modes in $V_b$, ordered by branch, and set
$D_b:=\operatorname{diag}(|a_{b,1}|I_{p_b},\ldots,
|a_{b,s'_b}|I_{p_b})$.
Thus $V_b$ has $p_bs'_b$ orthonormal columns and, by
Prop.~\ref{prop:correlated-block-rdm},
$\Gamma_{\omega'_b}^{(1)}=V_bD_b^2V_b^\dagger$.
The adjustment should reflect each mode's contribution to the state: every core mode is occupied, whereas a mode of branch $(b,\ell)$ has occupation $|a_{b,\ell}|^2$. We therefore collect the columns and their amplitude weights as
\begin{align}
    V&:=[V_0,V_1,\ldots,V_d],\qquad D:=\operatorname{diag}(I_q,D_1,\ldots,D_d).
    \label{eq:iu-common-mode-matrices}
\end{align}
Let $t:=q+\sum_b p_bs'_b$ be the number of columns of $V$.
Return \textsc{Fail} if $t>m$. On the local success event,
$p_bs'_b\leq\dim B_b$, so orthogonality of the exact supports gives
$t\leq q+\sum_b\dim B_b\leq m$.
Thus the branch-count bounds ensure that the recovered modes can
be made mutually orthonormal without discarding further branches.

\paragraph*{Step 2: Make all modes mutually orthogonal.}
Apply Lem.~\ref{lem:iu-weighted-orthogonalization} with $A=V$ and the weights $D$ to obtain $Q^\dagger Q=I_t$ minimizing $\|(Q-V)D\|_{\mathrm{HS}}$. Keep the columns in their original core and branch groups and leave all coefficients unchanged.

To bound this adjustment, we compare each $V_b$ with orthonormal
columns inside the exact support $B_b$. The state error first
controls the weighted component outside that support.
Let $N_{B_b^\perp}$ count particles outside $B_b$.
It annihilates $\omega_b$ and satisfies
$0\leq N_{B_b^\perp}\leq p_bI$ on the $p_b$-particle sector.
After aligning the phase of $\omega_b$,
\begin{align}
    \|(I-P_{B_b})V_bD_b\|_{\mathrm{HS}}^2
    =\operatorname{Tr}(P_{B_b^\perp}\Gamma_{\omega'_b}^{(1)})
    =\langle\omega'_b-\omega_b,
        N_{B_b^\perp}(\omega'_b-\omega_b)\rangle
      \leq p_be_b^2.
    \label{eq:iu-quadratic-occupation-bound}
\end{align}
The weights are important: a mode from a branch with a small
coefficient need not itself be close to $B_b$. The following lemma
uses the weighted support bound to construct orthonormal comparison
columns inside $B_b$, even when some projected directions vanish.

\begin{lemma}[Projection of orthonormal mode columns]
\label{lem:iu-weighted-projection}
Let $V\in\mathbb C^{m\times t}$ satisfy $V^\dagger V=I_t$, and let
$P$ be an orthogonal projector of rank at least $t$.
There is a matrix $W$ with orthonormal columns in
$\operatorname{ran}P$ such that, for every nonnegative diagonal $D$,
\begin{align}
    V^\dagger W&=(V^\dagger PV)^{1/2},\qquad \|(W-V)D\|_{\mathrm{HS}}^2 \leq2\|(I-P)VD\|_{\mathrm{HS}}^2.
    \label{eq:iu-projected-mode-columns}
\end{align}
The columns of $PV$ need not be linearly independent.
\end{lemma}
\begin{proof}
Take the polar factor of $PV$ on its nonzero singular directions
and complete it to an isometry inside $\operatorname{ran}P$.
The rank assumption permits choosing the added directions
orthogonal to $\operatorname{ran}(PV)$.
Each such direction $w$ satisfies
$V^\dagger w=(PV)^\dagger w=0$, so the completed matrix $W$
satisfies the first identity. With $A:=V^\dagger PV$, we have
$0\leq A\leq I$ and
\begin{align*}
    \|(W-V)D\|_{\mathrm{HS}}^2&=2\operatorname{Tr}[D^2(I-A^{1/2})] \leq2\operatorname{Tr}[D^2(I-A)] =2\|(I-P)VD\|_{\mathrm{HS}}^2.
\end{align*}
The inequality follows from $1-\sqrt x\leq1-x$ on $[0,1]$.
\end{proof}

Apply the lemma with $V=V_b$, $P=P_{B_b}$, and $D=D_b$.
Its rank condition holds because $p_bs'_b\leq\dim B_b$.
Together with Eq.~\eqref{eq:iu-quadratic-occupation-bound}, it gives
orthonormal columns $W_b$ in $B_b$ satisfying
$\|(W_b-V_b)D_b\|_{\mathrm{HS}}^2\leq2p_be_b^2$.

For the core, choose an orthonormal basis $W_0$ of
$S_{\mathrm{core}}$ aligned with $V_0$ by polar decomposition.
If the principal angles are $\theta_j$, the occupied-state overlap
and the core error imply
$\prod_j\cos\theta_j\geq1-\eta_0^2/2$. Hence
\begin{align}
    \|W_0-V_0\|_{\mathrm{HS}}^2&=2\sum_j(1-\cos\theta_j) \leq-2\log(1-\eta_0^2/2)\leq C\eta_0^2.
    \label{eq:iu-core-column-error}
\end{align}
Use empty matrices when $q=0$.
The columns of $W:=[W_0,W_1,\ldots,W_d]$ are mutually orthonormal,
since their exact supports are orthogonal.
We can therefore apply Lem.~\ref{lem:iu-weighted-orthogonalization}
to this comparison matrix and the computed $Q$:
\begin{align}
    \|(Q-W)D\|_{\mathrm{HS}}&\leq2\|(V-W)D\|_{\mathrm{HS}} \leq C_r\left(\sum_{b=1}^d e_b^2+\eta_0^2\right)^{1/2}.
    \label{eq:iu-common-column-error}
\end{align}
Thus the common adjustment controls the total weighted mode error.

\paragraph*{Step 3: Construct the input and the unitary.}
Complete $Q$ to an orthonormal basis of $H$ and set
$U_{\mathrm{out}}:=[Q,Q_\perp]\in\operatorname{U}(m)$.
Assign disjoint input index sets $I_{0,\mathrm{out}}$ to its $q$
core columns and $I_{b,\ell,\mathrm{out}}$ to the $p_b$ columns
of each branch, preserving the stored order.
Using the recovered coefficients, define
\begin{align}
    \Xi_{\mathrm{in}}&:=\hat c_{I_{0,\mathrm{out}}}^\dagger\prod_{b=1}^d\left(\sum_{\ell=1}^{s'_b}a_{b,\ell}\hat c_{I_{b,\ell,\mathrm{out}}}^\dagger\right)\ket{\vac},\qquad \widetilde\Psi:=\hat U_{\mathrm{out}}\Xi_{\mathrm{in}}.
    \label{eq:iu-explicit-input}
\end{align}
The modes assigned to $Q_\perp$ remain in the vacuum.
Here $\hat U_{\mathrm{out}}$ is the passive Gaussian unitary
associated with $U_{\mathrm{out}}$, with its vacuum phase fixed.
It maps each input mode to the corresponding adjusted column.
The input and output are normalized because the input mode sets
are disjoint, the coefficient lists are normalized, and
$\hat U_{\mathrm{out}}$ is unitary.

The returned input and unitary specify a preparation, not the original
input labels or the original unitary. A change of basis within a
branch contributes a determinant phase that is absorbed into its
coefficient; a change of core basis or fixed factor order contributes
only an overall phase. A one-branch factor may be absorbed into the
occupied core. The returned input may contain occupied and vacuum
modes in both target settings; it need not reproduce the input
promised in the homogeneous setting.

It remains to bound the error of this prepared state.
Let $N:=q+\sum_b p_b$, and let $\chi\in\wedge^N\mathbb C^t$
be the input on the first $t$ column labels, before adding the
vacuum modes. Then $\Gamma_\chi^{(1)}=D^2$ and
\begin{align*}
    \widetilde\Psi&=(\wedge^NQ)\chi,\qquad \Psi_W:=(\wedge^NW)\chi.
\end{align*}
We first compare $\widetilde\Psi$ with $\Psi_W$, which uses the same
coefficients on the exact orthogonal supports, and then compare
$\Psi_W$ with $\Psi_{\mathrm{tr}}$. Applying Lem.~\ref{lem:iu-fixed-particle-mode-error} directly to
$\chi$ would introduce a factor $\sqrt N$.
Its block-product structure instead gives the following bound, with a prefactor depending only on the largest block particle number $r$.
\begin{lemma}[Effect of changing modes in a block product]
\label{lem:iu-block-product-mode-error}
Let $V,W:\mathbb C^t\to\mathbb C^m$ have orthonormal columns,
and let $\chi\in\wedge^N\mathbb C^t$ be a normalized product of fixed-particle-number states on mutually orthogonal one-particle subspaces, each with at most $r$ particles, together with an arbitrary occupied core. Then
\begin{align}
    d_2\bigl((\wedge^NV)\chi,(\wedge^NW)\chi\bigr)&\leq\frac{\pi\sqrt r}{2} \|(V-W)(\Gamma_\chi^{(1)})^{1/2}\|_{\mathrm{HS}}.
    \label{eq:iu-block-product-mode-error}
\end{align}
\end{lemma}

To prove this bound, we control the variance of the one-body generator of a mode rotation. For a Hermitian one-particle matrix $h$, write $\widehat h:=\sum_{i,j}h_{ij}\hat c_i^\dagger\hat c_j$. The following estimate depends only on the block particle numbers.

\begin{lemma}[Variance of a one-body operator]
\label{lem:iu-one-body-variance}
Let $\psi$ be a normalized product of fixed-particle-number states
on mutually orthogonal one-particle subspaces, each with at most
$r$ particles, together with an arbitrary occupied core.
For any Hermitian one-particle matrix $h$,
\begin{align}
    \langle\psi,\widehat h^2\psi\rangle -\langle\psi,\widehat h\psi\rangle^2&\leq r\operatorname{Tr}(h^2\Gamma_\psi^{(1)}).
    \label{eq:iu-one-body-variance}
\end{align}
The individual blocks need not have a branch decomposition.
\end{lemma}
\begin{proof}
Split the core into occupied one-mode factors and include the
unoccupied complement as another subspace.
Let $P_a$ be the resulting one-particle projectors and put
$h_{ab}:=P_ahP_b$.
The terms from $h_{aa}$ preserve all block particle numbers.
Their centered contributions have zero cross expectations because
the state is a product on orthogonal supports.
For a $p_a$-particle block, Cauchy--Schwarz applied to the sum of
its $p_a$ single-particle actions bounds the variance by
$p_a\operatorname{Tr}(h_{aa}^2\Gamma_a^{(1)})$,
where $\Gamma_a^{(1)}$ is its $1$-RDM.

In a basis adapted to these subspaces, the term
$T_{ab}:=\sum_{i\in a,j\in b}h_{ij}\hat c_i^\dagger\hat c_j$,
$a\neq b$, transfers one particle from $b$ to $a$.
Distinct ordered pairs give orthogonal particle-number sectors,
also orthogonal to the diagonal contributions.
Anticommutation and product-state factorization give
\begin{align*}
    \|T_{ab}\psi\|^2&=\operatorname{Tr}\!\left[ h_{ab}^\dagger(I-\Gamma_a^{(1)})h_{ab}\Gamma_b^{(1)} \right] \leq\operatorname{Tr}\!\left[ h_{ab}^\dagger h_{ab}\Gamma_b^{(1)}\right].
\end{align*}
Here $0\leq\Gamma_a^{(1)}\leq I$ by
Lem.~\ref{lem:df-one-rdm-lipschitz}; the unoccupied factor has zero
RDM. Sum these squared norms and the diagonal variance bounds.
Since $p_a\leq r$ and $\Gamma_\psi^{(1)}$ is block diagonal,
the result is Eq.~\eqref{eq:iu-one-body-variance}.
\end{proof}

\begin{proof}[Proof of Lem.~\ref{lem:iu-block-product-mode-error}]
Complete the two column lists to orthonormal bases and choose a
unitary $R$ with $RW=V$. Write $R=e^{ih}$ with $h$ Hermitian and
spectrum in $[-\pi,\pi]$.
Starting from $\psi:=(\wedge^NW)\chi$, consider
$e^{-i\tau\langle\psi,\widehat h\psi\rangle}
(\wedge^Ne^{i\tau h})\psi$, $0\leq\tau\leq1$.
The derivative has constant norm equal to the square root of the
variance of $\widehat h$ in $\psi$.
The path length bounds the phase-aligned distance between its
endpoints, so Lem.~\ref{lem:iu-one-body-variance} gives
$\sqrt{r\operatorname{Tr}(h^2W\Gamma_\chi^{(1)}W^\dagger)}$
as an upper bound.

For $|x|\leq\pi$, $x^2\leq(\pi^2/4)|e^{ix}-1|^2$.
Apply this inequality to the spectral decomposition of $h$ and
take the trace against $W\Gamma_\chi^{(1)}W^\dagger\geq0$.
Since $(R-I)W=V-W$, the result is
Eq.~\eqref{eq:iu-block-product-mode-error}.
No small operator-norm distance between $V$ and $W$ is required.
\end{proof}

Apply Lem.~\ref{lem:iu-block-product-mode-error} with $V=Q$, the comparison columns $W$, and
$\Gamma_\chi^{(1)}=D^2$.
Together with Eq.~\eqref{eq:iu-common-column-error},
this bounds $d_2(\widetilde\Psi,\Psi_W)$ by
$C_r(\sum_b e_b^2+\eta_0^2)^{1/2}$.

For the second comparison, let $\chi_b$ be the normalized local
input of block $b$, so that $\Gamma_{\chi_b}^{(1)}=D_b^2$ and
$(\wedge^{p_b}V_b)\chi_b=\omega'_b$.
The corresponding block in $\Psi_W$ is
$\omega_{W,b}:=(\wedge^{p_b}W_b)\chi_b$.
Apply Lem.~\ref{lem:iu-fixed-particle-mode-error} to this individual
block and add its local reconstruction error:
\begin{align*}
    d_2(\omega_{W,b},\omega_b)&\leq\sqrt{p_b}\,\|(W_b-V_b)D_b\|_{\mathrm{HS}}+e_b \leq C_re_b.
\end{align*}
The reference core agrees with $\sigma_{\mathrm{core}}$ up to
phase. Both reference and exact blocks lie in the same mutually
orthogonal spaces $B_b$, so their squared overlaps multiply. Each block has infidelity at most $C_re_b^2$, giving
\begin{align*}
    1-|\langle\Psi_W,\Psi_{\mathrm{tr}}\rangle|^2&=1-\prod_{b=1}^d|\langle\omega_{W,b},\omega_b\rangle|^2 \leq C_r\sum_{b=1}^d e_b^2.
\end{align*}
For normalized pure states, the trace distance is $\sqrt{1-|\langle x,y\rangle|^2}\leq d_2(x,y)$. Its triangle inequality therefore gives the error of the prepared state:
\begin{align}
    \sqrt{1-|\langle\widetilde\Psi,\Psi_{\mathrm{tr}}\rangle|^2}&\leq C_r\left(\sum_{b=1}^d e_b^2+\eta_0^2\right)^{1/2}.
    \label{eq:iu-common-state-error}
\end{align}

\begin{proof}[Proof of Prop.~\ref{prop:iu-postprocessing}]
We now combine the local recovery bounds with the common mode adjustment. Assign each local search in Sec.~\ref{app:iu-estimated-branches} failure budget $\beta/\max\{1,d\}$. Since $s_b\leq m$, the smallness condition in Eq.~\eqref{eq:iu-supplied-errors} ensures
$s_b^{5/2}\eta\leq c_{p_b}$ for every higher-particle block,
as well as $\eta_0\leq1/2$. The two-particle bound and Eq.~\eqref{eq:iu-local-recovery-bound}
give, simultaneously with probability at least $1-\beta$,
$s'_b\leq s_b$ and valid error bounds
$e_b=3\eta$ for $p_b=2$ and
$e_b=C_{p_b}s_b^{5/2}\eta$ for $p_b\geq3$.
All subsequent operations are deterministic. Substituting these
bounds into Eq.~\eqref{eq:iu-common-state-error} yields
\begin{align}
    \sqrt{1-|\langle\widetilde\Psi,\Psi_{\mathrm{tr}}\rangle|^2}&\leq C_r\left[\left(\sum_{b:p_b\geq3}s_b^5\eta^2+\sum_{b:p_b=2}\eta^2\right)^{1/2}+\eta_0\right]\leq C_r(m^{5/2}\eta+\eta_0).
    \label{eq:iu-postprocessing-error}
\end{align}
Indeed, orthogonality of the exact supports and $p_b\geq2$ imply
$\sum_b s_b\leq m/2$, so
$\sum_b s_b^5\leq(\sum_b s_b)^5\leq(m/2)^5$.
If all blocks have two particles, the bound is instead
$C(\sqrt d\,\eta+\eta_0)$, and the weaker condition
$\sqrt d\,\eta+\eta_0\leq c_2$ suffices.
For $d=0$, only the core and unitary completion are needed, and
the error is at most $\eta_0$.

For a degree-$k$ block with $k\geq3$, one local trial costs
$O_k(m^{k+1})$ arithmetic operations, including contraction,
the SVD, branch completion, coefficient evaluation, and weighted
orthogonalization. There are at most $m/k$ tested counts per
block, $O(\log(2m/\beta))$ trials per count, and at most $m/2$
blocks. The total preparation cost is therefore bounded by
\begin{align}
    O_r\!\left(m^{r+3}\log\frac{2m}{\beta}\right).
    \label{eq:iu-postprocessing-cost}
\end{align}
The two-particle decomposition and truncation, the common weighted
SVD, and the unitary completion also fit this bound.
\end{proof}

\subsection{Application to the learning algorithms}
\label{app:iu-learning-guarantee}

We now complete the proofs of Thms.~\ref{thm:equal-blocks} and~\ref{thm:heterogeneous-blocks} by combining the reconstruction and preparation guarantees with RDM estimation and the resource bounds.

\paragraph*{RDM accuracy.}
We next verify that the RDM accuracies used for reconstruction also
satisfy Prop.~\ref{prop:iu-postprocessing}. Keep
$s_0^2=\varepsilon_{\mathrm{fid}}/(48m)$ and the exponents
$\mathfrak a_p^{\mathrm{hom}}$ and $\mathfrak a_r^{\mathrm{het}}$
from Eqs.~\eqref{eq:equal-p-direct-rdm-exponent}
and~\eqref{eq:explicit-delta-free-rdm-exponent}. We use
\begin{align}
    \mu&\leq c_p\frac{\varepsilon_{\mathrm{fid}}} {m^{\mathfrak a_p^{\mathrm{hom}}}}, \qquad \nu\leq c_r\frac{\varepsilon_{\mathrm{fid}}} {m^{\mathfrak a_r^{\mathrm{het}}}},
    \label{eq:iu-sufficient-rdm-accuracy}
\end{align}
in the homogeneous and heterogeneous settings, respectively.
The constants may be decreased to satisfy the preparation bounds;
the exponents remain unchanged. We condition on accurate RDMs and
successful Gram-splitting calls. For this verification, the uniform
bound in Eq.~\eqref{eq:iu-postprocessing-error} suffices, including
when all blocks have two particles.

In the homogeneous setting, Prop.~\ref{prop:equal-p-input-accuracy}
verifies the reconstruction conditions and the core dimension.
Its block and core bounds in
Eq.~\eqref{eq:homogeneous-block-core-error-contract} imply
\begin{align}
    \eta&\leq\mathsf A_pm^5\frac{\mu}{s_0},\qquad \eta_0\leq\mathsf A_pm^{7/2}\frac{\mu}{s_0}.
    \label{eq:iu-homogeneous-supplied-errors}
\end{align}
Thus both the preparation conditions and its error are controlled by
\begin{align}
    m^{5/2}\eta+\eta_0&\leq\mathsf C_pm^{15/2}\frac{\mu}{s_0} \leq\mathsf C_pc_p\sqrt{\varepsilon_{\mathrm{fid}}}\, m^{8-\mathfrak a_p^{\mathrm{hom}}}.
    \label{eq:iu-homogeneous-postprocessing-error}
\end{align}
Since $\mathfrak a_p^{\mathrm{hom}}=\lceil(p+13)/2\rceil\geq8$
for every $p\geq2$, decreasing $c_p$ makes this quantity small
enough for the construction and gives
$\sqrt{1-|\langle\widetilde\Psi,\Psi_{\mathrm{tr}}\rangle|^2}
\leq\sqrt{\varepsilon_{\mathrm{fid}}}/4$
on its success event.

In the heterogeneous setting, the reconstruction conditions and
core dimension follow from
Prop.~\ref{prop:delta-free-explicit-accuracy}.
Put $a:=(7\lfloor r/2\rfloor+3)/2$.
Eqs.~\eqref{eq:df-recursive-power-bound}
and~\eqref{eq:df-core-power-bound} yield
\begin{align}
    \eta&\leq C_rm^a\frac{\nu}{s_0},\qquad \eta_0\leq C_rm^{a+3/2}\frac{\nu}{s_0}.
    \label{eq:iu-heterogeneous-supplied-errors}
\end{align}
Consequently,
\begin{align}
    m^{5/2}\eta+\eta_0&\leq C_rm^{a+5/2}\frac{\nu}{s_0} \leq C_rc_r\sqrt{\varepsilon_{\mathrm{fid}}}\, m^{a+3-\mathfrak a_r^{\mathrm{het}}}.
    \label{eq:iu-heterogeneous-postprocessing-error}
\end{align}
For $r\geq3$,
$\mathfrak a_r^{\mathrm{het}}=\lceil a+(r+3)/2\rceil\geq a+3$;
for $r=2$, $a=5$ and $\mathfrak a_2^{\mathrm{het}}=8=a+3$.
The exponent is therefore nonpositive in every case.
Decreasing $c_r$ again verifies the preparation conditions and
gives $\sqrt{1-|\langle\widetilde\Psi,\Psi_{\mathrm{tr}}\rangle|^2}
\leq\sqrt{\varepsilon_{\mathrm{fid}}}/4$
on its success event.

\paragraph*{Success probability and target fidelity.}
Allocate failure probability $\delta/3$ to simultaneous RDM
accuracy, $\delta/3$ to Gram splitting, and set $\beta=\delta/3$
for the local branch searches above.
In the homogeneous setting, each of the two RDM estimates and
each of the two splitting calls receives budget $\delta/6$.
In the heterogeneous setting, use $\delta/(3r)$ per RDM order
and $\delta/[6(r-1)]$ per splitting call.
Each local branch search thus receives
$\delta/(3\max\{1,d\})$.

Conditional on RDM accuracy and successful preceding calls,
the reconstruction and preparation inputs satisfy the required
bounds. Summing the conditional probabilities of the first failed
call, followed by the local-search failure bound, induces total
failure probability at most $\delta$.
This does not require independence between RDM estimates or
between reconstruction levels.
On the joint success event, the preparation error and the truncation bounds in
Eqs.~\eqref{eq:equal-p-df-truncation-bound}
and~\eqref{eq:ideal-truncation-bound} each give trace distance at most $\sqrt{\varepsilon_{\mathrm{fid}}}/4$ from $\Psi_{\mathrm{tr}}$. The triangle inequality then yields
\begin{align}
    \sqrt{1-|\langle\widetilde\Psi,\Psi\rangle|^2}&\leq\frac{\sqrt{\varepsilon_{\mathrm{fid}}}}{2},\qquad |\langle\widetilde\Psi,\Psi\rangle|^2\geq1-\varepsilon_{\mathrm{fid}}.
    \label{eq:iu-learning-fidelity}
\end{align}
Thus the returned pair $(\Xi_{\mathrm{in}},U_{\mathrm{out}})$ specifies the final estimate $\widetilde\Psi$ with the required fidelity.

\paragraph*{Classical post-processing time.}
With $\beta=\delta/3$, Eq.~\eqref{eq:iu-postprocessing-cost} yields $O_r(m^{r+3}\log(2m/\delta))$ arithmetic operations, with $r$ replaced by $p$ for homogeneous blocks.
Under Eq.~\eqref{eq:iu-sufficient-rdm-accuracy}, this cost is
dominated by the RDM-estimation costs in
Eqs.~\eqref{eq:equal-rdm-estimation-runtime} and~\eqref{eq:bounded-rdm-estimation-runtime}, as in Secs.~\ref{sec:homogeneous-time-complexity} and~\ref{sec:heterogeneous-time-complexity}. The RDM orders and accuracy exponents are unchanged, and the
failure allocation changes only constants inside logarithms. Hence the sample and classical time bounds of Thms.~\ref{thm:equal-blocks} and~\ref{thm:heterogeneous-blocks} remain valid. The construction uses only the recovered classical data and requires no additional target copies.

\end{document}